\documentclass{article}

\usepackage{amsmath,amsfonts}
\usepackage[hmargin=1.2in,vmargin=1.2in]{geometry}
\usepackage{url}
\usepackage{enumitem}

\newcommand\blfootnote[1]{%
  \begingroup
  \renewcommand\thefootnote{}\footnote{\hspace*{-2mm}#1}%
  \addtocounter{footnote}{-1}%
  \endgroup
}

\newtheorem{definition}{Definition}[section]
\newtheorem{lemma}{Lemma}[section]
\newtheorem{theorem}{Theorem}[section]
\newtheorem{corollary}{Corollary}[section]

\def\Box{\hbox
   {\vbox
     {\hrule height .033em
      \hbox to .5em
        {\vrule width .033em
         \hbox{\vbox to .5em{\vss}\hss}
         \vrule width .033em}
      \hrule height .033em
     }}}

 \newif\ifqed\qedfalse 
 \def\endofpfmarker{\mbox{$\Box$}}
 \def\endofpf{\hbox{\ }\hfill\endofpfmarker\qedtrue}
 \def\proofstyle{}
 \def \qedhere {\endofpf}

 \newenvironment{pfenv}[1]{%
   \global\def\qedhere{\endofpf\global\def\qedhere{}}
   \proofstyle
   \trivlist 
   \item[\hskip \labelsep{\bf #1}]%
   }{%
   \qedhere
   \endtrivlist}

\newenvironment{proof}     {\begin{pfenv}{Proof:     }}{\end{pfenv}}

\def\vect#1{\widetilde{#1}}

\def\eqns#1{\begin{eqnarray*}#1\end{eqnarray*}}

\newcommand{\rulename}[1]{{\sc #1}}
\newcommand{\brn}[1]{\mbox{\rulename{#1}}}

\makeatletter
\newenvironment{restate}[2]{%
    \trivlist
    \item[%
        \hskip \labelsep
        {\bf #1\hskip 5\p@\relax{#2}\enspace}]%
        \itshape\hskip-\labelsep}%
 {}
\newenvironment{restatenamed}[3]{%
    \trivlist
    \item[%
        \hskip \labelsep
        {{\bf{#1\ #2\unskip\hskip5pt\relax(#3)%
	   \unskip\hskip5pt}}}]
        \itshape\hskip-\labelsep}%
 {}
\makeatother

\newcommand{\bicong}{\approx }

\newcommand{\eqstr}{\equiv }

\newcommand{\rel}{\mathrel{\mathcal{R}}}
\newcommand{\wkbisim}{\mathrel{\approx_l}}
\newcommand{\altbisim}{\mathrel{\approx_L}}

\newcommand{\eqdef}{\mathrel{\stackrel{\rm\mbox{\tiny def}}{=}}}
\newcommand{\enveq}{\mathrel{\bicong_{s}}}
\newcommand{\barb}[1]{\Downarrow{\!#1}}
\newcommand{\hole}{\_}
\newcommand{\Const}[1]{\mathsf{#1}}

\newcommand{\youwin}{\Const{\,oops\,!\,}}

\newcommand{\hash}{\Const{h}}

\newcommand{\mac}[2]{\Const{mac}(#1,#2)}
\newcommand{\hbin}[2]{\Const{h}(#1,#2)}  
\newcommand{\fbin}[2]{\Const{f}(#1,#2)}  

\def\subst#1#2{^{#1} \!/\! _{#2}}
\newcommand{\someframe}{\varphi}
\newcommand{\frameof}[1]{\varphi(#1)}
\newcommand{\sx}[1]{\subst{#1}{x}}

\newcommand{\smx}{\{\sx{M}\}}
\newcommand{\smxvect}{\{\subst{\vect M}{\vect x}\}}
\def\ltr#1{\xrightarrow{#1}}
\newcommand{\infrule}[2]{\cfrac{\mbox{$#1$}}{\mbox{$#2$}}}
\newcommand{\enf}[1]{#1^\circ}
\newcommand{\doubleopensqbr}{{\rm [\hspace{-1.67pt}[}}
\newcommand{\doubleclosesqbr}{{\rm ]\hspace{-1.67pt}]}}
\newcommand{\tr}[1]{\doubleopensqbr #1 \doubleclosesqbr}

\newenvironment{defn}{\begin{tabbing}
  \hspace{1.5em} \= \hspace{.35\linewidth - 1.5em} \= \hspace{1.5em} \= \kill
  }{
  \end{tabbing}}

\newcommand{\entry}[2]{\>$#1$\>\>#2}
\newcommand{\clause}[2]{$#1$\>\>#2}
\newcommand{\syntaxcategory}[2]{\clause{#1::=}{#2}}

\newcommand{\Msg}[1]{\langle#1\rangle}
\newcommand{\Abs}[1]{(#1)}
\newcommand{\kw}[1]{\mathit{#1}}

\newcommand{\fn}{\kw{fn}}            
\newcommand{\fv}{\kw{fv}}            
\newcommand{\bv}{\kw{bv}}
\newcommand{\dom}{\kw{dom}}

\newcommand{\Rcv}[2]{#1\Abs{#2}}
\newcommand{\Snd}[2]{\overline{#1}\Msg{#2}}
\newcommand{\Let}[2]{\kw{let}\ #1 = #2\ \kw{in}\ }
\newcommand{\Res}[1]{\nu #1.}
\newcommand{\parop}{\mathbin{\mid}}
\newcommand{\Repl}[1]{\mathord{! #1}}
\newcommand{\nil}{\mathbf{0}}
\newcommand{\IfThen}[3]{\kw{if}\ #1 = #2\ \kw{then}\ #3}
\newcommand{\IfThenElse}[4]{\kw{if}\ #1 = #2\ \kw{then}\ #3\ \kw{else}\ #4}
\newcommand{\CTX}{E} 

\newcommand{\sdecrypt}[3]

\newcommand{\hassort}[2]{{#1}: {#2}}
\newcommand{\judgM}[2]{\vdash {#1}: {#2}}
\newcommand{\judgP}[1]{\vdash {#1}}

\newcommand{\rew}{R}
\newcommand{\substinj}{{\sc SubstInj}}
\newcommand{\size}{\mathit{size}}

\newcommand{\nuc}[1]{\nu\vect u.(#1 \parop C)} 
\newcommand{\pnf}{\mathrm{pnf}}
\newcommand{\equivpnf}{\mathrel{\smash[t]{\stackrel{\circ}{\equiv}}}}
\newcommand{\equivP}{\mathrel{\smash[t]{\stackrel{\diamond}{\equiv}}}}
\newcommand{\redpnf}{\mathrel{\rightarrow_\circ}}
\newcommand{\redP}{\mathrel{\rightarrow_\diamond}}
\newcommand{\ltrpnf}[1]{\mathrel{\ltr{#1}_{\circ}}}
\newcommand{\ltrP}[1]{\mathrel{\ltr{#1}_\diamond}}

\newcommand{\comp}[2]{#1\mathbin{\uplus}#2}
\newcommand{\prop}{\mathit{Prop}}

\newcommand{\Key}{\Const{Key}}
\newcommand{\Bool}{\Const{Bool}}
\newcommand{\Data}{\Const{Data}}
\newcommand{\DataList}{\Const{List}}
\newcommand{\Block}{\Const{Block}}
\newcommand{\BlockList}{\Const{BlockList}}
\newcommand{\BlockPair}{\Const{Block2}}
\newcommand{\BlockT}{\Const{Block3}}
\newcommand{\BlockBlockList}{\Const{Block2Blocks}}
\newcommand{\BlockBlockListList}{\Const{Block2Blocks\_List}}
\newcommand{\neseq}{\Const{ne\_list}}
\newcommand{\cons}[2]{{#1}::{#2}}
\newcommand{\bappend}{\mathbin{+\hspace*{-1.5mm}+}}
\newcommand{\len}{\mathit{length}}
\newcommand{\hval}{\mathit{hval}}

\newcommand{\fset}{I}
\newcommand{\fsetRel}{I'_\textit{h}}
\newcommand{\fsetNRel}{I'_\textit{alt}}
\newcommand{\fsetOut}{I_\textit{done}}
\newcommand{\fsetBOut}{I_\textit{out}}
\newcommand{\hset}{J}
\newcommand{\hsetOut}{J_{\textit{done}}}
\newcommand{\hsetBOut}{J_{\textit{out}}}
\newcommand{\hsetInt}{J_{\textit{test}}}
\newcommand{\hsetDrop}{J_{\textit{fail}}}

\newcommand{\env}{\mathcal{N}}

\title{The Applied Pi Calculus: Mobile Values,  \linebreak[2] New Names, and Secure Communication}

\author{Mart{\'\i}n Abadi
\\
Google Brain\\
\url{abadi@google.com}
\and
Bruno Blanchet\\
Inria\\
\url{Bruno.Blanchet@inria.fr}
\and
C\'edric Fournet\\
Microsoft Research\\
\url{fournet@microsoft.com}}

\begin{document}

\maketitle

\blfootnote{This work was started while Mart{\'\i}n Abadi
  was at Bell Labs Research, and continued while he was at the University of California at Santa Cruz and at Microsoft Research.

Authors' addresses: Mart{\'\i}n Abadi, Google Brain, Mountain View, CA, USA;
Bruno Blanchet, Inria, Paris, France;
C\'edric Fournet, Microsoft Research, Cambridge, UK.}

\begin{abstract}
We study the interaction of the programming construct ``new'', which
generates statically scoped names, with communication via messages on channels.
This interaction is crucial in security protocols, which are the main
motivating examples for our work; it also appears in other
programming-language contexts.

We define the applied pi calculus, a simple,
general extension of the pi calculus in which values can be formed
from names via the application of built-in functions, subject to
equations, and be sent as messages.  (In contrast, the pure pi
calculus lacks built-in functions; its only messages are atomic names.)
We develop semantics and proof techniques for this extended language
and apply them in reasoning about security protocols.

This paper essentially subsumes the conference paper that introduced
the applied pi calculus in 2001. 
It fills gaps, incorporates improvements, and further explains and
studies the applied pi calculus.  Since 2001, the applied pi calculus
has been the basis for much further work, described in many research
publications and sometimes embodied in useful software, such as the
tool ProVerif, which relies on the applied pi calculus to support the
specification and automatic analysis of security protocols. Although
this paper does not aim to be a complete review of the subject, it
benefits from that further work and provides better foundations for
some of it. In particular, the applied pi calculus has evolved through
its implementation in ProVerif, and the present definition reflects
that evolution.
\end{abstract}

\section{A Case for Impurity}\label{sec:intro}

Purity often comes before convenience and even before faithfulness in
the lambda calculus, the pi calculus, and other foundational
programming languages.  For example, in the standard pi calculus, the
only messages are atomic names~\cite{milner:communicating-mobile}.
This simplicity is extremely appealing from a foundational viewpoint,
and helps in developing the theory of the pi calculus. Furthermore,
ingenious encodings demonstrate that it may not entail a loss of
generality. In particular, integers, objects, and even higher-order
processes can be represented in the pure pi calculus. Similarly,
various encodings of cryptographic operations in the pi calculus have
been
considered~\cite{spi2four,carbonemaffeis,baldamustranslation,martinhoravara}.

On the other hand, this purity has a price.  In applications, the
encodings can be futile, cumbersome, and even misleading.  For
instance, in the study of programming languages based on the pi
calculus (such as Pict~\cite{pierce.turner:pict-programming},
JoCaml~\cite{JoCAML}, or occam-pi~\cite{Occampi}), there is little point in pretending that
integers are not primitive. The encodings may also hinder careful
reasoning about communication (for example, because they require extra messages), 
and they may complicate static analysis and proofs.

These difficulties are often circumvented through on-the-fly
extensions. The extensions range from quick punts (``for the next
example, let's pretend that we have a datatype of integers'') to the
laborious development of new calculi, such as the spi calculus~\cite{spi2four} (a calculus with cryptographic operations) and its
variants.  Generally, the extensions bring us closer to a realistic
programming language or modeling language---that is not always a bad
thing.

Although many of the resulting calculi are ad hoc and poorly
understood, others are robust and uniform enough to have a rich theory
and a variety of applications.  In particular, impure extensions of
the lambda calculus with function symbols and with equations among
terms (``delta rules'') have been developed systematically, with
considerable success. Similarly, impure versions of CCS and CSP with
value-passing are not always deep but often neat and convenient~\cite{MilCC89}.

In this paper, we introduce, study, and use an analogous uniform
extension of the pi calculus, which we call the applied pi calculus
(by analogy with ``applied lambda calculus'').  
From the pure pi calculus, we inherit
constructs for communication and concurrency, and for generating
statically scoped new names (``new'').  We add functions and
equations, much as is done in the lambda calculus.  Messages may then
consist not only of atomic names but also of values constructed from
names and functions. This embedding of names into the space of values
gives rise to an important interaction between the ``new'' construct
and value-passing communication, which appears in neither the pure pi
calculus nor value-passing CCS and CSP. Further, we add an auxiliary
substitution construct, roughly similar to a floating ``let''; this
construct is helpful in programming examples and especially in
semantics and proofs, and serves to capture the partial knowledge that
an environment may have of some values.

The applied pi calculus builds on the pure pi calculus and its
substantial theory, but it shifts the focus away from encodings. In
comparison with ad hoc approaches, it permits a general, systematic
development of syntax, operational semantics, equivalences, and proof
techniques.

Using the calculus, we can write and reason about 
programming examples where ``new'' and value-passing appear.
First, we can easily treat standard datatypes
(integers, pairs, arrays, etc.).  We can also model unforgeable capabilities
as new names, then model the application of certain functions to those
capabilities.  For instance, we may construct a pair of capabilities.
More delicately, the capabilities may be pointers
to composite structures, and then adding an offset to a pointer to a
pair may yield a pointer to its second component (e.g., as
in~\cite{liblit}).  Furthermore, we can study a variety of security
protocols. For this purpose, we represent
fresh channels, nonces, and keys as new names, and primitive
cryptographic operations as functions, obtaining a simple
but useful programming-language perspective on security protocols
(much as in the spi calculus).
A distinguishing characteristic of the present approach is that we need not
craft a special calculus and develop its proof techniques for
each choice of cryptographic operations.
Thus, we can express and analyze fairly sophisticated protocols that combine
several cryptographic primitives (encryptions, hashes,
signatures, XORs,~\dots). We can also describe attacks against the
protocols that rely on (equational) properties of some of those
primitives.
In our work to date, security protocols are our main
source of examples.

The next section defines the applied pi calculus. Section~\ref{sec:sample}
introduces some small, informal examples.
Section~\ref{sec:equivalences} defines semantic concepts, such as
process equivalence, and develops proof techniques.
Sections \ref{sec:diffie-hellman} and~\ref{sec:extension-attacks}
treat larger, instructive examples; they concern 
a Diffie-Hellman key exchange, cryptographic hash functions, and message authentication codes.
(The two sections are independent.)
Many other examples now appear in the literature, as explained below.
Section~\ref{sec:related} discusses related work, and 
Section~\ref{sec:conclusion} concludes.
The body of the paper contains some proofs and outlines others; many details of the proofs, however, are in appendices.

This paper essentially subsumes the conference paper that introduced
the applied pi calculus in 2001.  It fills gaps, incorporates various
improvements, and further explains and studies the applied pi
calculus. Specifically, it presents a revised language, with a revised
semantics, as explained in Sections~\ref{sec:calculus}
and~\ref{sec:equivalences}.  It also includes precise definitions and
proofs; these address gaps in the conference paper, discussed in
further detail in Section~\ref{sec:equivalences}.  Finally, some of the
examples in Sections~\ref{sec:sample}, \ref{sec:diffie-hellman}, and
especially~\ref{sec:extension-attacks} are polished or entirely new.

Since 2001, the applied pi calculus has been the basis for much
further work, described in many research publications (some of which
are cited below) and
tutorials~\cite{abadi:fosad,CortierKremer14,2011-Applied-pi-calculus}.  This further
work includes semantics, proof techniques, and applications in diverse
contexts (key exchange, electronic voting, certified email,
cryptographic file systems, encrypted Web storage, website
authorization, zero-knowledge proofs, and more).  It is sometimes
embodied in useful software, such as the tool ProVerif~\cite{Blanchet2001,Blanchet04,Blanchet07b,ProVerifSurvey}.  This tool,
which supports the specification and automatic analysis of security
protocols, relies on the applied pi calculus as
input language. 
Other software that builds on ProVerif
targets protocol implementations, Web-security mechanisms, or stateful
systems such as hardware
devices~\cite{Bhargavan:dec08,webspi,statverif}. 
Finally, the applied pi calculus has also been implemented in other settings, such as the prover Tamarin~\cite{Meier:tamarin,Kremer14}.

Although this paper does not aim to offer a complete review of the
subject and its growth since 2001, it benefits from that further work
and provides better foundations for some of it. In particular, the
applied pi calculus has evolved through its implementation in
ProVerif, and the present definition reflects that evolution.

\section{The Applied Pi Calculus}\label{sec:calculus}

In this section we define the applied pi calculus: its syntax and
informal semantics (Section~\ref{sec:syntax}), then its operational
semantics (Section~\ref{sec:ope-sem}). We also discuss a few variants and extensions of our definitions (Section~\ref{sect:variants}).

\subsection{Syntax and Informal Semantics}\label{sec:syntax}

A {\em signature} $\Sigma$ consists of a finite set of function symbols,
such as $\Const{f}$, $\Const {encrypt}$, and $\Const {pair}$, each with
an arity.
A function symbol with arity 0 is a constant symbol.

Given a signature $\Sigma$, an infinite set of names, and an infinite
set of variables, the set of \emph{terms} is defined by the grammar:
\begin{defn}
\syntaxcategory{L, M, N, T, U, V}             {terms}\\
\entry{a, b, c,\dots,k,\dots, m,n, \dots,s} {name}\\
\entry{x,y,z}                             {variable}\\
\entry{f(M_1,\dots,M_l)}            {function application}
\end{defn}
where $f$ ranges over the functions of $\Sigma$ and $l$ matches the
arity of $f$. 

Although names, variables, and constant symbols have similarities,
we find it clearer to keep them separate.
A term is ground when it does not contain 
variables (but it may contain names and constant symbols).
We use meta-variables $u,v,w$ to range over both names
and variables. 
We abbreviate tuples $u_1,\dots,u_l$ and $M_1,\dots,M_l$ to
$\vect{u}$ and~$\vect{M}$, respectively.

The grammar for \emph{processes} is similar to the one in the pi
calculus, but here messages can contain terms 
(rather than only names) and names need not be just channel names:

\begin{defn}
\syntaxcategory{P, Q, R}            {processes (or plain processes)}\\
\entry{\nil}                            {null process}\\
\entry{P \parop Q}                      {parallel composition}\\ 
\entry{\Repl P}                         {replication}\\
\entry{\Res n P}                        {name restriction (``new'')}\\
\entry{\IfThenElse{M}{N}{P}{Q}}         {conditional}\\
\entry{\Rcv N x.P}                      {message input}\\
\entry{\Snd N M.P}                      {message output}
\end{defn}
The null process $\nil$ does nothing; $P \parop Q$ is the
parallel composition of $P$ and~$Q$; the 
replication $\Repl P$ behaves as an infinite number of copies of $P$
running in parallel.
The process $\Res n P$ makes a new, private name $n$ then behaves as $P$.
The conditional construct $\IfThenElse{M}{N}{P}{Q}$ is standard,
but we should stress that $M=N$ represents equality, rather than
strict syntactic identity. 
We abbreviate it $\IfThen{M}{N}{P}$ when $Q$ is $\nil$.
Finally, $\Rcv N x.P$ is ready to input from channel $N$, then
to run $P$ with the actual message replaced for the formal parameter $x$,
while $\Snd N M.P$ is ready to output $M$ on channel $N$, then to run~$P$.
In both of these, we may omit $P$ when it is $\nil$.

Further, we extend processes with {\em active substitutions}:
\begin{defn}
\syntaxcategory{A, B, C}                     {extended processes}\\
\entry{P}                               {plain process}\\
\entry{A \parop B}                      {parallel composition}\\
\entry{\Res n A}                        {name restriction}\\
\entry{\Res x A}                        {variable restriction}\\
\entry{\smx}                            {active substitution}
\end{defn}
We write $\smx$ for the substitution that replaces the variable $x$
with the term~$M$.  Considered as a process, $\smx$ is like
$\Let{x}{M}\ldots$, and is similarly useful.
However, unlike a ``let'' definition, $\smx$ floats and applies to any 
process that comes into contact with it.
To control this contact, we may add a restriction: 
$\Res x (\smx \parop P)$ corresponds exactly to $\Let{x}{M}P$.
The substitution $\smx$ typically appears when the term $M$ has been sent
to the environment, but the environment may not have
the atomic names that appear in $M$; the variable~$x$ is just
a way to refer to $M$ in this situation.
Although the substitution $\smx$ concerns only one variable,
we can build bigger substitutions by parallel composition,
and may write $$\{\subst{M_1}{x_1}, \dots, \subst{M_l}{x_l}\} \quad {\mbox{for}}
\quad \{\subst{M_1}{x_1}\} \parop
\dots \parop \{\subst{M_l}{x_l}\}$$ 
We write $\sigma$, $\smx$, $\{\subst{\vect{M}}{\vect{x}}\}$ for substitutions,
$x\sigma$ for the image of $x$ by $\sigma$, and $T\sigma$ for the
result of applying $\sigma$ to the free variables of~$T$.
We identify the empty substitution and the null process $\nil$.

\begin{figure}[tp]
\eqns{
\fv(x) & \eqdef & \{ x\}\\
\fv(n) & \eqdef & \emptyset\\
\fv(f(M_1, \ldots, M_l)) & \eqdef & \fv(M_1) \cup \dots \cup \fv(M_l)\\[1ex]
\fv(\nil) & \eqdef & \emptyset\\
\fv(P \parop Q) & \eqdef & \fv(P) \cup \fv(Q)\\
\fv(\Repl P) & \eqdef & \fv(P)\\
\fv(\Res n P) & \eqdef & \fv(P)\\
\fv(\IfThenElse{M}{N}{P}{Q}) & \eqdef & \fv(M) \cup \fv(N) \cup \fv(P) \cup \fv(Q)\\
\fv(\Rcv N x.P) & \eqdef & \fv(N) \cup \left(\fv(P) \setminus \{ x \}\right)\\
\fv(\Snd N M.P) & \eqdef & \fv(N) \cup \fv(M) \cup \fv(P)\\[1ex]
\fv(A \parop B) & \eqdef & \fv(A) \cup \fv(B)\\
\fv(\Res n A) & \eqdef & \fv(A)\\
\fv(\Res x A) & \eqdef & \fv(A) \setminus \{ x \}\\
\fv(\smx) & \eqdef & \fv(M) \cup \{x\} 
}
$\fn(\cdot)$ is defined as $\fv(\cdot)$, except that
\eqns{
\fn(x) & \eqdef & \emptyset\\
\fn(n) & \eqdef & \{ n \}\\[1ex]
\fn(\Res n P) & \eqdef & \fn(P) \setminus \{ n \}\\
\fn(\Rcv N x.P) & \eqdef & \fn(N) \cup \fn(P)\\[1ex]
\fn(\Res n A) & \eqdef & \fn(A) \setminus \{ n \}\\
\fn(\Res x A) & \eqdef & \fn(A)\\
\fn(\smx) & \eqdef & \fn(M)\\[2ex]
\dom(P) & \eqdef & \emptyset\\
\dom(A \parop B) & \eqdef & \dom(A) \cup \dom(B)\\
\dom(\Res n A) & \eqdef & \dom(A)\\
\dom(\Res x A) & \eqdef & \dom(A) \setminus \{ x \}\\
\dom(\smx) & \eqdef & \{ x \}
}
\caption{Free variables, free names, and domain}\label{fig:fv-fn}
\end{figure}

As usual, names and variables have scopes, which are delimited by
restrictions and by inputs.
We write $\fv(A)$ and $\fn(A)$ for the sets of
free variables and free names of~$A$,
respectively.
These sets are inductively defined, as detailed in Figure~\ref{fig:fv-fn}.
The domain $\dom(A)$ of an extended process $A$ is the set of 
variables that $A$ exports 
(those variables~$x$ for which $A$ contains a substitution $\smx$ 
not under a restriction on $x$). Figure~\ref{fig:fv-fn} also defines
$\dom(A)$ formally.
We consider that expressions (processes and extended processes) are equal modulo renaming of bound names 
and variables.

We always assume that our substitutions are cycle-free, that is, 
by reordering,
they can be written 
$\{\subst{M_1}{x_1}, \dots, \subst{M_l}{x_l}\}$ where
$x_i \notin \fv(M_j)$ for all $i \leq j \leq l$.
For instance, we exclude substitutions such as $\{ \subst{f(y)}{x}, \subst{f(x)}{y} \}$.
We also assume that, in an extended process, there is at most one 
substitution for each variable, and there is 
exactly one when the variable is restricted, that is,
$\dom(A) \cap \dom(B) = \emptyset$ in every extended process $A \parop B$,
and $x \in \dom(A)$ in every extended process $\Res x A$. 
An extended process $A$ is {\em closed} when 
its free variables are all defined by an active substitution,
that is, $\dom(A) = \fv(A)$.
We use the abbreviation $\nu \vect{u}$ for the (possibly empty) series of
pairwise distinct restrictions $\nu u_1.\nu u_2.\dots \nu u_l$.

A {\em frame} is an extended process built up from $\nil$ and 
active substitutions of the form $\smx$ by parallel composition and
restriction.
We let $\varphi$ and $\psi$ range over frames.
Every extended
process $A$ can be mapped to a frame $\frameof{A}$ by replacing every
plain process embedded in $A$ with~$\nil$.
The frame $\frameof{A}$ can be viewed as an
approximation of $A$ that accounts for the static knowledge exposed by
$A$ to its environment, but not for $A$'s dynamic behavior.  
Assuming that all bound names and variables are pairwise distinct, and do not clash with free ones,
one can ignore all restrictions in a frame, thus obtaining an underlying
substitution; we require that, for each extended process, this resulting 
substitution be cycle-free.

\begin{figure}
\begin{gather*}
\frac{\hassort{u}{\tau}}{\judgM{u}{\tau}} \qquad 
\frac{\hassort{f}{\tau_1 \times \dots \times \tau_l \rightarrow \tau}
\quad \judgM{M_1}{\tau_1} \quad \dots \quad \judgM{M_l}{\tau_l}
}{\judgM{f(M_1, \ldots, M_l)}{\tau}}\\[1.5mm]
\judgP{0}\qquad
\frac{\judgP{P} \quad \judgP{Q}}{\judgP{P \parop Q}}\qquad
\frac{\judgP{P}}{\judgP{\Repl P}}\qquad
\frac{\judgP{P}}{\judgP{\Res n P}} \qquad 
\frac{\judgM{M}{\tau}\quad \judgM{N}{\tau} \quad \judgP{P} \quad \judgP{Q}}{\judgP{\IfThenElse{M}{N}{P}{Q}}}\\[1.5mm]
\frac{\judgM{N}{\Const{Channel}}\quad \judgP{P}}{\judgP{\Rcv N x.P}} \qquad 
\frac{\judgM{N}{\Const{Channel}}\quad \judgM{M}{\tau} \quad \judgP{P}}{\judgP{\Snd N M.P}}\\[1.5mm]
\frac{\judgP{A} \quad \judgP{B}}{\judgP{A \parop B}}\qquad
\frac{\judgP{A}}{\judgP{\Res u A}}\qquad
\frac{\hassort{x}{\tau} \quad \judgM{M}{\tau}}{\judgP{\smx}}
\end{gather*}
\caption{Sort system}\label{fig:sortsystem}
\end{figure}

We rely on a sort system for terms and processes. It includes a sort
$\Const{Channel}$ for channels. It may also include other sorts 
such as $\Const{Integer}$, $\Const {Key}$, or simply a universal 
sort for data $\Const {Data}$.  
Each variable and each name comes with a sort;
we write $\hassort{u}{\tau}$ to mean that $u$ has sort $\tau$.
There are an infinite number of variables and an infinite number
of names of each sort.
We typically use $a$, $b$, and $c$ as names of sort $\Const{Channel}$, 
$s$ and $k$ as names of some other sort (e.g., $\Const {Data}$), and 
$m$ and $n$ as names of any sort.  
Function symbols also come with the sorts of their arguments and 
of their result. We write
$\hassort{f}{\tau_1 \times \dots \times \tau_l \rightarrow \tau}$
to mean that $f$ has arguments of sorts $\tau_1, \ldots, \tau_l$
and a result of sort $\tau$.
Figure~\ref{fig:sortsystem} gives the rules of the sort system.
It defines the following judgments: $\judgM{M}{\tau}$ means that $M$ is a
term of sort $\tau$; $\judgP{P}$ means that the process $P$ is well-sorted; 
$\judgP{A}$ means that the extended process $A$ is well-sorted.
This sort system enforces 
that function applications are well-sorted, 
that $M$ and $N$ are of the same sort in conditional expressions, 
that $N$ has sort $\Const{Channel}$ in input and output expressions,
that $M$ is well-sorted (with an arbitrary sort $\tau$) in output expressions, and 
that active substitutions preserve sorts. 
We always assume 
that expressions are well-sorted, and that
substitutions preserve sorts.

\subsection{Operational Semantics}\label{sec:ope-sem}

We give an operational semantics for the applied pi calculus 
in the now customary ``chemical style''~\cite{berry92,milner92a}.
At the center of this operational semantics is a reduction
relation $\rightarrow$ on extended processes, which basically models the steps
of computations. For example, ${\Snd a M} \parop {\Rcv a x.\Snd b
  x} \mathrel{\rightarrow} \Snd b M$ represents the transmission of the message
$M$ on the channel $a$ to a process that will forward the message on
the channel $b$; the formal $x$ is replaced with its actual value $M$
in this reduction. The axioms for the reduction relation
$\rightarrow$, which are remarkably simple, rely on auxiliary rules
for a structural equivalence relation $\equiv$ that permits the
rearrangement of processes, for example the use of commutativity and
associativity of parallel composition. Furthermore, both structural
equivalence and reduction depend on an underlying equational
theory. Therefore, this section introduces equational theories, then
defines structural equivalence and reduction.

Given a signature $\Sigma$, we equip it with an equational theory,
that is, with a congruence relation on terms that is closed under
substitution of terms for variables and names.
(See for example Mitchell's textbook~\cite[chapter 3]{mitchell:book} and its
references for background on universal algebra and algebraic data
types from a programming-language perspective.)  
We further require that this equational theory respect
the sort system, that is, two equal terms are of the same sort,
and that it be non-trivial, that is, there exist two different 
terms in each sort.

An equational theory may be generated from a finite set of equational
axioms, or from rewrite rules, but this property is not essential for
us.  We tend to ignore the mechanics of specifying equational
theories, but give several examples in Section~\ref{sec:sample}.

We write $\Sigma \vdash M = N$ when the equation $M=N$ is in the
theory associated with $\Sigma$. Here we keep the theory implicit,
and we may even abbreviate $\Sigma \vdash M = N$ to $M = N$ when
$\Sigma$ is clear from context or unimportant.
We write $\Sigma \not\vdash M = N$ for the negation of $\Sigma \vdash
M = N$.

As usual, a context is an expression 
with a hole.  An {\em evaluation context} is a context whose hole is
not under a replication, a conditional, an input, or an output. A
context $\CTX[\hole]$ \emph{closes}~$A$ when $\CTX[A]$ is closed.

\emph{Structural equivalence} $\equiv$ is the smallest equivalence
relation on extended processes
that is closed 
by application of evaluation contexts, and such that:
\[\begin{array}{lrcll}
\brn{Par-\mbox{$\nil$}} & 
A  & \equiv & A \parop \nil \\
\brn{Par-A} & 
A \parop (B \parop C)  & \equiv & (A \parop B) \parop C \\
\brn{Par-C} & 
A  \parop B &  \equiv & B \parop A \\
\brn{Repl} & 
 \Repl P  & \equiv & P \parop  \Repl P\\[.5em]
\brn{New-\mbox{$\nil$}} & 
\nu n.\nil  & \equiv & \nil \\
\brn{New-C} & 
\nu u.\nu v.A  & \equiv & \nu v.\nu u.A \\
\brn{New-Par} & 
A \parop \nu u.B  & \equiv & \nu u.(A \parop B) 
\quad 
\mbox{when } u \not\in \fv(A) \cup \fn(A)\\[.5em]
\brn{Alias} & 
\nu x. \smx & \equiv & \nil 
\\
\brn{Subst} &
\smx \parop A & \equiv & \smx \parop A\smx\\
\brn{Rewrite} & \{\sx M\} & \equiv & \{ \sx N \} 
\quad \mbox{when } \Sigma \vdash M = N
\end{array}\]

The rules for parallel composition and restriction are standard.
\rulename{Alias} enables the introduction of an arbitrary active
substitution.
\rulename{Subst} describes the application of an
  active substitution to a process that is in contact with it.
\rulename{Rewrite} deals with equational rewriting.
\rulename{Subst} implicitly requires that $x: \tau$ and $\vdash M : \tau$ for some sort $\tau$. 
In combination, \rulename{Alias} and \rulename{Subst} yield 
$A\{\subst{M}{x}\} \equiv \nu x.( \{\subst{M}{x}\} \parop A )$
for $x \notin \fv(M)$:
$$\begin{array}{rcll}
A\{\subst{M}{x}\} & \equiv &  A\{\subst{M}{x}\} \parop \nil & \mbox{by \rulename{Par-\mbox{$\nil$}}}\\
& \equiv & A\{\subst{M}{x}\} \parop \nu x.\{\subst{M}{x}\}& \mbox{by \rulename{Alias}}\\
& \equiv & \nu x.( A\{\subst{M}{x}\} \parop \{\subst{M}{x}\}) & \mbox{by \rulename{New-Par}}\\
& \equiv & \nu x.(  \{\subst{M}{x}\} \parop A\{\subst{M}{x}\}) & \mbox{by \rulename{Par-C}}\\
& \equiv & \nu x.( \{\subst{M}{x}\} \parop A ) & \mbox{by \rulename{Subst}}
\end{array}$$

Using structural equivalence, every closed extended proc\-ess $A$ can
be rewritten to consist of a substitution and a closed plain process with
some restricted names:
\eqns{ A & \equiv & 
  \nu \vect{n}. (\{\subst{\vect{M}}{\vect{x}}\} \parop P)}%
where $\fv(P) = \emptyset$, $\fv(\vect{M}) = \emptyset$, and $\{\vect{n}\} \subseteq
\fn(\vect{M})$. 
In particular, every closed frame $\varphi$ can be rewritten to
consist of a substitution with some restricted names:
\eqns{ \varphi & \equiv & 
\nu \vect n. \{\subst{\vect{M}}{\vect{x}}\}}
where $\fv(\vect{M}) = \emptyset$ and $\{\vect{n}\} \subseteq
\fn(\vect{M})$. 
The set $\{\vect{x}\}$ is the domain of $\varphi$.

\emph{Internal reduction}~$\rightarrow$ is the smallest relation on
extended processes closed by
structural equivalence and application of evaluation contexts such that:
\[\begin{array}{lrcll}
\brn{Comm} &
{\Snd N x.P} \parop {\Rcv N x.Q} & \rightarrow & P \parop Q \\[.5em]
\brn{Then} &
\IfThenElse{M}{M}{P}{Q} & \rightarrow & P \\[.5em]
\brn{Else} &
\IfThenElse{M}{N}{P}{Q} & \rightarrow & Q \\
& \multicolumn{4}{l}{\mbox{\quad for any ground terms $M$ and $N$ such that $\Sigma \not\vdash M = N$}}\\
\end{array}\]%
Communication (\rulename{Comm}) is remarkably simple because the message
concerned is a variable; this simplicity entails no loss
of generality because \rulename{Alias} and \rulename{Subst} 
can introduce a variable to stand for a term:
\eqns{%
  {\Snd N M.P} \parop {\Rcv N x.Q} & \equiv & 
  \nu x.( \{\subst{M}{x}\} \parop  {\Snd N x.P} \parop {\Rcv N x.Q})\\
  & \rightarrow & 
  \nu x.( \{\subst{M}{x}\} \parop  P \parop Q) \quad \mbox{by \rulename{Comm}}\\
  & \equiv &  P \parop Q\{\subst{M}{x}\}
  }%
(This derivation assumes that $x \notin \fv(N) \cup \fv(M) \cup \fv(P)$, which can be established
by renaming as needed.)

Comparisons (\rulename{Then} and \rulename{Else}) directly depend 
on the underlying equational theory. Using \rulename{Else} sometimes 
requires that active substitutions in the context be applied first,
to yield ground terms $M$ and $N$. For example, rule \rulename{Else}
does not allow us to reduce $\{\subst{n}{x}\} \parop \IfThenElse{x}{n}{P}{Q}$.

This use of the equational theory may be reminiscent of initial
algebras. 
In an initial algebra, 
the principle of ``no confusion'' dictates that two elements
are equal only if this is required by the corresponding equational
theory. Similarly, $\IfThenElse{M}{N}{P}{Q}$ reduces to $P$ only
if this is required by the equational theory, and reduces to $Q$ otherwise.
Initial algebras
also obey the principle of ``no junk'', which says that all elements
correspond to terms built exclusively from function symbols of the signature.
In contrast, a fresh name need not equal any such term in the applied pi calculus.

\subsection{Variants and Extensions}\label{sect:variants}

Several variants of the syntax of the applied pi calculus 
appear in the literature, and further variants may be considered. We
discuss a few:
\begin{itemize}

\item 
In the conference paper,
there are several sorts for channels:
the sort $\Const{Channel}\langle \tau\rangle$ is the sort of channels
that convey messages of sort $\tau$. The sort $\Const{Channel}$
without argument is more general, in the sense that all processes 
well-sorted with 
$\Const{Channel}\langle \tau\rangle$ are also well-sorted with 
$\Const{Channel}$. 
Having a single sort for channels simplifies some models,
for instance when all public messages are sent on the same channel,
even if they have different types.
Moreover, by using $\Const{Channel}$ as only sort,
we can encode an untyped version of the applied pi calculus.
The tool ProVerif also uses the sort $\Const{Channel}$ without argument.

\item 
In a more refined version of the sort system, we could allow names
only in a distinguished set of sorts. For instance, we could consider
a sort of booleans, containing as only values the constants $\Const{true}$
and $\Const{false}$. Such a sort would not contain names.
Sorts without names would have to be treated with special care in proofs,
since our proofs often use fresh names. 

On the other hand, letting all sorts contain names does not prevent
modeling booleans by a sort. For example, we can treat as 
false all terms of the sort different from $\Const{true}$, including not only the constant $\Const{false}$
but also all names. Analogous treatments apply to other common datatypes.

\item In the conference paper, channels in inputs and outputs are
  names or variables rather than any term. 
  Allowing any term as channel
  yields a more general calculus and avoids some side conditions in
  theorems.  It is also useful for some encodings~\cite{Abadi04f}.  
  Finally, it is in line with the syntax of ProVerif, where this design choice
  was adopted in order to simplify the untyped version of the calculus.

  Nevertheless, the sort
  system can restrict the terms that appear as channels: if no
  function symbol returns a result of sort $\Const{Channel}$, then
  channels can be only names or variables.

\item Function symbols can also be defined by rewrite rules instead of
an equational theory. This approach is taken in ProVerif~\cite{Blanchet08c}:
a destructor $g$ is a partial function defined by rewrite rules 
$g(M_1, \ldots, M_l) \rightarrow M$;
the destructor application $g(N_1, \ldots, N_l)$ fails when 
no rewrite rule applies,
and this failure can be tested in the process calculus. 

A destructor $g: \tau_1 \times \dots \times \tau_l \rightarrow \tau$
with rewrite rule $g(M_1, \ldots, M_l) \rightarrow M$ can be encoded
in the applied pi calculus by function symbols $g : \tau_1 \times
\dots \times \tau_l \rightarrow \tau$ and $\Const{test}_g : \tau_1
\times \dots \times \tau_l \rightarrow \Const{bool}$ with the
equations
\eqns{
g(M_1, \ldots, M_l) &=& M\\
\Const{test}_g(M_1, \ldots, M_l) &=& \Const{true}
}
The function $\Const{test}_g$ allows one to test whether $g(N_1, \ldots, N_l)$ is defined, by checking whether $\Const{test}_g(N_1, \ldots, N_l) = \Const{true}$ holds. (See Section~\ref{sec:sample} for examples of such test functions.) 
The function $g$ may be applied even when its arguments are not instances of $(M_1, \ldots, M_l)$, thus yielding terms $g(N_1, \ldots, N_l)$ that do not exist in the calculus with rewrite rules. These ``stuck'' terms may be simulated with distinct fresh names in that variant of the calculus. 

Destructors are easy to implement in a tool. They also provide a
built-in error-handling construct: the error handling is triggered
when no rewrite rule applies.  However, they complicate the semantics
because they require a notion of evaluation of terms. 
Moreover, many useful functions can be
defined by equations but not as destructors (for instance, encryption 
without redundancy, XOR, and modular exponentiation, which we use in 
the rest of this paper).  Therefore, ProVerif supports both destructors 
and equations~\cite{Blanchet07b}. 
Thus, the language of ProVerif
is a superset of the applied pi calculus as defined in this 
paper~\cite[Chapter~4]{ProVerifSurvey}, with the caveat that ProVerif does not 
support all equational theories and that it considers only plain processes.

\item An extension that combines the applied pi calculus with
ambients and with a built-in construct for evaluating messages
as programs has also been studied~\cite{Blanchet03e}.
This extended calculus mixes many notions, so the corresponding proofs
are complex. Considering a single notion at a time yields a simpler
and more elegant calculus. Furthermore, although the applied pi calculus has few primitives,
it supports various other constructs via encodings; in 
particular, the message-evaluation construct could be
represented by defining an interpreter in the calculus.

\item Our equational theories are closed under substitution 
of terms for names. This property yields a simple and uniform treatment
of variables and names. An alternative definition, which may suffice, assumes only that 
equational theories  are closed under one-to-one
renaming and do not equate names.
That definition makes it possible to define a function that tests whether
a term is a name.

\end{itemize}
Some other variations concern the definition of the semantics:
\begin{itemize}

\item As in other papers~\cite{Blanchet07b,Liu11},
we can handle the replication by a reduction
step $\Repl P \rightarrow P \parop \Repl P$ instead of the
structural equivalence rule $\Repl P \equiv P \parop \Repl P$.
This modification prevents transforming $P \parop \Repl P$ into $\Repl P$, 
and thus simplifies some proofs.

\item As Section~\ref{sec:ope-sem} indicates, we can rewrite extended processes by
pulling restrictions to the top, so that every closed extended process
$A$ becomes an extended process $\enf A$ such that
\[A \equiv {\enf A} = \Res {\vect{n}} (\{ \subst{\vect{M}}{\vect{x}} \}
\parop P_1 \parop \ldots \parop P_l)\]
where $\fv(P_1 \parop \ldots \parop P_l) = \emptyset$, $\fv(\vect{M}) = \emptyset$, and 
$P_1$, \ldots, $P_l$ are replication, conditional, input, or output
expressions. 
We can then modify the definitions of structural equivalence and
internal reduction to act on processes in the form above.
Structural equivalence says that the parallel composition $P_1 \parop \ldots \parop P_l$ is associative
and commutative and that the names in~$\vect{n}$ can be reordered. Internal reduction is the smallest relation on closed extended processes, closed by structural equivalence, such that:
\[\begin{array}{rcll}
\CTX[{\Snd N M.P} \parop {\Rcv {N'} x.Q}] & \rightarrow & \CTX[P \parop \enf{Q\{\subst{M}{x}\}]} &\text{if }\Sigma \vdash N = N'\\[.5em]
\CTX[\IfThenElse{M}{N}{P}{Q}] & \rightarrow & \enf{\CTX[P]} &\text{if }\Sigma \vdash M = N\\[.5em]
\CTX[\IfThenElse{M}{N}{P}{Q}] & \rightarrow & \enf{\CTX[Q]} &\text{if }\Sigma \not\vdash M = N\\[.5em]
\CTX[\Repl P] & \rightarrow & \enf{ \CTX[P \parop \Repl P] }&
\end{array}\]%
for any evaluation context $\CTX$.
A similar idea appears in the intermediate applied pi 
calculus of Delaune et al.~\cite{Delaune10} and Liu et al.~\cite{Liu11,LiuLin12}. There, all restrictions not under replication are pulled to the top of processes, over conditionals, inputs, outputs, and parallel compositions; the processes $P_1$, \dots, $P_l$ may be $\nil$; and channels are names or variables.

\item Pushing the previous idea further, we can represent the extended
process 
\[A \equiv \Res {\vect{n}} (\{ \subst{\vect{M}}{\vect{x}} \} 
\parop P_1 \parop \ldots \parop P_l)\]
as a configuration $(\env, \sigma, \mathcal{P}) = (\{\vect{n}\}, \{\subst{\vect{M}}{\vect{x}} \}, \{P_1, \ldots, P_l \})$, where $\env$ is a set of names, $\sigma$ is a substitution, and $\mathcal{P}$ is a multiset of processes. We can then define internal reduction on such configurations, without any structural equivalence. (Sets and multisets allow us to ignore the ordering of restrictions and parallel processes.) This idea is used in semantics of the calculus of ProVerif~\cite{Abadi04f,Allamigeon05,Blanchet08c,ProVerifSurvey}.
\end{itemize}
Semantics based on global configurations are closer to abstract machines.
Such semantics simplify proofs, because they leave
only few choices in reductions. They also make it easier to define
further extensions of the calculus, such as tables and phases in
ProVerif~\cite{ProVerifSurvey}. However, our compositional semantics is more
convenient in order to model interactions between a process and a
context. It is also closer to the traditional semantics of the 
pi calculus. The two kinds of semantics are of course connected. In particular, 
Blanchet~\cite[Chapter~4]{ProVerifSurvey} formally relates 
the semantics of ProVerif based on configurations to our semantics.

\section{Brief Examples}\label{sec:sample}

This section collects several examples, focusing on signatures,
equations, and some simple processes.  We start with pairs; this
trivial example serves to introduce some notations and issues.  We then
discuss lists, cryptographic hash functions, encryption functions, digital
signatures, and the XOR function~\cite{Menezes96,Schnei96}, as well as 
a form of multiplexing, which demonstrates the use of channels that 
are terms rather than names. Further examples
appear in Sections~\ref{sec:diffie-hellman} and~\ref{sec:extension-attacks}.
More examples, such as blind signatures~\cite{Kremer05} and zero-knowledge proofs~\cite{Backes08}, have appeared in the literature since 2001.

Of course, at least some of these functions appear in most
formalizations of cryptography and security protocols.  In comparison
with the spi calculus, the applied pi calculus permits a more uniform
and versatile treatment of these functions, their variants, and their
properties.  Like the spi calculus, however, the  applied pi calculus takes
advantage of notations, concepts, and techniques from programming
languages.

\paragraph{Pairs}

Algebraic datatypes such as pairs, tuples, arrays, and lists occur in many
examples. Encoding them in the pure pi calculus is not hard, but
neither is representing them as primitive.
For instance, the signature~$\Sigma$ may contain the binary function
symbol $\Const{pair}$ and the unary function symbols $\Const{fst}$ and
$\Const{snd}$, with the abbreviation 
$(M,N)$  for $\Const{pair}(M,N)$, and with the evident equations:
\begin{eqnarray}
  \Const{fst}( (x,y) ) & =  & x \label{eq:fst}\\
  \Const{snd}( (x,y) ) & =  & y \label{eq:snd}
\end{eqnarray}
(So the equational theory consists of these equations, and all 
the equations obtained by reflexivity, symmetry, transitivity, 
applications of function symbols, and substitutions of terms for variables.)
These function symbols may for instance be sorted as follows:
\eqns{
  \Const{pair}&: &\Const{Data} \times \Const{Data} \rightarrow \Const{Data}\\
  \Const{fst}&: &\Const{Data}\rightarrow \Const{Data} \\
  \Const{snd}&: &\Const{Data}\rightarrow \Const{Data}
}%
We may use the test 
$(\Const{fst}(M), \Const{snd}(M)) = M$ to check that $M$ is a pair before using the values of $\Const{fst}(M)$ and $\Const{snd}(M)$. 
Alternatively, we may add a boolean function $\Const{is\_pair}$
that recognizes pairs, defined by the equation:
\eqns{
  \Const{is\_pair}((x,y)) &= & \Const{true}
}%
With this equation, the conditional $\IfThenElse{\Const{is\_pair}(M)}{\Const{true}}{P}{Q}$ runs $P$ if $M$ is a pair and $Q$ otherwise.
Using pairs, we may, for instance, define the process:
\[\nu s. \big( 
\Snd{a}{(M,s)}
\parop
\Rcv{a}{z}.\IfThen{\Const{snd}(z)}{s}{\Snd{b}{\Const{fst}(z)}}
\big) \]
One of its components sends a pair consisting of a term $M$
and a fresh name $s$ on a channel $a$. The other receives
a message on~$a$ and, if its second component is~$s$, it forwards
the first component on a channel~$b$.  Thus, we may say that
$s$ serves as a capability (or password) for the forwarding.
However, this capability is not protected from eavesdroppers
when it travels on~$a$. Any other process can listen on $a$
and can apply $\Const{snd}$ to the message received, 
thus learning~$s$. We can represent
such an attacker within the calculus, for example by the 
following process:
\[ \Rcv{a}{z}.\Snd{a}{(N,\Const{snd}(z))}\]
which may receive $(M,s)$ on $a$ and send $(N,s)$ on~$a$.
Composing this attacker in parallel with the process, we may obtain $N$ instead of
$M$ on~$b$.

Such attacks can be thwarted by the use of restricted channel names, as in 
the process
\[\Res{a} \nu s. \big( 
\Snd{a}{(M,s)}
\parop
\Rcv{a}{z}.\IfThen{\Const{snd}(z)}{s}{\Snd{b}{\Const{fst}(z)}}
\big) \]
Alternatively, they can be thwarted by the use of cryptography, as
discussed below.

\paragraph{Lists}

We may treat lists similarly, with the following function symbols and corresponding sorts:
\eqns{
\Const{nil} &: &\DataList\\
\Const{cons} &: &\Data \times \DataList \rightarrow \DataList\\
\Const{hd} &: &\DataList \rightarrow \Data\\
\Const{tl} &: &\DataList \rightarrow \DataList
}
The constant $\Const{nil}$ is the empty list;
$\Const{cons}(x,y)$ represents the concatenation of the element $x$ at the beginning of the list $y$, and we write it with infix notation as $\cons{x}{y}$,
where the symbol $::$ associates to the right;
and $\Const{hd}$ and $\Const{tl}$ are head and tail functions 
with equations:
\begin{equation}
\Const{hd}(\cons{x}{y}) = x
\qquad
\Const{tl}(\cons{x}{y}) = y\label{eq:seq}
\end{equation}
Further, we write $M \bappend N$ for 
the concatenation of an element $N$ at the end of a list $M$, 
where the function 
$\bappend : \DataList \times \Data \rightarrow \DataList$
associates to the left, and satisfies the equations:
\begin{equation}
\Const{nil} \bappend x = \cons{x}{\Const{nil}}
\qquad (\cons{x}{y}) \bappend z = \cons{x}{(y\bappend z)}
\label{eq:++}
\end{equation}

\paragraph{Cryptographic Hash Functions}

We represent a cryptographic hash function as a unary function
symbol $\hash$ with no equations.
The absence of an inverse for $\hash$ models the
one-wayness of $\hash$.  The fact that $\hash(M) = \hash(N)$ only when
$M = N$ models that $\hash$ is collision-free.

Modifying our first example, we may now write the process:
\[\nu s. \left( 
\Snd{a}{(M,\hash((s,M)))}
\parop 
\Rcv{a}{x}.\IfThen{\hash((s,\Const{fst}(x)))}{\Const{snd}(x)}
{\Snd{b}{\Const{fst}(x)}}
\right) \]
Here the value $M$ is authenticated by pairing it with the fresh
name $s$ and then hashing the pair.
Although $(M,\hash((s,M)))$ travels on the public channel
$a$, no other process can extract $s$ from this message, or produce 
$(N,\hash((s,N)))$ for some other $N$ using the available functions.
Therefore, we may reason that this process will forward only the intended
  term $M$ on channel~$b$. 

This example is a typical cryptographic application of hash
functions.  In light of the practical importance of those
applications, our treatment of hash functions is 
attractively straightforward. Still, we may question whether our formal
model of these functions is not too strong and simplistic 
in comparison with the
properties of actual implementations based on algorithms such as 
SHA. In Section~\ref{sec:extension-attacks}, we consider a somewhat
weaker, subtler model for hash functions.

\paragraph{Symmetric Encryption} 

In order to model symmetric cryptography (that is, shared-key cryptography), we take binary 
function symbols $\Const{enc}$ and $\Const{dec}$ for encryption
and decryption, respectively, with the equation:
\eqns{
  \Const {dec} (\Const {enc}(x, y), y) & =  & x 
  }%
Here $x$ represents the plaintext and $y$ the key.
We often use fresh names as
keys in examples; for instance, the (useless) process:
\[\nu k. 
\Snd{a}{\Const{enc}(M,k)}
\]
sends the term $M$ encrypted under a fresh key $k$.

In applications of encryption, it is frequent to assume that
each encrypted message comes with sufficient redundancy so that
decryption with the ``wrong'' key is evident. 
Accordingly, we can test whether the decryption of $M$ with the key $k$
succeeds by testing whether $\Const {enc} (\Const {dec}(M, k), k) =  M$.
Alternatively, we could also add a test function $\Const{test}_{\Const{dec}}$
with the equation 
\eqns{
  \Const{test}_\Const {dec} (\Const {enc}(x, y), y) & =  & \Const{true}
}%
Provided that we check that decryption succeeds before using the decrypted
message, this model of encryption basically yields the spi 
calculus~\cite{spi2four}. 

On the other hand, in modern
cryptology, such redundancy is not usually viewed as part of the encryption function
proper, but rather an addition. The 
redundancy can be implemented with message authentication codes.
We can model an encryption scheme without redundancy 
with the two equations:
\eqns{
  \Const {dec} (\Const {enc}(x, y), y) & =  & x \\
  \Const {enc} (\Const {dec}(z, y), y) & =  & z
  }%
These equations model that decryption is the inverse bijection of encryption,
a property that is typically satisfied by block ciphers.

\paragraph{Asymmetric Encryption}

It is only slightly harder to model 
asymmetric (public-key) cryptography, where the keys for encryption
and decryption are different. We introduce
two new unary function symbols
$\Const{pk}$ and $\Const{sk}$ for generating public and private keys from
a seed, and the equation:
\eqns{
  \Const{dec}(\Const{enc}(x,\Const{pk}(y)), \Const{sk}(y)) & =  & x 
}%

We may now write the process:
\[\nu s. \big( 
\Snd{a}{\Const{pk}(s)}
\parop
\Rcv{b}{x}.\Snd{c}{\Const{dec}(x,\Const{sk}(s))}
\big) \]
The first component publishes the public key $\Const{pk}(s)$
by sending it on $a$. The second receives a message on $b$,
uses the corresponding private key $\Const{sk}(s)$ to decrypt it,
and forwards the resulting plaintext on $c$.
As this example indicates, we essentially view name restriction ($\nu s$) as a generator of 
unguessable seeds. In some cases, those seeds may be directly
used as passwords or keys; in others, some transformations are needed.

Some encryption schemes have additional properties. In particular,
$\Const{enc}$ and $\Const{dec}$ may be the same function.  This
property matters in implementations, and sometimes permits attacks. 
Moreover, certain encryptions and decryptions commute in some
schemes.  For example, we have
$\Const{dec}(\Const{enc}(x,y),z)  = \Const{enc}(\Const{dec}(x,z),y)$
if the encryptions and decryptions are performed using RSA with the same modulus.
The treatment of such properties is left open in the spi 
calculus~\cite{spi2four}.
In contrast, it is easy to express the properties 
in the applied pi calculus, and to study the protocols and attacks
that depend on them. 

\paragraph{Non-Deterministic (``Probabilistic'') Encryption} 

Going further, we may add a third argument to $\Const{enc}$,
so that the encryption of a plaintext with a key is not unique.
This non-determinism is an essential property of 
probabilistic encryption~\cite{gm}. The equation for decryption becomes:
\eqns{
  \Const{dec}(\Const{enc}(x,\Const{pk}(y),z), \Const{sk}(y)) & =  & x 
}%

With this variant, we may write the process:
\[\Rcv{a}{x}. \big( 
\nu m. \Snd{b}{\Const{enc}(M,x,m)} \parop 
\nu n. \Snd{c}{\Const{enc}(N,x,n)}
\big) \]
which receives a message $x$ and uses it as an encryption key
for two messages, $\Const{enc}(M,x,m)$ and $\Const{enc}(N,x,n)$.
An observer who does not have the corresponding decryption key
cannot tell whether the underlying plaintexts 
$M$ and $N$ are identical by comparing the ciphertexts,
because the ciphertexts rely on different fresh names $m$ and $n$.
Moreover, even if the observer learns $x$, $M$, and $N$ (but not
the decryption key), it cannot
verify that the messages contain $M$ and $N$ because it
does not know $m$ and $n$.

\paragraph{Public-Key Digital Signatures} 

Like public-key encryption schemes,
digital signature schemes rely on pairs of public and private keys. In each pair, 
the private key serves for computing signatures and the public key for verifying
those signatures. 
In order to model key generation, we use again the two unary function
symbols $\Const{pk}$ and $\Const{sk}$ for generating public and private
keys from a seed. For signatures and their verification, we 
use a new binary function symbol $\Const{sign}$, a ternary
function symbol $\Const{check}$, and a constant symbol $\Const{ok}$,
with the equation:
\eqns{
  \Const{check}(x,\Const{sign}(x,\Const{sk}(y)), \Const{pk}(y))
  & =  & \Const{ok}
}%
(Several variants are possible.)

Modifying once more our first example, we may now write the process:
\[\begin{array}{l}
\left( \nu s. 
\{\subst{\Const{pk}(s)}{y}\} \parop
\Snd{a}{(M,\Const{sign}(M,\Const{sk}(s)))}
\right) \parop \\[.3em]
\Rcv{a}{x}.
\IfThen{\Const{check}(\Const{fst}(x),\Const{snd}(x),y)}{\Const{ok}}
{\Snd{b}{\Const{fst}(x)}}
\end{array}
\]
Here the value $M$ is signed using the private key $\Const{sk}(s)$. 
Although $M$ and its signature travel on the public channel~$a$, 
no other process can produce 
$N$ and its signature for some other $N$.
Therefore, again, we may reason that only the intended term $M$ will be
forwarded on channel~$b$. 
This property holds despite the publication of 
$\Const{pk}(s)$ (but not $\Const{sk}(s)$),
which is 
represented by the active substitution that maps $y$ to $\Const{pk}(s)$. 
Despite the restriction on $s$,
processes outside the restriction can use $\Const{pk}(s)$
through~$y$. In particular, 
$y$ refers to  $\Const{pk}(s)$ in the process that checks
the signature on~$M$.

\paragraph{XOR}

We may model the XOR function, some of its uses
in cryptography, and some of the protocol flaws connected with it. 
Some of these flaws (e.g.,~\cite{RyanSchneider98})
stem from the intrinsic equational properties of XOR,
such as associativity, commutativity, the existence of a neutral element, 
and the 
cancellation property that we may write:
\begin{eqnarray*}
  \Const{xor}(\Const{xor}(x,y),z) & = & \Const{xor}(x,\Const{xor}(y,z))\\
  \Const{xor}(x,y) & = & \Const{xor}(y,x)\\
  \Const{xor}(x,0) & = & x\\
  \Const{xor}(x,x) & = & 0
\end{eqnarray*}
Others arise because of the interactions between XOR and other operations 
(e.g.,~\cite{stubblebine:integrity,sshattack}).
For example, CRCs (cyclic redundancy checks)
can be poor proofs of integrity,
partly because of the equation
\eqns{ \Const{crc}(\Const{xor}(x,y)) & = & 
  \Const{xor}(\Const{crc}(x),\Const{crc}(y))}%

\paragraph{Multiplexing}

Finally, we illustrate a possible usage of channels that are not names.
Consider for instance a pairing function for building channels 
$\Const{pair}: \Const{Data} \times \Const{Port} \rightarrow \Const{Channel}$
with its associated projections $\Const{fst}: \Const{Channel} \rightarrow \Const{Data}$ and $\Const{snd}:  \Const{Channel} \rightarrow \Const{Port}$, and equations~\eqref{eq:fst} and~\eqref{eq:snd} from our first example.
We may use this function for multiplexing as follows:
\begin{align*}
\Res{s}(&\Snd{\Const{pair}(s, \Const{port}_1)}{M_1}
\parop \Snd{\Const{pair}(s, \Const{port}_2)}{M_2}\\
{}\parop{} &\Rcv{\Const{pair}(s, \Const{port}_1)}{x_1}
\parop \Rcv{\Const{pair}(s, \Const{port}_2)}{x_2})
\end{align*}
In this process, the first output can be received only by the
first input, and the second output can be received only by the second input.

\section{Equivalences and Proof Techniques}\label{sec:equivalences}

In examples, we frequently argue that two given processes cannot be
distinguished by any context, that is, that the processes are
observationally equivalent. The spi calculus developed the idea that
the context represents an active attacker, and equivalences capture
authenticity and secrecy properties in the presence of the attacker.
More broadly, a wide variety of security properties can be expressed as
equivalences.

In this section we define observational equivalence for the applied pi
calculus. We also introduce a notion of static equivalence for frames,
a labelled semantics for processes, and a labelled equivalence
relation. We prove that labelled equivalence and observational
equivalence coincide, obtaining a convenient proof technique for
observational equivalence. 

\subsection{Observational Equivalence}\label{subsec:obseq}

We write $A\barb{a}$ when $A$ can send a message on name $a$, that is, when
$A \rightarrow^*\equiv \CTX[ \Snd{a}{M}.P ]$ for some evaluation context
$\CTX[\hole]$ that does not bind~$a$.

\begin{definition}\label{def:bicong}
An {\em observational bisimulation} is a symmetric
relation $\rel$ between closed extended processes with the
same domain such that $A \rel B$ implies:
  \begin{enumerate}
  \item\label{bcone} if $A \barb{a}$, then $B \barb{a}$;
  \item\label{bctwo} if $A \rightarrow^* A'$ and $A'$ is closed, then 
    $B \rightarrow^* B'$ and $A' \rel B'$~for some $B'$;
  \item\label{bcthree} $\CTX[A] \rel \CTX[B]$ for all closing evaluation contexts $\CTX[\hole]$.
\end{enumerate}
{\em Observational equivalence} ($\bicong$) is the largest such relation.
\end{definition}

For example, when $\Const{h}$ is a unary function symbol with no equations, 
we obtain that $\nu s. \Snd{a}{s} \bicong \nu s. \Snd{a}{\Const{h}(s)}$.

These definitions are standard in the pi calculus, where $\barb{a}$ is
called a {\em barb} on~$a$, and where $\bicong$ is one of the two
usual notions of weak barbed bisimulation congruence.
(See Section~\ref{sec:bigpf} and~\cite{FournetGonthier98:equivalences} for a detailed discussion.)
In the applied pi calculus, one could also
define barbs on arbitrary terms, not just on names; we do not need that
generalization for our purposes.
The set of closing evaluation contexts for $A$ depends only
on $A$'s domain; hence, in Definition~\ref{def:bicong}, $A$ and $B$ have the
same closing evaluation contexts.
In Definition~\ref{def:bicong}\eqref{bctwo}, since $\rel$ is a
relation between closed extended processes, we require that $A'$ also
be closed. Being closed is not preserved by all reductions, since
structural equivalence may introduce free unused variables.  For
instance, we have $\nil \equiv \nu x.\{\subst{y}{x}\}$ by
\rulename{Alias} and $\{\sx M\} \equiv \{\sx {\Const{fst}((M,y))} \}$
by \rulename{Rewrite} using the equation \hbox{$\Const{fst}((x,y)) = x$}.

Although observational equivalence is undecidable
in general, various tools support certain automatic proofs of
observational equivalence and other equivalence relations, in the
applied pi calculus and related languages 
(e.g.,~\cite{Baudet05ccs,Blanchet07b,Chadha12,CCD-tcs13}).

\subsection{Static Equivalence}\label{subsec:static}

Two substitutions may be seen as equivalent when they behave equivalently
when applied to terms. We write $\enveq$ for this notion
of equivalence, and call it static equivalence. 
In the presence of the ``new'' construct, defining $\enveq$ is 
somewhat delicate and interesting. 
For instance, consider two functions $\Const f$ and~$\Const g$
with no equations (intuitively, two independent hash functions), and the
three frames:
\eqns{
  \someframe_0 & \eqdef & 
  \nu k. \{\subst{k}{x}\} \parop \nu s.\{\subst{s}{y}\}
  \\ \someframe_1 & \eqdef & 
  \nu k. \{\subst{\Const {f}(k)}{x},\subst{\Const {g}(k)}{y}\}
  \\ \someframe_2 & \eqdef & 
  \nu k. \{\subst{k}{x},\subst{\Const f(k)}{y}\}
}%
In $\someframe_0$, the variables $x$ and $y$ are mapped to two unrelated
values that are different from any value that the context may build
(since $k$ and $s$ are new). These properties also hold, but more subtly,
for $\someframe_1$; although $\Const {f}(k)$ and $\Const {g}(k)$ are
based on the same underlying fresh name, they look unrelated.
(Analogously, it is common to derive apparently unrelated keys
by hashing from a single underlying secret, as in SSL and TLS~\cite{SSLthreeotwo,TLS12}.) 
Hence, a context
that obtains the values for $x$ and $y$ cannot distinguish
$\someframe_0$ and $\someframe_1$.
On the other hand, the context can discriminate $\someframe_2$ by testing the
predicate $\Const f(x) = y$.
Therefore, we would like to define static equivalence so that 
$\someframe_0 \enveq \someframe_1 \not\enveq \someframe_2$. 

This example relies on a concept of equality of terms in a frame,
which the following definition captures.
\begin{definition}\label{def:eqframe}
Two terms $M$ and $N$ are equal in the frame $\varphi$,
written $(M=N)\varphi$, if and only if 
$\fv(M)\cup\fv(N) \subseteq \dom(\varphi)$,
$\varphi \equiv \nu \vect{n}.\sigma$, 
$M\sigma = N\sigma$, and
$\{\vect{n}\} \cap (\fn(M)\cup\fn(N)) = \emptyset$ 
for some names $\vect{n}$ and substitution $\sigma$.
\end{definition}
In Definition~\ref{def:eqframe}, 
the equality $M\sigma = N\sigma$ is independent
of the representative $\Res{\vect n}\sigma$ 
chosen for the frame $\varphi$ such that
$\varphi \equiv \nu \vect{n}.\sigma$ and
$\{\vect{n}\} \cap (\fn(M)\cup\fn(N)) = \emptyset$. 
(Lemma~\ref{lem:equivframes} in Appendix~\ref{app:disclosure} establishes this property.)

\begin{definition}\label{def:enveq}
  Two closed frames $\varphi$ and $\psi$ are
  \emph{statically equivalent}, written $\varphi \enveq \psi$, when 
  $\dom(\varphi)= \dom(\psi)$ and when, for all terms $M$ and $N$, we
  have $(M=N)\varphi$ if and only if $(M=N)\psi$.
  
  Two closed extended processes are statically equivalent,
  written $A \enveq B$, when their frames are statically equivalent.
\end{definition}

For instance, in our example, we have $(\Const f(x) = y)\someframe_2$ but not $(\Const f(x) =
y)\someframe_1$, hence $\someframe_1 \not\enveq \someframe_2$.

Depending on $\Sigma$, static equivalence can be quite hard to check,
but at least it does not depend on the dynamics of processes.  Some
simplifications are possible in common cases, in particular when terms
can be put in normal forms (for example, in the proof of Theorems~\ref{th:mac} and~\ref{th:indif}).
Decisions procedures exist for static equivalence in large classes of equational theories~\cite{Abadi06b}, some implemented in tools~\cite{Baudet09,CDK-jar10}.

The next lemma establishes closure properties of static equivalence:
it shows that static equivalence is invariant by structural equivalence 
and reduction, and closed by application of closing evaluation contexts.  
Its proof appears in Appendix~\ref{app:simplecontexts}.

\begin{lemma}\label{LEM:INVARIANT-STATIC-EQ}
Let $A$ and $B$ be closed extended processes. 
If $A \equiv B$ or $A \rightarrow B$, then $A \enveq B$.
If $A \enveq B$, then $\CTX[A] \enveq \CTX[B]$ for all closing evaluation contexts $\CTX[\hole]$.
\end{lemma}

As the next two lemmas demonstrate, static equivalence coincides with 
observational equivalence on frames, but is coarser on extended processes.

\begin{lemma}\label{lem:enveq1}
  Observational equivalence and static equivalence coincide on frames.
\end{lemma}
This lemma is an immediate corollary of 
Theorem~\ref{THM:OBSERVATIONAL-LABELED} below. (See 
Corollary~\ref{cor:enveq-frame} in Appendix~\ref{app:twodirection}.)

\begin{lemma}\label{lem:enveq2}
  Observational equivalence is strictly finer than static equivalence
  on extended processes: ${\bicong} \subset {\enveq}$.
\end{lemma}
To see that observational equivalence implies static equivalence, note
that if $A$ and $B$ are observationally equivalent then 
$A \parop C$ and $B \parop C$ have the same barbs for every $C$ 
with $\fv(C)\subseteq\dom(A)$,
and that they are statically equivalent when 
$A \parop C$ and $B \parop C$ have the same barb $\barb{a}$ for every $C$
of the special form $\IfThen{M}{N}{\Snd{a}{n}}$, where $a$ does not occur 
in $A$ or $B$ and $\fv(C) \subseteq \dom(A)$. (See Lemma~\ref{lem:enveq-gen} in Appendix~\ref{app:twodirection}.)
The converse does not hold, as the following counter-example shows:
letting $A = \Snd{a}{n}$ and $B = \Snd{b}{n}$, we have $A \not \bicong B$, but 
$A \enveq B$ because $\frameof{A} = \frameof{B} = \nil$.

\subsection{Labelled Operational Semantics and Equivalence}\label{subsec:labopsem}

A labelled operational semantics extends the chemical semantics
of Section~\ref{sec:ope-sem}, enabling us to reason
about processes that interact with their context while keeping it implicit.
The labelled semantics defines a relation $A \ltr{\alpha} A'$,
where $\alpha$ is a label of one of the following forms:
\begin{itemize}
\item a label $\Rcv{N}{M}$, 
which corresponds to an input of $M$ on~$N$;
  
\item a label $\nu x.\Snd{N}{x}$, where $x$ is a variable that must not occur in $N$, which corresponds to an
  output of $x$ on $N$. 
\end{itemize}
The variable $x$ is bound in the label $\nu x.\Snd{N}{x}$, so we define 
the bound variables of labels
by $\bv(\Rcv{N}{M}) \eqdef \emptyset$
and $\bv(\nu x.\Snd{N}{x}) \eqdef \{ x \}$. The free variables of labels
are defined by $\fv(\Rcv{N}{M}) \eqdef \fv(N) \cup \fv(M)$
and $\fv(\nu x.\Snd{N}{x}) \eqdef \fv(N)$
(since $x$ does not occur in $N$ in the latter label).

In addition to the rules for structural equivalence and reduction of
Section~\ref{sec:calculus},
we adopt the following rules:
\[\begin{array}{lc}
  \brn{In} &
  \Rcv{N}{x}.P 
  \ltr{ \Rcv{N}{M} }  
  P \{\subst{M}{x}\}  
  \\[3ex]
  \brn{Out-Var} &
  \infrule{ 
    x \notin \fv(\Snd{N}{M}.P)}{
  \Snd{N}{M}.P \ltr{ \nu x.\Snd{N}{x} } P \parop \{\subst{M}{x}\}}
  \\[4ex]
  \brn{Scope} &
  \infrule{ 
    A \ltr{\alpha} A' \hspace{5ex} 
    u \mbox{ does not occur in } \alpha
    }{
    \nu u.A \ltr{\alpha} \nu u.A'}
  \\[3ex]
  \brn{Par} & \hspace{-3ex}
  \infrule{
    A \ltr{\alpha} A' \hspace{5ex} 
    \bv(\alpha)\cap \fv(B) = \emptyset }{
    A \parop B \ltr{\alpha} A'\parop B }
  \\[3ex]
  \brn{Struct} & \hspace{-3ex}
  \infrule{
    A \equiv B \hspace{5ex} B \ltr{\alpha} B' \hspace{5ex} B' \equiv A' }{
    A \ltr{\alpha} A'}
\end{array}\]

According to \rulename{In}, a term $M$ may be input.
On the other hand, \rulename{Out-Var} permits output for terms ``by reference'': 
a fresh variable is associated with the term in question and output.

\begin{figure}[t]
\eqns{
&& \nu k. \Snd{a}{\Const{enc}(M,k)}. \Snd{a}{k} . \Rcv{a}{z}.\IfThen z M {\Snd{c}{\youwin}} \\
& \ltr{\nu x.\Snd{a}{x}} &
\nu k. \big(\{\subst{\Const{enc}(M,k)}{x}\} \parop \Snd{a}{k}. \Rcv{a}{z}.\IfThen z M {\Snd{c}{\youwin}} \big)\\
& \ltr{\nu y.\Snd{a}{y}} & 
\nu k. \big(\{\subst{\Const{enc}(M,k)}{x}\} \parop \{\subst{k}{y}\} \parop 
\Rcv{a}{z}.\IfThen z M {\Snd{c}{\youwin}} \big)\\
& \ltr{\Rcv{a}{\Const {dec}(x,y)}} &
\nu k. \big(\{\subst{\Const{enc}(M,k)}{x}\} \parop \{\subst{k}{y}\} \parop
\IfThen 
{\Const {dec} (x,y)}
 M {\Snd{c}{\youwin}} \big)\\
& \rightarrow  &
\nu k. \big( \{\subst{\Const{enc}(M,k)}{x}\} \parop \{\subst{k}{y}\} \big) \parop
\Snd{c}{\youwin} 
}
\caption{Example transitions}\label{fig:oops}
\end{figure}

For example, using the signature and equations for symmetric encryption, 
and the new constant symbol $\youwin$,
we have the sequence of transitions of Figure~\ref{fig:oops}.
The first two transitions do not directly reveal the term $M$. However, 
they give enough information to the environment to compute $M$ as $\Const{dec}(x,y)$,
and to input it in the third transition.

The labelled operational semantics leads to an equivalence relation:
\begin{definition}\label{def:wkbisim}
A \emph{labelled bisimulation} is a
symmetric relation~$\rel$ 
  on closed extended proc\-esses such that $A \rel B$ implies:
  \begin{enumerate}
  \item $A \enveq B$;\label{ppone}
  \item if $A \rightarrow A'$ and $A'$ is closed, then $B
    \rightarrow^* B'$ and $A' \rel B'$ for some $B'$;\label{pptwo}
  \item if $A \ltr{\alpha} A'$, $A'$ is closed, and 
    $\fv(\alpha) \subseteq \dom(A)$,
    then $B \rightarrow^*
    \ltr{\alpha} \rightarrow^* B'$ and $A' \rel B'$ for some $B'$.\label{ppthree}
\end{enumerate}
{\em Labelled bisimilarity} ($\wkbisim$) is the largest such relation.
\end{definition}
Conditions~\ref{pptwo} and~\ref{ppthree} are standard;
condition~\ref{ppone}, which requires that 
bisimilar processes be statically equivalent, is necessary
for example in order to distinguish the frames $\someframe_0$ and $\someframe_2$ of
Section~\ref{subsec:static}.
As in Definition~\ref{def:bicong}, we explicitly require that $A'$ be closed and $\fv(\alpha) \subseteq \dom(A)$
in order to exclude transitions that introduce free unused variables.

Our main result is that this relation coincides with observational
equivalence. Although such results are fairly common in process calculi,
they are important and non-trivial.
\begin{theorem}\label{THM:OBSERVATIONAL-LABELED}
  Observational equivalence is labelled bisimilarity: ${\bicong} = {\wkbisim}$.
\end{theorem}
The proof of this theorem is outlined in Section~\ref{sec:bigpf} and completed in the appendix. 

The theorem implies that 
$\wkbisim$ is closed by application of closing evaluation contexts.  However,
unlike the definition of $\bicong$, the definition of ${\wkbisim}$
does not include a condition about contexts.  It therefore permits
simpler proofs.

In addition, labelled bisimilarity can probably be established via 
standard ``bisimulation up to context'' techniques~\cite{San98MFCS}, which 
enable useful on-the-fly simplifications in frames after output steps.
We do not develop the theory of ``up to context'' techniques, since 
we do not use them in this paper.

The following lemmas provide methods for simplifying frames:

\begin{lemma}[Alias elimination]\label{lem:upto-frame}
Let $A$ and $B$ be closed extended processes, $M$ be a term such that
$\fv(M) \subseteq \dom(A)$, and $x$ be a variable such that $x \notin \dom(A)$.
We have $A \wkbisim B$ if and only if 
\eqns{ 
  \smx \parop A & \wkbisim & \smx \parop B}%
\end{lemma}

\begin{proof} Both directions follow from context closure of $\wkbisim$, for the contexts $\smx \parop \_$ and $\Res{x}{\_}$, respectively. 
In the converse direction, since $x$ is not free in $A$ or $B$, 
we have $A \equiv \Res{x}{(\smx \parop A)}$, $\Res{x}{(\smx \parop A)} \wkbisim \Res{x}{(\smx \parop B)}$,
and $\Res{x}{(\smx \parop B)} \equiv B$ hence $A \wkbisim B$. 
\end{proof}

\begin{lemma}[Name disclosure]\label{LEM:DISCLOSURE}
  Let $A$ and $B$ be closed extended processes and $x$ be a variable such that
  $x \notin \dom(A)$.
We have $A \wkbisim B$ if and only if
\eqns{
  \nu n.(\{\subst{n}{x}\}
  \parop A )
  & \wkbisim & 
  \nu n.(\{\subst{n}{x}\}
  \parop B )}%
\end{lemma}

\begin{proof} The direct implication follows from context closure of
  $\wkbisim$. Conversely, we show that the relation $\rel$ defined by
$A \rel B$ if and only if $A$ and $B$ are closed extended processes and
  $\Res{n}(\{\subst{n}{x}\} \parop A )
  \wkbisim  
  \Res{n}(\{\subst{n}{x}\} \parop B )$ for some $x \notin \dom(A)$ 
  is a labelled bisimulation. This proof is detailed in Appendix~\ref{app:disclosure}.
\end{proof}

In Lemma~\ref{lem:upto-frame}, the substitution $\smx$ can  affect only
the context, since $A$ and $B$ are closed. However, the lemma implies that the substitution
does not give or mask any information about $A$ and $B$ to the context.
In Lemma~\ref{LEM:DISCLOSURE}, the restriction on $n$ and the substitution
$\{\subst{n}{x}\}$ mean that the context can access $n$ only indirectly, through the
free variable $x$. 
Intuitively, the lemma says that indirect access is equivalent
to direct access in this case.

\sloppy 
Our labelled operational semantics contrasts with a more naive
semantics carried over from the pure pi calculus, with output labels
of the form $\nu \vect{u}.\Snd{N}{M}$ and rules that permit direct output of any
term, such as:
\[\begin{array}{@{}lc@{}}
  \brn{Out-Term} & 
  \Snd{N}{M}.P 
  \ltr{ \Snd{N}{M} }  
  P
  \\[3ex]
  \brn{Open} & 
  \infrule{ 
    A \ltr{\nu\vect{u}.\Snd{N}{M}} A'  {\hspace{5ex}}
    v \in \fv(M) \cup \fn(M) \setminus (\fv(N) \cup \fn(N) \cup \{\vect{u}\})
    }{
    \nu v. A \ltr{\nu v,\vect{u}.\Snd{N}{M}} A' }
\end{array}\]
These rules lead to a different, finer equivalence relation,
which for example would distinguish $\nu k, s. \Snd{a}{(k,s)}$
and $\nu k. \Snd{a}{(\Const{f}(k),\Const{g}(k))}$.
This equivalence relation is often inadequate in applications 
(as in~\cite[Section 5.2.1]{spi2four}), hence our definitions. 

\fussy 

We have also studied intermediately liberal rules for output, which
permit direct output of certain terms.  In particular, the rules of
the conference paper permit
direct output of channel names. That feature implies that it is not
necessary to export variables of channel types; as
Section~\ref{sec:bigpf} explains, this property is needed for
Theorem~\ref{THM:OBSERVATIONAL-LABELED} for those rules. That feature
makes little sense in the present calculus, in which arbitrary terms
may be used as channels, so we abandon it in the rules above.
Nevertheless, certain rules with more explicit labels can still be
helpful.  We explain those rules next.

\subsection{Making the Output Labels More Explicit}\label{SUBSEC:REFINING}

In the labelled operational semantics of
Section~\ref{subsec:labopsem}, the labels for outputs do not reveal
anything about the terms being output: those terms are represented by
fresh variables.
Often, however, more explicit labels can be convenient in reasoning about
protocols, and they do not cause harm as long as they only make explicit
information that is immediately available to the environment.
For instance, for the process
$\nu k. \Snd{a}{(\Const{Header},\Const{enc}(M,k))}$, the
label $\nu y.\Snd{a}{(\Const{Header},y)}$
is more informative than $\nu x.\Snd{a}{x}$. 
In this example, the environment could anyway observe that $x$ is a pair such that $\Const{fst}(x) = \Const{Header}$ and use $\Const{snd}(x)$ for $y$.
More generally, we rely on  the following definition
to characterize the information that the environment can derive.

\begin{definition}\label{defn:derived}
  Variables $\vect x$ 
  \emph{resolve to} $\vect M$ \emph{in} $A$ if and only if $A \eqstr \smxvect \parop
  \Res{\vect x}{A}$.  
  They are \emph{solvable} in $A$ if and only if they resolve to some terms in $A$.
\end{definition}
Hence, when variables $\vect x$ resolve to terms $\vect M$ in $A$, 
they are in  $\dom(A)$ and
we can erase the restriction of $\nu \vect x.A$ 
by applying the context $\smxvect  \parop \hole$
and by structural equivalence.
Intuitively, 
$A$ does not reveal more information than $\nu \vect x.A$, because the
environment can build the terms $\vect M$ and use them instead of~$\vect x$.

In general, when variables $\vect x$ are in $\dom(A)$, there exist 
$\vect{n}$, $\vect M$, and~$A'$ such that 
$A \equiv
 \nu\vect{n}.(\smxvect \parop A')$. 
If variables $\vect x$ resolve to $\vect M$ in~$A$, then $\vect{n}$ can be chosen empty, 
so that the terms $\vect M$ are not under restrictions.
The following lemma provides two reformulations of Definition~\ref{defn:derived}, including a converse to this observation. Its proof appears in Appendix~\ref{app:refining}.
\begin{lemma}
\label{lem:not-so-new} 
The following three properties are equivalent:
\begin{enumerate}
\item\label{equiv:resolve} the variables $\vect x$ resolve to $\vect M$ in $A$;
\item\label{equiv:caract-equiv} there exists $A'$ such that $A \equiv \smxvect \parop A'$;
\item\label{equiv:caract-phi} 
$(\vect x = \vect M)\frameof{A}$
and the substitution $\smxvect$ is cycle-free.
\end{enumerate}
\end{lemma}

For example, using pairs and symmetric encryption, 
we let:
\eqns{
  \varphi & \eqdef & \nu k.\{ \subst{M}{x}, \subst{\Const{enc}(x,k)}{y},
  \subst{(\Const{Header},y)}{z}\}
}%
The variable $y$ resolves to $\Const{snd}(z)$ in $\varphi$, 
since 
\eqns{
  \varphi & \equiv & \{\subst{\Const{snd}(z)}{y}\} \parop \nu k.\{ \subst{M}{x}, 
  \subst{(\Const{Header},\Const{enc}(x,k))}{z}\}
}%
and $z$ resolves to $(\Const{Header},y)$ in $\varphi$, since 
\eqns{
   \varphi & \equiv & 
   \{ \subst{(\Const{Header},y)}{z} \} \parop 
   \nu k.\{ \subst{M}{x}, \subst{\Const{enc}(x,k)}{y} \}
}%
In contrast, $x$ is not always solvable in $\varphi$ (for instance, when $M$ is $k$).

A second lemma shows that Definition~\ref{defn:derived} is robust in the sense that it is preserved by static equivalence,
so a fortiori by labelled bisimilarity: 

\begin{lemma}\label{lem:enveq-preserves-resolve}
If $A \enveq B$ and 
$\vect x$ resolve to $\vect M$ in $A$, then 
$\vect x$ resolve to $\vect M$ in $B$.
\end{lemma}

\begin{proof}
Static equivalence preserves property~\ref{equiv:caract-phi} of Lemma~\ref{lem:not-so-new}, so we conclude by Lemma~\ref{lem:not-so-new}.
\end{proof}

We introduce an alternative semantics in which the 
rules permit composite terms in output labels but require
that every restricted variable that is exported be solvable.
In this semantics, the label $\alpha$ in the relation $A \ltr{\alpha}
A'$ ranges over the same input labels $\Rcv{N}{M}$ 
as in Section~\ref{subsec:labopsem}, and over
generalized output labels of the form
$\nu\vect{x}.\Snd{N}{M}$, 
where $\{\vect{x}\} \subseteq \fv(M)\setminus \fv(N)$.
The label $\nu\vect{x}.\Snd{N}{M}$ corresponds to an output of $M$ on $N$ that
reveals the variables~$\vect{x}$.
We retain the rules for structural equivalence and reduction, and
rules \rulename{In}, \rulename{Par}, and \rulename{Struct} of Section~\ref{subsec:labopsem}.
We also keep rule \rulename{Scope}, but only for labels with no
  extrusion, that is, for labels $\Rcv{N}{M}$ and $\Snd{N}{M}$.
This restriction is necessary because variables may not remain solvable
after the application of a context $\Res{u}\hole$.
As a replacement for the rule \rulename{Out-Var}, 
we use the rule \rulename{Out-Term} discussed in Section~\ref{subsec:labopsem} and:
\[\begin{array}{lc}
  \brn{Open-Var} & 
  \infrule{ 
   \begin{array}{l}
    A \ltr{\Snd{N}{M}} A' \hspace{5ex}
    \{ \vect x \} \subseteq \fv(M) \setminus \fv(N)
    \\ \vect x \mbox{ solvable in } 
    \{\subst{M}{z}\} \parop A' \text{ for some }z \notin \fv(A') \cup \{ \vect x \}
   \end{array}}{
    \nu \vect x. A \ltr{\Res{\vect{x}} \Snd{N}{M}} A' }
\end{array}\]
These rules are more liberal
than those of Section~\ref{subsec:labopsem}.
For instance, consider $A_1 = \nu k.\Snd{a}{(\Const f(k),\Const g(k))}$ and 
$A_2 =  \nu k. \Snd{a}{(k, \Const f(k))}$. 
With the rules of Section~\ref{subsec:labopsem}, we have: 
$$A_i \ltr{\nu z.\Snd{a}{z}} \nu
x,y. (\{\subst{(x,y)}{z}\} \parop \someframe_i)$$ where 
$\someframe_i$ is as in Section~\ref{subsec:static}. 
With the new rules, we also have: 
\begin{equation}
A_i \ltr{\nu x,y.\Snd{a}{(x,y)}} \someframe_i\label{Ai-new}
\end{equation}
Indeed, $A_i \equiv \Res{x,y}(\Snd a {(x,y)} \parop \someframe_i)$ and
the variables $x,y$ are solvable in $\{\subst{(x,y)}{z}\} \parop \someframe_i$
because $\{\subst{(x,y}{z}\} \parop \someframe_i \equiv \{\subst{\Const{fst}(z)}{x}, \subst{\Const{snd}(z)}{y}\} \parop \Res{x,y}(\{\subst{(x,y)}{z}\} \parop \someframe_i)$, so we derive:
\begin{align}
\Snd a {(x,y)} &\ltr{\Snd a {(x,y)}} \nil \tag*{by $\brn{Out-Term}$}\\
\Snd a {(x,y)} \parop \someframe_i &\ltr{\Snd a {(x,y)}} \someframe_i \tag*{by $\brn{Par}$ and $\brn{Struct}$}\\
\Res{x,y}(\Snd a {(x,y)} \parop \someframe_i) &\ltr{\Res{x,y}\Snd a {(x,y)}} \someframe_i \tag*{by $\brn{Open-Var}$}\\
A_i &\ltr{\nu x,y.\Snd{a}{(x,y)}} \someframe_i \tag*{by $\brn{Struct}$}
\end{align}
Transition~\eqref{Ai-new} is the most informative for $A_1$ since $x$ and $y$ behave like fresh, 
independent values in $\someframe_1$.
For $A_2$, we also have the more informative transition: $$A_2 \ltr{\nu
  x.\Snd{a}{(x,\Const f(x))}} \nu k. \{\subst{k}{x}\}$$
that reveals
the link between $x$ and $y$, but not that $x$ is a name.
As in this example, several output transitions are sometimes possible,
each transition leading to an extended process with a different frame. 
In reasoning (for example, in proving that a relation is included in labelled bisimilarity), 
it often suffices to consider any one of the transitions, 
so one may be chosen so as to limit the complexity of the resulting extended processes.

We name ``simple semantics'' the labelled semantics of Section~\ref{subsec:labopsem} and ``refined semantics'' the semantics of this section, and ``simple labels'' and ``refined labels'' the corresponding labels.
The next theorem states that the two labelled semantics yield the same notion of
equivalence. Thus, making the output labels more explicit only makes apparent
some of the information that is otherwise kept in the static, equational part of
labelled bisimilarity. 

\begin{theorem}\label{thm:relate-lts} 
  Let $\altbisim$ be the relation of labelled
  bisimilarity obtained by applying Definition~\ref{def:wkbisim} to
  the refined semantics. We have ${\wkbisim} = {\altbisim}$.
\end{theorem}

The proof of Theorem~\ref{thm:relate-lts} relies on the next two 
lemmas, which relate simple and refined output transitions.

\begin{lemma}\label{lem:output-correspondence} 
  $A \ltr{ \nu\vect{x}.\Snd{N}{M}} A'$ if and only if, for some $z$ that does
  not occur in any of $A$, $A'$, $\vect{x}$, $N$, and $M$,
  $A \ltr{\nu z.\Snd{N}{z} } \nu\vect{x}. (\{\subst{M}{z}\} \parop A')$,
  $\{\vect x\} \subseteq \fv(M) \setminus \fv(N)$,
  and the variables $\vect{x}$ are solvable in $\{\subst{M}{z}\} \parop A'$. 
\end{lemma}

In Lemma~\ref{lem:output-correspondence}, 
the transition $A \ltr{ \nu\vect{x}.\Snd{N}{M}} A'$
is performed in the refined semantics,
while the transition $A \ltr{ \nu z.\Snd{N}{z} } \nu\vect{x}. (\{\subst{M}{z}\} \parop A')$
is performed in the simple semantics.
However, Lemma~\ref{lem:single-var-output-correspondence} below shows that the
choice of the semantics does not matter. 
Lemma~\ref{lem:single-var-output-correspondence} is a consequence of
Lemma~\ref{lem:output-correspondence}.

\begin{lemma}\label{lem:single-var-output-correspondence}
$A \ltr{ \nu x.\Snd{N}{x}} A'$ in the refined semantics 
if and only if 
$A \ltr{ \nu x.\Snd{N}{x}} A'$ in the simple semantics.
\end{lemma}

Theorem~\ref{thm:relate-lts} is then proved as follows. By Lemma~\ref{lem:single-var-output-correspondence}, 
$\altbisim$ is a simple-labelled bisimulation, and thus ${\altbisim} \subseteq {\wkbisim}$. 
Conversely, to show that $\wkbisim$ is a refined-labelled bisimulation,
it suffices to prove its bisimulation property for any refined output label.
This proof,
which relies on Lemma~\ref{lem:output-correspondence}, 
and the proofs of Lemmas~\ref{lem:not-so-new}, \ref{lem:output-correspondence},
and~\ref{lem:single-var-output-correspondence} are detailed in 
Appendix~\ref{app:refining}.

\subsection{Proving Theorem~\ref{THM:OBSERVATIONAL-LABELED} (${\bicong} = {\wkbisim}$)}\label{sec:bigpf}

A claim of Theorem~\ref{THM:OBSERVATIONAL-LABELED} appears, without
proof, in the conference version of this paper, for the calculus as
presented in that version.  There, the channels in labels cannot be variables.
The claim neglects to include a corresponding
hypothesis that exported variables must not be of channel type.  This
hypothesis is implicitly assumed, as it holds trivially for plain
processes and is maintained, as an invariant, by output transitions.
Without it, the two extended processes $\Res{a}(\{\subst{a}{x}\})$
and $\Res{a}(\{\subst{a}{x}\} \parop \Snd{a}{N})$ (where the exported
variable $x$ stands for the channel $a$) would constitute a counterexample:
they would not be observationally equivalent but they would be bisimilar in the
labelled semantics, since neither could make a labelled transition.
Delaune et al.~\cite{DelauneKremerRyan07,Delaune10}
included the hypothesis in their study of
symbolic bisimulation.
Avik Chaudhuri (private communication, 2007) pointed out
this gap in the statement of the theorem, and 
Bengtson et al.~\cite{bengtson:psi} discussed it as motivation for their work on alternative
calculi, the psi calculi, with a more abstract treatment of terms and a
mechanized metatheory.  
On the other hand, 
Liu~\cite{Liu11} presented a proof of the theorem, making
explicit the necessary hypothesis. Her proof demonstrated that the
theorem was basically right---no radical changes or new languages were
needed.  More recently, Liu and others have also developed an
extension of the proof for a stateful variant of the applied pi
calculus~\cite{Arapinis:post2014}.

Theorem~\ref{THM:OBSERVATIONAL-LABELED}, in its present form, does not
require that hypothesis because of some of the details of the calculus
as we define it in this paper. Specifically, the labelled semantics
allows variables that stand for channels in labels. Therefore, extended processes such as $\Res{a}(\{\subst{a}{x}\} \parop 
\Snd{a}{N})$ can make labelled transitions.

This section outlines the proof of
Theorem~\ref{THM:OBSERVATIONAL-LABELED}.  The appendix gives further
details, including all proofs that this section omits.  Those details
are fairly long and technical. In particular, they rely on a
definition of ``partial normal forms'' for extended processes, which
are designed to simplify reasoning about reductions. (In an extended
process $A \parop B$, the frame of $A$ may affect~$B$ and vice versa,
so $A$ and $B$ may not reduce independently of each other; partial
normal forms are designed to simplify the analysis of reductions in
such situations.)
We believe that these partial normal
forms may be useful in other proofs on the applied pi calculus.
In this section, we omit further specifics on
partial normal forms, since they are not essential to understanding
our main arguments. 

The proof of Theorem~\ref{THM:OBSERVATIONAL-LABELED} starts with a
fairly traditional definition of ``labelled bisimulation up to~$\equiv$'':

\begin{definition}\label{def:weakbisimstruct}
A relation~$\rel$ on closed extended proc\-esses is a \emph{labelled
bisimulation up to $\equiv$} if and only if $\rel$ is symmetric and
$A \rel B$ implies: 
\begin{enumerate} 

\item $A \enveq B$;\label{ppweakbisimstruct-one} 

\item if $A \rightarrow A'$ and $A'$
is closed, then $B \rightarrow^* B'$ and $A' \equiv \rel \equiv B'$
for some closed $B'$;\label{ppweakbisimstruct-two} 

\item if $A \ltr{\alpha}
A'$, $A'$ is closed, and
    $\fv(\alpha) \subseteq \dom(A)$,
    then $B \rightarrow^*
    \ltr{\alpha} \rightarrow^* B'$ and $A' \equiv \rel \equiv B'$ for some closed $B'$.\label{ppweakbisimstruct-three}

\end{enumerate}
\end{definition}

This definition implies that, if $\rel$ is a labelled bisimulation up to $\equiv$, 
then $\equiv \rel\equiv$ restricted to closed processes is
a labelled bisimulation (since, by Lemma~\ref{LEM:INVARIANT-STATIC-EQ}, 
static equivalence is invariant by structural equivalence).

We use the definition to establish the following lemma:
\begin{lemma}\label{lem:bisim-context-closed}
$\wkbisim$ is closed by application of closing evaluation contexts.
\end{lemma}
In the proof of this lemma (which is given in Appendix~\ref{app:onedirection}), we 
show that we can restrict attention to contexts of the form $\nuc{\hole}$.
  To every relation $\rel$ on closed extended processes, we
  associate a relation ${\rel'} = \{ (\nuc A, \nuc B) \mid A \rel B, \allowbreak \nuc{\hole} \text{ closing for $A$ and $B$} \}$.
  We prove that, if $\rel$ is a labelled bisimulation,
  then $\rel'$ is a labelled bisimulation up to $\equiv$, hence ${\rel}
  \subseteq {\equiv\rel'\equiv} \subseteq {\wkbisim}$.
  For ${\rel} = {\wkbisim}$, this property entails that $\wkbisim$ is
  closed by application of evaluation contexts $\nuc{\hole}$.

Another lemma characterizes barbs in terms of labelled transitions:
\begin{lemma}\label{lem:caract-barb}
Let $A$ be a closed extended process.
We have $A \barb{a}$ if and only if $A \rightarrow^* \ltr{\nu x.\Snd{a}{x}} A'$ for some fresh variable $x$ and some $A'$.
\end{lemma}

We then obtain Lemma~\ref{lem:th1-first-dir}, which is one direction of Theorem~\ref{THM:OBSERVATIONAL-LABELED}:
\begin{lemma} \label{lem:th1-first-dir}
${\wkbisim} \subseteq {\bicong}$.
\end{lemma}

\begin{proof} 
We show that $\wkbisim$ satisfies the three properties of Definition~\ref{def:bicong}, as follows.
\begin{enumerate}
\item To show that $\wkbisim$ preserves barbs, we apply Lemma~\ref{lem:caract-barb} and use Properties~\ref{pptwo} and~\ref{ppthree} of Definition~\ref{def:wkbisim}. 

\item Suppose that $A \wkbisim B$, $A \rightarrow^* A'$, and $A'$ is closed.
Given the trace $A = A_0 \rightarrow A_1 \rightarrow \ldots \rightarrow A_n = A'$, we instantiate all variables in $\bigcup_{i=0}^n (\fv(A_i)\setminus \dom(A_i))$ with fresh names. This instantiation yields a trace in which all intermediate processes are closed. We can then conclude that $B \rightarrow^* B'$ and $A' \wkbisim B'$ for some $B'$ by Property~\ref{pptwo} of Definition~\ref{def:wkbisim}. 

\item $\wkbisim$ is closed
  by application of closing evaluation contexts by Lemma~\ref{lem:bisim-context-closed}.
\end{enumerate}
Moreover, $\wkbisim$ is symmetric. Since $\bicong$ is the largest relation that satisfies these properties, we obtain ${\wkbisim} \subseteq {\bicong}$.
\end{proof}

The other direction of Theorem~\ref{THM:OBSERVATIONAL-LABELED} relies on two lemmas that characterize input and output transitions.
The first lemma characterizes inputs $\Rcv{N}{M}$ using processes of the
form $T^p_{\Rcv N M} \eqdef \Snd p p \parop \Snd N M.\Rcv p x$. Here, the use
of $p$ as a message in $\Snd p p$ is arbitrary: we could equally use processes
of the form $\Snd{p}{M'}$ for any term $M'$.

\begin{lemma}\label{lem:caract-RcvNM}
Let $A$ be a closed extended process.
Let $N$ and $M$ be terms such that $\fv(\Snd{N}{M}) \subseteq \dom(A)$. Let $p$ be a name that does not occur in $A$, $M$, and $N$.
\begin{enumerate}
\item\label{prop:caract-RcvNM-1} If $A \ltr{\Rcv{N}{M}} A'$ and $p$ does not occur in $A'$, 
then $A \parop T^p_{\Rcv N M} \rightarrow \rightarrow A'$ and $A' \not\barb{p}$.
\item\label{prop:caract-RcvNM-2} If $A \parop T^p_{\Rcv N M} \rightarrow^* A'$ and
$A' \not\barb{p}$, then
$A \rightarrow^* \ltr{\Rcv{N}{M}} \rightarrow^* A'$.
\end{enumerate}
\end{lemma}

The second lemma characterizes outputs $\Res{x}\Snd N x$ using processes
of the form $T^{p,q}_{\Res{x}\Snd N x} \eqdef\; \Snd p p \parop \Rcv{N}{x}. \Rcv p y.\Snd q x$.

\begin{lemma}\label{lem:caract-SndN}
Let $A$ be a closed extended process.
Let $N$ be a term such that $\fv(N) \subseteq \dom(A)$. Let $p$ and $q$ be names that do not occur in $A$ and $N$.
\begin{enumerate}
\item\label{prop:caract-SndN-1} If $A \ltr{\Res{x}\Snd{N}{x}} A'$ and $p$ and $q$ do not occur in $A'$, 
then $A \parop T^{p,q}_{\Res{x}\Snd N x}  \rightarrow\rightarrow \Res{x}(A' \parop \Snd{q}{x})$,
$\Res{x}(A' \parop \Snd{q}{x}) \not\barb{p}$, and $x \notin \dom(A)$.
\item\label{prop:caract-SndN-2} Let $x$ be a variable such that $x \notin \dom(A)$.
If $A \parop T^{p,q}_{\Res{x}\Snd N x} \rightarrow^* A''$ and $A'' \not\barb{p}$, then
$A \rightarrow^*\allowbreak \ltr{\Res{x}\Snd{N}{x}}\allowbreak \rightarrow^* A'$
and $A'' \equiv \Res{x}(A' \parop \Snd{q}{x})$ for some $A'$.
\end{enumerate}
\end{lemma}

A further lemma provides a way of proving the equivalence of two extended processes with the same domain by putting them in a context that binds the variables in their domain and extrudes them.
  Given a family of processes $P_i$ for $i$ in a finite set $I$, we
  write $\prod_i P_i$ for the parallel composition of the processes
  $P_i$ if $I$ is not empty, and for $\nil$ otherwise.

\begin{lemma}\label{lem:extrusion}
Let $A$ and $B$ be two closed extended processes with a same domain that contains $\vect{x}$.
Let $\CTX_{\vect x}[\hole] \eqdef \nu \vect{x}.( \prod_{x \in \vect{x}}\Snd {n_x} x \parop \hole\,)$
using names $n_x$ that do not occur in $A$ or $B$.
If $\CTX_{\vect x}[A] \bicong \CTX_{\vect x}[B]$, then $A \bicong B$.
\end{lemma}

The final lemma is the other direction of Theorem~\ref{THM:OBSERVATIONAL-LABELED}:

\begin{lemma}\label{lem:th1-second-dir}
  $\bicong$ is a labelled bisimulation, and thus ${\bicong}\subseteq {\wkbisim}$.
\end{lemma}

\begin{proof}
The relation $\bicong$ is symmetric. We show that it satisfies the three properties of Definition~\ref{def:wkbisim}.
\begin{enumerate}
\item If $A \bicong B$, then $A \enveq B$, by Lemma~\ref{lem:enveq2}.
\item If $A \bicong B$, $A \rightarrow A'$, and $A'$ is closed, then $B
    \rightarrow^* B'$ and $A' \bicong B'$ for some $B'$, by Property~\ref{bctwo} of the definition of $\bicong$.
\item If $A \bicong B$, $A \ltr{\alpha} A'$, $A'$ is closed, and 
    $\fv(\alpha) \subseteq \dom(A)$,
    then $B \rightarrow^*
    \ltr{\alpha} \rightarrow^* B'$ and $A' \bicong B'$ for some $B'$.
To prove this property, we rely on characteristic contexts $\hole \parop T_\alpha$ that unambiguously test
for a labelled transition $\ltr{\alpha}$ using the disappearance of a
barb $\barb{p}$, and do not otherwise affect $\bicong$.

Assume $A \bicong B$, $A \ltr{\alpha} A'$, $A'$ is closed, and 
    $\fv(\alpha) \subseteq \dom(A)$.
\begin{enumerate}
\item For input $\alpha = \Rcv{N}{M}$ (where $N$ and $M$ may contain variables exported by $A$ and $B$)
  and some fresh name $p$,  
  we have $A \parop T^p_{\Rcv N M} \rightarrow \rightarrow A' \not\barb{p}$
  by Lemma~\ref{lem:caract-RcvNM}\eqref{prop:caract-RcvNM-1}, 
  hence $B \parop T^p_{\Rcv N M} \rightarrow^* B' \not\barb{p}$ with $A'
  \bicong B'$, hence $B \rightarrow^* \ltr{\Rcv N M} \rightarrow^* B'$
  by Lemma~\ref{lem:caract-RcvNM}\eqref{prop:caract-RcvNM-2}.

\item For output $\alpha = \Res{x}\Snd{N}{x}$ and some fresh names $p$ and $q$, 
we have $A \parop T^{p,q}_{\Res{x}\Snd N x} \rightarrow \rightarrow \Res{x}(A' \parop \Snd{q}{x}) \not\barb{p}$ and $x \notin \dom(A)$
by Lemma~\ref{lem:caract-SndN}\eqref{prop:caract-SndN-1}, 
hence $B \parop T^{p,q}_{\Res{x}\Snd N x} \rightarrow^* B'' \not\barb{p}$ for some $B''$,
hence $B \rightarrow^* \ltr{\Res{x}\Snd N x} \rightarrow^* B'$ and $B'' \equiv \Res{x}(B' \parop \Snd{q}{x})$ for some $B'$
by Lemma~\ref{lem:caract-SndN}\eqref{prop:caract-SndN-2}.
We obtain a pair 
$\nu x.(A' \parop \Snd q x) \bicong \nu x.(B' \parop \Snd q x)$, 
and conclude by applying Lemma~\ref{lem:extrusion}.

\end{enumerate}
\end{enumerate}
Hence $\bicong$ is a labelled bisimulation, and ${\bicong} \subseteq {\wkbisim}$, since $\wkbisim$ is the largest labelled bisimulation.
\end{proof}

Theorem~\ref{THM:OBSERVATIONAL-LABELED} is an immediate consequence of Lemmas~\ref{lem:th1-first-dir} and~\ref{lem:th1-second-dir}.

Considering this proof of Theorem~\ref{THM:OBSERVATIONAL-LABELED}, we can explain further some aspects of 
our definition of observational equivalence (Definition~\ref{def:bicong}). That definition includes
conditions related to barbs, reductions, and evaluation contexts
(Conditions~\eqref{bcone} to~\eqref{bcthree}, respectively), 
as is done in work on the $\nu$-calculus~\cite{Honda95} and on the join calculus~\cite{AbadiFournetGonthier}.
In an alternative approach, used in CCS~\cite{Milner1992} and in the pi calculus~\cite{Sangiorgi93}, equivalence is defined in two stages:
\begin{enumerate}
\item First, barbed bisimilarity is defined as the largest barbed bisimulation, that is, the largest symmetric relation $\rel$ such that $A \rel B$ implies Conditions~\eqref{bcone} and~\eqref{bctwo} of Definition~\ref{def:bicong}. 
\item Second, equivalence is defined as the largest congruence (that is, the largest relation $\rel$ such that $A \rel B$ implies Condition~\eqref{bcthree} of Definition~\ref{def:bicong}) contained in barbed bisimilarity. 
\end{enumerate}
The two approaches do not necessarily yield the same equivalence relation; see \cite{FournetGonthier98:equivalences} for positive and negative examples in variants of the pi calculus. 
The advantage of our approach is that, in reasoning about process equivalences, we can add a context at any point after reductions, as we do in the proof of Lemma~\ref{lem:th1-second-dir}. With the alternative approach, we can add a context only at the beginning, before any reduction, so we need to build contexts that test for all possible sequences of
labelled transitions that the processes under consideration may make, and that manifest them as different combinations of barbs. This testing is not possible for all processes, so with the alternative approach, analogues of Theorem~\ref{THM:OBSERVATIONAL-LABELED} would typically require a restriction to so-called \emph{image finite} processes~\cite{Milner1992}. 
Our definition of observational equivalence avoids this restriction.

It would be interesting to formalize the proofs of this section (and
also those of the rest of the paper) with a theorem prover such as
Coq. This formalization may perhaps benefit from past Coq developments
on bisimulations for the pi calculus~\cite{Hirschkoff97,honsell2001pi}
and the spi calculus~\cite{Briais2008}. However, the applied pi
calculus introduces additional difficulties (because of the role of terms with
equational theories), and proving our results with Coq would certainly
require a major effort.

\section{Diffie-Hellman Key Agreement}\label{sec:diffie-hellman}

\newcommand{\Pair}[1]{\Const {pair} #1} 
\newcommand{\Crypt}[1]{\Const {enc} #1}
\newcommand{\F}[2]{\Const {f} (#1,#2) }
\newcommand{\Gf}[1]{\Const {g}(#1)}
\newcommand{\SH}[1]{\Const {h}(#1)}

The fundamental Diffie-Hellman protocol allows two principals to
establish a shared secret by exchanging messages over public channels~\cite{DH}.
The principals need not have any shared secrets in advance.  The basic
protocol, on which we focus here as an example, does not provide
authentication; therefore, a ``bad'' principal may play the role of
either principal in the protocol. On the other hand, the two
principals that follow the protocol will communicate securely with one
another afterwards, even in the presence of active attackers. In
extended protocols, such as the Station-to-Station protocol~\cite{STS}
and SKEME~\cite{Krawczyk:SKEME}, additional messages perform authentication.

We program the basic protocol in terms of the binary function symbol $\Const{f}$
and the unary function symbol $\Const{g}$, with the equation:
\begin{eqnarray}
\F x {\Gf y} & = & \F y {\Gf x} \label{pp:dh}
\end{eqnarray}
Concretely, the functions are 
$\F x y = y^x\ mod\ p$ and 
$\Gf x = \alpha^x\ mod\ p$ for a prime $p$ and a generator $\alpha$ of $\mathbb{Z}^*_p$, 
and we have the equation 
$\F x {\Gf y} = (\alpha^y)^x = \alpha^{y \times x}
= \alpha^{x \times y} = (\alpha^x)^y = \F y {\Gf x}$.
However, we ignore the underlying number theory,
working abstractly with $\Const{f}$ and $\Const{g}$.

The protocol has two symmetric participants, which we represent by the processes
$A_0$ and $A_1$. 
The protocol establishes a shared key, then the
participants  respectively run $P_0$ and $P_1$ using the key.
We use the public channel $c_{01}$ for messages from $A_0$ to $A_1$
and the public channel $c_{10}$ for communication in the opposite direction.
(Although the use of two distinct public channels is of no value for
security, it avoids some trivial confusions, so makes for a cleaner
presentation.)
We assume that none of the values introduced in the protocol appears in
$P_0$ and $P_1$, except for the key.

In order to establish the key, $A_0$ invents a name $n_0$, sends $\Gf
{n_0}$ to $A_1$, and $A_1$ proceeds symmetrically. Then $A_0$ computes
the key as $\F {n_0}{\Gf {n_1}}$ and $A_1$ computes it as $\F {n_1}{\Gf
{n_0}}$, with the same result. 
We find it convenient to use the
following substitutions for $A_0$'s message and key:
\eqns{
  \sigma_0  & \eqdef & \{ \subst{\Gf {n_0}}{x_0} \} \\
  \phi_0   & \eqdef & \{ \subst{\F {n_0}{x_1}}{y} \}
  }%
and the corresponding substitutions $\sigma_1$ and $\phi_1$,
as well as the frame:
\eqns{
  \varphi & \eqdef & (\nu n_0 .\; (\phi_0 \parop \sigma_0)) \parop (\nu n_1.\;\sigma_1)
  }%
With these notations, $A_0$ is:
\newcommand{\sndrcv}[2]{\Snd {c_{01}}{#1} \parop {\Rcv {c_{10}}{#2}}.}
\newcommand{\altsndrcv}[2]{\Snd {c_{10}}{#1} \parop {\Rcv {c_{01}}{#2}}.}
\eqns{
  A_0   & \eqdef & \nu n_0. (
  \sndrcv {x_0 \sigma_0} {x_1} P_0\phi_0)
  }%
and $A_1$ is analogous.

Two reductions represent a normal run of the
protocol:
\begin{eqnarray}
A_0 \parop A_1 & \rightarrow \rightarrow & 
\nu x_0, x_1, n_0, n_1.\;
(P_0\phi_0 \parop P_1\phi_1 \parop \sigma_0 \parop \sigma_1) \label{op:comms}
\\
& \equiv & 
\nu x_0, x_1, n_0, n_1, y.\;
(P_0 \parop P_1 \parop \phi_0 \parop \sigma_0 \parop \sigma_1)\quad \label{op:dh}
\\
& \equiv & 
\nu y.
(P_0 \parop P_1 \parop \nu x_0,x_1.\;\varphi) \label{op:push}
\end{eqnarray}
The two communication steps (\ref{op:comms}) use 
structural equivalence to activate the substitutions $\sigma_0$ and
$\sigma_1$ and extend the scope of the secret values $n_0$ and $n_1$.
The structural equivalence (\ref{op:dh}) crucially relies on equation~(\ref{pp:dh}) 
in order to reuse the active substitution $\phi_0$ instead of $\phi_1$ after
the reception of $x_0$ in $A_1$.
The next structural equivalence (\ref{op:push}) tightens the scope for
restricted
names and variables, then uses the definition of~$\varphi$.

We model an eavesdropper as a process 
$\Rcv{c_{01}}{x_0}.\Snd{c_{01}}{x_0}.\Rcv{c_{10}}{x_1}.\Snd{c_{10}}{x_1}.P$
that intercepts messages on $c_{01}$ and~$c_{10}$,
remembers them, but forwards them unmodified. 
Using the labelled semantics to represent
the interaction of $A_0 \parop A_1$ with 
such a passive attacker, we obtain:
\begin{eqnarray*}
A_0 \parop A_1 & \ltr{\Snd{c_{01}}{x_0}} & \Res{n_0}(\sigma_0 \parop \Rcv{c_{10}}{x_1}.P_0\phi_0) \parop A_1 \\
&\ltr{\Rcv{c_{01}}{x_0}} & \Res{n_0}(\sigma_0 \parop \Rcv{c_{10}}{x_1}.P_0\phi_0) \parop 
\Res{n_1}(\Snd{c_{10}}{\sigma_1x_1} \parop P_1\phi_1) \\
& \ltr{\Snd{c_{10}}{x_1}} & \Res{n_0}(\sigma_0 \parop \Rcv{c_{10}}{x_1}.P_0\phi_0) \parop 
\Res{n_1}(\sigma_1 \parop P_1\phi_1) \\
&\ltr{\Rcv{c_{10}}{x_1}} & \Res{n_0}(\sigma_0 \parop P_0\phi_0) \parop 
\Res{n_1}(\sigma_1 \parop P_1\phi_1) \\
& \equiv & 
\nu n_0, n_1, y.\;
(P_0 \parop P_1 \parop \phi_0 \parop \sigma_0 \parop \sigma_1)
\\
& \equiv & 
\nu y.
(P_0 \parop P_1 \parop \varphi) 
\end{eqnarray*}
The labelled transitions $\ltr{\Snd{c_{01}}{x_0}} \ltr{\Rcv{c_{01}}{x_0}}$ show
that the eavesdropper obtains the message sent on $c_{01}$ by $A_0$, stores it in $x_0$,
and forwards it to $A_1$. The transitions $\ltr{\Snd{c_{10}}{x_1}} \ltr{\Rcv{c_{10}}{x_1}}$
deal with the message on~$c_{10}$ in a similar way.
The absence of the restrictions on $x_0$ and $x_1$
corresponds to the fact that the eavesdropper has obtained the values
of these variables.

The following theorem relates this process to $$\nu k. (P_0 \parop
P_1)\{\subst{k}{y}\}$$ which represents the bodies $P_0$ and $P_1$ of
$A_0$ and $A_1$ sharing a key $k$.  This key appears as a simple
shared name, rather than as the result of communication and
computation. Intuitively, we may read $\nu k. (P_0 \parop
P_1)\{\subst{k}{y}\}$ as the ideal outcome of the protocol: $P_0$ and
$P_1$ execute using a shared key, without concern for how the key was
established, and without any side-effects from weaknesses in the
establishment of the key.  The theorem says that this ideal outcome is
essentially achieved, up to some ``noise''. This ``noise'' is a
substitution that maps $x_0$ and $x_1$ to unrelated, fresh names. 
It accounts for the fact that an attacker may have the key-exchange
messages, and that they look just like unrelated values to the
attacker.
In particular, the key in use between $P_0$ and $P_1$ has no
observable relation to those messages, or to any other left-over
secrets.
We view this independence of the shared key as an important forward-secrecy property.

\begin{theorem} Let $P_0$ and $P_1$ be processes with free
  variable $y$ where the name $k$ does not appear.
We have:
$$\begin{array}{l}
  \nu y. (P_0 \parop P_1 \parop \varphi) 
\;\bicong\; 
  \nu k. (P_0 \parop P_1)\{\subst{k}{y}\} \parop 
  \nu s_0.\{\subst{s_0}{x_0}\} \parop
  \nu s_1.\{\subst{s_1}{x_1}\} 
\end{array}$$
\end{theorem}
\begin{proof}
The theorem follows from Lemma~\ref{lem:enveq1} and the static equivalence
$\varphi \enveq \nu s_0,s_1,k.\{\subst{s_0}{x_0},\allowbreak\subst{s_1}{x_1},\allowbreak \subst{k}{y}\}$, which
says that the frame $\varphi$ generated by the protocol execution 
is equivalent to one that maps variables to fresh names.
This static equivalence is proved automatically by ProVerif,
using the technique presented in~\cite{Blanchet07b}.
We conclude by applying the context $\nu y. (P_0 \parop P_1 \parop \hole)$.
\end{proof}

Extensions of the basic protocol add rounds of communication that
confirm the key and authenticate the principals. We have studied one
such extension with key confirmation. There, the shared secret $\F
{n_0}{\Gf {n_1}}$ is used in confirmation messages. 
Because of these messages, the shared secret can no longer be equated with a virgin
key for $P_0$ and~$P_1$. Instead, the final key is
computed by hashing the shared secret. This hashing guarantees the
independence of the final key.

We have also studied more advanced protocols that rely on a
Diffie-Hellman key exchange, such as the JFK protocol~\cite{jfk-tissec}.
The analysis of JFK in the applied pi calculus~\cite{abf06-jfk}
illustrates the composition of manual reasoning with invocations of
ProVerif.

\section{Hash Functions and Message Authentication Codes}
\label{sec:extension-attacks}

Section~\ref{sec:sample} briefly discusses cryptographic hash functions. In this section we continue their study, and also treat 
message authentication codes (MACs). We consider constructions of both hash functions and MACs.
These examples provide a further illustration
of the usefulness of equations in the applied pi calculus. On the other
hand, some aspects of the constructions are rather low-level, and we
would not expect to account for all their combinatorial details (e.g.,
the ``birthday attacks''~\cite{Menezes96}). A higher-level task is to express and
reason about protocols treating hash functions and MACs as primitive; this is squarely
within the scope of our approach.

\subsection{Using MACs}\label{sec:usingMACs}

MACs serve to authenticate messages using shared keys.  
When $k$ is a key and $M$ is a message, and $k$ is known only to 
a certain principal $A$ and to the recipient~$B$ of the message,
$B$ may take $\mac{k}{M}$ as proof that $M$ comes from~$A$. More precisely, 
$B$ can check $\mac{k}{M}$ by recomputing it
upon receipt of $M$ and $\mac{k}{M}$, 
and reason that $A$ must be the sender of~$M$. 
This property should hold even if $A$ generates MACs for other messages as well; 
those MACs should not permit forging a MAC for $M$. In the worst case, it
should hold even if $A$ generates MACs for other messages on demand. 

Using a new binary function symbol $\Const{mac}$, we may describe this scenario 
by the following processes:
\eqns{
A & \eqdef & \Repl {\Rcv{a}{x}.\Snd{b}{(x,\mac k x)}} \\
B & \eqdef & \Rcv{b}{y}.\IfThen{\mac k {\Const{fst}(y)}}{\Const{snd}(y)}{\Snd{c}{\Const{fst}(y)}}\\
S & \eqdef & \nu k.(A \parop B)
}%
The process $S$ represents the complete system, composed of $A$ and $B$; 
the restriction on $k$ means that $k$ is private to $A$ and $B$. 
The process $A$ receives messages on a public channel $a$ and returns them MACed
on the public channel $b$. When $B$ receives a message on $b$, it checks its MAC and  
acts upon it, here simply by forwarding on a channel $c$. 
Intuitively, we would expect that 
$B$ forwards on~$c$ only a message that $A$ has MACed. In other words,
although an attacker may intercept, modify, and inject messages on $b$,
it should not be able to forge a MAC and trick $B$ into forwarding some other message.
Hence, 
every message output on $c$ equals a
  preceding input on $a$, as illustrated in Figure~\ref{fig:goodtrace}.

\begin{figure*}[t]
\eqns{
\nu k.(A \parop B) & \ltr{{\Rcv{a}{M}}}  & \nu k.(A \parop B \parop \Snd{b}{(M,\mac k M)})\\
   & \ltr{{\nu x.\Snd{b}{x}}}  & \nu k.(A \parop B \parop \{\subst{(M,\mac k M)}{x}\})\\
   & \ltr{\Rcv{b}{x}}\rightarrow &  
   \nu k.(A \parop \Snd{c}{M} \parop \{\subst{(M,\mac k M)}{x}\})\\
   & \ltr{\nu y.\Snd{c}{y}}& \nu k.(A \parop \{\subst{(M,\mac k M)}{x}, \subst{M}{y}\})
}%

\caption{A correct trace}\label{fig:goodtrace}
\end{figure*}

This property can be expressed precisely in terms of the labelled semantics
and it can be checked without too much difficulty
when $\Const{mac}$ is a primitive function symbol with no equations. 
The property remains true even if there is a function $\Const{extract}$ that
maps a MAC $\Const{mac}(x,y)$ to the underlying cleartext~$y$,
with the equation $\Const{extract}(\Const{mac}(x,y)) = y$. Since MACs are
not supposed to guarantee secrecy, such a function may well exist,
so it is safer to assume that it is available to the attacker.

The property is more delicate if $\Const{mac}$ is defined from other
operations, as it invariably is in practice.  In that case, the property
may even be taken as \emph{the} specification of MACs~\cite{GoldwasserBellare}. Thus, a MAC
implementation may be deemed correct if and only if the process $S$ works as
expected when $\Const{mac}$ is instantiated with that implementation.
More specifically, the next section deals with the question of whether
the property remains true when $\Const{mac}$ is defined from hash
functions.

\subsection{Constructing Hash Functions and MACs}\label{SUBSEC:CONSTRUCTING}

In Section~\ref{sec:sample}, we give no equations for hash
functions. In practice, following Merkle and Damg{\r{a}}rd, hash functions are commonly defined by
iterating a basic binary compression function, which maps two input
blocks to one output block~\cite{Menezes96}.  Furthermore, keyed hash functions
include a key as an additional argument. Thus, we may have:
\begin{eqnarray}
\hbin{x}{\cons{y_0}{\cons{y_1}{z}}} &=& \hbin{\fbin{x}{y_0}}{\cons{y_1}{z}}\label{eq:h-iter}\\
\hbin{x}{\cons{y}{\Const{nil}}} &=& \fbin{x}{y}\label{eq:h-end}
\end{eqnarray} 
Here, we use the sorts $\Block$ for blocks and $\BlockList$ for sequences of blocks,  defined as lists as in Section~\ref{sec:sample},
with sorts $\Block$ and $\BlockList$ instead of $\Data$ and $\DataList$, respectively.
The function $\Const{h} :  \Block \times \BlockList \rightarrow \Block$ is the keyed hash function, 
$\Const{f} : \Block \times \Block \rightarrow \Block$ is the compression function.

In these equations we are rather abstract in our treatment of
blocks, their sizes, and therefore of padding and other related
issues.  We also ignore two common twists: some
functions use initialization vectors to start the iteration, and
some append a length block to the input.  Nevertheless, we can explain
various MAC constructions, describing flaws in some and reasoning
about the properties of others.

A first, classical definition of a MAC from a keyed hash
function $\Const{h}$ is:
\eqns{
\mac{x}{y} & \eqdef & \hbin{x}{y}
}%
For instance, the MAC of a three-block message $M = \cons{M_1}{\cons{M_2}{\cons{M_3}{\Const{nil}}}}$ with key~$k$ is $\mac{k}{M} = \fbin{\fbin{\fbin{k}{M_1}}{M_2}}{M_3}$.
More generally, the MAC of a $n$-block message $M = \cons{M_1}{\cons{\dots}{\cons{M_{n}}{\Const{nil}}}}$ is $\mac{k}{M} = \Const f(\dots(\fbin{k}{M_1},\ldots),\allowbreak M_{n})$.
This implementation is subject to a well-known extension attack.
Given the MAC of $M = \cons{M_1}{\cons{\dots}{\cons{M_{n}}{\Const{nil}}}}$, an attacker can compute the MAC of 
any extension $M\bappend N = \cons{M_1}{\cons{\dots}{\cons{M_{n}}{\cons{N}{\Const{nil}}}}}$ without knowing the MAC key, since
$\mac{k}{M\bappend N}  = \fbin{\mac{k}{M}}{N}$.

We describe the extension attack formally, through the operational semantics of
the process $S$ of Section~\ref{sec:usingMACs},
in Figures~\ref{fig:attack} and~\ref{fig:attackbis}. These figures use the semantics of Sections~\ref{subsec:labopsem} and~\ref{SUBSEC:REFINING}
respectively. In both cases, we assume $k \not\in\fn(M)\cup\fn(N)$.
Additionally, we adopt the sorts
  $\Const{pair}: \BlockList \times \Block \rightarrow \Const{Data}$,
  $\Const{fst}: \Const{Data}\rightarrow \BlockList$,
and $\Const{snd}: \Const{Data}\rightarrow \Block$,
the abbreviation $(M,N)$ for $\Const{pair}(M,N)$, and the equations~\eqref{eq:fst} and~\eqref{eq:snd} of Section~\ref{sec:sample}.
In Figures~\ref{fig:attack} and~\ref{fig:attackbis}, we see that the message $M$ that the system MACs differs 
from the message $M\bappend N$ that it forwards on~$c$.
These transitions are not enabled with the primitive MAC of Section~\ref{sec:usingMACs},
hence $S$ with the proposed
MAC implementation is not labelled bisimilar to $S$ with the primitive MAC.

\begin{figure*}[t]
\eqns{
\nu k.(A \parop B) & \ltr{{\Rcv{a}{M}}}  & \nu k.(A \parop B \parop \Snd{b}{(M,\mac k M)})\\
   & \ltr{{\nu x.\Snd{b}{x}}}  & \nu k.(A \parop B \parop \{\subst{(M,\mac k M)}{x}\})\\
   & \ltr{\Rcv{b}{(M\bappend N,\fbin{\Const{snd}(x)}{N})}}\rightarrow &  
   \nu k.(A \parop \Snd{c}{M\bappend N} \parop \{\subst{(M,\mac k M)}{x}\})\\
   & \ltr{\nu y.\Snd{c}{y}}& \nu k.(A \parop \{\subst{(M,\mac k M)}{x}, \subst{M\bappend N}{y}\})
}%
\caption{An attack scenario}\label{fig:attack}
\end{figure*}

\begin{figure*}
\eqns{
\nu k.(A \parop B) & \ltr{{\Rcv{a}{M}}}  & \nu k.(A \parop B \parop \Snd{b}{(M,\mac k M)})\\
   & \ltr{{\nu x.\Snd{b}{(M,x)}}}  & \nu k.(A \parop B \parop \{\subst{\mac k M}{x}\})\\
   & \ltr{\Rcv{b}{(M\bappend N,\fbin{x}{N})}}\rightarrow &  
   \nu k.(A \parop \Snd{c}{M\bappend N} \parop \{\subst{\mac k M}{x}\})\\
   & \ltr{\Snd{c}{M\bappend N}}& \nu k.(A \parop \{\subst{\mac k M}{x}\})
}%
\caption{An attack scenario (with refined labels)}\label{fig:attackbis}
\end{figure*}

There are several ways to address extension attacks, and indeed the
literature contains many MAC constructions that are not subject to
these attacks. We have considered some of them.
Here we describe a construction that uses the MAC key twice:
\eqns{
\mac x y       & \eqdef & \fbin{x}{\hbin{x}{y}}
}%
Under this definition, the MAC of $M = \cons{M_1}{\cons{M_2}{\cons{M_3}{\Const{nil}}}}$ with key~$k$ is $\mac{k}{M} = \Const{f}(k,\allowbreak\Const{f}(\Const{f}(\Const{f}(k,\allowbreak M_1), \allowbreak M_2), \allowbreak M_3))$, and
the process $S$ forwards
on $c$ only a message that it has previously MACed, as desired.

Looking beyond the case of $S$, 
we can prove a more general result by comparing 
the situation where $\Const{mac}$ is primitive (and has no special equations)
and one with the definition of $\mac x y$ as $\fbin{x}{\hbin{x}{y}}$.
Given a name $k$ and an extended process $C$ that uses 
the symbol $\Const{mac}$, we write
$\tr{C}$ for the translation of $C$ in which the definition of $\Const{mac}$
is expanded wherever the key $k$ is used, 
with $\fbin {k}{\hbin{k}{M}}$ replaced for $\mac {k} M$.
The theorem says that this translation yields an equivalent process
(so, intuitively, the constructed MACs work as well as the primitive ones).
It applies to a class of equational theories generated by rewrite rules.

\begin{theorem} \label{th:mac}
Suppose that the signature $\Sigma$ is equipped with 
an equational theory generated by a convergent rewrite system such that
$\Const{mac}$ and $\Const{f}$ do not occur in the left-hand sides of rewrite rules;
the only rewrite rules with $\Const{h}$ at the root of the left-hand side are
those of~\eqref{eq:h-iter} and~\eqref{eq:h-end} oriented from left to right;
there are no rewrite rules with $::$ nor $\Const{nil}$ at the root of the left-hand side;
and names do not occur in rewrite rules.
  Suppose that $C$ is closed and the name $k$ appears only as first argument of $\Const{mac}$ in $C$. 
Then $\nu k. C \bicong \nu k.\tr{C}$.
\end{theorem}

In the proof of this theorem (which is given in Appendix~\ref{app:constructing}), 
we use the same notion of partial normal form as in the proof
of Theorem~\ref{THM:OBSERVATIONAL-LABELED}. We define a relation $\rel$ by 
$A \rel B$ if and only if
$A$ and $B$ are closed,
$A \equiv \Res{k}C$, $B \equiv \Res{k}\tr{C}$,
$C$ is a closed extended process in partial normal form, 
and the name $k$ appears only as MAC key in $C$.
We show that the relation $\rel \cup \rel^{-1}$ (that is, the union of $\rel$ with its inverse relation) is a labelled bisimulation.
Static equivalence follows from the preservation of equality 
by the translation $\tr{\cdot}$ for terms in which $k$ occurs 
only as MAC key;
reductions commute with the translation $\tr{\cdot}$ and
  preserve the restriction on the occurrences of the key~$k$.
We conclude by Theorem~\ref{THM:OBSERVATIONAL-LABELED}.
An alternative proof of similar complexity would show that
$\rel \cup \rel^{-1}$ is an observational bisimulation.

  Theorem~\ref{th:mac} considers a single MAC key at a time. For an extended
  process with several MAC keys $k_1, \dots, k_n$, we can apply
  Theorem~\ref{th:mac} once for each key $k_i$, using structural equivalence to
  move each restriction $\nu k_i$ to the root of the extended process.

  Theorem~\ref{th:mac} allows cryptographic primitives other than hash functions and
MACs, provided the assumptions on the equational theory are satisfied.
The following corollary states a simple special case for the primitives
mentioned in this section. It suffices for treating the system $S$.

\begin{corollary}\label{cor:mac}
Suppose that the signature $\Sigma$ is equipped with the equational theory
defined by the equations~\eqref{eq:fst}, \eqref{eq:snd}, \eqref{eq:seq}, \eqref{eq:++}, \eqref{eq:h-iter}, and
\eqref{eq:h-end}.
Suppose that $C$ is closed and the name $k$ appears only as first argument of $\Const{mac}$ in $C$. 
Then $\nu k. C \bicong \nu k.\tr{C}$.
\end{corollary}

\subsection{Constructing Robust Hash Functions}\label{SUBSEC:INDIF}

Constructions of hash functions, of the kind described in
Section~\ref{SUBSEC:CONSTRUCTING}, typically impose constraints on the
use of these functions. For example, some care is needed in order to
thwart extension attacks in the definition of MACs. The possibility of
such attacks stems from structural flaws in the constructions; details
such as the iteration of a compression function are not completely
hidden, lead to unwanted additional properties, and can be exploited.

A line of work in cryptography studies safer hash functions with
stronger guarantees~\cite{Coron}. 
Although these functions are
generally built much as in Section~\ref{SUBSEC:CONSTRUCTING} by
iterating a compression function, their design conceals their inner
structure. The functions thus aim to behave like abstract ``random oracles''
on inputs of arbitrary length. A notion of indifferentiability captures this goal.

In this section, as a final, more advanced example, we describe one
design that strengthens the Merkle-Damg{\r{a}}rd approach, following Coron et al.~\cite[Section 3.4]{Coron}. 
In this example, the attacker is given only indirect access to functions
such as the hash function $\Const{h}$. We model this restriction
by inserting a private name $k$ as the first argument of $\Const{h}$.
(Cryptographically, the name $k$ may reflect the initial random sampling
of $\Const{h}$.)
We refer to this argument as a key, of sort $\Key$.
We use sorts $\Block$ for blocks and $\BlockList$ for sequences of blocks, defined as lists as in Section~\ref{sec:sample},
with sorts $\Block$ and $\BlockList$ instead of $\Data$ and $\DataList$, respectively.
We use sort $\BlockPair$ for pairs of blocks, with
$\Const{pair} : \Block \times \Block \rightarrow \BlockPair$,
$\Const{fst} : \BlockPair \rightarrow \Block$,
and
$\Const{snd} : \BlockPair \rightarrow \Block$,
the abbreviation $(x,y)$ for $\Const{pair}(x,y)$,
and the equations
\begin{equation}
\Const{fst}((x,y)) = x \qquad \Const{snd}((x,y)) = x
\qquad (\Const{fst}(x),\Const{snd}(x)) = x
\label{eq:pairs-all}
\end{equation}
The third equation of~\eqref{eq:pairs-all} is not present in Section~\ref{sec:sample}; 
it models that all elements of sort $\BlockPair$ are pairs.
We use sort $\BlockT$ for pairs of a $\BlockPair$ and a $\Block$
defined in the same way with overloaded function symbols $\Const{pair}$,
$\Const{fst}$, and $\Const{snd}$, and sort $\Bool$ for booleans.

We define the hash function $\Const{h} : \Key \times \BlockList \rightarrow \Block$ by:
\begin{align}
\Const{h}(k,z) & {} = \Const{h}_2(k,(0,0),z)\label{eq:h}
\end{align}
where
\begin{align}
\Const{h}_2(k,x,\Const{nil}) & {} = \Const{fst}(x) \label{eq:h2nil}\\
\Const{h}_2(k,x,\cons{y}{z}) & {} = \Const{h}_2(k,\Const{f}(k,(x,y)),z) \label{eq:h2cons}
\end{align}
The function $\Const{h}_2  : \Key \times \BlockPair \times \BlockList \rightarrow \Block$ uses a compression function $\Const{f} : \Key \times \BlockT \rightarrow \BlockPair$. 
In $\Const{h}_2(k,x,z)$, the variable $x$ represents the fixed-size
internal state of the hash function and $z$ is the remainder of the
input. The internal state starts at $(0,0)$ and is updated by
applications of the compression function $\Const{f}$ as input blocks are
processed. Finally, only the first half of the internal state is returned.

For instance, 
the hash of a two-block message $M = \cons{M_1}{\cons{M_2}{\Const{nil}}}$ with key~$k$ is 
$\Const{h}(k,M)
 = \Const{fst}(\Const f 
 ( k, (\Const f(k, ((0,0), M_1)),M_2)))$.
More generally, we have
\[\Const h(k, \cons{M_1}{\cons{\dots}{ \cons{M_n}{\Const{nil}}}}) 
= \Const{fst}( \Const f(k,(\dots \Const f(k,((0,0),M_1))\dots,M_n)))\]

Indifferentiability requires that the hash function behave like a
black box (like a ``random oracle''), even in interaction with an
adversary that also has access to the underlying compression function.
The compression function and the hash function are related, of course.
However, as far as the adversary can tell, it is the compression function
that may be defined from the hash function (in fact, from an ideal hash function without equations as in Section~\ref{sec:sample}) rather than the other way
around. Thus, we express indifferentiability as the
equivalence of two systems, each of which provides access to the hash
function and the compression function. In the applied pi calculus, one of the systems is:
\[\Res{k}{(A^0_h \parop A^0_f)}\]
where the processes
\begin{align*}
A^0_h & {} = \Repl{\Rcv{c_h}{y}. \IfThen{\neseq(y)}{\Const{true}}{\Snd{c'_h}{ \Const{h}(k,y) }}} \\
A^0_f & {} = \Repl{\Rcv{c_f}{x}. \Snd{c'_f}{ \Const{f}(k,x) }} 
\end{align*}
answer requests to evaluate $\Const{h}$ and $\Const{f}$ with key $k$.
We restrict ourselves to hashes of non-empty sequences of blocks. In practice,
one never hashes the empty string, because the input of the hash function is padded to a non-zero multiple of the block length. This restriction is important in this example, because the definition of $\Const h$ yields $\Const h(k, \Const{nil}) = 0$, and this special hash value would break indifferentiability.
In order to enforce this restriction,
we use symbols $\Const{true}: \Bool$ and
$\neseq : \BlockList \rightarrow \Bool$,
with equations
\begin{equation}
\neseq(\cons{x}{\Const{nil}}) = \Const{true}\qquad\qquad
\neseq(\cons{x}{\cons{y}{z}}) = \neseq(\cons{y}{z})
\label{eq:neseq}
\end{equation}
The term $\neseq(M)$ is equal to $\Const{true}$ when $M$ is a non-empty list.

The other system offers an analogous interface, for an ideal hash function
$\Const{h}' : \Key \times \BlockList \rightarrow \Block$ 
and for a stateful compression function built from
$\Const{h}'$:
\[\Res{k}{(A^1_h \parop A^1_f)}\]
The process $A^1_h$ answers requests to evaluate an ideal hash function $\Const{h}'$:
\begin{align*}
A^1_h & {} = \Repl{\Rcv{c_h}{y}. \IfThen{\neseq(y)}{\Const{true}}{\Snd{c'_h}{ \Const{h}'(k,y) }}}
\end{align*}
and $A^1_f$ simulates the compression function using $\Const{h}'$. The
code for $A^1_f$, which is considerably more intricate, captures the
core of the security argument as it might appear in the cryptography literature.
(The paper by Coron et al.~\cite{Coron} omits this argument and, as far as we know, this argument does not appear elsewhere.)
{\allowdisplaybreaks\begin{align*}
A^1_f & {} = \Res{\ell,c_s}({\Repl{\Rcv{c_s}{s}. \Rcv{c_f}{x}. \Snd{\ell}{x,s,s}}  \parop \Repl{Q} \parop \Snd{c_s}{\cons{((0,0),\Const{nil})}{\Const{nil}}}})
\\
Q & {} = \begin{array}[t]{@{}l}
\Rcv{\ell}{x,t,s}.
  \IfThenElse{ t }{ \Const{nil} } { P_0 }{}\\ 
    \IfThenElse{ \Const{fst}(\Const{hd}(t))}{\Const{fst}(x)}{ P_1 }{ 
    \Snd{\ell}{x,\Const{tl}(t),s}}
\end{array}\\
P_0 & {} = \Snd{c'_f}{
  \Const{f}'(k, x) }
  \parop  \Snd{c_s}{s} 
\\
P_1 & {} = 
\begin{array}[t]{@{}l}
\Let z { \Const{snd}(\Const{hd}(t)) } \\*
\Let {z'}{ z \bappend \Const{snd}(x) } \\*
\Let r {(\Const{h}'(k, z'), \Const{f}_c(k,z'))}\\*
\Snd{c'_f}{r} \parop \Snd{c_s}{\cons{(r,z')}{s}}
\end{array}
\end{align*}}%
In this definition,  $\Let x M P$ is syntactic sugar for 
$P\{\subst{M}{x}\}$, $\Rcv{\ell}{x,t,s}.P$ is syntactic sugar for
$\Rcv{\ell}{y}.\Let x {\Const{fst}(y)} \Let t {\Const{fst}(\Const{snd}(y))} \Let s {\Const{snd}(\Const{snd}(y))} P$ where $y$ is a fresh variable,
and $\Snd{\ell}{x,t,s}$ is syntactic sugar for $\Snd{\ell}{(x,(t,s))}$,
with the appropriate sorts and overloading of the function symbols for pairs.
The function symbol 
$\Const{f}' : \Key \times \BlockT \rightarrow \BlockPair$
represents the compression function outside the domain used for implementing
the hash function,
and the function symbol $\Const{f}_c : \Key \times \BlockList \rightarrow \Block$ represents the second projection of the compression function inside that domain.
The channel $c_s$ maintains global private state, a lookup table
that maps each term $(\Const{h}'(k, M), \Const{f}_c(k, M))$ with $M = \cons{M_1}{\cons{\dots}{ \cons{M_n}{\Const{nil}}}}$ built as a result
of previous compression requests to the term
$M$, and initially maps $(0,0)$ to $\Const{nil}$.
This lookup table is represented as a list of pairs.
Each table element, of sort $\BlockBlockList$, is a pair of a $\BlockPair$ and a $\BlockList$; the table, of sort $\BlockBlockListList$, is a list of $\BlockBlockList$; we overload the function symbols for pairs and lists.
Upon a compression request with input~$x$, 
the process $Q$ looks up $\Const{fst}(x)$ in the table:
$Q$ receives as input $x$, the initial state of the table $s$,
and the tail $t$ of the lookup table. It uses a local channel $l$
for encoding the recursive call.
The auxiliary processes $P_0$ and $P_1$ complete compression requests
in the cases where lookups fail and succeed, respectively.
When a lookup fails, the compression request is outside the domain used for
implementing the hash function, so $P_0$ answers it using $\Const{f}'$,
and leaves the table unchanged.
When a lookup succeeds, we have either
$\Const{fst}(x) = (\Const{h}'(k, M), \Const{f}_c(k, M))$ with $M = \cons{M_1}{\cons{\dots}{ \cons{M_{n-1}}{\Const{nil}}}}$ and $n>0$
or $\Const{fst}(x) = (0,0)$ and we let $M = \Const{nil}$.
The lookup yields $z = M$,
$P_1$ computes $z' = z \bappend \Const{snd}(x)$ and 
returns $r = (\Const{h}'(k, z'), \Const{f}_c(k, z'))$ as result
of the compression request.
The table
is extended by adding the mapping from $r$ to $z'$.

Let us now explain, informally, why this code ensures that the results of the compression function are consistent with those of hash computations. 
The result of a compression request with argument~$x$ needs to be made consistent with the hash function
when 
\begin{equation}
x = (\Const f(k, (\dots \Const f(k,((0,0),M_1))\dots,M_{n-1})), M_n)
\label{eq:consist}
\end{equation}
for some $M_1, \ldots, M_n$ ($n>0$), because in that case
\begin{equation}
\Const h(k, \cons{M_1}{\cons{\dots}{ \cons{M_n}{\Const{nil}}}}) 
= \Const{fst}( \Const f(k,x))
\label{eq:consist2}
\end{equation}
that is, in the system $\Res{k}{(A^0_h \parop A^0_f)}$, the result of the hash request with argument $\cons{M_1}{\cons{\dots}{ \cons{M_n}{\Const{nil}}}}$ computed by $A^0_h$ is equal to the first block of the result of the compression request with argument $x$ computed by $A^0_f$.
We need to have an analogous equality in the system $\Res{k}{(A^1_h \parop A^1_f)}$.
In the system $\Res{k}{(A^0_h \parop A^0_f)}$, equality~\eqref{eq:consist} holds exactly when $\Const{fst}(x)$ is the result of previous compression requests $\Const{f}(k, (\dots \Const f(k,((0,0),M_1))\dots,M_{n-1}))$
for some $M_1, \ldots, M_{n-1}$.
In the system $\Res{k}{(A^1_h \parop A^1_f)}$, the table lookup tests a corresponding condition and, when it succeeds, $P_1$ retrieves
$z = \cons{M_1}{\cons{\dots}{ \cons{M_{n-1}}{\Const{nil}}}}$,
computes $z' = \cons{M_1}{\cons{\dots}{ \cons{M_{n-1}}{\cons{M_{n}}{\Const{nil}}}}}$ since $\Const{snd}(x) = M_n$, and returns $r = (\Const{h}'(k, z'), \Const{f}_c(k, z'))$.
Hence, $\Const{fst}(r) = \Const{h}'(k, z')$ and
the result of the hash request with argument $z' = \cons{M_1}{\cons{\dots}{ \cons{M_n}{\Const{nil}}}}$ computed by  $A^1_h$ is equal to the first block of the result $r$ of the compression request with argument~$x$ computed by $A^1_f$.

Formally, we obtain the following observational equivalence:
\begin{theorem}\label{th:indif} 
$\Res{k}{(A^0_h \parop A^0_f)} 
\bicong 
\Res{k}{( A^1_h \parop A^1_f)} $.
\end{theorem}
In the proof of this theorem (which is given in Appendix~\ref{app:indif}), 
we define a relation $\rel$ between configurations of the two systems,
and show that $\rel \cup \rel^{-1}$ is a labelled bisimulation. A key step of
this proof consists in proving static equivalence between related
configurations; this step formalizes the informal explanation of the 
process $A_1^f$ given above. We conclude by 
Theorem~\ref{THM:OBSERVATIONAL-LABELED}.

\section{Related Work}\label{sec:related}

This section aims to position the applied pi calculus with respect to
research on process calculi and on the analysis of security protocols.
As discussed in Section~\ref{sec:intro}, the applied pi calculus has
been the basis for much further work since its initial publication; 
this section does not discuss many papers that build on the applied pi
calculus. (Some of those papers, and others, are mentioned elsewhere 
in this paper.)

\subsection{Process Calculi}

The applied pi calculus has many commonalities with the original pi
calculus~\cite{milner:communicating-mobile} and its relatives, such as the spi calculus~\cite{spi2four} (discussed
in Sections~\ref{sec:sample} and~\ref{sec:equivalences}).  In particular, the model of communication adopted in the
applied pi calculus is deliberately classical: communication is
through named channels, and value computation is rather separate from
communication. 

Furthermore, active substitutions are reminiscent of the
constraints of the fusion calculus~\cite{victor:fusion-calculus}.
They are especially close to the substitution environments that Boreale et
al.~employ in their proof techniques for a variant of the spi calculus
with a symmetric cryptosystem~\cite{boreale:techniques}. We
incorporate substitutions into processes, systematize them, and
generalize from symmetric cryptosystems to arbitrary operations and
equations. 

Extensions of the pi calculus are not limited to modelling cryptography:
many extensions and variants of the pi calculus have been designed for diverse applications.
Examples include calculi for mobility, such as the ambient calculus~\cite{Cardelli00},
calculi for modelling biological processes, such as the enhanced pi calculus~\cite{Curti04},
and calculi for service-oriented computing, which
model the contracts that services implement, 
the composition of services, and their protocols~\cite{Orchestration,Lucchi07,CCpi,CruzFilipe14}. 
The psi calculi~\cite{bengtson:psi} provide a general framework parameterized by nominal data types for terms, conditions (generalizing comparison between terms),
and assertions (generalizing our notion of frames) and their operational semantics.
They also give sufficient conditions on these parameters to ensure that the resulting observational equivalence coincides with labelled bisimilarity.
The framework accommodates encodings of the pi calculus and several of its variants, for example ones with fusion~\cite{ExplicitFusion} and
concurrent constraints~\cite{CCpi}.
In particular, Bengtson et al.~give an encoding of the applied pi calculus into their framework.
However, as they explain, the result of this encoding
differs from the applied pi calculus in the way processes interact with contexts.
In particular, an important difference is that, when an encoded process sends a ciphertext, the ciphertext
appears on the label of the transition, and an agent that receives this
message will immediately learn the cleartext and the key. 
In psi calculi, one can avoid such counter-intuitive disclosures by 
explicitly creating and using aliases.
A recent extension of psi calculi~\cite{Borgstrom14,Borgstrom16} addresses these difficulties 
with a new form of pattern matching. 
In contrast, the management of aliases is built into the applied pi calculus, 
facilitating the modelling of 
security-protocol attackers as contexts. 

\subsection{Analysis of Security Protocols}

The analysis of a security protocol generally requires reasoning about its possible executions.
However, the ways of talking about the executions and
their properties vary greatly. We use a process calculus whose
semantics provides a detailed specification for interactions with a context.
Because the process calculus has a proper ``new'' construct 
(like the pi calculus but unlike CSP), it provides a direct
account of the generation of new keys and other fresh quantities. 
It also enables reasoning with equivalence and implementation
relations. 

Reasoning with those
relations is often more challenging than reasoning about trace properties, but it can be worthwhile.
Equivalences are useful, in particular, for modeling privacy properties~\cite{Pfitzmann01},
for instance in electronic voting~\cite{DKR08}.
While proofs of equivalences are difficult to automate in 
general---and observational equivalence is undecidable as noted in Section~\ref{subsec:obseq}---,
several tools support certain automatic proofs of equivalences
in the applied pi calculus and similar languages:
tools have focused on establishing particular kinds of equivalences such as trace
equivalence for bounded processes (that is, processes without 
replication)~\cite{TiuDawson10,Chadha12,CCD-tcs13}
or for restricted classes of unbounded processes~\cite{Chretien15,Chretien15b}.
Although ProVerif initially focused on proofs of trace properties~\cite{Blanchet08c},
it also supports automatic proofs of diff-equivalences, which are equivalences between processes 
that share the same structure and differ only in the choice of
terms~\cite{Blanchet07b}.
A diff-equivalence between two processes requires that the two processes reduce in the
same way, in the presence of any adversary. In particular, the two
processes must have the same branching behaviour. Hence,
diff-equivalence is much stronger than observational equivalence.
Maude-NPA~\cite{Santiago14} and Tamarin~\cite{Basin15} 
also employ that notion. Baudet~\cite{Baudet05ccs,Baudet07} showed that diff-equivalence is decidable for bounded processes: he treated a class of equivalences that model security against off-line guessing attacks in~\cite{Baudet05ccs} and proved the full result in~\cite{Baudet07}.

Furthermore, the use of a process calculus permits treating security
protocols as programs written in a programming notation---subject to
typing~\cite{abadi:secrecy,cardelli:mark,Gordon2004}, to other static analyses~\cite{bodei:control}, and to translations~\cite{Abadi:protection,AbadiFournetGonthier,AbadiFournetGonthier:auth}.
Thus, language-based approaches have led to tools such as ProVerif 
where protocols can be described by programs, and 
analyzed using automated techniques
that leverage 
type systems and Horn clauses~\cite{Abadi04c}. 

The applied pi calculus is also convenient as an intermediate language.
Translations to ProVerif have been implemented from TulaFale (a language for standardized Web-services protocols)~\cite{Bhargavan04}, from F\#~\cite{Bhargavan:dec08}, and from JavaScript
in order to verify protocols, including TLS~\cite{Bhargavan08b}.

As in many other works
(e.g.,~\cite{DY83,DLM,Merritt:thesis,JCRYPT::KemmererMM1994,Lowe,Schneider,Paulson:induct,Mitchell:murphi,strand,spi2four,Dam97:Proving_Trust,amadiolugiez,Escobar06,AVISPA,Cremers08b,Schmidt12}),
our use of the applied pi calculus conveniently avoids matters of
computational complexity and probability. In contrast, other
techniques for the analysis of security protocols employ more concrete
computational models, where principals are basically Turing machines
that manipulate bitstrings, and security depends on the computational
limitations of attackers (e.g.,~\cite{y,gm,gmr,br,GoldwasserBellare}).

Although these two approaches remained rather distinct during the
1980s and 1990s, fruitful connections have now been established
(e.g.,~\cite{Lincoln98,PfScWa00,Abadi2002b,datta:jcs05,BlanchetOakland06,BlanchetPointchevalCrypto06,CLC:ccs,Backes:CoSP,abadi:progress,barthe:techniques}). In
particular, some work interprets symbolic proofs in terms of concrete,
bitstring-based models~\cite{Abadi2002b}, in some cases specifically
studying the ``computational soundness'' of the applied pi
calculus~\cite{bck2009,CLC:ccs,Backes:CoSP}.  
Other work focuses directly
on those concrete models but benefits from notations and ideas from
process calculi and programming languages. For example, the tool
CryptoVerif~\cite{BlanchetOakland06,BlanchetPointchevalCrypto06}
provides guarantees in terms bitstrings, running times, and
probabilities, but its input language is strongly reminiscent of the
applied pi calculus, which influenced it---rather than of Turing
machines.

\section{Conclusion}\label{sec:conclusion}

In this paper, we describe a uniform extension of the pi calculus, the
applied pi calculus, in which messages may be compound values, not
just channel names.  We study its theory, developing its semantics and proof
techniques.  Although the calculus has no special, built-in features
to deal with security, it has proved useful in the analysis of security
protocols.

Famously, the pi calculus is the language of those lively social
occasions where all conversations are exchanges of names.  The applied
pi calculus opens the possibility of more substantial, structured
conversations; the cryptic character of some of these conversations
can only add to their charm and to their tedium.

The previous paragraph closed 
the conference paper that introduced the applied pi calculus in 2001. 
We are now in a
better position to assess the possibility to which it refers.  As we
hoped in 2001, the applied pi calculus has been extensively used for
modeling and for reasoning about security protocols, particularly
ones that rely heavily on cryptography (and less for ones that rely on
simple capabilities). 
The flexibility of the applied
pi calculus is a key enabler for those applications. This flexibility
did not render unnecessary or uninteresting the exploration of
variants and extensions. However, it did allow the applied pi calculus
to remain a relevant core system---it was not displaced by an extended
language with many more constructs.

We are also in a better position to comment on matters of charm and
tedium, alas. It is debatable whether security protocols have become
more charming or more tedious since 2001. It is clear, however, that
they play an ever-growing role, and that 
their security remains problematic.
The evolution of TLS exemplifies these
points. The literature now contains many attacks on TLS, 
e.g.,~\cite{3shake,logjam,vanhoef2015all,drown},
but also several
partial specifications and proofs~\cite{paterson2011tag,DBLP:conf/crypto/JagerKSS12,kpw13,DBLP:conf/sp/CremersHSM16}, sometimes relying on the applied pi
calculus~\cite{tls-toplas,BhargavanSP17}, and sometimes with language-based methods of the kind that
the applied pi calculus started to explore~\cite{BhargavanFKPS13,DBLP:conf/crypto/BhargavanFKPSB14}.  
The state-of-the-art approaches~\cite{paterson2011tag,DBLP:conf/crypto/JagerKSS12,kpw13,BhargavanFKPS13,DBLP:conf/crypto/BhargavanFKPSB14,BhargavanSP17b,BhargavanSP17} 
rely on refined frameworks that consider matters of
computational complexity and probability, which are beyond the
(explicit) scope of the applied pi calculus, as explained above.
In such applications, tools play a helpful role, often an essential
one.  Although they sometimes lead to important insights, manual proofs---in particular, manual proofs of
equivalences---can be rather painful
and tedious. 
(We may have suspected this fact
in 2001; on the basis of our experience since then, we now know it
with certainty.) On the other hand, the relative ease of use of
ProVerif has contributed greatly to the spread of the applied pi
calculus. The applied pi calculus has evolved through
its implementation in ProVerif, and as a result of its use in this context. 
That evolution is, in our opinion, an improvement, so the present paper
aims to reflect it.

Finally, independently of the merits and the future of the applied pi
calculus, we believe that languages, and in particular the formalization 
of attackers as contexts, should continue to play a role in the analysis
of security protocols. The alternatives (defining protocols as
interacting Turing machines?) are not easier. Describing protocols
in a programming notation not only makes them precise but also brings them
into the realm where ideas and tools from programming can support analysis.

\appendix
\setlist[itemize]{leftmargin=*}

\section*{Supplementary Material: Proofs}

The appendix contains proofs and auxiliary definitions needed for
these proofs.  After introducing the notion of simple
contexts and proving Lemma~\ref{LEM:INVARIANT-STATIC-EQ} (Section~\ref{app:simplecontexts}), the bulk of the appendix
is devoted to lemmas and definitions that contribute to the
proof of Theorem~\ref{THM:OBSERVATIONAL-LABELED} (Sections~\ref{app:bigpfpnf} and~\ref{app:bigpfmain}).
Section~\ref{app:disclosure} presents the proof of Lemma~\ref{LEM:DISCLOSURE}.
Section~\ref{app:refining} presents the proofs related to refined labels.
Finally, Sections~\ref{app:constructing} and~\ref{app:indif} present the 
proofs related to the two constructions of hash functions given in Sections~\ref{SUBSEC:CONSTRUCTING} 
and~\ref{SUBSEC:INDIF} respectively.

\section{Simple Contexts and Proof of Lemma~\ref{LEM:INVARIANT-STATIC-EQ}}\label{app:simplecontexts}

In order to work with definitions that refer to contexts, such as 
Definition~\ref{def:bicong}, it is convenient to
generalize structural equivalence from extended processes to contexts.
For this generalization, 
we use the rules of Section~\ref{sec:ope-sem}, except that 
(1) we do not rename bound names and variables whose scope includes the hole; and
(2) in rule \rulename{New-Par}, the hole is considered to contain any name and variable.

Further, in order to avoid special cases in proofs, we often adopt simplifying assumptions on contexts. We say that an evaluation context $\CTX$ is \emph{simple} when (1) no name is both free in $\CTX$ and restricted above the hole; and 
(2) no variable is both in $\dom(\CTX)$ and restricted above the hole.
We say that $\CTX$ is \emph{simple for}~$A$ if, in addition, it is closing for $A$.
These conditions on scopes exclude, for example,
$\Snd{a}{s} \parop \Res{s}(\hole)$
and 
$\{\subst{
  \Const{s}}{
  x}\} \parop \Res{x}(\hole)$.

\begin{lemma}\label{lem:simplecontexts}
Let $A$ be a closed extended process.
Given a simple context $\CTX$ for~$A$, there exists a context $\CTX'$ of the form $\Res{\vect u}(\hole \parop B)$ such that $\CTX \equiv \CTX'$ and all subcontexts of $\CTX'$ are simple for $A$.
\end{lemma}

\begin{proof}
We construct the context $\CTX'$ from $\CTX$ as follows:
\begin{enumerate}
\item\label{step:ren} We rename all names and variables bound by restrictions that 
are not above the hole to distinct fresh names and variables.
\item We move all restrictions above the hole in $\CTX$ to the root of $\CTX$.
These moves are possible because the names and variables bound by these restrictions do not occur elsewhere: they are not free since $\CTX$ is simple for $A$ and they are not bound by other restrictions by the renaming of step~\ref{step:ren}. 
\item We reorganize parallel compositions by associativity and commutativity so that the obtained context is of the form $\Res{\vect u}(\hole \parop B)$.
\end{enumerate}
Hence we obtain a context $\CTX' = \Res{\vect u}(\hole \parop B)$ such that
$\CTX \equiv \CTX'$ and $\CTX'$ is closing for $A$, that is, $\CTX'[A]$ is closed.

The subcontexts of $\CTX'$ are $\hole$, $\hole \parop B$, and contexts of the form
$\Res{\vect u'}(\hole \parop B)$ where $\vect u'$ is a suffix of  $\vect u$.
The contexts $\hole$ and $\hole \parop B$ have no names and variables
bound in the hole.
In the contexts $\Res{\vect u'}(\hole \parop B)$, the names and variables
bound in the hole are $\vect u'$, and they are not free.
So all these contexts are simple.

We show that all subcontexts of $\CTX'$ are closing for $A$.
For the context $\CTX'' = \hole$, we have $\CTX''[A] = A$ and we know by hypothesis
that $A$ is closed.
For the other subcontexts, we proceed by removing one by one the elements of $\vect u'$.
\begin{itemize}
\item Suppose that $\Res{x}\CTX''[A]$ is closed where $\CTX'' = \Res{\vect u''}(\hole \parop B)$. We have 
$\fv(\CTX''[A]) \setminus \dom(\CTX''[A]) = \fv(\Res{x}\CTX''[A]) \setminus \dom(\Res{x}\CTX''[A]) = \emptyset$, so $\CTX''[A]$ is closed.
\item Suppose that $\Res{n}\CTX''[A]$ is closed where $\CTX'' = \Res{\vect u''}(\hole \parop B)$. Then $\CTX''[A]$ is obviously also closed.
\qedhere
\end{itemize}
\end{proof}

\begin{lemma}\label{lem:rename-static-eq}
Let $A$ and $B$ be two closed extended processes.
\begin{enumerate}
\item Let $\sigma$ be a bijective renaming.
We have $A \enveq B$ if and only if $A\sigma \enveq B\sigma$.
\item Let $A'$ and $B'$ be obtained from $A$ and $B$, respectively, by replacing all variables (including their occurrences in domains of active substitutions) with distinct variables. 
We have $A \enveq B$ if and only if $A' \enveq B'$.
\end{enumerate}
\end{lemma}
\begin{proof}
To show the first point, suppose that $A \enveq B$. Hence 
for all terms $M$ and $N$, 
$(M=N)\frameof{A}$ if and only if $(M=N)\frameof{B}$.
So  $(M\sigma = N\sigma)\frameof{A\sigma}$
if and only if $(M\sigma = N\sigma) \frameof{B\sigma}$, since the equational theory is closed under renaming.
So $A\sigma \enveq B\sigma$.
The same argument also shows the converse, via the inverse renaming.

To show the second point, suppose that $A \enveq B$. Hence 
for all terms $M$ and $N$, 
$(M=N)\frameof{A}$ if and only if $(M=N)\frameof{B}$.
We let $M'$ and $N'$ be obtained from $M$ and $N$, respectively, by
the same variable replacement as the one that transforms
$A$ and $B$ into $A'$ and $B'$.
So  $(M' = N')\frameof{A'}$
if and only if $(M' = N') \frameof{B'}$, since 
$M\sigma = M'\sigma'$ where $\frameof{A} \equiv \Res{\vect n}\sigma$
and $\frameof{A'} \equiv \Res{\vect n}\sigma'$.
So $A' \enveq B'$.
As above, the same argument also shows the converse, via the inverse variable replacement.
\end{proof}

\begin{restate}{Lemma}{\ref{LEM:INVARIANT-STATIC-EQ}}
Let $A$ and $B$ be closed extended processes. 
If $A \equiv B$ or $A \rightarrow B$, then $A \enveq B$.
If $A \enveq B$, then $\CTX[A] \enveq \CTX[B]$ for all closing evaluation contexts $\CTX[\hole]$.
\end{restate}%
\begin{proof}
  We show that, if $A \equiv B$, then $\frameof{A} \equiv \frameof{B}$, by an easy induction on the derivation of $A \equiv B$. We then show that, if $A \rightarrow B$, then $\frameof{A} \equiv \frameof{B}$ since the frame is not affected by \brn{Comm}, \brn{Then}, and \brn{Else}. Since Definition~\ref{def:eqframe} considers frames up to structural equivalence, we conclude that static equivalence is invariant by structural equivalence and reduction.

For the context-closure property, we suppose that $A \enveq B$ and we show that for all closing evaluation contexts $\CTX$, we have $\CTX[A] \enveq \CTX[B]$. 
We first rename the free names and variables of $\CTX$, so that the obtained context is simple,
and apply Lemma~\ref{lem:rename-static-eq}. 
Then by Lemma~\ref{lem:simplecontexts}, we construct a context $\CTX'$ such that $\CTX \equiv \CTX'$ and all subcontexts of $\CTX'$ are simple. Since static equivalence is invariant by structural equivalence, it is sufficient to show that $\CTX'[A] \enveq \CTX'[B]$. All subcontexts of $\CTX'$ are closing evaluation contexts, so we proceed by structural induction on $\CTX'$.
The cases of name restriction and variable restriction hold because they restrict the range of $M$ and $N$ in Definition~\ref{def:eqframe}.
In the case of parallel composition, we have $\CTX' \equiv \hole \parop \Res{\vect{n}}(\{\subst{\vect{M'}}{\vect{x}}\}\parop P)$ with $\fv(\vect{M'}) \cup \fv(P) \subseteq \dom(A) = \dom(B)$ and $\{ \vect{x} \} \cap \dom(A) = \emptyset$. By renaming the names $\vect{n}$ so that they do not occur free in $A$ and $B$, we have $\frameof{\CTX'[A]} \equiv \Res{\vect{n}}(\frameof{A} \parop \{\subst{\vect{M'}}{\vect{x}}\})$ and
$\frameof{\CTX'[B]} \equiv \Res{\vect{n}}(\frameof{B} \parop \{\subst{\vect{M'}}{\vect{x}}\})$. Since we have already handled the case of name restriction, it suffices to consider the closing context $\hole \parop \{\subst{M'}{x}\}$ with $\fv(M') \subseteq \dom(A)$ and $x\notin \dom(A)$. In this case, we apply the inductive hypothesis using $M\{\subst{M'}{x}\}$ and $N\{\subst{M'}{x}\}$ instead of $M$ and $N$ in Definition~\ref{def:eqframe}.
\end{proof}

\section{Proof of Theorem~\ref{THM:OBSERVATIONAL-LABELED}: Partial Normal Forms}\label{app:bigpfpnf}

Our proof of Theorem~\ref{THM:OBSERVATIONAL-LABELED}, outlined in
Section~\ref{sec:bigpf}, requires a definition of partial normal forms, which we 
present in Section~\ref{app:pnf}. A semantics on partial normal forms
corresponds to the standard semantics (Section~\ref{app:relpnf}). We
can soundly restrict attention to reductions between closed processes
in the semantics of partial normal forms (Section~\ref{app:respnd}). Moreover, partial normal forms enable helpful compositions and decompositions of reductions (Section~\ref{app:comppnf}).

\subsection{Definition of Partial Normal Forms}\label{app:pnf}

In this section, we define partial normal forms and prove two of their basic properties.

We first define the normalization of the parallel composition of two substitutions.
The composition $\comp{\sigma}{\sigma'}$ of 
two substitutions $\sigma$ and $\sigma'$ such that $\sigma \parop \sigma'$ is cycle-free
is defined as follows: 
we reorder $\sigma \parop \sigma'$ into $\{ \subst{M_1}{x_1}, \ldots, \subst{M_l}{x_l} \}$
where $x_i \notin \fv(M_j)$ for all $i \leq j \leq l$; we let $\sigma_0 = \nil$ and $\sigma_{i+1} = \sigma_i \{\subst{M_{i+1}}{x_{i+1}}\} \parop \{ \subst{M_{i+1}}{x_{i+1}} \}$ for $0 \leq i \leq l-1$; then $\comp{\sigma}{\sigma'} = \sigma_l$.
By definition, we have $\comp{\sigma}{\sigma'} \equiv \sigma \parop \sigma'$.

The partial normal form $\pnf(A)$ of an extended process $A$ is 
an extended process of the form 
$\Res{\vect n}(\{ \subst{\vect M}{\vect x}\} \parop P)$ 
such that $(\fv(P) \cup \fv(\vect M)) \cap \{ \vect x \} = \emptyset$. 
The sequence of restrictions
$\nu \vect n$ may be empty, in which case the partial normal form is 
written $\{ \subst{\vect M}{\vect x}\} \parop P$. The substitution $\{ \subst{\vect M}{\vect x}\}$ may be empty, in which case it is written $\nil$.
The partial normal form of $A$ is defined by induction on $A$ as follows:
\begin{align*}
\pnf(P) &= \nil \parop P\\
\pnf(\{ \subst{M}{x} \}) &= \{ \subst{M}{x} \} \parop \nil\\
\pnf(\Res{n}A) &= \Res{n,\vect n}(\sigma  \parop P)
\text{ where }\pnf(A) = \Res{\vect n}( \sigma \parop P) \text{ and } n \notin \{\vect n\}\\
\pnf(\Res{x}A) &= \Res{\vect n}( \sigma_{|\dom(\sigma) \setminus \{ x \}} \parop P)
\text{ where }\pnf(A) = \Res{\vect n}(\sigma \parop P)\\
\pnf(A \parop B) &= \Res{\vect n,\vect n'}(\comp{\sigma}{\sigma'} \parop (P \parop P')(\comp{\sigma}{\sigma'}))\\
&\text{\begin{minipage}{10cm}
where $\pnf(A) = \Res{\vect n}(\sigma \parop P)$,
$\pnf(B) = \Res{\vect n'}(\sigma' \parop P')$ and
$\vect n$ and $\vect n'$ are renamed so that they are disjoint,
the names of $\vect n$ are not free in $\sigma' \parop P'$, and the names of
$\vect n'$ are not free in $\sigma \parop P$.
\end{minipage}}
\end{align*}
The last four cases apply only when the argument of $\pnf$ is not a plain process.
We define a {\it normal process} as an extended process in partial normal form.

Two simple lemmas provide some basic properties of partial normal forms.

\begin{lemma}\label{lem:equivpnf}
$A \equiv \pnf(A)$.
\end{lemma}
\begin{proof}
By induction on the syntax of $A$. 
\end{proof}

\begin{lemma}\label{lem:pnf-closed}
If $A$ is closed, then $\pnf(A)$ is closed.
\end{lemma}
\begin{proof}
We prove by induction on the syntax of $A$ that $\fv(\pnf(A)) \subseteq \fv(A)$
and $\dom(\pnf(A)) = \dom(A)$. The result follows.
\end{proof}

\subsection{Relation between the Standard Semantics and the Semantics on Partial Normal Forms}\label{app:relpnf}

In this section, we define an operational semantics on partial normal forms, by defining
structural equivalence, internal reduction, and labelled transitions.
We relate this semantics to the standard semantics of the applied
pi calculus given in Sections~\ref{sec:ope-sem} and~\ref{subsec:labopsem}.

We begin with the definition of structural
equivalence on partial normal forms.
Let $\equivP$ be the smallest equivalence relation on plain processes
closed by application of evaluation contexts such that
\[\begin{array}{lrcll}
\brn{Par-\mbox{$\nil$}}' & 
P  \parop \nil & \equivP & P \\
\brn{Par-A}' & 
P \parop (Q \parop R)  & \equivP & (P \parop Q) \parop R \\
\brn{Par-C}' & 
P  \parop Q &  \equivP & Q \parop P \\
\brn{Repl}' & 
\Repl P  & \equivP & P \parop  \Repl P\\
\brn{New-\mbox{$\nil$}}' & 
\nu n.\nil  & \equivP & \nil \\
\brn{New-C}' & 
\nu n.\nu n'.P  & \equivP & \nu n'.\nu n.P \\
\brn{New-Par}' & 
P \parop \nu n.Q  & \equivP & \nu n.(P \parop Q) &\mbox{when } n \not\in \fn(P)\\
\brn{Rewrite}' &
P\{\subst{M}{x}\} & \equivP & P\{\subst{N}{x}\} &\text{when }\Sigma \vdash M = N\\
\end{array}\]
and
let $\equivpnf$ be the smallest equivalence relation on
normal processes such that
\[\begin{array}{@{}lrcll@{}}
\brn{Plain}'' &
\Res{\vect n}(\sigma \parop P) & \equivpnf & \Res{\vect n}(\sigma \parop P')&\\
&\multicolumn{4}{@{}l@{}}{\quad\qquad\text{when $P \equivP P'$ and }(\fv(P) \cup \fv(P')) \cap \dom(\sigma) = \emptyset}\\
\brn{New-C}'' &
\Res{\vect n}(\sigma \parop P) & \equivpnf & \Res{\vect n'}(\sigma \parop P)&\text{when $\vect n'$ is a reordering of $\vect n$}\\
\brn{New-Par}'' &
\Res{\vect n}(\sigma \parop \Res{n'}P) & \equivpnf & \Res{\vect n, n'}(\sigma \parop P)&\text{when $n' \notin \fn(\sigma)$}\\
\brn{Rewrite}'' &
\Res{\vect n}(\sigma \parop P) & \equivpnf & \Res{\vect n}(\sigma' \parop P)&\\
&\multicolumn{4}{@{}l@{}}{\quad\qquad\text{when $\dom(\sigma) = \dom(\sigma')$}} \\
&\multicolumn{4}{@{}l@{}}{\quad\qquad\text{and $\Sigma \vdash x \sigma = x\sigma'$ for all $x \in \dom(\sigma)$}}\\
&\multicolumn{4}{@{}l@{}}{\quad\qquad\text{and $(\fv(x\sigma) \cup \fv(x\sigma')) \cap \dom(\sigma) = \emptyset$ for all $x \in \dom(\sigma)$}}
\end{array}\]
In $\brn{Plain}''$ and $\brn{Rewrite}''$, the hypotheses on free variables
ensure that the process remains normalized in case fresh variables are introduced (respectively, 
via $\brn{Rewrite}'$
and 
by rewriting the substitution $\sigma$ to $\sigma'$).

We also introduce the corresponding reduction relation.
Let $\redP$ be the smallest relation on plain processes closed by 
$\equivP$ and by application of evaluation contexts such that:
\[\begin{array}{lrcll}
\brn{Comm}' &
{\Snd N M.P} \parop {\Rcv{N}{x}.Q} & \redP & P \parop Q\{\subst{M}{x}\}\\[.5em]
\brn{Then}' &
\IfThenElse{M}{M}{P}{Q} & \redP & P\\[.5em]
\brn{Else}' &
\IfThenElse{M}{N}{P}{Q} & \redP & Q \\
& \multicolumn{4}{l}{\mbox{\quad for any ground terms $M$ and $N$ such that $\Sigma \not\vdash M = N$}}
\end{array}\]%
and let $\redpnf$ be the smallest relation on normal processes closed by $\equivpnf$ such that 
$\Res{\vect n}(\sigma \parop P) \redpnf \Res{\vect n}(\sigma \parop P')$ when $P \redP P'$.

\begin{lemma}\label{lem:instance}
\begin{enumerate}
\item\label{instance:equiv} If $P \equivP P'$, then $P \sigma \equivP P' \sigma$.
\item\label{instance:red} If $P \redP P'$, then $P \sigma \redP P' \sigma$.
\end{enumerate}
\end{lemma}
\begin{proof}
These properties are immediate by induction on derivations.
The proof of Property~\ref{instance:red} relies on Property~\ref{instance:equiv} in the case in which one applies $\equivP$. 
Note that the change from \brn{Comm} to $\brn{Comm}'$ is crucial for Property~\ref{instance:red}.
\end{proof}

\begin{lemma}\label{lem:equiv-red-induct}
Assume that $\Res{\vect n}(\sigma \parop P)$, $\Res{\vect n'}(\sigma' \parop P')$, and $\Res{\vect n''}(\sigma'' \parop P'')$ are normal processes.
If $\Res{\vect n}(\sigma \parop P) \equivpnf  \Res{\vect n'}(\sigma' \parop P')$, then 
\begin{enumerate}
\item\label{equiv-Resx} $\Res{\vect n}(\sigma_{|\dom(\sigma) \setminus \{ x\}} \parop P) \equivpnf  \Res{\vect n'}(\sigma'_{|\dom(\sigma') \setminus \{ x\}} \parop P')$;
\item\label{equiv-Resn} $\Res{n,\vect n}(\sigma \parop P) \equivpnf  \Res{n,\vect n'}(\sigma' \parop P')$; and
\item\label{equiv-Par} if $\sigma \parop \sigma''$ and $\sigma' \parop \sigma''$ are cycle-free, then
$\Res{\vect n,\vect n''}(\comp{\sigma}{\sigma''} \parop (P \parop P'')(\comp{\sigma}{\sigma''})) \equivpnf  \Res{\vect n',\vect n''}(\comp{\sigma'}{\sigma''} \parop (P' \parop P'')(\comp{\sigma'}{\sigma''}))$.
\end{enumerate}
If $\Res{\vect n}(\sigma \parop P) \redpnf  \Res{\vect n'}(\sigma' \parop P')$, then
\begin{enumerate}
\setcounter{enumi}{3}
\item\label{red-Resx}  
$\Res{\vect n}(\sigma_{|\dom(\sigma) \setminus \{ x\}} \parop P) \redpnf  \Res{\vect n'}(\sigma'_{|\dom(\sigma') \setminus \{ x\}} \parop P')$;
\item\label{red-Resn}  $\Res{n,\vect n}(\sigma \parop P) \redpnf  \Res{n,\vect n'}(\sigma' \parop P')$; and 
\item\label{red-Par} if $\sigma \parop \sigma''$ and $\sigma' \parop \sigma''$ are cycle-free, then
$\Res{\vect n,\vect n''}(\comp{\sigma}{\sigma''} \parop (P \parop P'')(\comp{\sigma}{\sigma''})) \redpnf  \Res{\vect n',\vect n''}(\comp{\sigma'}{\sigma''} \parop (P' \parop P'')(\comp{\sigma'}{\sigma''}))$.
\end{enumerate}
\end{lemma}
\begin{proof}
We establish these properties by induction on derivations.
To prove Property~\ref{equiv-Par} in the case $\equivP$, we use that
if $P \equivP P'$ then $(P \parop P'')(\comp{\sigma}{\sigma''}) \equivP
(P' \parop P'')(\comp{\sigma}{\sigma''})$, which follows from Lemma~\ref{lem:instance}\eqref{instance:equiv}.
To prove Properties~\ref{red-Resx} to~\ref{red-Par}, we use Properties~\ref{equiv-Resx} to~\ref{equiv-Par}, respectively, in the case in which we apply $\equivpnf$. 
Additionally, to prove Property~\ref{red-Par} in the case $\redP$,
we use that if $P \redP P'$ then $(P \parop P'')(\comp{\sigma}{\sigma''}) \redP
(P' \parop P'')(\comp{\sigma}{\sigma''})$, which follows from Lemma~\ref{lem:instance}\eqref{instance:red}.
\end{proof}

\begin{lemma}\label{lem:struct-std-to-pnf}
If $A \equiv B$, then $\pnf(A) \equivpnf \pnf(B)$.
\end{lemma}
\begin{proof}
By induction on the derivation of $A \equiv B$.
We first notice that, if $P$ and $Q$ are plain processes, 
then $\pnf(P \parop Q) = \nil \parop (P \parop Q)$ 
can also be obtained by applying the definition of $\pnf(A \parop B)$ for extended processes, with the same result: $\pnf(P) = \nil \parop P$, $\pnf(Q) = \nil \parop Q$, so $\pnf(P \parop Q) = \nil \parop (P \parop Q)$. We use this property to avoid distinguishing whether $A$, $B$, $C$ are plain processes or not in the first three cases and in the last case.

\begin{itemize}
\item Case $A \equiv A \parop \nil$.

Let $\pnf(A) = \Res{\vect n}(\sigma \parop P)$.
We have $\pnf(A \parop \nil) = \Res{\vect n}(\sigma \parop (P \parop \nil))$. 
Since $P \equivP P \parop \nil$, we have $\pnf(A) \equivpnf \pnf(A \parop \nil)$.

\item Case $A \parop (B \parop C)  \equiv (A \parop B) \parop C$.

Let $\pnf(A) = \Res{\vect n}(\sigma \parop P)$, $\pnf(B) = \Res{\vect n'}(\sigma' \parop P')$, and $\pnf(C) = \Res{\vect n''}(\sigma'' \parop P'')$.
To compute  $\pnf(A \parop (B \parop C))$, we first rename $\vect n'$ and $\vect n''$ so that they are disjoint,
the names of $\vect n'$ are not free in $\sigma'' \parop P''$, and the names of
$\vect n''$ are not free in $\sigma' \parop P'$. Then 
$\pnf(B \parop C) = \Res{\vect n',\vect n''}(\comp{\sigma'}{\sigma''} \parop (P' \parop P'')(\comp{\sigma'}{\sigma''}))$. Then, we rename $\vect n$ and $\vect n',\vect n''$ so that they are disjoint,
the names of $\vect n$ are not free in $\comp{\sigma'}{\sigma''} \parop (P' \parop P'')(\comp{\sigma'}{\sigma''})$, and the names of
$\vect n',\vect n''$ are not free in $\sigma \parop P$. Hence, $\vect n$, $\vect n'$ and $\vect n''$ are renamed so that they are disjoint,
the names of $\vect n$ are not free in $\sigma' \parop P'$ and $\sigma'' \parop P''$, the names of $\vect n'$ are not free in $\sigma \parop P$ and $\sigma'' \parop P''$, and the names of $\vect n''$ are not free in $\sigma \parop P$ and $\sigma' \parop P'$. This condition is the same as the one obtained when we compute $\pnf((A \parop B) \parop C)$.
Let $\sigma_0 = \comp{\sigma}{(\comp{\sigma'}{\sigma''})} = \comp{(\comp{\sigma}{\sigma'})}{\sigma''}$.
So 
\begin{align*}
\pnf(A \parop (B \parop C)) 
&= \Res{\vect n,\vect n', \vect n''}(\sigma_0 \parop (P \parop (P' \parop P'')(\comp{\sigma'}{\sigma''}))\sigma_0) \\
&= \Res{\vect n,\vect n', \vect n''}(\sigma_0 \parop (P\sigma_0 \parop (P'\sigma_0 \parop P'' \sigma_0)))\\
&\equivpnf \Res{\vect n,\vect n', \vect n''}(\sigma_0 \parop ((P\sigma_0 \parop P'\sigma_0) \parop P''\sigma_0)) = \pnf((A \parop B) \parop C)
\end{align*}
since $P\sigma_0 \parop (P'\sigma_0 \parop P''\sigma_0) \equivP (P\sigma_0 \parop P'\sigma_0) \parop P''\sigma_0$.

\item Case $A  \parop B \equiv B \parop A$.

Let $\pnf(A) = \Res{\vect n}(\sigma \parop P)$ and $\pnf(B) = \Res{\vect n'}(\sigma' \parop P')$. We rename $\vect n$ and $\vect n'$ so that they are disjoint,
the names of $\vect n$ are not free in $\sigma' \parop P'$, and the names of
$\vect n'$ are not free in $\sigma \parop P$.
Let $\sigma_0 = \comp{\sigma}{\sigma'} = \comp{\sigma'}{\sigma}$.
We have $\pnf(A \parop B) = \Res{\vect n,\vect n'}(\sigma_0 \parop (P\sigma_0 \parop P'\sigma_0)) \equivpnf \Res{\vect n',\vect n}(\sigma_0 \parop (P'\sigma_0 \parop P\sigma_0)) = \pnf(B \parop A)$ since $P\sigma_0 \parop P'\sigma_0 \equivP P'\sigma_0 \parop P\sigma_0$.

\item Case $\Repl P  \equiv P \parop  \Repl P$.

We have $\pnf(\Repl P) = \nil \parop \Repl P \equivpnf \nil \parop (P \parop \Repl P) = \pnf(P \parop \Repl P)$, since $\Repl P \equivP P \parop \Repl P$.  

\item Case $\nu n.\nil \equiv \nil$.

We have $\pnf(\nu n.\nil) = \nil \parop \nu n.\nil \equivpnf \nil \parop \nil = \pnf(\nil)$, since $\nu n.\nil \equivP \nil$.

\item Case $\nu u.\nu v.A  \equiv \nu v.\nu u.A$.

If $A$ is a plain process, then $u$ and $v$ are names. (If $u$ or $v$ were variables, these variables would be in the domain of $A$, so $A$ would not be a plain process.)
In this case, $\pnf(\nu u.\nu v.A ) = \nil \parop \nu u.\nu v.A \equivpnf \nil \parop \nu v.\nu u.A = \pnf(\nu v.\nu u.A)$ since $ \nu u.\nu v.A \equivP \nu v.\nu u.A$.

If $A$ is not a plain process, let $\pnf(A) = \Res{\vect n}(\sigma \parop P)$.
If $u$ and $v$ are names, then we have $\pnf(\nu u.\nu v.A) = \Res{u,v,\vect n}(\sigma \parop P) \equivpnf \Res{v,u,\vect n}(\sigma \parop P) = \pnf(\nu v.\nu u.A)$.
If $u$ and $v$ are variables, then $\pnf(\nu u.\nu v.A) = \Res{\vect n}(\sigma_{|\dom(\sigma) \setminus \{u,v\}} \parop P) =  \pnf(\nu v.\nu u.A)$.
If $u$ is a name and $v$ is a variable, then $\pnf(\nu u.\nu v.A) = \Res{u,\vect n}(\sigma_{|\dom(\sigma) \setminus \{v\}} \parop P) = \pnf(\nu v.\nu u.A)$. The remaining case is symmetric.

\item Case $A \parop \nu u.B  \equiv \nu u.(A \parop B)$ with $u \not\in \fv(A) \cup \fn(A)$.

If $B$ is a plain process, then $u$ is a name and $\pnf(\nu u.B) = \nil \parop \nu u.B$. 
\begin{itemize}
\item If $A$ is also a plain process, then $\pnf(A \parop \nu u.B) = \nil \parop (A \parop \nu u.B) \equivpnf \nil \parop \nu u.(A \parop B) = \pnf(\nu u.(A \parop B))$ since $A \parop \nu u.B \equivP \nu u.(A \parop B)$.
\item If $A$ is not a plain process, then 
let $\pnf(A) = \Res{\vect n}(\sigma \parop P)$.
We rename $\vect n$ so that the names of $\vect n$ are not free in $B$ and do
not contain $u$.
We have $\pnf(A \parop \nu u.B) = \Res{\vect n}(\sigma \parop (P \parop \nu u.B)\sigma) \equivpnf \Res{\vect n}(\sigma \parop \nu u.(P \parop B)\sigma) \equivpnf \Res{u,\vect n}(\sigma \parop (P \parop B)\sigma) = \pnf(\nu u.(A \parop B))$.
\end{itemize}

If $B$ is not a plain process, then let $\pnf(B) = \Res{\vect n'}(\sigma' \parop P')$. Let $\pnf(A) = \Res{\vect n}(\sigma \parop P)$.
We rename $\vect n$ and $\vect n'$ so that $\vect n$ and $\vect n'$ are disjoint and do not contain $u$, the names of $\vect n$ are not free in $\sigma' \parop P'$, and the names of $\vect n'$ are not free in $\sigma \parop P$.
\begin{itemize}
\item If $u$ is a name, then $\pnf(A \parop \nu u.B) = 
\Res{\vect n, u, \vect n'}(\comp{\sigma}{\sigma'} \parop (P \parop P')(\comp{\sigma}{\sigma'}))
\equivpnf \Res{u, \vect n, \vect n'}(\comp{\sigma}{\sigma'} \parop (P \parop P')(\comp{\sigma}{\sigma'})) = \pnf( \nu u.(A \parop B))$.
\item If $u$ is a variable, then 
\begin{align*}
\pnf(A \parop \nu u.B)
&= \Res{\vect n, \vect n'}(\comp{\sigma}{(\sigma'_{|\dom(\sigma') \setminus\{u\}})} \parop (P \parop P')(\comp{\sigma}{(\sigma'_{|\dom(\sigma') \setminus\{u\}})}))\\
&=
\Res{\vect n, \vect n'}((\comp{\sigma}{\sigma'})_{|\dom(\comp{\sigma}{\sigma'}) \setminus\{u\}} \parop (P \parop P')(\comp{\sigma}{\sigma'}))\\
&= \pnf( \nu u.(A \parop B))
\end{align*}
\end{itemize}

\item Case $\nu x. \smx \equiv \nil$.

We have $\pnf(\nu x. \smx) = \nil \parop \nil = \pnf(\nil)$.

\item Case $\{\sx M\} \parop A \equiv \{\sx M\} \parop A \{\sx M\}$.

Let $\pnf(A) = \Res{\vect n}(\sigma \parop P)$. We rename $\vect n$ 
so that these names do not occur in $M$.
We have $\pnf(\{\sx M\} \parop A) = \Res{\vect n}((\comp{\{\sx M\}}{\sigma}) \parop (\nil \parop P)(\comp{\{\sx M\}}{\sigma}))$
and $\pnf(\{\sx M\} \parop A \{\sx M\}) = \Res{\vect n}((\comp{\{\sx M\}}{\sigma\{\sx M\}}) \parop (\nil \parop P\{\sx M\})(\comp{\{\sx M\}}{\sigma}))$,
so $\pnf(\{\sx M\} \parop A) = \pnf(\{\sx M\} \parop A \{\sx M\})$.

\item Case $\{\sx M\} \equiv \{ \sx N \}$ with $\Sigma \vdash M = N$.

We have $\pnf(\{\sx M\}) = \{\sx M\} \parop \nil \equivpnf \{ \sx N \} \parop \nil = \pnf(\{ \sx N \})$.

\item Case $\Res{x}A \equiv \Res{x}B$ knowing that $A \equiv B$.

Let $\pnf(A) = \Res{\vect n}(\sigma \parop P)$ and
$\pnf(B) = \Res{\vect n'}(\sigma' \parop P')$.
By induction hypothesis, we have $\pnf(A) \equivpnf \pnf(B)$, that is,
$\Res{\vect n}(\sigma \parop P) \equivpnf  \Res{\vect n'}(\sigma' \parop P')$. 
By Lemma~\ref{lem:equiv-red-induct}\eqref{equiv-Resx}, we conclude that 
\[\pnf(\Res{x}A) = \Res{\vect n}(\sigma_{|\dom(\sigma) \setminus \{ x\}} \parop P) \equivpnf  \Res{\vect n'}(\sigma'_{|\dom(\sigma') \setminus \{ x\}} \parop P') = \pnf(\Res{x}B)\]

\item Case $\Res{n}A \equiv \Res{n}B$ knowing that $A \equiv B$.
By induction hypothesis, we have $\pnf(A)\equivpnf \pnf(B)$.

If $A$ and $B$ are plain processes, then $\pnf(A) = \nil \parop A \equivpnf \nil \parop B = \pnf(B)$, so $\Res{n}(\nil \parop A) \equivpnf \Res{n}(\nil \parop B)$ by Lemma~\ref{lem:equiv-red-induct}\eqref{equiv-Resn}, so 
\[\pnf(\Res{n} A) = \nil \parop \Res{n} A \equivpnf \nil \parop \Res{n}B = \pnf(\Res{n}B)\]
by $\brn{New-Par}''$. 

If $A$ is a plain process and $B$ is not a plain process, then let $\pnf(B) = \Res{\vect n}(\sigma \parop P)$. We have $\pnf(A) = \nil \parop A \equivpnf  \Res{\vect n}(\sigma \parop P) = \pnf(B)$, so $\Res{n}(\nil \parop A) \equivpnf \Res{n, \vect n}(\sigma \parop P)$  by Lemma~\ref{lem:equiv-red-induct}\eqref{equiv-Resn}, so 
\[\pnf(\Res{n} A) = \nil \parop \Res{n} A \equivpnf \Res{n, \vect n}(\sigma \parop P) = \pnf(\Res{n}B)\]
by $\brn{New-Par}''$.

If $A$ and $B$ are not plain processes, then let $\pnf(A) = \Res{\vect n}(\sigma \parop P)$ and $\pnf(B) = \Res{\vect n'}(\sigma' \parop P')$.
We have 
\[\pnf(\Res{n} A) = \Res{n,\vect n}(\sigma \parop P) \equivpnf \Res{n,\vect n'}(\sigma' \parop P') = \pnf(\Res{n} B)\]
by Lemma~\ref{lem:equiv-red-induct}\eqref{equiv-Resn}.

\item Case $A \parop A'' \equiv B \parop A''$ knowing that $A \equiv B$.

Let $\pnf(A) = \Res{\vect n}(\sigma \parop P)$, $\pnf(B) = \Res{\vect n'}(\sigma' \parop P')$, and  $\pnf(A'') = \Res{\vect n''}(\sigma'' \parop P'')$.
We rename $\vect n$, $\vect n'$, and $\vect n''$ so that 
$\vect n$ and $\vect n'$ are disjoint from $\vect n''$, the names of $\vect n$ and $\vect n'$ are not free in $\sigma'' \parop P''$, the names of $\vect n''$ are not free in $\sigma \parop P$ and $\sigma' \parop P'$.
By induction hypothesis, we have $\pnf(A) \equivpnf \pnf(B)$, that is, $\Res{\vect n}(\sigma \parop P) \equivpnf  \Res{\vect n'}(\sigma' \parop P')$, so 
\begin{align*}
\pnf(A \parop A'') &= \Res{\vect n,\vect n''}(\comp{\sigma}{\sigma''} \parop (P \parop P'')(\comp{\sigma}{\sigma''})) \\
&\equivpnf  \Res{\vect n',\vect n''}(\comp{\sigma'}{\sigma''} \parop (P' \parop P'')(\comp{\sigma'}{\sigma''})) = \pnf(B \parop A'')
\end{align*}
by Lemma~\ref{lem:equiv-red-induct}\eqref{equiv-Par}.
\qedhere
\end{itemize}
\end{proof}

\begin{lemma}\label{lem:struct-process-std-to-pnf}
If $P \equiv Q$, then $P \equivP Q$.
\end{lemma}
\begin{proof}
By Lemma~\ref{lem:struct-std-to-pnf}, $P \equiv Q$ implies $\pnf(P) = \nil \parop P \equivpnf \pnf(Q) = \nil \parop Q$. We show that, if $\nil \parop P \equivpnf \Res{\vect n}(\sigma \parop Q)$, then $\sigma = \nil$ and $P \equivP \Res{\vect n}Q$, by an easy induction on the derivation of $\nil \parop P \equivpnf \Res{\vect n}(\sigma \parop Q)$. By applying this result to $\nil \parop P \equivpnf \nil \parop Q$, we obtain $P \equivP Q$.
\end{proof}

\begin{lemma}\label{lem:struct-pnf-to-std}
If $P \equivP Q$, then $P \equiv Q$.
If $A \equivpnf B$, then $A \equiv B$.
\end{lemma}
\begin{proof}
By induction on the derivations of $P \equivP Q$ and of $A \equivpnf B$, respectively.
\end{proof}

\begin{lemma}\label{lem:red-std-to-pnf}
If $A \rightarrow B$, then $\pnf(A) \redpnf \pnf(B)$.
\end{lemma}
\begin{proof}
By induction on the derivation of $A \rightarrow B$.
\begin{itemize}

\item Cases \brn{Comm}, \brn{Then}, and \brn{Else}. $A$ and $B$
are plain processes, so 
$\pnf(A) = \nil \parop A \redpnf \nil \parop B = \pnf(B)$, since
$A \redP B$ by \brn{Comm}$'$, \brn{Then}$'$, and \brn{Else}$'$ respectively.

\item Case $\Res{x}A \rightarrow \Res{x}B$ knowing $A \rightarrow B$.
The result follows easily from Lemma~\ref{lem:equiv-red-induct}\eqref{red-Resx}.

\item Case $\Res{n}A \rightarrow \Res{n}B$ knowing $A \rightarrow B$.
The result follows easily from Lemma~\ref{lem:equiv-red-induct}\eqref{red-Resn}, by distinguishing cases depending on whether $A$ and $B$ are plain processes or not, as in the proof of Lemma~\ref{lem:struct-std-to-pnf}.

\item Case $A \parop A'' \rightarrow B \parop A''$ knowing $A \rightarrow B$.
The result follows easily from Lemma~\ref{lem:equiv-red-induct}\eqref{red-Par}, as in the proof of Lemma~\ref{lem:struct-std-to-pnf}. 

\item If we apply $\equiv$, the result follows immediately from
Lemma~\ref{lem:struct-std-to-pnf}.
\qedhere
\end{itemize}
\end{proof}

\begin{lemma}\label{lem:red-pnf-to-std}
If $P \redP Q$, then $P \rightarrow Q$.
If $A \redpnf B$, then $A \rightarrow B$.
\end{lemma}
\begin{proof}
By induction on the derivations of  $P \redP Q$ and $A \redpnf B$, respectively.
In the cases in which we apply $\equivP$ or $\equivpnf$, we rely on Lemma~\ref{lem:struct-pnf-to-std}.
\end{proof}

Similarly, we define restricted labelled transitions. First, for plain processes, we define $P \ltrP{\alpha} A$ as follows:
\[\begin{array}{lc}
  \brn{In}' &
  \Rcv{N}{x}.P 
  \ltrP{ \Rcv{N}{M} }  
  P \{\subst{M}{x}\}  
  \\[3ex]
  \brn{Out-Var}' & 
  \infrule{ 
    x \notin \fv(\Snd{N}{M}.P)}{
    \Snd{N}{M}.P \ltrP{\nu x.\Snd{N}{x}} P \parop \{ \subst{M}{x} \} }
  \\[3ex]
  \brn{Scope}' &
  \infrule{ 
    P \ltrP{\alpha} A \hspace{5ex} 
    n \mbox{ does not occur in } \alpha
    }{
    \nu n.P \ltrP{\alpha} \nu n.A}
  \\[3ex]
  \brn{Par}' & \hspace{-3ex}
  \infrule{
    P \ltrP{\alpha} A \hspace{5ex} 
    \bv(\alpha)\cap \fv(Q) = \emptyset }{
    P \parop Q \ltrP{\alpha} A\parop Q }
  \\[3ex]
  \brn{Struct}' & \hspace{-3ex}
  \infrule{
    P \equivP Q \hspace{5ex} Q \ltrP{\alpha} B \hspace{5ex} B \equiv A }{
    P \ltrP{\alpha} A}
\end{array}\]
We define $A \ltrpnf{\alpha} B$, where $A$ is a normal process and $B$ is an extended process, 
as follows: there exist $\vect n$, $\sigma$, $P$, $\alpha'$, $B'$
such that
$A \equivpnf \Res{\vect n}(\sigma \parop P)$, 
$P \ltrP{\alpha'} B'$, $B \equiv \Res{\vect n}(\sigma \parop B')$, 
$\fv(\sigma) \cap \bv(\alpha') = \emptyset$, 
$\Sigma \vdash \alpha \sigma = \alpha'$, and the elements of $\vect n$ 
do not occur in $\alpha$.

We give below an alternative formulation of $\ltrP{\alpha}$.
\begin{lemma}\label{lem:caract-ltrP}
We have $P \ltrP{\alpha} A$ if and only if for some $\vect n$, $P_1$, $P_2$, $A_1$, $N$, $M$, $P'$, $x$, we have
$P \equivP \Res{\vect n}(P_1 \parop P_2)$,
$A \equiv \Res{\vect n}(A_1 \parop P_2)$,
$\{\vect n\} \cap \fn(\alpha) = \emptyset$,
$\bv(\alpha) \cap \fv(P_1 \parop P_2) = \emptyset$,
and
one of the following two cases holds:
\begin{enumerate}
\item $\alpha = \Rcv{N}{M}$, $P_1 = \Rcv{N}{x}.P'$, and $A_1 = P'\{\subst{M}{x}\}$; or
\item $\alpha = \Res{x}\Snd{N}{x}$, $P_1 = \Snd{N}{M}.P'$, and $A_1 = P' \parop \{\subst{M}{x}\}$.
\end{enumerate}
\end{lemma}
\begin{proof}
For the implication from left to right, we proceed by induction on this derivation of $P \ltrP{\alpha} A$.
\begin{itemize}
\item Case $\brn{In}'$: We are in the first case with $P_2 = \nil$ and $\vect n = \emptyset$.
\item Case $\brn{Out-Var}'$: We are in the second case with $P_2 = \nil$ and $\vect n = \emptyset$.
\item Case $\brn{Scope}'$: $P \ltrP{\alpha} A$ has been derived from
$Q \ltrP{\alpha} B$ with $P = \Res{n}Q$, $A = \Res{n}B$, and $n$ does not occur in $\alpha$. By induction hypothesis, $Q \equivP \Res{\vect n'}(Q_1 \parop Q_2)$,
$B \equiv \Res{\vect n'}(B_1 \parop Q_2)$,
$\{\vect n'\} \cap \fn(\alpha) = \emptyset$ and
$\bv(\alpha) \cap \fv(Q_1 \parop Q_2) = \emptyset$.
So $P \equivP \Res{n,\vect n'}(Q_1 \parop Q_2)$, $A \equivP \Res{n,\vect n'}(B_1 \parop Q_2)$, $\{n,\vect n'\} \cap \fn(\alpha) = \emptyset$ and
$\bv(\alpha) \cap \fv(Q_1 \parop Q_2) = \emptyset$.

\item Case $\brn{Par}'$: $P \ltrP{\alpha} A$ has been derived from
$Q \ltrP{\alpha} B$ with $P = Q \parop Q'$, $A = B \parop Q'$,
and $\bv(\alpha) \cap \fv(Q') = \emptyset$.
By induction hypothesis, $Q \equivP \Res{\vect n'}(Q_1 \parop Q_2)$,
$B \equiv \Res{\vect n'}(B_1 \parop Q_2)$,
$\{\vect n'\} \cap \fn(\alpha) = \emptyset$ and
$\bv(\alpha) \cap \fv(Q_1 \parop Q_2) = \emptyset$.
So $P \equivP \Res{\vect n'}(Q_1 \parop Q_2)\parop Q'
\equivP \Res{\vect n}(Q_1\{\subst{\vect n}{\vect n'}\} \parop (Q_2\{\subst{\vect n}{\vect n'}\} \parop Q'))$ and
$A \equiv \Res{\vect n'}(B_1 \parop Q_2)\parop Q'
\equiv \Res{\vect n}(B_1\{\subst{\vect n}{\vect n'}\} \parop (Q_2\{\subst{\vect n}{\vect n'}\} \parop Q'))$ where $\vect n$ consists of fresh names that do not occur in $\alpha$ nor $Q'$.
Let $P_1 = Q_1\{\subst{\vect n}{\vect n'}\}$, $P_2 = Q_2\{\subst{\vect n}{\vect n'}\} \parop Q'$, and $A_1 = B_1\{\subst{\vect n}{\vect n'}\}$.
Then $P \equivP \Res{\vect n}(P_1 \parop P_2)$, $A \equiv \Res{\vect n}(A_1 \parop P_2)$, 
$\{\vect n\} \cap \fn(\alpha) = \emptyset$,
$\bv(\alpha) \cap \fv(P_1 \parop P_2) = \emptyset$,
and the two cases are preserved because the renaming of $\vect n'$ into $\vect n$ leaves $\alpha$ unchanged, so in the first case, it leaves $N$ and $M$ unchanged, and just renames inside $P'$, and in the second case, it leaves $N$ unchanged and renames inside $M$ and $P'$.

\item Case $\brn{Struct}'$:
$P \ltrP{\alpha} A$ has been derived from
$Q \ltrP{\alpha} B$ with $P \equivP Q$ and $B \equiv A$.
By induction hypothesis, $Q \equivP \Res{\vect n'}(Q_1 \parop Q_2)$,
$B \equiv \Res{\vect n'}(B_1 \parop Q_2)$,
$\{\vect n'\} \cap \fn(\alpha) = \emptyset$ and
$\bv(\alpha) \cap \fv(Q_1 \parop Q_2) = \emptyset$.
So $P \equivP \Res{\vect n'}(Q_1 \parop Q_2)$,
$A \equiv \Res{\vect n'}(B_1 \parop Q_2)$,
$\{\vect n'\} \cap \fn(\alpha) = \emptyset$ and
$\bv(\alpha) \cap \fv(Q_1 \parop Q_2) = \emptyset$.

\end{itemize}
For the converse implication, we have $P_1 \ltrP{\alpha} A_1$ 
by $\brn{In}'$ in Case 1 and by $\brn{Out-Var}'$ in Case 2.
Then $P_1 \parop P_2 \ltrP{\alpha} A_1 \parop P_2$ by $\brn{Par}'$, 
$\Res{\vect n}(P_1 \parop P_2) \ltrP{\alpha} \Res{\vect n}(A_1 \parop P_2)$ by $\brn{Scope}'$,
and $P \ltrP{\alpha} A$ by $\brn{Struct}'$.
\end{proof}

\begin{lemma}\label{lem:instance-ltrP}
If $P \ltrP{\alpha} A$ and $\fv(\sigma) \cap \bv(\alpha) = \emptyset$,
then $P \sigma \ltrP{\alpha \sigma} A \sigma$.
\end{lemma}
\begin{proof}
By Lemma~\ref{lem:instance}\eqref{instance:equiv},
if $P \equivP P'$, then $P \sigma \equivP P' \sigma$,
We show that, if $A \equiv B$ and $\dom(\sigma) \cap \dom(A) = \emptyset$, 
then $A \sigma \equiv B \sigma$,
by noticing that $A \sigma \equiv \Res{\vect x}(A \parop \sigma)$
where $\{ \vect x \} = \dom(\sigma)$.
Then, we use the characterization of Lemma~\ref{lem:caract-ltrP},
after renaming the elements of $\vect n$ so that $\{ \vect n \}
\cap \fn(\sigma) = \emptyset$.
\end{proof}

\begin{lemma}\label{lem:redalpha-std-to-pnf}
If $A \ltr{\alpha} B$, then $\pnf(A) \ltrpnf{\alpha} B$.
\end{lemma}
\begin{proof}
By induction on the derivation of $A \ltr{\alpha} B$.
\begin{itemize}
\item In all cases in which $A$ is a plain process, we have
$\pnf(A) = \nil \parop A \ltrP{\alpha} B$ since, for plain processes,
the rules that define $A \ltrP{\alpha} B$ are the same as those
that define $A \ltr{\alpha} B$.
So $\pnf(A) = \nil \parop A \ltrpnf{\alpha} B$, with
$\alpha' = \alpha$, $\sigma = \nil$, and $\vect n = \emptyset$.

\item Case \brn{Scope} with $u = n$. We have $A' \ltr{\alpha} B'$,
$n$ does not occur in $\alpha$, $A = \Res{n}A'$, and $B = \Res{n}B'$.
By induction hypothesis, $\pnf(A') \ltrpnf{\alpha} B'$, so
$\pnf(A') \equivpnf \Res{\vect n}(\sigma \parop P)$, 
$P \ltrP{\alpha'} B''$, $B' \equiv \Res{\vect n}(\sigma \parop B'')$, 
$\fv(\sigma) \cap \bv(\alpha') = \emptyset$, 
$\Sigma \vdash \alpha \sigma = \alpha'$, and the elements of $\vect n$ 
do not occur in $\alpha$.
So $\pnf(A) \equivpnf \Res{n, \vect n}(\sigma \parop P)$
and $B \equiv \Res{n, \vect n}(\sigma \parop B'')$,
so $\pnf(A) \ltrpnf{\alpha} B$.

\item Case \brn{Scope} with $u = x$. We have $A' \ltr{\alpha} B'$,
$x$ does not occur in $\alpha$, $A = \Res{x}A'$, and $B = \Res{x}B'$.
By induction hypothesis, $\pnf(A') \ltrpnf{\alpha} B'$, so
$\pnf(A') \equivpnf \Res{\vect n}(\sigma \parop P)$, 
$P \ltrP{\alpha'} B''$, $B' \equiv \Res{\vect n}(\sigma \parop B'')$, 
$\fv(\sigma) \cap \bv(\alpha') = \emptyset$, 
$\Sigma \vdash \alpha \sigma = \alpha'$, and the elements of $\vect n$ 
do not occur in $\alpha$.
Let $\sigma' = \sigma_{|\dom(\sigma) \setminus \{x\}}$.
So $\pnf(A) \equivpnf \Res{\vect n}(\sigma' \parop P)$,
$P \ltrP{\alpha'} B''$, $B \equiv \Res{\vect n}(\sigma' \parop B'')$,
$\fv(\sigma') \cap \bv(\alpha') = \emptyset$,
$\Sigma \vdash \alpha \sigma' = \alpha'$ since $x$ does not occur in $\alpha$, 
and the elements of $\vect n$ do not occur in $\alpha$,
so $\pnf(A) \ltrpnf{\alpha} B$.

\item Case \brn{Par}. We have $A' \ltr{\alpha} B'$, $\bv(\alpha) \cap \fv(B_0) = \emptyset$, $A = A' \parop B_0$, and $B = B' \parop B_0$.
By induction hypothesis, $\pnf(A') \ltrpnf{\alpha} B'$, so
$\pnf(A') \equivpnf \Res{\vect n}(\sigma \parop P)$, 
$P \ltrP{\alpha'} B''$, $B' \equiv \Res{\vect n}(\sigma \parop B'')$, 
$\fv(\sigma) \cap \bv(\alpha') = \emptyset$, 
$\Sigma \vdash \alpha \sigma = \alpha'$, and the elements of $\vect n$ 
do not occur in $\alpha$.
Let $\pnf(B_0) = \Res{\vect n'}(\sigma' \parop P')$, where $\vect n$ and
$\vect n'$ are renamed so that they are disjoint,
the names of $\vect n$ are not free in $\sigma' \parop P'$, and the names of
$\vect n'$ are not free in $\sigma \parop P$, in $\alpha$, nor in $B''$.
Then $\pnf(A) \equivpnf \Res{\vect n,\vect n'}(\comp{\sigma}{\sigma'} \parop (P \parop P')(\comp{\sigma}{\sigma'}))$.
By $\brn{Par}'$, $P \parop P' \ltrP{\alpha'} B'' \parop P'$.
(We have $\bv(\alpha') \cap \fv(P') = \emptyset$ because
$\fv(P') \subseteq \fv(\pnf(B_0)) \subseteq \fv(B_0)$, 
$\bv(\alpha') = \bv(\alpha)$, and
$\bv(\alpha) \cap \fv(B_0) = \emptyset$.)
By Lemma~\ref{lem:instance-ltrP}, $(P \parop P')(\comp{\sigma}{\sigma'}) \ltrP{\alpha'(\comp{\sigma}{\sigma'})} (B'' \parop P')(\comp{\sigma}{\sigma'})$.
(We have $\fv(\comp{\sigma}{\sigma'}) \cap \bv(\alpha') = \emptyset$
because $\fv(\sigma) \cap \bv(\alpha') = \emptyset$
and $\fv(\sigma') \cap \bv(\alpha') = \emptyset$.)
Moreover, 
\begin{align*}
&B = B' \parop B_0 \equiv \Res{\vect n}(\sigma \parop B'')
\parop \Res{\vect n'}(\sigma' \parop P')
\equiv \Res{\vect n,\vect n'}(\comp{\sigma}{\sigma'} \parop (B'' \parop P')(\comp{\sigma}{\sigma'}))\\
&\fv(\comp{\sigma}{\sigma'}) \cap \bv(\alpha'(\comp{\sigma}{\sigma'})) 
= \fv(\comp{\sigma}{\sigma'}) \cap \bv(\alpha') = \emptyset\\
&\Sigma \vdash \alpha(\comp{\sigma}{\sigma'}) = \alpha\sigma(\comp{\sigma}{\sigma'})
= \alpha'(\comp{\sigma}{\sigma'})
\end{align*}
and the elements of $\vect n, \vect n'$ do not occur in $\alpha$,
so $\pnf(A) \ltrpnf{\alpha} B$.

\item Case \brn{Struct}.
We have $A' \ltr{\alpha} B'$, $A \equiv A'$, and $B \equiv B'$.
By induction hypothesis, $\pnf(A') \ltr{\alpha} B'$, so
$\pnf(A') \equivpnf \Res{\vect n}(\sigma \parop P)$, 
$P \ltrP{\alpha'} B''$, $B' \equiv \Res{\vect n}(\sigma \parop B'')$, 
$\fv(\sigma) \cap \bv(\alpha') = \emptyset$, 
$\Sigma \vdash \alpha \sigma = \alpha'$, and the elements of $\vect n$ 
do not occur in $\alpha$.
By Lemma~\ref{lem:struct-std-to-pnf}, $\pnf(A) \equivpnf \pnf(A')$,
so $\pnf(A) \equivpnf \Res{\vect n}(\sigma \parop P)$, 
and $B \equiv \Res{\vect n}(\sigma \parop B'')$,
hence $\pnf(A) \ltrpnf{\alpha} B$.
\qedhere
\end{itemize}
\end{proof}

\begin{lemma}\label{lem:redalpha-pnf-to-std}
If $P \ltrP{\alpha} A$, then $P \ltr{\alpha} A$.
If $A \ltrpnf{\alpha} B$, then $A \ltr{\alpha} B$.
\end{lemma}
\begin{proof}
The first point is proved by induction on the derivation of $P \ltrP{\alpha} A$. In the case $\brn{Struct}'$, we use Lemma~\ref{lem:struct-pnf-to-std}.

For the second point, we have $A \equivpnf \Res{\vect n}(\sigma \parop P)$, 
$P \ltrP{\alpha'} B'$, $B \equiv \Res{\vect n}(\sigma \parop B')$, 
$\fv(\sigma) \cap \bv(\alpha') = \emptyset$, 
$\Sigma \vdash \alpha \sigma = \alpha'$, and the elements of $\vect n$ 
do not occur in $\alpha$, for some $\vect n$, $\sigma$, $P$, $\alpha'$, $B'$.
By Lemma~\ref{lem:caract-ltrP}, we have
$P \equivP \Res{\vect n'}(P_1 \parop P_2)$,
$B' \equiv \Res{\vect n'}(B_1 \parop P_2)$,
$\{\vect n'\} \cap \fn(\alpha') = \emptyset$,
$\bv(\alpha') \cap \fv(P_1 \parop P_2) = \emptyset$,
and one of the following two cases holds:
\begin{enumerate}
\item $\alpha' = \Rcv{N'}{M'}$, $P_1 = \Rcv{N'}{x}.P'$, and $B_1 = P'\{\subst{M'}{x}\}$;
\item $\alpha' = \Res{x}\Snd{N'}{x}$, $P_1 = \Snd{N'}{M'}.P'$, and $B_1 = P' \parop \{\subst{M'}{x}\}$
\end{enumerate}
for some $\vect n'$, $P_1$, $P_2$, $B_1$, $N'$, $M'$, $P'$, $x$.
We rename the elements of $\vect n'$ so that 
$\{\vect n'\} \cap \fn(\alpha) = \emptyset$.

In Case~1, $\alpha = \Rcv{N}{M}$ for some $N$ and $M$.
We have 
\[A \equiv \Res{\vect n,\vect n'}(\sigma \parop \Rcv{N'}{x}.P' \parop P_2) \equiv \Res{\vect n,\vect n'}(\sigma \parop \Rcv{N}{x}.P' \parop P_2)\]
using Lemma~\ref{lem:struct-pnf-to-std} and \brn{Rewrite}, since $\Sigma \vdash N \sigma = N'$.
We have 
\[B \equiv \Res{\vect n,\vect n'}(\sigma \parop P'\{\subst{M'}{x}\} \parop P_2) \equiv \Res{\vect n,\vect n'}(\sigma \parop P'\{\subst{M}{x}\} \parop P_2)\]
using \brn{Rewrite}, since $\Sigma \vdash M \sigma = M'$.
Hence, we derive
\begin{align*}
&\Rcv{N}{x}.P' \ltr{\alpha} P'\{\subst{M}{x}\}\tag*{by \brn{In}}\\
&\Rcv{N}{x}.P' \parop P_2 \parop \sigma \ltr{\alpha} 
P'\{\subst{M}{x}\} \parop P_2 \parop \sigma\tag*{by \brn{Par}}\\
&\Res{\vect n,\vect n'}(\Rcv{N}{x}.P' \parop P_2 \parop \sigma) 
\ltr{\alpha} 
\Res{\vect n,\vect n'}(P'\{\subst{M}{x}\} \parop P_2 \parop \sigma)\tag*{by \brn{Scope}}\\
&A \ltr{\alpha} B\tag*{by \brn{Struct}}
\end{align*}
To apply \brn{Par}, we notice that $\fv(P_2 \parop \sigma) \cap \bv(\alpha) = \emptyset$ since 
$\bv(\alpha) = \bv(\alpha')$.

In Case~2, $\alpha = \Res{x}\Snd{N}{x}$ for some $N$.
We have 
\[A \equiv \Res{\vect n,\vect n'}(\sigma \parop \Snd{N'}{M'}.P' \parop P_2) \equiv \Res{\vect n,\vect n'}(\sigma \parop \Snd{N}{M'}.P' \parop P_2)\]
using Lemma~\ref{lem:struct-pnf-to-std} and \brn{Rewrite}, since $\Sigma \vdash N \sigma = N'$.
We have 
\[B \equiv \Res{\vect n,\vect n'}(\sigma \parop P' \parop \{\subst{M'}{x}\} \parop P_2)\]
Hence, we derive
\begin{align*}
&\Snd{N}{M'}.P' \ltr{\alpha} P' \parop \{\subst{M'}{x}\}\tag*{by \brn{Out-Var}}\\
&\Snd{N}{M'}.P' \parop P_2 \parop \sigma \ltr{\alpha} 
P' \parop \{\subst{M'}{x}\} \parop P_2 \parop \sigma\tag*{by \brn{Par}}\\
&\Res{\vect n,\vect n'}(\Snd{N}{M'}.P' \parop P_2 \parop \sigma) 
\ltr{\alpha} 
\Res{\vect n,\vect n'}(P' \parop \{\subst{M'}{x}\} \parop P_2 \parop \sigma)
\tag*{by \brn{Scope}}\\
&A \ltr{\alpha} B\tag*{by \brn{Struct}}
\end{align*}
\end{proof}

\subsection{Restriction to Closed Processes}\label{app:respnd}

Next, we show that we can restrict ourselves to reductions between closed processes
in the semantics on partial normal forms.
Let $\rel$ be an inductive relation on processes.
We say that a derivation of $\rel$ is \emph{closed} when all processes that appear in the derivation are closed,
and that a derivation of $\rel$ is \emph{closed on the left} when all processes that appear in the derivation before applying $\rel$ 
are closed.

Let $A$ and $B$ be two normal processes.
We write $\Sigma \vdash A = B$ when $B$ is obtained from $A$ by replacing
some terms $M$ with terms $N$ such that $\Sigma \vdash M = N$.
When $\Sigma \vdash P = Q$, we have $P \equivP Q$ by (possibly several) applications of $\brn{Rewrite}'$.
When $\Sigma \vdash A = B$, we have $A \equivpnf B$ by (possibly several) applications of $\brn{Rewrite}'$ and $\brn{Rewrite}''$.

\begin{lemma}\label{lem:inv-subst}
\begin{enumerate}
\item \label{inv-subst:equivP}
If $P \equivP Q$ and $\Sigma \vdash P = P\sigma$, 
then $P\sigma \equivP Q\sigma$ and $\Sigma \vdash Q = Q\sigma$.

\item \label{inv-subst:equivpnf}
If $A \equivpnf B$ and $\Sigma \vdash A = A\sigma$, 
then $A\sigma \equivpnf B\sigma$ and $\Sigma \vdash B = B\sigma$.
\end{enumerate}
\end{lemma}
\begin{proof}
We prove these properties by induction on the derivations.
All cases are straightforward. (When $y \in \dom(A)$, we consider that $A \{\subst{a}{y}\} = A$.) 
In the cases $\brn{Rewrite}'$ and $\brn{Rewrite}''$,
we use that the equational theory is closed under substitution of 
terms for variables.
In the cases of transitivity of $\equivP$ and $\equivpnf$, 
we use the induction hypothesis twice.
\end{proof}

\begin{lemma}\label{lem:instance-equiv}
If $A \equiv B$, then $A\sigma \equiv B\sigma$.
\end{lemma}
\begin{proof}
We prove this lemma by induction on the derivation of $A \equiv B$. 
\end{proof}

\begin{lemma}\label{lem:closing} In all the cases below, $Y$ 
  is a set of variables, $\sigma$ is a substitution from $Y$ to
  pairwise distinct fresh names.
\begin{enumerate} 
\item\label{closing:equivP} If 
$P \equivP Q$, 
$Y \eqdef \fv(P) \cup \fv(Q)$, and
$\Sigma \vdash P = P\sigma$, then
$\Sigma \vdash Q = Q\sigma$ and
$P\sigma \equivP Q\sigma$ by a closed derivation.

\item\label{closing:redP} If 
$P \redP Q$, 
$Y \eqdef \fv(P) \cup \fv(Q)$, and
$\Sigma \vdash P = P\sigma$, then
$\Sigma \vdash Q = Q\sigma$ and
$P\sigma \redP Q\sigma$ by a closed derivation.

\item\label{closing:ltrP} If 
$P \ltrP{\alpha} A$, 
$Y \eqdef \fv(P) \cup \fv(\alpha) \cup (\fv(A) \setminus \dom(A))$, and
$\Sigma \vdash P = P\sigma$, then
$\Sigma \vdash \alpha = \alpha\sigma$,
$A \equiv A\sigma$,
and 
$P\sigma \ltrP{\alpha\sigma} A\sigma$ 
by a derivation closed on the left.

\item\label{closing:equivpnf} If 
$A \equivpnf B$, 
$Y \eqdef (\fv(A) \setminus \dom(A)) \cup (\fv(B) \setminus \dom(B))$, and
$\Sigma \vdash A = A\sigma$, then 
$\Sigma \vdash B = B\sigma$ and
$A\sigma\equivpnf B\sigma$ by a closed derivation.

\item\label{closing:redpnf} If 
$A \redpnf B$, 
$Y \eqdef (\fv(A) \setminus \dom(A)) \cup (\fv(B)\setminus \dom(B))$, and
$\Sigma \vdash A = A\sigma$, then 
$\Sigma \vdash B = B\sigma$ and 
$A\sigma\redpnf B\sigma$ by a closed derivation.
\end{enumerate}
\end{lemma}
\begin{proof}
We prove the lemma by induction on the derivations.
All cases are straightforward.
In the cases $\brn{Rewrite}'$, $\brn{Rewrite}''$, and $\brn{Else}'$,
we use that the equational theory is closed under substitution of 
names for variables.
In the cases of transitivity of $\equivP$ and $\equivpnf$, 
we use the induction hypothesis
and notice that, if a variable does not occur free in a certain process,
then we can substitute it or not without changing the result.
We use a similar argument when we apply a structural equivalence
step and a reduction step.
In the case $\brn{Struct}'$, we additionally use Lemma~\ref{lem:instance-equiv}.
\end{proof}

\begin{lemma}\label{lem:closed-interm}
\begin{enumerate}
\item \label{prop:closed-interm-redP}
  If $P \redP Q$ and $P$ is closed, then $P \redP Q$ by a derivation closed on the left.
\item \label{prop:closed-interm-redpnf}
If $A \redpnf B$ and $A$ is closed, then $A \redpnf B$ by a derivation closed on the left.
\item \label{prop:closed-interm-ltrP}
If $P \ltrP{\alpha} A$, $P$ is closed, and $\alpha$ is $\Res{x}\Snd{N}{x}$ or
$\Rcv{N}{M}$ for some ground term $N$, then
$P \ltrP{\alpha} A$ 
by a derivation closed on the left.
 
\item \label{prop:closed-interm-ltrpnf}
  If $A \ltrpnf{\alpha} B$,
$A$ is closed, and $\fv(\alpha) \subseteq \dom(A)$, then $A \ltrpnf{\alpha}
B$ 
by a derivation closed on the left,
and the label $\alpha'$ of the transition $P \ltrP{\alpha'} B'$
used in the definition of $A \ltrpnf{\alpha} B$ is closed.
\end{enumerate}
\end{lemma}
  \begin{proof} In the proof below, $Y$ ranges over sets of variables and $\sigma$ maps $Y$ to pairwise distinct fresh names.
The first two properties immediately follow from Lemma~\ref{lem:closing}.
For instance, if $P \redP Q$ and $P$ is closed,
let $Y = \fv(P) \cup \fv(Q) = \fv(Q)$.
We have $P = P\sigma$, so a fortiori $\Sigma \vdash P = P\sigma$. 
By Lemma~\ref{lem:closing}\eqref{closing:redP},
we have
$\Sigma \vdash Q = Q\sigma$ and
$P = P\sigma \redP Q\sigma$ by a closed derivation, 
so $Q \equivP Q\sigma$.
Hence $P \redP Q$ by a derivation closed on the left.

\emph{Property~3:} Suppose that $\alpha = \Rcv{N}{M}$ where
$N$ is a ground term. By Lemma~\ref{lem:caract-ltrP}, 
$P \equivP \Res{\vect n}(\Rcv{N}{x}.P_1 \parop P_2)$,
$A \equiv \Res{\vect n}(P_1\{\subst{M}{x}\} \parop P_2)$,
and $\{\vect n\} \cap \fn(\alpha) = \emptyset$.
Let $Y = \fv(\Rcv{N}{x}.P_1 \parop P_2)$.
We rename $x$ so that $x \notin Y$.
Since $P$ is closed, $P = P\sigma$, so a fortiori 
$\Sigma \vdash P = P\sigma$. 
By Lemma~\ref{lem:closing}\eqref{closing:equivP},
$P = P\sigma \equivP \Res{\vect n}(\Rcv{N}{x}.P_1\sigma \parop P_2\sigma)$
by a closed derivation
and $\Sigma \vdash \Res{\vect n}(\Rcv{N}{x}.P_1 \parop P_2)
= \Res{\vect n}(\Rcv{N}{x}.P_1\sigma \parop P_2\sigma)$,
so $\Sigma \vdash P_1 = P_1 \sigma$ and $\Sigma \vdash P_2 = P_2 \sigma$.
Hence $A \equiv \Res{\vect n}(P_1\{\subst{M}{x}\} \parop P_2)
\equiv \Res{\vect n}(P_1\sigma\{\subst{M}{x}\} \parop P_2\sigma)$.
We derive 
\begin{align*}
&\Rcv{N}{x}.P_1\sigma \ltrP{\Rcv{N}{M}} P_1\sigma\{\subst{M}{x}\}\tag*{by $\brn{In}'$}\\
&\Rcv{N}{x}.P_1\sigma \parop P_2\sigma \ltrP{\Rcv{N}{M}} P_1\sigma\{\subst{M}{x}\} \parop P_2\sigma\tag*{by $\brn{Par}'$}\\
&\Res{\vect n}(\Rcv{N}{x}.P_1\sigma \parop P_2\sigma) \ltrP{\Rcv{N}{M}} \Res{\vect n}(P_1\sigma\{\subst{M}{x}\} \parop P_2\sigma)\tag*{by $\brn{Scope}'$}\\
&P \ltrP{\Rcv{N}{M}} A\tag*{by $\brn{Struct}'$}
\end{align*}
using the previous closed derivation of 
$P \equivP \Res{\vect n}(\Rcv{N}{x}.P_1\sigma \parop P_2\sigma)$.
In the resulting derivation, 
all intermediate processes before $\ltrP{\alpha}$ are closed.
The case $\alpha = \Res{x}\Snd{N}{x}$ where $N$ is a ground term can
be proved in a similar way, or by using
Lemma~\ref{lem:closing}\eqref{closing:ltrP} since $\alpha$ is closed.

\emph{Property~4:} 
Suppose that $A$ is closed and $A \ltrpnf{\alpha} B$.
Then $A \equivpnf \Res{\vect n}(\sigma' \parop P)$, 
$P \ltrP{\alpha'} B'$, $B \equiv \Res{\vect n}(\sigma' \parop B')$, 
$\fv(\sigma') \cap \bv(\alpha') = \emptyset$, 
$\Sigma \vdash \alpha \sigma' = \alpha'$, and the names~$\vect n$ 
do not occur in $\alpha$.
Let $Y = (\fv(\sigma') \setminus \dom(\sigma')) \cup \fv(P) \cup \fv(\alpha') \cup (\fv(B') \setminus \dom(B'))$.
Then 
$\Sigma \vdash A = A\sigma$, so by Lemma~\ref{lem:closing}\eqref{closing:equivpnf},
$A\sigma \equivpnf \Res{\vect n}(\sigma'\sigma \parop P\sigma)$
by a closed derivation and 
$\Sigma \vdash  \Res{\vect n}(\sigma' \parop P) = \Res{\vect n}(\sigma'\sigma \parop P\sigma)$, so 
$\Sigma \vdash \sigma' = \sigma' \sigma$ and $\Sigma \vdash P = P \sigma$.
Hence by Lemma~\ref{lem:closing}\eqref{closing:ltrP},
$P\sigma \ltrP{\alpha'\sigma} B'\sigma$ by a derivation
closed on the left,
and $B' \equiv B'\sigma$.
So $B \equiv \Res{\vect n}(\sigma'\sigma \parop B'\sigma)$,
$\fv(\sigma'\sigma) \cap \bv(\alpha'\sigma) = \emptyset$, 
$\Sigma \vdash \alpha \sigma'\sigma = \alpha'\sigma$, and the names $\vect n$ 
do not occur in $\alpha$.
Hence we obtain the desired derivation using 
$\alpha'\sigma$ instead of $\alpha'$, 
$\sigma'\sigma$ instead of $\sigma'$, 
$P\sigma$ instead of $P$, and 
$B'\sigma$ instead of~$B'$. 
\end{proof}

\subsection{Decomposition and Composition of Reductions on Partial Normal Forms}\label{app:comppnf}

The next few lemmas allow us to analyze internal reductions and labelled transitions on partial normal forms.
Most of these lemmas describe the possible reductions of a process.
Lemma~\ref{lem:comp-ltr} composes two reductions: if two processes perform labelled transitions, one an output transition
and the other an input transition on the same channel, then their parallel composition performs an internal reduction. 

\begin{lemma}\label{lem:decomp-ltrP}
Suppose that $P_0$ is closed, $\alpha$ is $\Res{x}\Snd{N'}{x}$ or $\Rcv{N'}{M'}$ for some ground term $N'$, and $P_0 \ltrP{\alpha} A$. Then one of the following cases holds:
\begin{enumerate}

\item $P_0 = P \parop Q$ and either $P \ltrP{\alpha} A'$ and $A \equiv A' \parop Q$, or $Q \ltrP{\alpha} A'$ and $A \equiv P \parop A'$, for some $P$, $Q$, and~$A'$;

\item $P_0 = \Res{n} P$, $P \ltrP{\alpha} A'$, and $A \equiv \Res{n} A'$ for some $P$, $A'$, and $n$ that does not occur in $\alpha$;

\item $P_0 = \Repl{P}$, $P \ltrP{\alpha} A'$, and $A \equiv A' \parop \Repl{P}$ for some $P$ and $A'$;

\item $P_0 = \Rcv{N}{x}.P$, $\alpha = \Rcv{N'}{M'}$, $\Sigma \vdash N = N'$, and $A \equiv P\{\subst{M'}{x}\}$ for some $N$, $x$, $P$, $N'$, and $M'$;

\item $P_0 = \Snd{N}{M}.P$, $\alpha = \Res{x}\Snd{N'}{x}$, $\Sigma \vdash N = N'$, $x \notin \fv(P_0)$, and $A \equiv P \parop \{\subst{M}{x}\}$ for some $N$, $M$, $P$, $x$, and $N'$.
\end{enumerate}
\end{lemma}
\begin{proof}
An obvious approach for proving this result is to proceed
by induction on the derivation of $P_0 \ltrP{\alpha} A$.
However, the statement is not strong enough to provide an inductive
invariant. For instance, in case $P_0 \ltrP{\alpha} A$ is derived
from $P_0 = \Res{n}\Res{n'}P \equivP \Res{n'}\Res{n}P \ltrP{\alpha} A$, 
we can apply the statement to $\Res{n'}\Res{n}P \ltrP{\alpha} A$ 
by induction hypothesis, because 
$\Res{n'}\Res{n}P \ltrP{\alpha} A$ is derived by a derivation smaller 
than that of $P_0 \ltrP{\alpha} A$.
Hence, we obtain that
$\Res{n}P \ltrP{\alpha} A'$ and $A \equiv \Res{n'}A'$ for some $A'$.
However, we cannot apply the result to $\Res{n}P \ltrP{\alpha} A'$,
because we are not sure that the derivation of $\Res{n}P \ltrP{\alpha} A'$
is smaller than that of $P_0 \ltrP{\alpha} A$.
For this reason, we strengthen the induction hypothesis as shown below,
to make sure that it can be applied to a labelled transition,
such as $\Res{n}P \ltrP{\alpha} A'$, obtained by applying the desired result
itself.

Let $\prop(P_0 \ltrP{\alpha} A)$ be the greatest property such that $\prop(P_0 \ltrP{\alpha} A)$ holds if and only if one of the following cases holds:
\begin{enumerate}
\item\label{prop1:PAR} $P_0 = P \parop Q$ and either 
$P \ltrP{\alpha} A'$, $\prop(P \ltrP{\alpha} A')$, and $A \equiv A' \parop Q$, 
or 
$Q \ltrP{\alpha} A'$, $\prop(Q \ltrP{\alpha} A')$, and $A \equiv P \parop A'$,
for some $P$, $Q$, and $A'$;

\item\label{prop1:RES} $P_0 = \Res{n} P$, $P \ltrP{\alpha} A'$, $\prop(P \ltrP{\alpha} A')$, and $A \equiv \Res{n} A'$ for some $P$, $A'$, and $n$ that does not occur in $\alpha$;

\item\label{prop1:REPL} $P_0 = \Repl{P}$, $P \ltrP{\alpha} A'$, $\prop(P \ltrP{\alpha} A')$, and $A \equiv A' \parop \Repl{P}$ for some $P$ and $A'$;

\item\label{prop1:RCV} $P_0 = \Rcv{N}{x}.P$, $\alpha = \Rcv{N'}{M'}$, $\Sigma \vdash N = N'$, and $A \equiv P\{\subst{M'}{x}\}$ for some $N$, $x$, $P$, $N'$, and $M'$;

\item\label{prop1:SND} $P_0 = \Snd{N}{M}.P$, $\alpha = \Res{x}\Snd{N'}{x}$, $\Sigma \vdash N = N'$, $x \notin \fv(P_0)$, and $A \equiv P \parop \{\subst{M}{x}\}$ for some $N$, $M$, $P$, $x$, and $N'$.
\end{enumerate}
Let us show that, if $P_0 \ltrP{\alpha} A$ is derived by a derivation closed on the left, then $\prop(P_0 \ltrP{\alpha} A)$, by induction on the derivation of $P_0 \ltrP{\alpha} A$. 
\begin{itemize}
\item Case $\brn{In}'$. We have $P_0 = \Rcv{N}{x}.P$, $\alpha = \Rcv{N}{M}$, and
$A = P\{\subst{M}{x}\}$, so we are in Case~\ref{prop1:RCV} of $\prop(P_0 \ltrP{\alpha} A)$ with $N' = N$ and $M' = M$.
\item Case $\brn{Out-Var}'$. We have $P_0 = \Snd{N}{M}.P$, $\alpha = \Res{x}\Snd{N}{x}$, $x \notin \fv(P_0)$, and $A = P \parop \{\subst{M}{x}\}$, so we are in Case~\ref{prop1:SND} of $\prop(P_0 \ltrP{\alpha} A)$ with $N' = N$.
\item Case $\brn{Scope}'$. We have $P_0 = \Res{n} P$, $n$ does not occur in $\alpha$, $P \ltrP{\alpha} A'$, and $A = \Res{n} A'$ for some $A'$. We obtain $\prop(P \ltrP{\alpha} A')$ by induction hypothesis, so we are in Case~\ref{prop1:RES} of $\prop(P_0 \ltrP{\alpha} A)$.
\item Case $\brn{Par}'$. We have $P_0 = P \parop Q$, $P \ltrP{\alpha} A'$ and $A = A' \parop Q$ for some $P$, $Q$, and $A'$. We obtain $\prop(P \ltrP{\alpha} A')$ by induction hypothesis, so we are in Case~\ref{prop1:PAR} of $\prop(P_0 \ltrP{\alpha} A)$.
\item Case $\brn{Struct}'$. We have $P_0 \equivP Q_0 \ltrP{\alpha} B \equiv A$. The case in which $P_0 = Q_0$ is obvious. Let us consider the case in which the structural equivalence $P_0 \equivP Q_0$ consists of applying a single structural equivalence step. (The case in which it consists of several steps can be transformed into several applications of $\brn{Struct}'$.) 
The process $Q_0$ is closed, and by induction hypothesis $\prop(Q_0 \ltrP{\alpha} B)$. We show that, if $P_0 \equivP Q_0 \ltrP{\alpha} B \equiv A$, $\prop(Q_0 \ltrP{\alpha} B)$, and all processes in the derivation of $P_0 \equivP Q_0$ are closed, then $\prop(P_0 \ltrP{\alpha} A)$, by induction on the derivation of $P_0 \equivP Q_0$.
\begin{itemize}
\item Case $P_0 = Q_0\parop \nil \equivP Q_0$. We have $Q_0 \ltrP{\alpha} B$, $\prop(Q_0 \ltrP{\alpha} B)$, and $A \equiv B \equiv B \parop \nil$, so we are in Case~\ref{prop1:PAR} of $\prop(P_0 \ltrP{\alpha} A)$.

\item Case $P_0 \equivP P_0 \parop \nil$. Since $\prop(P_0 \parop \nil \ltrP{\alpha} B)$, we have 
either $P_0 \ltrP{\alpha} B'$, $\prop(P_0 \ltrP{\alpha} B')$, and $B \equiv B' \parop \nil$ for some $B'$, or 
$\nil \ltrP{\alpha} B'$, $\prop(\nil \ltrP{\alpha} B')$, and $B \equiv P \parop B'$ for some $B'$. By definition of $\prop$, $\prop(\nil \ltrP{\alpha} B')$ is impossible, so we are in the first case: $\prop(P_0 \ltrP{\alpha} B')$ and $A \equiv B'$. Since $\prop(P_0 \ltrP{\alpha} B')$ is invariant by structural equivalence applied to $B'$, we can then conclude that $\prop(P_0 \ltrP{\alpha} A)$.

\item Case $P_0 = P \parop (Q \parop R) \equivP (P \parop Q) \parop R$.
We have $\prop((P \parop Q) \parop R  \ltrP{\alpha} B)$, so 
either $P \parop Q \ltrP{\alpha} B'$, $\prop(P \parop Q \ltrP{\alpha} B')$, and $B \equiv B' \parop R$ for some $B'$, 
or $R \ltrP{\alpha} B'$, $\prop(R \ltrP{\alpha} B')$, and $B \equiv (P \parop Q) \parop B'$ for some $B'$.
In the first case, either $P \ltrP{\alpha} B''$, $\prop(P \ltrP{\alpha} B'')$, and $B' \equiv B'' \parop Q$ for some $B''$, 
or $Q \ltrP{\alpha} B''$, $\prop(Q \ltrP{\alpha} B'')$, and $B' \equiv P \parop B''$ for some $B''$.
Consider for instance the last case, in which $Q$ reduces. The other two cases are similar. 
Since $R$ is closed, $\bv(\alpha) \cap \fv(R) = \emptyset$, so by $\brn{Par}'$, 
$Q \parop R \ltrP{\alpha} B'' \parop R$, $\prop(Q \parop R \ltrP{\alpha} B'' \parop R)$, and $A \equiv B \equiv B' \parop R \equiv (P \parop B'') \parop R \equiv P \parop (B'' \parop R)$, so we are in Case~\ref{prop1:PAR} of $\prop(P_0 \ltrP{\alpha} A)$.

\item Case $P_0 = (P \parop Q) \parop R \equivP P \parop (Q \parop R)$. This case is similar to the previous one.

\item Case $P_0 = P \parop Q \equivP Q \parop P$. (This case is its own symmetric.) This case is immediate, since the desired result is invariant by swapping $P$ and $Q$.

\item Case $P_0 = \Repl{P} \equivP P \parop \Repl{P}$. Since $\prop(P \parop \Repl{P}  \ltrP{\alpha} B)$, 
either $P \ltrP{\alpha} B'$, $\prop(P \ltrP{\alpha} B')$, and $B \equiv B' \parop \Repl{P}$ for some $B'$, 
or $\Repl{P} \ltrP{\alpha} B'$, $\prop(\Repl{P} \ltrP{\alpha} B')$, and $B \equiv P \parop B'$ for some $B'$.
In the first case, we are in Case~\ref{prop1:REPL} of $\prop(P_0 \ltrP{\alpha} A)$ with $A' = B'$.
In the second case, since  $\prop(\Repl{P} \ltrP{\alpha} B')$, we have  $P \ltrP{\alpha} B''$, $\prop(P \ltrP{\alpha} B'')$, and $B' \equiv B'' \parop \Repl{P}$ for some $B''$. Hence, $A \equiv B \equiv B' \parop P \equiv B'' \parop \Repl{P}$, so we are in Case~\ref{prop1:REPL} of $\prop(P_0 \ltrP{\alpha} A)$ with $A' = B''$.

\item Case $P_0 = P \parop \Repl{P} \equivP \Repl{P}$. 
Since $\prop(\Repl{P}  \ltrP{\alpha} B)$, we have  $P \ltrP{\alpha} B'$, $\prop(P \ltrP{\alpha} B')$, and $B \equiv B' \parop \Repl{P}$ for some $B'$.
We are in Case~\ref{prop1:PAR} of $\prop(P_0 \ltrP{\alpha} A)$ with $A' = B'$.

\item Case $P_0 = \Res{n}\nil \equivP \nil$. We have that $\prop(\nil  \ltrP{\alpha} B)$ is impossible, so this case never happens.

\item Case $P_0 = \nil \equivP \Res{n}\nil$. Since $\prop(\Res{n}\nil  \ltrP{\alpha} B)$, we have $\prop(\nil  \ltrP{\alpha} B')$, which is impossible,  so this case never happens.

\item Case $P_0 = \Res{n}\Res{n'}P \equivP \Res{n'}\Res{n}P$. (This case is its own symmetric.) We rename $n$ and $n'$ so that they do not occur in $\alpha$.
Since $\prop(\Res{n'}\Res{n}P  \ltrP{\alpha} B)$, we have 
$\Res{n}P \ltrP{\alpha} B'$, $\prop(\Res{n}P \ltrP{\alpha} B')$, and $B \equiv \Res{n'} B'$ for some $B'$, so
$P \ltrP{\alpha} B''$, $\prop(P \ltrP{\alpha} B'')$, and $B' \equiv \Res{n} B''$ for some $B''$.
Hence, by \brn{Scope'}, $\Res{n'}P \ltrP{\alpha} \Res{n'}B''$, 
$\prop(\Res{n'}P \ltrP{\alpha} \Res{n'}B'')$, and we have $A \equiv B \equiv \Res{n'}B' \equiv \Res{n'}\Res{n}B'' \equiv \Res{n}\Res{n'}B''$, so we are in Case~\ref{prop1:RES} of $\prop(P_0 \ltrP{\alpha} A)$ with $A' = \Res{n'}B''$.

\item Case $P_0 = P \parop \Res{n}Q \equivP \Res{n}(P \parop Q)$ and $n \notin \fn(P)$. We rename $n$ so that it does not occur in $\alpha$.
Since $\prop(\Res{n}(P \parop Q)  \ltrP{\alpha} B)$, we have 
$P \parop Q \ltrP{\alpha} B'$, $\prop(P \parop Q \ltrP{\alpha} B')$, and $B \equiv \Res{n} B'$ for some $B'$, 
so either
$P \ltrP{\alpha} B''$, $\prop(P \ltrP{\alpha} B'')$, and $B' \equiv B'' \parop Q$ for some $B''$, or 
$Q \ltrP{\alpha} B''$, $\prop(Q \ltrP{\alpha} B'')$, and $B' \equiv P \parop B''$ for some $B''$.
In the first case, we rename $n$ in $Q$ so that $n \notin \fn(B'')$, hence $A \equiv B \equiv \Res{n}B' \equiv \Res{n}(B'' \parop Q) \equiv B'' \parop \Res{n}Q$, so we are in Case~\ref{prop1:PAR} of $\prop(P_0 \ltrP{\alpha} A)$ with $A' = B''$.
In the second case, by $\brn{Scope}'$, $\Res{n}Q \ltrP{\alpha} \Res{n}B''$, $\prop(\Res{n}Q \ltrP{\alpha} \Res{n}B'')$, and $A \equiv B \equiv \Res{n}B' \equiv \Res{n}(P \parop B'') \equiv P \parop \Res{n}B''$ since $n \notin \fn(P)$, so we are in Case~\ref{prop1:PAR} of $\prop(P_0 \ltrP{\alpha} A)$ with $A' = \Res{n}B''$.

\item Case $P_0 = \Res{n}(P \parop Q) \equivP P \parop \Res{n}Q$ and $n \notin \fn(P)$. This case is fairly similar to the previous one.

\item Case $P_0 = P_1 \{\subst{M}{x}\} \equivP P_1\{\subst{N}{x}\}$ and $\Sigma \vdash M=N$. (This case is its own symmetric.)
We have $\prop(P_1\{\subst{N}{x}\}\ltrP{\alpha} B)$.
We show by induction on the syntax of $P_1$ that, if $\prop(P_1\{\subst{N}{x}\}\ltrP{\alpha} B)$, $\Sigma \vdash M=N$, and $A \equiv B$, then $\prop(P_1 \{\subst{M}{x}\} \ltrP{\alpha} A)$. 
\begin{itemize}

\item Case $P_1 = P \parop Q$. 
We have $\prop(P\{\subst{N}{x}\}  \parop Q\{\subst{N}{x}\} \ltrP{\alpha} B)$, 
so we have either $P\{\subst{N}{x}\} \ltrP{\alpha} B'$, $\prop(P\{\subst{N}{x}\} \ltrP{\alpha} B')$, and $B \equiv B' \parop Q\{\subst{N}{x}\}$, or $Q\{\subst{N}{x}\} \ltrP{\alpha} B'$, $\prop(Q\{\subst{N}{x}\} \ltrP{\alpha} B')$, and $B \equiv P\{\subst{N}{x}\} \parop B'$ for some $B'$. 
Hence $P_1 \{\subst{M}{x}\} = P\{\subst{M}{x}\}  \parop Q\{\subst{M}{x}\}$ and either $P\{\subst{M}{x}\} \equivP P\{\subst{N}{x}\} \ltrP{\alpha} B'$, $\prop(P\{\subst{M}{x}\} \ltrP{\alpha} B')$ by induction hypothesis, and $A \equiv B \equiv B' \parop Q\{\subst{N}{x}\} \equiv B' \parop Q\{\subst{M}{x}\} $, or $Q\{\subst{M}{x}\} \equivP Q\{\subst{N}{x}\} \ltrP{\alpha} B'$, $\prop(Q\{\subst{M}{x}\} \ltrP{\alpha} B')$ by induction hypothesis, and $A \equiv B \equiv P\{\subst{N}{x}\} \parop B' \equiv P\{\subst{M}{x}\} \parop B' $, so we are in Case~\ref{prop1:PAR} of $\prop(P_1 \{\subst{M}{x}\} \ltrP{\alpha} A)$.

\item Case $P_1 = \Res{n}P$. We rename $n$ so that it does not occur in $\alpha$. We have $\prop(\Res{n}P\{\subst{N}{x}\}  \ltrP{\alpha} B)$, so $P\{\subst{N}{x}\}\ltrP{\alpha} B'$, $\prop(P\{\subst{N}{x}\}  \ltrP{\alpha} B')$, and $B \equiv \Res{n}B'$ for some $B'$. Hence $P_1 \{\subst{M}{x}\} = \Res{n}P\{\subst{M}{x}\}$, $P\{\subst{M}{x}\} \equivP P\{\subst{N}{x}\}\ltrP{\alpha} B'$, $\prop(P\{\subst{M}{x}\}  \ltrP{\alpha} B')$ by induction hypothesis, and $A \equiv B \equiv \Res{n}B'$, so we are in Case~\ref{prop1:RES} of $\prop(P_1 \{\subst{M}{x}\} \ltrP{\alpha} A)$.

\item Case $P_1 = \Repl{P}$. We have $\prop(\Repl{P\{\subst{N}{x}\}}  \ltrP{\alpha} B)$, so $P\{\subst{N}{x}\}\ltrP{\alpha} B'$, $\prop(P\{\subst{N}{x}\}\ltrP{\alpha} B')$, and $B \equiv B' \parop \Repl{P\{\subst{N}{x}\}}$ for some $B'$. Hence $P_1 \{\subst{M}{x}\} = \Repl{P\{\subst{M}{x}\}}$, $P\{\subst{M}{x}\} \equivP P\{\subst{N}{x}\}\ltrP{\alpha} B'$, $\prop(P\{\subst{M}{x}\}\ltrP{\alpha} B')$ by induction hypothesis, and $A \equiv B \equiv B' \parop \Repl{P\{\subst{N}{x}\}} \equiv B' \parop \Repl{P\{\subst{M}{x}\}}$, so we are in Case~\ref{prop1:REPL} of $\prop(P_1 \{\subst{M}{x}\} \ltrP{\alpha} A)$.

\item Case $P_1 = \Rcv{N_1}{x_1}.P$. We rename $x_1$ so that $x_1 \neq x$. 
We have $\prop(\Rcv{N_1\{\subst{N}{x}\}}{x_1}.P\{\subst{N}{x}\} \ltrP{\alpha} B)$, so $\alpha = \Rcv{N'}{M'}$, $\Sigma \vdash N_1\{\subst{N}{x}\} = N'$,
and $B \equiv P\{\subst{N}{x}\}\{\subst{M'}{x_1}\}$. So
$\Sigma \vdash N_1\{\subst{M}{x}\} = N'$,
and $A \equiv B \equiv P\{\subst{M}{x}\}\{\subst{M'}{x_1}\}$,
hence $\prop(\Rcv{N_1\{\subst{M}{x}\}}{x_1}.P\{\subst{M}{x}\} \ltrP{\alpha} A)$,
so we are in Case~\ref{prop1:RCV} of $\prop(P_1\{\subst{M}{x}\} \ltrP{\alpha} A)$.

\item Case $P_1 = \Snd{N_1}{M_1}.P$.
We have $\prop(\Snd{N_1\{\subst{N}{x}\}}{M_1\{\subst{N}{x}\}}.P\{\subst{N}{x}\} \ltrP{\alpha} B)$, so $\alpha = \Res{x_1}\Snd{N'}{x_1}$, $\Sigma \vdash N_1\{\subst{N}{x}\} = N'$,
and $B \equiv P\{\subst{N}{x}\} \parop \{\subst{M_1\{\subst{N}{x}\}}{x_1}\}$. So
$\Sigma \vdash N_1\{\subst{M}{x}\} = N'$,
and $A \equiv B \equiv P\{\subst{N}{x}\} \parop \{\subst{M_1\{\subst{M}{x}\}}{x_1}\}$,
hence $\prop(\Snd{N_1\{\subst{M}{x}\}}{M_1\{\subst{M}{x}\}}.P\{\subst{M}{x}\} \ltrP{\alpha} A)$,
so we are in Case~\ref{prop1:SND} of $\prop(P_1\{\subst{M}{x}\} \ltrP{\alpha} A)$.

\end{itemize}
Using this result, we obtain $\prop(P_0 \ltrP{\alpha} A)$.

\item Case $P_0 = P \parop Q \equivP P' \parop Q$ knowing $P \equivP P'$.
Since $\prop(P' \parop Q \ltrP{\alpha} B)$, 
either 
$P' \ltrP{\alpha} A'$, $\prop(P' \ltrP{\alpha} A')$, and $A \equiv A' \parop Q$ for some $A'$, or 
$Q \ltrP{\alpha} A'$, $\prop(Q \ltrP{\alpha} A')$, and $A \equiv P' \parop A'$ for some $A'$.
In the first case, by $\brn{Struct}'$, $P \ltrP{\alpha} A'$. By induction hypothesis, $\prop(P \ltrP{\alpha} A')$. Moreover, $A \equiv A' \parop Q$,
so we are in Case~\ref{prop1:PAR} of $\prop(P_0 \ltrP{\alpha} A)$.
In the second case, $Q \ltrP{\alpha} A'$, $\prop(Q \ltrP{\alpha} A')$, and $A \equiv P' \parop A' \equiv P \parop A'$,
so we are in Case~\ref{prop1:PAR} of $\prop(P_0 \ltrP{\alpha} A)$.

\item Case $P_0 = \Res{n}P \equivP \Res{n}P'$ knowing $P \equivP P'$. 
We rename $n$ so that it does not occur in $\alpha$.
Since $\prop(\Res{n}P' \ltrP{\alpha} B)$, we have
$P' \ltrP{\alpha} A'$, $\prop(P' \ltrP{\alpha} A')$, and $A \equiv \Res{n}A'$ for some $A'$.
By $\brn{Struct}'$, $P \ltrP{\alpha} A'$. By induction hypothesis, $\prop(P \ltrP{\alpha} A')$. Moreover, $A \equiv \Res{n}A'$, so we are in Case~\ref{prop1:RES} of $\prop(P_0 \ltrP{\alpha} A)$.

\end{itemize}
\end{itemize}
Since $P_0$ is closed and $\alpha$ is $\Res{x}\Snd{N'}{x}$ or $\Rcv{N'}{M'}$ for some ground term $N'$, by Lemma~\ref{lem:closed-interm}\eqref{prop:closed-interm-ltrP}, there exists a derivation of $P_0 \ltrP{\alpha} A$ closed on the left. So by applying the previous result, $\prop(P_0 \ltrP{\alpha} A)$, which yields the desired property.
\end{proof}

\begin{lemma}\label{lem:decomp-ltrpnf}
If $\Res{\vect n}(\sigma \parop P)$ is a closed normal process,
$\Res{\vect n}(\sigma \parop P) \ltrpnf{\alpha} A$, 
$\fv(\alpha) \subseteq \dom(\sigma)$,
and the elements of $\vect n$ do not occur in $\alpha$,
then 
$P \ltrP{\alpha\sigma} A'$, $A \equiv \Res{\vect n}(\sigma \parop A')$, and $\bv(\alpha) \cap \dom(\sigma) = \emptyset$ for some $A'$.
\end{lemma}
\begin{proof}
By Lemma~\ref{lem:closed-interm}\eqref{prop:closed-interm-ltrpnf}, we consider a derivation of $\Res{\vect n}(\sigma \parop P) \ltrpnf{\alpha} A$ closed on the left and the label $\alpha'$ below is closed.
By definition of $\ltrpnf{\alpha}$, we have $\Res{\vect n}(\sigma \parop P) \equivpnf \Res{\vect n'}(\sigma' \parop P')$, $P' \ltrP{\alpha'} B$, $A \equiv \Res{\vect n'}(\sigma' \parop B)$, $\Sigma \vdash \alpha\sigma' = \alpha'$ for some $\vect n'$, $\sigma'$, $P'$, $\alpha'$, $B$ such that the elements of $\vect n'$ do not occur in $\alpha$ and $\fv(\sigma') \cap \bv(\alpha') = \emptyset$.
By applying Lemma~\ref{lem:caract-ltrP} back and forth, since $P' \ltrP{\alpha'} B$ and $\Sigma \vdash \alpha\sigma' = \alpha'$, we have $P' \ltrP{\alpha\sigma'} B$.
We proceed by induction on the derivation of $\Res{\vect n}(\sigma \parop P) \equivpnf \Res{\vect n'}(\sigma' \parop P')$.
\begin{itemize}

\item Base case: $\Res{\vect n}(\sigma \parop P) = \Res{\vect n'}(\sigma' \parop P')$, and the desired result holds.

\item Transitivity: the result is proved by applying the induction hypothesis twice.

\item Case $\brn{Plain}''$: $\vect n = \vect n'$, $\sigma = \sigma'$, $P \equivP P' \ltrP{\alpha\sigma} B$, and $A \equiv \Res{\vect n}(\sigma \parop B)$, so the result holds.

\item Case $\brn{New-C}''$: $\vect n'$ is a reordering of $\vect n$, 
$\Res{\vect n}(\sigma \parop P) \equivpnf \Res{\vect n'}(\sigma \parop P)$, $P \ltrP{\alpha\sigma} B$ and $A \equiv \Res{\vect n'}(\sigma \parop B)$, 
so $P \ltrP{\alpha\sigma} B$ and $A \equiv \Res{\vect n'}(\sigma \parop B) \equiv \Res{\vect n}(\sigma \parop B)$, hence the result holds.

\item Case $\brn{New-Par}''$: $P = \Res{n'}P'$, $\Res{\vect n}(\sigma \parop \Res{n'}P') \equivpnf \Res{\vect n, n'}(\sigma \parop P')$, $P' \ltrP{\alpha\sigma} B$, and $A \equiv \Res{\vect n,n'}(\sigma \parop B)$ where
the elements of $\vect n,n'$ do not occur in $\alpha$ and $n' \notin \fn(\sigma)$.
By $\brn{Scope}'$, $P = \Res{n'}P'\ltrP{\alpha\sigma} \Res{n'}B$
and by $\brn{New-Par}$, $A \equiv \Res{\vect n,n'}(\sigma \parop B) \equiv
\Res{\vect n}(\sigma \parop \Res{n'}B)$, hence the result holds.

\item Case $\brn{New-Par}''$ reversed: 
$\Res{\vect n, n'}(\sigma \parop P) \equivpnf \Res{\vect n}(\sigma \parop \Res{n'}P)$, $\Res{n'}P \ltrP{\alpha\sigma} B$ and $A \equiv \Res{\vect n}(\sigma \parop B)$ where the elements of $\vect n$ do not occur in $\alpha$ and $n' \notin \fn(\sigma)$.
We rename $n'$ so that it does not occur in $\alpha$.
By Lemma~\ref{lem:decomp-ltrP}, $P \ltrP{\alpha\sigma} B'$ and $B \equiv \Res{n'}B'$
for some $B'$.
Hence, $P \ltrP{\alpha\sigma} B'$ and $A \equiv \Res{\vect n}(\sigma \parop B) 
\equiv \Res{\vect n}(\sigma \parop \Res{n'}B') 
\equiv \Res{\vect n, n'}(\sigma \parop B')$ by $\brn{New-Par}$, 
so the result holds.

\item Case $\brn{Rewrite}''$:
$\Res{\vect n}(\sigma \parop P) \equivpnf \Res{\vect n}(\sigma' \parop P)$, $P \ltrP{\alpha\sigma'} B$ and $A \equiv \Res{\vect n}(\sigma' \parop B)$
where $\dom(\sigma) = \dom(\sigma')$, $\Sigma \vdash x \sigma = x\sigma'$ for all $x \in \dom(\sigma)$, and $(\fv(x\sigma) \cup \fv(x\sigma')) \cap \dom(\sigma) = \emptyset$ for all $x \in \dom(\sigma)$.
Hence $P \ltrP{\alpha\sigma'} B$ and $\Sigma \vdash \alpha \sigma = \alpha \sigma'$, so by applying Lemma~\ref{lem:caract-ltrP} back and forth,
$P \ltrP{\alpha\sigma} B$.
Moreover, $A \equiv \Res{\vect n}(\sigma' \parop B)
\equiv \Res{\vect n}(\sigma \parop B)$ by several applications of $\brn{Rewrite}$, so the result holds.
\qedhere
\end{itemize}
\end{proof}

\begin{lemma}\label{lem:comp-ltr}
If $P$ and $Q$ are closed processes, $N$ is a ground term, $P \ltrP{\Rcv{N}{x}} A$, 
and $Q \ltrP{\Res{x}\Snd{N}{x}} B$, 
then 
$P \parop Q \redP R$ and $R\equiv \Res{x}(A \parop B)$ for some $R$.
\end{lemma}
\begin{proof}
By Lemma~\ref{lem:caract-ltrP}, we have
$P \equivP \Res{\vect n}(\Rcv{N}{y}.P_1 \parop P_2)$,
$A \equiv \Res{\vect n}(P_1\{\subst{x}{y}\} \parop P_2)$,
$\{\vect n\} \cap \fn(N) = \emptyset$
and
$Q \equivP \Res{\vect n'}(\Snd{N}{M}.Q_1 \parop Q_2)$,
$B \equiv \Res{\vect n'}(Q_1 \parop \{\subst{M}{x}\} \parop Q_2)$,
$\{\vect n'\} \cap \fn(N) = \emptyset$,
$x \notin \fv(\Snd{N}{M}.Q_1 \parop Q_2)$.

Let $Y = \fv(\Res{\vect n}(\Rcv{N}{y}.P_1 \parop P_2))$.
We rename $y$ so that $y \notin Y$.
Let $\sigma$ be a substitution from $Y$ to pairwise distinct fresh names.
Since $P$ and $N$ are closed, 
$P\sigma = P$ and $N\sigma = N$, so a fortiori
$\Sigma \vdash P = P\sigma$.
By Lemma~\ref{lem:closing}\eqref{closing:equivP}, $P \equivP \Res{\vect n}(\Rcv{N}{y}.P_1\sigma \parop P_2\sigma)$ and $\Sigma \vdash \Res{\vect n}(\Rcv{N}{y}.P_1 \parop P_2) = \Res{\vect n}(\Rcv{N}{y}.P_1\sigma \parop P_2\sigma)$.
Then, by introducing fresh names $\vect n_1, \vect n'_1$,
\[\begin{split}
P \parop Q 
&\equivP
\Res{\vect n_1, \vect n'_1}(\Rcv{N}{y}.P_1\sigma\{\subst{\vect n_1}{\vect n}\} 
\parop P_2\sigma\{\subst{\vect n_1}{\vect n}\} \\
&\phantom{\equivP\Res{\vect n_1, \vect n'_1}}
\parop \Snd{N}{M\{\subst{\vect n'_1}{\vect n'}\}}.Q_1\{\subst{\vect n'_1}{\vect n'}\} \parop Q_2\{\subst{\vect n'_1}{\vect n'}\})\\
&\redP 
\Res{\vect n_1, \vect n'_1}(P_1\sigma\{\subst{\vect n_1}{\vect n}\}\{\subst{M\{\subst{\vect n'_1}{\vect n'}\}}{y}\} 
\parop P_2\sigma\{\subst{\vect n_1}{\vect n}\} 
\parop Q_1\{\subst{\vect n'_1}{\vect n'}\} \parop Q_2\{\subst{\vect n'_1}{\vect n'}\}) = R
\end{split}\]
and
\[\begin{split}
\Res{x}(A \parop B) 
&\equiv 
\Res{x, \vect n_1, \vect n'_1}(P_1\sigma\{\subst{x}{y},\subst{\vect n_1}{\vect n}\}  \parop P_2\sigma\{\subst{\vect n_1}{\vect n}\}\\
&\phantom{\equiv \Res{x, \vect n_1, \vect n'_1}} \parop Q_1\{\subst{\vect n'_1}{\vect n'}\} \parop \{\subst{M\{\subst{\vect n'_1}{\vect n'}\}}{x}\} \parop Q_2\{\subst{\vect n'_1}{\vect n'}\})\\
&\equiv \Res{\vect n_1, \vect n'_1}(P_1\sigma\{\subst{M\{\subst{\vect n'_1}{\vect n'}\}}{y},\subst{\vect n_1}{\vect n}\}  \parop P_2\sigma\{\subst{\vect n_1}{\vect n}\} \parop Q_1\{\subst{\vect n'_1}{\vect n'}\} \parop Q_2\{\subst{\vect n'_1}{\vect n'}\})\\
& \equiv R
\end{split}\]
because $x$ is not free in $P_1\sigma$, $P_2\sigma$, $Q_1$, $Q_2$, 
since $\Res{\vect n}(\Rcv{N}{y}.P_1\sigma \parop P_2\sigma)$ is closed
and $x \notin \fv(\Snd{N}{M}.Q_1 \parop Q_2)$.
\end{proof}

\begin{lemma}\label{lem:decomp-redP}
Suppose that $P_0$ is a closed process and $P_0 \redP R$. Then 
one of the following cases holds:
\begin{enumerate}

\item $P_0 = P \parop Q$ for some $P$ and $Q$, and one of the following cases holds:
\begin{enumerate}
\item $P \redP P'$ and $R \equiv P' \parop Q$ for some $P'$,
\item $P \ltrP{\Rcv{N}{x}} A$, $Q \ltrP{\Res{x}\Snd{N}{x}} B$, and $R \equiv \Res{x}(A \parop B)$ for some $A$, $B$, $x$, and ground term $N$, 
\end{enumerate}
and two symmetric cases obtained by swapping $P$ and $Q$;

\item $P_0 = \Res{n}P$, $P \redP Q'$, and $R \equiv \Res{n} Q'$ for some $n$, $P$, and $Q'$;

\item $P_0 = \Repl{P}$, $P \parop P \redP Q'$, and $R \equiv Q' \parop \Repl{P}$ for some $P$ and $Q'$;

\item $P_0 = \IfThenElse{M}{N}{P}{Q}$ and either $\Sigma \vdash M = N$ and $R \equiv P$, or $\Sigma \vdash M \neq N$ and $R \equiv Q$, for some $M$, $N$, $P$, and $Q$.
\end{enumerate}
\end{lemma}
\begin{proof}
We proceed similarly to Lemma~\ref{lem:decomp-ltrP}.
Let $\prop(P_0 \redP R_0)$ be the greatest property such that $\prop(P_0 \redP R_0)$ holds if and only if one of the following cases holds:
\begin{enumerate}
\item\label{prop2:PAR} $P_0 = P \parop Q$ for some $P$ and $Q$, and one of the following cases holds:
\begin{enumerate}
\item $P \redP P'$, $\prop(P \redP P')$, and $R_0 \equiv P' \parop Q$ for some $P'$,
\item $P \ltrP{\Rcv{N}{x}} A$, $Q \ltrP{\Res{x}\Snd{N}{x}} B$, and $R_0 \equiv \Res{x}(A \parop B)$ for some $A$, $B$, $x$, and ground term $N$, 
\end{enumerate}
and two symmetric cases obtained by swapping $P$ and $Q$, named (a') and (b') respectively;

\item\label{prop2:RES} $P_0 = \Res{n} P$, $P \redP Q'$, $\prop(P \redP Q')$, and $R_0 \equiv \Res{n} Q'$ for some $n$, $P$, and $Q'$;

\item\label{prop2:REPL} $P_0 = \Repl{P}$, $P \parop P \redP Q'$, $\prop(P \parop P \redP Q')$, and $R_0 \equiv Q' \parop \Repl{P}$ for some $P$ and $Q'$;

\item\label{prop2:IF} $P_0 = \IfThenElse{M}{N}{P}{Q}$ and either $\Sigma \vdash M = N$ and $R_0 \equiv P$, or $\Sigma \vdash M \neq N$ and $R_0 \equiv Q$, for some $M$, $N$, $P$, and $Q$.

\end{enumerate}
Let us show that, if $P_0 \redP R_0$ is derived by a derivation closed on the left, then $\prop(P_0 \redP R_0)$, by induction on the derivation of $P_0 \redP R_0$.
\begin{itemize}
\item In the case $\brn{Comm}'$, $P_0 = P \parop Q$ where $P = \Snd{N}{M}.P'$, $Q = \Rcv{N}{x}.Q'$,
and $R_0 = P' \parop Q'\{\subst{M}{x}\}$,
so by choosing a fresh variable $y$,
$P \ltrP{\Res{y}\Snd{N}{y}} P' \parop \{\subst{M}{y}\}$, $Q \ltrP{\Rcv{N}{y}} Q'\{\subst{y}{x}\}$,
$\Res{y}(P' \parop \{\subst{M}{y}\} \parop Q'\{\subst{y}{x}\}) \equiv P' \parop Q'\{\subst{M}{x}\} \equiv R_0$, 
and $N$ is ground since $P_0$ is closed.
Therefore, we are in Case~\ref{prop2:PAR}.(a') of $\prop(P_0 \redP R_0)$.
\item In the case $\brn{Then}'$, we are in Case~\ref{prop2:IF} of $\prop(P_0 \redP R_0)$ with $P_0 = \IfThenElse{M}{M}{P}{Q}$ and $R_0 = P$.
\item In the case $\brn{Else}'$, we are in Case~\ref{prop2:IF} of $\prop(P_0 \redP R_0)$ with $P_0 = \IfThenElse{M}{N}{P}{Q}$,
$\Sigma \vdash M \neq N$, and $R_0 = Q$.
\item If we apply a reduction under an evaluation context $\CTX$, then 
$P_0 = \CTX[P] \redP \CTX[P'] = R_0$ is derived from $P \redP P'$.
By induction hypothesis, we have $\prop(P \redP P')$, and 
we are in Case~\ref{prop2:PAR}.(a) (respectively, \ref{prop2:PAR}.(a') or
\ref{prop2:RES}) of $\prop(P_0 \redP R_0)$ when $E$ is $E' \parop Q$ (respectively, $Q \parop E'$ or $\Res{n}E'$).
\item Finally, suppose that we use $\equivP$. We have $P_0 \equivP Q_0 \redP R_1 \equivP R_0$. The case in which $P_0 = Q_0$ is obvious by induction, since $R_1 \equivP R_0$ implies $R_1 \equiv R_0$ by Lemma~\ref{lem:struct-pnf-to-std}. 
Let us consider the case in which the structural equivalence $P_0 \equivP Q_0$ consists of applying a single structural equivalence step. (The case in which it consists of several steps can be transformed into several applications of the rule.) The process $Q_0$ is closed, and by induction hypothesis $\prop(Q_0 \redP R_1)$. We show that, if $P_0 \equivP Q_0 \redP R_1 \equiv R_0$, $\prop(Q_0 \redP R_1)$, and all processes in the derivation of $P_0 \equivP Q_0$ are closed, then $\prop(P_0 \redP R_0)$, by induction on the derivation of $P_0 \equivP Q_0$.
\begin{itemize}
\item Case $P_0 = Q_0\parop \nil \equivP Q_0$. We have $Q_0 \redP R_1$, $\prop(Q_0 \redP R_1)$, and $R_0 \equiv R_1 \equiv R_1 \parop \nil$, so we are in Case~\ref{prop2:PAR}.(a) of $\prop(P_0 \redP R_0)$.

\item Case $P_0 \equivP P_0 \parop \nil$. Since $\prop(P_0 \parop \nil \redP R_1)$, we have 
either 
\begin{enumerate}
\item $P_0 \redP R'$, $\prop(P_0 \redP R')$, and $R_1 \equiv R' \parop \nil$ for some $R'$;
\item $P_0 \ltrP{\Rcv{N}{x}} A$, $\nil \ltrP{\Res{x}\Snd{N}{x}} B$;
\item $P_0 \ltrP{\Res{x}\Snd{N}{x}} A$, $\nil \ltrP{\Rcv{N}{x}} B$; or
\item $\nil \redP R'$, $\prop(\nil \redP R')$, and $R_1 \equiv P_0 \parop R'$ for some $R'$.
\end{enumerate}
Cases 2 and 3 are impossible by Lemma~\ref{lem:decomp-ltrP}.
Case 4 is impossible since, by definition of $\prop$, $\prop(\nil \redP R')$ does not hold. 
So we are in the first case: $\prop(P_0 \redP R')$ and $R_0 \equiv R_1 \equiv R' \parop \nil \equiv R'$. Since $\prop(P_0 \redP R')$ is invariant by structural equivalence applied to $R'$, we can then conclude that $\prop(P_0 \redP R_0)$.

\item Case $P_0 = P \parop (Q \parop R) \equivP (P \parop Q) \parop R$.
Since $\prop((P \parop Q) \parop R \redP R_1)$, we have four cases:
\begin{itemize}
\item $P \parop Q \redP R_2$, $\prop(P \parop Q \redP R_2)$, and $R_1 \equiv R_2 \parop R$. We have again four cases.
\begin{itemize}

\item $P \redP R_3$, $\prop(P \redP R_3)$, and $R_2 \equiv R_3 \parop Q$.
Then $R_0 \equiv R_1 \equiv R_2 \parop R \equiv (R_3 \parop Q) \parop R \equiv R_3 \parop (Q \parop R)$, so we are in Case~\ref{prop2:PAR}.(a) of $\prop(P_0 \redP R_0)$.

\item $Q \redP R_3$, $\prop(Q \redP R_3)$, and $R_2 \equiv P \parop R_3$.
Then $Q \parop R \redP R_3 \parop R$, $\prop(Q \parop R \redP R_3 \parop R)$,
and $R_0 \equiv R_1 \equiv R_2 \parop R \equiv (P \parop R_3) \parop R
\equiv P \parop (R_3 \parop R)$, so we are in Case~\ref{prop2:PAR}.(a') of $\prop(P_0 \redP R_0)$.

\item $P \ltrP{\Rcv{N}{x}} A$, $Q \ltrP{\Res{x}\Snd{N}{x}} B$, and $R_2 \equiv \Res{x}(A \parop B)$ for some $A$, $B$, $x$, and ground term $N$.
Then $Q \parop R \ltrP{\Res{x}\Snd{N}{x}} B \parop R$ by $\brn{Par}'$,
and $R_0 \equiv R_1 \equiv R_2 \parop R \equiv \Res{x}(A \parop B) \parop R
\equiv \Res{x}(A \parop (B \parop R))$ since $x \notin \fv(R)$, so we are in Case~\ref{prop2:PAR}.(b) of $\prop(P_0 \redP R_0)$.

\item $P \ltrP{\Res{x}\Snd{N}{x}} A$, $Q \ltrP{\Rcv{N}{x}} B$, and $R_2 \equiv \Res{x}(A \parop B)$ for some $A$, $B$, $x$, and ground term $N$.
This case can be handled similarly to the previous one.

\end{itemize}

\item $R \redP R_2$, $\prop(R \redP R_2)$, and $R_1 \equiv (P \parop Q) \parop R_2$. Then $Q \parop R \redP Q \parop R_2$, $\prop(Q \parop R \redP Q \parop R_2)$, and $R_0 \equiv R_1 \equiv (P \parop Q) \parop R_2 \equiv P \parop (Q \parop R_2)$, so we are in Case~\ref{prop2:PAR}.(a') of $\prop(P_0 \redP R_0)$.

\item $P \parop Q \ltrP{\Rcv{N}{x}} A$, $R \ltrP{\Res{x}\Snd{N}{x}} B$, and $R_1 \equiv \Res{x}(A \parop B)$ for some $A$, $B$, $x$, and ground term $N$. By Lemma~\ref{lem:decomp-ltrP}, either $P \ltrP{\Rcv{N}{x}} A'$ and $A \equiv A' \parop Q$ for some $A'$, or $Q \ltrP{\Rcv{N}{x}} A'$ and $A \equiv P \parop A'$ for some $A'$.
In the first case, $P \ltrP{\Rcv{N}{x}} A'$, $Q \parop R \ltrP{\Res{x}\Snd{N}{x}} Q \parop B$ by $\brn{Par}'$ and $\brn{Struct}'$, and $R_0 \equiv R_1 \equiv \Res{x}(A \parop B) \equiv \Res{x}((A' \parop Q) \parop B) \equiv \Res{x}(A' \parop (Q \parop B))$, so we are in Case~\ref{prop2:PAR}.(b) of $\prop(P_0 \redP R_0)$.
In the second case, by Lemma~\ref{lem:comp-ltr}, 
$Q \parop R \redP R_2$ and $R_2 \equiv \Res{x}(A' \parop B)$ for some $R_2$, 
$\prop(Q \parop R \redP R_2)$, and  
$R_0 \equiv R_1 \equiv \Res{x}(A \parop B) 
\equiv \Res{x}((P\parop A') \parop B) 
\equiv P \parop \Res{x}(A' \parop B) \equiv P \parop R_2$ 
since $x \notin \fv(P)$. Hence, we are in Case~\ref{prop2:PAR}.(a') of $\prop(P_0 \redP R_0)$.

\item $P \parop Q \ltrP{\Res{x}\Snd{N}{x}} A$, $R \ltrP{\Rcv{N}{x}} B$, and $R_1 \equiv \Res{x}(A \parop B)$ for some $A$, $B$, $x$, and ground term $N$.
This case can be handled similarly to the previous one.

\end{itemize}

\item Case $P_0 = (P \parop Q) \parop R \equivP P \parop (Q \parop R)$. This case is similar to the previous one.

\item Case $P_0 = P \parop Q \equivP Q \parop P$. (This case is its own symmetric.) This case is immediate, since the desired result is invariant by swapping $P$ and $Q$.

\item Case $P_0 = \Repl{P} \equivP P \parop \Repl{P}$. Since $\prop(P \parop \Repl{P} \redP R_1)$, we have four cases: 
\begin{itemize}
\item $P \redP P'$, $\prop(P \redP P')$, and $R_1 \equiv P' \parop \Repl{P}$ for some $P'$. Hence $P \parop P \redP P' \parop P$, $\prop(P \parop P \redP P' \parop P)$, and $R_0 \equiv R_1 \equiv P' \parop \Repl{P} \equiv (P' \parop P) \parop \Repl{P}$, so we are in Case~\ref{prop2:REPL} of $\prop(P_0 \redP R_0)$.
\item $\Repl{P} \redP P'$, $\prop(\Repl{P} \redP P')$, and $R_1 \equiv P \parop P'$ for some $P'$. Since $\prop(\Repl{P} \redP P')$, we have $P \parop P \redP R_2$,  $\prop(P \parop P \redP R_2)$, and $P' \equiv R_2 \parop \Repl{P}$.
So $P \parop P \redP R_2$,  $\prop(P \parop P \redP R_2)$, and $R_0 \equiv R_1 \equiv P \parop P' \equiv P \parop (R_2 \parop \Repl{P}) \equiv R_2 \parop \Repl{P}$, so we are in Case~\ref{prop2:REPL} of $\prop(P_0 \redP R_0)$.
\item $P \ltrP{\Rcv{N}{x}} A$, $\Repl{P} \ltrP{\Res{x}\Snd{N}{x}} B$, and $R_1 \equiv \Res{x}(A \parop B)$ for some $A$, $B$, $x$, and ground term $N$.
By Lemma~\ref{lem:decomp-ltrP}, $P \ltrP{\Res{x}\Snd{N}{x}} B'$ and $B \equiv B' \parop \Repl{P}$ for some $B'$.
So by Lemma~\ref{lem:comp-ltr}, $P \parop P \redP R_2$ and $R_2 \equiv \Res{x}(A \parop B')$ for some $R_2$.
So $P \parop P \redP R_2$, $\prop(P \parop P \redP R_2)$, and
$R_0 \equiv R_1 \equiv \Res{x}(A \parop B) \equiv \Res{x}(A \parop (B' \parop \Repl{P})) \equiv \Res{x}(A \parop B') \parop \Repl{P} \equiv R_2 \parop \Repl{P}$ since $x \notin \fv(\Repl{P})$. So we are in Case~\ref{prop2:REPL} of $\prop(P_0 \redP R_0)$.

\item $P \ltrP{\Res{x}\Snd{N}{x}} A$, $\Repl{P} \ltrP{\Rcv{N}{x}} B$, and $R_1 \equiv \Res{x}(A \parop B)$ for some $A$, $B$, $x$, and ground term $N$.
This case can be handled similarly to the previous one.
\end{itemize}

\item Case $P_0 = P \parop \Repl{P} \equivP \Repl{P}$. 
Since $\prop(\Repl{P} \redP R_1)$, we have 
$P \parop P \redP R_2$, $\prop(P \parop P \redP R_2)$, and
$R_1 \equiv R_2 \parop \Repl{P}$ for some $R_2$.
Since $\prop(P \parop P \redP R_2)$, we have four cases, which reduce to two
by symmetry:
\begin{itemize}
\item $P \redP P'$, $\prop(P \redP P')$, and $R_2 \equiv P' \parop P$ for some $P'$. Hence $R_0 \equiv R_1 \equiv R_2 \parop \Repl{P} \equiv P' \parop P \parop \Repl{P} \equiv P' \parop \Repl{P}$, so we are in Case~\ref{prop2:PAR}.(a) of $\prop(P_0 \redP R_0)$.

\item $P \ltrP{\Rcv{N}{x}} A$, $P \ltrP{\Res{x}\Snd{N}{x}} B$, and $R_2 \equiv \Res{x}(A \parop B)$ for some $A$, $B$, $x$, and ground term $N$.
Hence $P \ltrP{\Rcv{N}{x}} A$, $\Repl{P} \equivP P \parop \Repl{P} \ltrP{\Res{x}\Snd{N}{x}} B \parop \Repl{P}$ by $\brn{Par}'$ so $\Repl{P} \ltrP{\Res{x}\Snd{N}{x}} B \parop \Repl{P}$ by $\brn{Struct}'$, and $R_0 \equiv R_1 \equiv R_2 \parop \Repl{P} \equiv \Res{x}(A \parop B)\parop \Repl{P}\equiv \Res{x}(A \parop (B \parop \Repl{P}))$ since $x \notin \fv(\Repl{P})$.
So we are in Case~\ref{prop2:PAR}.(b) of $\prop(P_0 \redP R_0)$.

\end{itemize}

\item Case $P_0 = \Res{n}\nil \equivP \nil$. We have that $\prop(\nil  \redP R_1)$ is impossible, so this case never happens.

\item Case $P_0 = \nil \equivP \Res{n}\nil$. Since $\prop(\Res{n}\nil  \redP R_1)$, we have $\prop(\nil  \redP R'_1)$, which is impossible, so this case never happens.

\item Case $P_0 = \Res{n}\Res{n'}P \equivP \Res{n'}\Res{n}P$. (This case is its own symmetric.) Since $\prop(\Res{n'}\Res{n}P  \redP R_1)$, we have 
$\Res{n}P \redP R'_1$, $\prop(\Res{n}P \redP R'_1)$, and $R_1 \equiv \Res{n'} R'_1$ for some $R'_1$, so
$P \redP R''_1$, $\prop(P \redP R''_1)$, and $R'_1 \equiv \Res{n} R''_1$ for some $R''_1$.
Hence, $\Res{n'}P \redP \Res{n'}R''_1$, 
$\prop(\Res{n'}P \redP \Res{n'}R''_1)$, and $R_0 \equiv R_1 \equiv \Res{n'}R'_1 \equiv \Res{n'}\Res{n}R''_1 \equiv \Res{n}\Res{n'}R''_1$, so we are in Case~\ref{prop2:RES} of $\prop(P_0 \redP R_0)$ with $Q' = \Res{n'}R''_1$.

\item Case $P_0 = P \parop \Res{n}Q \equivP \Res{n}(P \parop Q)$ and $n \notin \fn(P)$.  Since $\prop(\Res{n}(P \parop Q) \redP R_1)$, we have
$P \parop Q \redP R_2$, $\prop(P \parop Q \redP R_2)$, and $R_1 \equiv \Res{n}R_2$ for some $R_2$. Since $\prop(P \parop Q \redP R_2)$, we have four cases:
\begin{itemize}
\item $P \redP P'$, $\prop(P \redP P')$, and $R_2 \equiv P' \parop Q$ for some $P'$. We rename $n$ in $Q$ so that $n \notin \fn(P')$. Hence $P \redP P'$, $\prop(P \redP P')$, and $R_0 \equiv R_1 \equiv \Res{n}R_2 \equiv \Res{n}(P' \parop Q) \equiv P' \parop \Res{n}Q$ since $n \notin \fn(P')$. 
So we are in Case~\ref{prop2:PAR}.(a) of $\prop(P_0 \redP R_0)$.
\item $Q \redP Q'$, $\prop(Q \redP Q')$, and $R_2 \equiv P \parop Q'$ for some $Q'$. Then $\Res{n}Q \redP \Res{n}Q'$, $\prop(\Res{n}Q \redP \Res{n}Q')$, and $R_0 \equiv R_1 \equiv \Res{n}R_2 \equiv \Res{n}(P \parop Q') \equiv P \parop \Res{n}Q'$ since $n \notin \fn(P)$.
So we are in Case~\ref{prop2:PAR}.(a') of $\prop(P_0 \redP R_0)$.
\item $P \ltrP{\Rcv{N}{x}} A$, $Q \ltrP{\Res{x}\Snd{N}{x}} B$, and $R_2 \equiv \Res{x}(A \parop B)$ for some $A$, $B$, $x$, and ground term $N$.

We need to rename $n$ so that $n \notin \fn(N) \cup \fn(A)$.
To do that, we first show that, for all processes $P$, $P'$, if $P \equivP P'$ and $\Sigma \vdash P\{\subst{n'}{n}\} = P$, then $\Sigma \vdash P'\{\subst{n'}{n}\} = P'$, by induction on the derivation of $P \equivP P'$.

We also show that, if $A \equiv A'$, then $A\{\subst{n'}{n}\} \equiv A'\{\subst{n'}{n}\}$, by induction on the derivation of $A \equiv A'$.

By Lemma~\ref{lem:caract-ltrP}, since $P \ltrP{\Rcv{N}{x}} A$, we have
$P \equivP \Res{\vect n}(\Rcv{N}{y}.P_1 \parop P_2)$,
$A \equiv \Res{\vect n}(P_1\{\subst{x}{y}\} \parop P_2)$, and
$\{\vect n\} \cap \fn(N) = \emptyset$,
for some $\vect n$, $P_1$, $P_2$, $y$,
and since $Q \ltrP{\Res{x}\Snd{N}{x}} B$, we have
$Q \equivP \Res{\vect n'}(\Snd{N}{M}.Q_1 \parop Q_2)$,
$A \equiv \Res{\vect n'}(Q_1 \parop \{\subst{M}{x}\} \parop Q_2)$,
$\{\vect n'\} \cap \fn(N) = \emptyset$, and
$x \notin \fv(\Snd{N}{M}.Q_1 \parop Q_2)$,
for some $\vect n'$, $Q_1$, $Q_2$, $M$.
Let $n'$ be a fresh name.
\begin{itemize}
\item First case: $n \notin \vect n$. 
We have $P\{\subst{n'}{n}\} = P$ since $n \notin \fn(P)$, 
so by the result shown above, $\Sigma \vdash (\Res{\vect n}(\Rcv{N}{y}.P_1 \parop P_2))\{\subst{n'}{n}\} = \Res{\vect n}(\Rcv{N}{y}.P_1 \parop P_2)$,
so $\Sigma \vdash N\{\subst{n'}{n}\} = N$,
$P \equivP \Res{\vect n}(\Rcv{N\{\subst{n'}{n}\}}{y}.P_1\{\subst{n'}{n}\} \parop P_2\{\subst{n'}{n}\})$,
$A\{\subst{n'}{n}\} \equiv \Res{\vect n}(P_1\{\subst{n'}{n}\}\{\subst{x}{y}\} \parop P_2\{\subst{n'}{n}\})$, and
$\fn(N\{\subst{n'}{n}\}) \linebreak[2] \cap \{\vect n\}  = \emptyset$,
so by Lemma~\ref{lem:caract-ltrP}, $P \ltrP{\Rcv{N\{\subst{n'}{n}\}}{x}} A\{\subst{n'}{n}\}$.
\item Second case: $n \in \vect n$, so $n \notin \fn(N)$. We have $N\{\subst{n'}{n}\} = N$.
So $P \equivP \Res{\vect n}(\Rcv{N\{\subst{n'}{n}\}}{y}.P_1\} \parop P_2)$,
$A\{\subst{n'}{n}\} \equiv \Res{\vect n}(P_1\{\subst{x}{y}\} \parop P_2)$, and
$\{\vect n\} \cap \fn(N\{\subst{n'}{n}\}) = \emptyset$,
so by Lemma~\ref{lem:caract-ltrP}, $P \ltrP{\Rcv{N\{\subst{n'}{n}\}}{x}} A\{\subst{n'}{n}\}$.
\end{itemize}
Let $N' = N\{\subst{n'}{n}\}$ and $A' = A\{\subst{n'}{n}\}$.
Hence in both cases, $P \ltrP{\Rcv{N'}{x}} A'$
and $\Sigma \vdash N' = N$, so
$Q \equivP \Res{\vect n'}(\Snd{N'}{M}.Q_1 \parop Q_2)$,
$A \equiv \Res{\vect n'}(Q_1 \parop \{\subst{M}{x}\} \parop Q_2)$,
$\{\vect n'\} \cap \fn(N') = \emptyset$, and
$x \notin \fv(\Snd{N'}{M}.Q_1 \parop Q_2)$,
so by Lemma~\ref{lem:caract-ltrP}, $Q \ltrP{\Res{x}\Snd{N'}{x}} B$.
Hence $P \ltrP{\Rcv{N'}{x}} A'$, $\Res{n}Q \ltrP{\Res{x}\Snd{N'}{x}} \Res{n}B$
by $\brn{Scope}'$, and $R_0 \equiv R_1 \equiv \Res{n}R_2 \equiv \Res{n}\Res{x}(A' \parop B) \equiv \Res{x}(A' \parop \Res{n}B)$,
so we are in Case~\ref{prop2:PAR}.(b) of $\prop(P_0 \redP R_0)$.
\item $P \ltrP{\Res{x}\Snd{N}{x}} A$, $Q \ltrP{\Rcv{N}{x}} B$, and $R_2 \equiv \Res{x}(A \parop B)$ for some $A$, $B$, $x$, and ground term $N$. This case can be handled similarly to the previous one.
\end{itemize}

\item Case $P_0 = \Res{n}(P \parop Q) \equivP P \parop \Res{n}Q$ and $n \notin \fn(P)$. This case is fairly similar to the previous one.

\item Case $P_0 = P_1 \{\subst{M}{x}\} \equivP P_1\{\subst{N}{x}\}$ and $\Sigma \vdash M=N$. (This case is its own symmetric.)
We have $\prop(P_1\{\subst{N}{x}\}\redP R_1)$.
We show by induction on the syntax of $P_1$ that, if $\prop(P_1\{\subst{N}{x}\}\redP R_1)$, $\Sigma \vdash M=N$, and $R_0 \equiv R_1$, then $\prop(P_1 \{\subst{M}{x}\} \redP R_0)$. 
\begin{itemize}
\item Case $P_1 = P \parop Q$. 
We have $\prop(P\{\subst{N}{x}\}  \parop Q\{\subst{N}{x}\} \redP R_1)$, 
so we have four cases:
\begin{itemize}
\item $P\{\subst{N}{x}\} \redP P'$, $\prop(P\{\subst{N}{x}\} \redP P')$, and $R_1 \equiv P' \parop Q\{\subst{N}{x}\}$ for some $P'$. We have $P\{\subst{M}{x}\} \redP P'$, $\prop(P\{\subst{M}{x}\} \redP P')$ by induction hypothesis, 
and $R_0 \equiv R_1 \equiv P' \parop Q\{\subst{N}{x}\}$, so we are in Case~\ref{prop2:PAR}.(a) of
$\prop(P\{\subst{M}{x}\}  \parop Q\{\subst{M}{x}\} \redP R_0)$.

\item $Q\{\subst{N}{x}\} \redP Q'$, $\prop(Q\{\subst{N}{x}\} \redP Q')$, and $R_1 \equiv P\{\subst{N}{x}\}  \parop Q'$ for some $Q'$. This case is obtained from the previous one by swapping $P$ and $Q$.

\item $P\{\subst{N}{x}\} \ltrP{\Rcv{N'}{x}} A$, $Q\{\subst{N}{x}\} \ltrP{\Res{x}\Snd{N'}{x}} B$, and $R_1 \equiv \Res{x}(A \parop B)$ for some $A$, $B$, $x$, and ground term $N'$.
Then 
$P\{\subst{M}{x}\} \ltrP{\Rcv{N'}{x}} A$, $Q\{\subst{M}{x}\} \ltrP{\Res{x}\Snd{N'}{x}} B$ by $\brn{Struct}'$, and $R_1 \equiv \Res{x}(A \parop B)$, so $R_0 \equiv R_1 \equiv \Res{x}(A \parop B)$. So we are in Case~\ref{prop2:PAR}.(b) of
$\prop(P\{\subst{M}{x}\}  \parop Q\{\subst{M}{x}\} \redP R_0)$.

\item $P\{\subst{N}{x}\} \ltrP{\Res{x}\Snd{N'}{x}} A$, $Q\{\subst{N}{x}\} \ltrP{\Rcv{N'}{x}} B$, and $R_1 \equiv \Res{x}(A \parop B)$ for some $A$, $B$, $x$, and ground term $N'$.
This case is obtained from the previous one by swapping $P$ and $Q$.
\end{itemize}

\item Case $P_1 = \Res{n}P$. 
We have $\prop(\Res{n}P\{\subst{N}{x}\}  \redP R_1)$, so $P\{\subst{N}{x}\}\redP R_2$, $\prop(P\{\subst{N}{x}\}  \redP R_2)$, and $R_1 \equiv \Res{n}R_2$ for some $R_2$. Hence $P_1 \{\subst{M}{x}\} = \Res{n}P\{\subst{M}{x}\}$, $P\{\subst{M}{x}\} \equivP P\{\subst{N}{x}\}\redP R_2$, $\prop(P\{\subst{M}{x}\}  \redP R_2)$ by induction hypothesis, and $R_0 \equiv R_1 \equiv \Res{n}R_2$, so we are in Case~\ref{prop2:RES} of $\prop(P_1 \{\subst{M}{x}\} \redP R_0)$.

\item Case $P_1 = \Repl{P}$. 
We have $\prop(\Repl{P\{\subst{N}{x}\}}  \redP R_1)$, so
$P\{\subst{N}{x}\} \parop P\{\subst{N}{x}\} \redP R_2$, 
$\prop(P\{\subst{N}{x}\} \parop P\{\subst{N}{x}\} \redP R_2)$, and 
$R_1 \equiv R_2 \parop \Repl{P\{\subst{N}{x}\}}$ for some $R_2$.
Since $\prop(P\{\subst{N}{x}\} \parop P\{\subst{N}{x}\} \redP R_2)$,
we have four cases, which reduce to two by symmetry:
\begin{itemize}
\item $P\{\subst{N}{x}\} \redP P'$, $\prop(P\{\subst{N}{x}\} \redP P')$, and $R_2 \equiv P' \parop P\{\subst{N}{x}\}$ for some $P'$. We have $P\{\subst{M}{x}\} \redP P'$, $\prop(P\{\subst{M}{x}\} \redP P')$ by induction hypothesis, 
so $P\{\subst{M}{x}\} \parop P\{\subst{M}{x}\} \redP P' \parop P\{\subst{M}{x}\}$
and $R_0 \equiv R_1 \equiv R_2 \parop \Repl{P\{\subst{N}{x}\}}
\equiv P' \parop P\{\subst{N}{x}\} \parop \Repl{P\{\subst{N}{x}\}}
\equiv P' \parop P\{\subst{M}{x}\} \parop \Repl{P\{\subst{M}{x}\}}$, so 
we are in Case~\ref{prop2:REPL} of $\prop(\Repl{P\{\subst{M}{x}\}} \redP R_0)$.

\item $P\{\subst{N}{x}\} \ltrP{\Rcv{N'}{x}} A$, $P\{\subst{N}{x}\} \ltrP{\Res{x}\Snd{N'}{x}} B$, and $R_2 \equiv \Res{x}(A \parop B)$ for some $A$, $B$, $x$, and ground term $N'$.
Hence $P\{\subst{M}{x}\} \ltrP{\Rcv{N'}{x}} A$, $P\{\subst{M}{x}\} \ltrP{\Res{x}\Snd{N'}{x}} B$ by $\brn{Struct}'$, and $R_0 \equiv R_1 \equiv R_2 \parop \Repl{P\{\subst{N}{x}\}} \equiv \Res{x}(A \parop B) \parop \Repl{P\{\subst{N}{x}\}}$.
By Lemma~\ref{lem:comp-ltr}, $P\{\subst{M}{x}\} \parop P\{\subst{M}{x}\} \redP R_3$ and $R_3 \equiv \Res{x}(A \parop B) $, so $\prop(P\{\subst{M}{x}\} \parop P\{\subst{M}{x}\} \redP R_3)$, and $R_0 \equiv R_3 \parop \Repl{P\{\subst{M}{x}\}}$, so we are in Case~\ref{prop2:REPL} of $\prop(\Repl{P\{\subst{M}{x}\}} \redP R_0)$.

\end{itemize}

\item Case $P_1 = \IfThenElse{M_1}{N_1}{P}{Q}$. We have $\prop( \kw{if}$ $M_1\{\subst{N}{x}\} = N_1\{\subst{N}{x}\}$ $\kw{then}$ $P\{\subst{N}{x}\}$ $\kw{else}$ $Q\{\subst{N}{x}\} \redP R_1)$, so we have two cases:
\begin{itemize}
\item $\Sigma \vdash M_1\{\subst{N}{x}\} = N_1\{\subst{N}{x}\}$ and $R_1 \equiv P\{\subst{N}{x}\}$. Hence, we have $\Sigma \vdash M_1\{\subst{M}{x}\} = N_1\{\subst{M}{x}\}$ and $R_0 \equiv R_1 \equiv P\{\subst{M}{x}\}$, so we are in Case~\ref{prop2:IF} of $\prop(P_1\{\subst{M}{x}\}\redP R_0)$.
\item $\Sigma \vdash M_1\{\subst{N}{x}\} \neq N_1\{\subst{N}{x}\}$ and $R_1 \equiv Q\{\subst{N}{x}\}$. Hence, we have $\Sigma \vdash M_1\{\subst{M}{x}\} \neq N_1\{\subst{M}{x}\}$ and $R_0 \equiv R_1 \equiv Q\{\subst{M}{x}\}$, so we are in Case~\ref{prop2:IF} of $\prop(P_1\{\subst{M}{x}\}\redP R_0)$.
\end{itemize}

\end{itemize}
Using this result, we obtain $\prop(P_0 \redP R_0)$.

\item Case $P_0 = P \parop Q \equivP P' \parop Q$ knowing $P \equivP P'$.
Since $\prop(P' \parop Q \redP R_1)$, we have four cases:
\begin{itemize}
\item $P' \redP P''$, $\prop(P' \redP P'')$, and $R_1 \equiv P'' \parop Q$ for some $P''$. 
Then $P \redP P''$, $\prop(P \redP P'')$ by induction hypothesis, 
and $R_0 \equiv R_1 \equiv P'' \parop Q$, so we are in Case~\ref{prop2:PAR}.(a) of $\prop(P_0 \redP R_0)$. 
\item $Q \redP Q'$, $\prop(Q \redP Q')$, and $R_1 \equiv P' \parop Q'$ for some $Q'$. 
Then $Q \redP Q'$, $\prop(Q \redP Q')$, and $R_0 \equiv R_1 \equiv P' \parop Q'$, so we are in Case~\ref{prop2:PAR}.(a') of $\prop(P_0 \redP R_0)$. 
\item $P' \ltrP{\Rcv{N}{x}} A$, $Q \ltrP{\Res{x}\Snd{N}{x}} B$, and $R_1 \equiv \Res{x}(A \parop B)$ for some $A$, $B$, $x$, and ground term $N$.
Then $P \ltrP{\Rcv{N}{x}} A$ by $\brn{Struct}'$, $Q \ltrP{\Res{x}\Snd{N}{x}} B$, and $R_0 \equiv R_1 \equiv \Res{x}(A \parop B)$, so we are in Case~\ref{prop2:PAR}.(b) of $\prop(P_0 \redP R_0)$.
\item $P' \ltrP{\Res{x}\Snd{N}{x}} A$, $Q \ltrP{\Rcv{N}{x}} B$, and $R_1 \equiv \Res{x}(A \parop B)$ for some $A$, $B$, $x$, and ground term $N$. This case can be handled similarly to the previous one.
\end{itemize}

\item Case $P_0 = \Res{n}P \equivP \Res{n}P'$ knowing $P \equivP P'$. 
Since $\prop(\Res{n}P' \redP R_1)$, 
$P' \redP R'_1$, $\prop(P' \redP R'_1)$, and $R_1 \equiv \Res{n}R'_1$ for some $R'_1$.
Then, $P \redP R'_1$. By induction hypothesis, $\prop(P \redP R'_1)$. Moreover, $R_0 \equiv R_1 \equiv \Res{n}R'_1$, so we are in Case~\ref{prop2:RES} of $\prop(P_0 \redP R_0)$.

\end{itemize}
\end{itemize}
If $P_0$ is closed and $P_0 \redP R_0$, then by Lemma~\ref{lem:closed-interm}\eqref{prop:closed-interm-redP}, there exists a derivation of $P_0 \redP R_0$ closed on the left. So by applying the previous result, $\prop(P_0 \redP R_0)$, which yields the desired property.
\end{proof}

\begin{lemma}\label{lem:decomp-redpnf}
If $\Res{\vect n}(\sigma \parop P)$ is a closed normal process
and 
$\Res{\vect n}(\sigma \parop P) \redpnf A$, 
then 
$P \redP P'$ and $A \equiv \Res{\vect n}(\sigma \parop P')$ for some $P'$.
\end{lemma}
\begin{proof}
We proceed similarly to Lemma~\ref{lem:decomp-ltrpnf}.
By Lemma~\ref{lem:closed-interm}\eqref{prop:closed-interm-redpnf}, we consider a derivation of $\Res{\vect n}(\sigma \parop P) \redpnf A$ closed on the left.
By definition of $\redpnf$, we have $\Res{\vect n}(\sigma \parop P) \equivpnf \Res{\vect n'}(\sigma' \parop P')$, $P' \redP Q'$ and $A \equivpnf \Res{\vect n'}(\sigma' \parop Q')$ for some $\vect n'$, $\sigma'$, $P'$, $Q'$.
We proceed by induction on the derivation of $\Res{\vect n}(\sigma \parop P) \equivpnf \Res{\vect n'}(\sigma' \parop P')$.
\begin{itemize}

\item Base case: $\Res{\vect n}(\sigma \parop P) = \Res{\vect n'}(\sigma' \parop P')$, and the desired result holds.

\item Transitivity: the result is proved by applying the induction hypothesis twice.

\item Case $\brn{Plain}''$: $P \equivP P' \redP Q'$ and $A \equiv \Res{\vect n}(\sigma \parop Q')$, so the result holds.

\item Case $\brn{New-C}''$: $\vect n'$ is a reordering of $\vect n$, 
$\Res{\vect n}(\sigma \parop P) \equivpnf \Res{\vect n'}(\sigma \parop P)$, $P \redP Q'$ and $A \equiv \Res{\vect n'}(\sigma \parop Q')$, 
so $P \redP Q'$ and $A \equiv \Res{\vect n'}(\sigma \parop Q') \equiv \Res{\vect n}(\sigma \parop Q')$, hence the result holds.

\item Case $\brn{New-Par}''$: 
$P = \Res{n'}P'$, $\Res{\vect n}(\sigma \parop \Res{n'}P') \equivpnf \Res{\vect n, n'}(\sigma \parop P')$, $P' \redP Q'$, and $A \equiv \Res{\vect n,n'}(\sigma \parop Q')$ where $n' \notin \fn(\sigma)$.
Therefore, $P = \Res{n'}P' \redP \Res{n'}Q'$
and by $\brn{New-Par}$, $A \equiv \Res{\vect n,n'}(\sigma \parop Q') \equiv
\Res{\vect n}(\sigma \parop \Res{n'}Q')$, hence the result holds.

\item Case $\brn{New-Par}''$ reversed: 
$\Res{\vect n, n'}(\sigma \parop P) \equivpnf \Res{\vect n}(\sigma \parop \Res{n'}P)$, $\Res{n'}P \redP Q'$ and $A \equiv \Res{\vect n}(\sigma \parop Q')$ where $n' \notin \fn(\sigma)$.
By Lemma~\ref{lem:decomp-redP}, $P \redP Q''$ and $Q' \equiv \Res{n'}Q''$
for some $Q''$.
Hence, $P \redP Q''$ and $A \equiv \Res{\vect n}(\sigma \parop Q') 
\equiv \Res{\vect n}(\sigma \parop \Res{n'}Q'') 
\equiv \Res{\vect n, n'}(\sigma \parop Q'')$ by $\brn{New-Par}$, 
so the result holds.

\item Case $\brn{Rewrite}''$:
$\Res{\vect n}(\sigma \parop P) \equivpnf \Res{\vect n}(\sigma' \parop P)$, $P \redP Q$ and $A \equiv \Res{\vect n}(\sigma' \parop Q)$
where $\dom(\sigma) = \dom(\sigma')$, $\Sigma \vdash \sigma x = \sigma' x$ for all $x \in \dom(\sigma)$, and $(\fv(\sigma x) \cup \fv(\sigma' x)) \cap \dom(\sigma) = \emptyset$ for all $x \in \dom(\sigma)$.
Hence $P \redP Q$ and $A \equiv \Res{\vect n}(\sigma' \parop Q)
\equiv \Res{\vect n}(\sigma \parop Q)$ by several applications of $\brn{Rewrite}$, so the result holds.
\qedhere
\end{itemize}
\end{proof}

We prove the following strengthened version of Lemma~\ref{lem:decomp-redpnf}, 
in which the process $P'$ is guaranteed to be closed.

\begin{lemma}\label{lem:decomp-redpnf-closed}
If $\Res{\vect n}(\sigma \parop P)$ is a closed normal process
and 
$\Res{\vect n}(\sigma \parop P) \redpnf A$, 
then 
$P \redP P'$ and $A \equiv \Res{\vect n}(\sigma \parop P')$ for some closed process $P'$.
\end{lemma}
\begin{proof}
By Lemma~\ref{lem:decomp-redpnf}, we get the existence of a process $P'$,
which may not be closed.
Let us apply Lemma~\ref{lem:closing}\eqref{closing:redP}.
Let $Y = \fv(P) \cup \fv(P') = \fv(P')$.
Let $\sigma'$ be a substitution from $Y$ to pairwise distinct fresh names.
Since $P$ is closed, $P = P\sigma'$, so a fortiori
$\Sigma \vdash P = P\sigma'$.
Hence $P\sigma' \redP P'\sigma'$
and $\Sigma \vdash P' = P'\sigma'$.
So $P \redP P'\sigma'$ and $A \equiv  \Res{\vect n}(\sigma \parop P')
\equiv \Res{\vect n}(\sigma \parop P'\sigma')$, so we get
the desired result by using the closed process $P'\sigma'$ instead of $P'$.
\end{proof}

The following strengthened version of Lemma~\ref{lem:decomp-redP}
is proved in a similar way.

\begin{lemma}\label{lem:decomp-redP-closed}
Suppose that $P_0$ is a closed process and $P_0 \redP R$. Then 
one of the following cases holds:
\begin{enumerate}

\item $P_0 = P \parop Q$ for some $P$ and $Q$, and one of the following cases holds:
\begin{enumerate}
\item $P \redP P'$ and $R \equiv P' \parop Q$ for some closed process $P'$,
\item $P \ltrP{\Rcv{N}{x}} A$, $Q \ltrP{\Res{x}\Snd{N}{x}} B$, and $R \equiv \Res{x}(A \parop B)$ for some $A$, $B$, $x$, and ground term $N$, 
\end{enumerate}
and two symmetric cases obtained by swapping $P$ and $Q$;

\item $P_0 = \Res{n}P$, $P \redP Q'$, and $R \equiv \Res{n} Q'$ for some $n$ and
some closed processes $P$ and $Q'$;

\item $P_0 = \Repl{P}$, $P \parop P \redP Q'$, and $R \equiv Q' \parop \Repl{P}$ for some closed processes $P$ and $Q'$.

\item $P_0 = \IfThenElse{M}{N}{P}{Q}$ and either $\Sigma \vdash M = N$ and $R \equiv P$, or $\Sigma \vdash M \neq N$ and $R \equiv Q$, for some $M$, $N$, $P$, and $Q$.
\end{enumerate}
\end{lemma}

\section{Proof of Theorem~\ref{THM:OBSERVATIONAL-LABELED}: Main Lemmas}\label{app:bigpfmain}

Relying on partial normal forms and their semantics, we prove the
remaining lemmas needed for the proof of
Theorem~\ref{THM:OBSERVATIONAL-LABELED}.  
Sections~\ref{app:onedirection} and~\ref{app:twodirection}
establish the two directions of Theorem~\ref{THM:OBSERVATIONAL-LABELED}. 
The argument for the first direction employs lemmas
about consequences of static equivalences; these lemmas are in 
Section~\ref{app:exploiting}.

\subsection{Exploiting Static Equivalence}\label{app:exploiting}

The lemmas in this section rely on static equivalences in order to analyze and to establish
structural equivalences or reductions. 
For all these lemmas, we consider the action of two equivalent frames
$\Res{\vect n}\sigma \enveq \Res{\vect n'}\sigma'$ on a process $P'$ such that 
$\fn(P') \cap \{ \vect n, \vect n'\} = \emptyset$: we suppose a structural equivalence or reduction of a process $P$ such that $\Sigma \vdash P' \sigma = P$, and prove a corresponding structural equivalence or reduction of the process $P'\sigma'$.  
Lemma~\ref{lem:equiv-enveq} deals with structural equivalence, Lemma~\ref{lem:red-enveq} with internal reduction, and Lemma~\ref{lem:ltr-enveq} with labelled transitions.

\begin{lemma}\label{lem:equiv-enveq}
Suppose that $\Res{\vect n}\sigma \enveq \Res{\vect n'}\sigma'$, 
 $\fn(P') \cap \{ \vect n, \vect n'\} = \emptyset$, and
 $\Sigma \vdash P' \sigma = P$.
 If $P \equivP Q$, then 
$P' \sigma' \equivP Q' \sigma'$ for some $Q'$ such that 
$\fn(Q') \cap \{ \vect n, \vect n'\} = \emptyset$;  
$\Sigma \vdash Q = Q' \sigma$; and,
(*) if $\sigma$, $\sigma'$, and $P'\sigma$ are closed,
then $Q'\sigma$ is closed.
\end{lemma}
\begin{proof}
  We first prove the lemma without property (*),
  by induction on the derivation of $P \equivP Q$. 
The only rule that depends on terms is $\brn{Rewrite}'$, and when
$P \equivP Q$ by $\brn{Rewrite}'$,
 $\Sigma \vdash P' \sigma = P = Q$, so taking $Q' = P'$, we have $P' \sigma' \equivP Q' \sigma'$, $\fn(Q') \cap \{ \vect n, \vect n'\} = \emptyset$, and $\Sigma \vdash Q = Q' \sigma$.
For all other base cases, the structural equivalence rule applied in $P \equivP Q$ also applies to $P'$ and yields a process $Q'$ such that $P' \equivP Q'$, $\fn(Q') \cap \{ \vect n, \vect n'\} = \emptyset$, and $\Sigma \vdash Q = Q' \sigma$;
by Lemma~\ref{lem:instance}\eqref{instance:equiv} we conclude $P' \sigma' \equivP Q' \sigma'$.
The case of transitivity is proved by applying the induction hypothesis twice.

We now prove the lemma with property~(*) by
applying Lemma~\ref{lem:closing}\eqref{closing:equivP} to the structural equivalence $P' \sigma' \equivP Q' \sigma'$
for the process $Q'$ obtained above. 
Let $Y = \fv(P'\sigma) \cup \fv(Q'\sigma) = \fv(Q'\sigma) = \fv(Q')\setminus \dom(\sigma)$ and let $\sigma''$ map $Y$ to pairwise distinct fresh names.
We have $P'\sigma\sigma'' = P'\sigma$ so a fortiori
$\Sigma \vdash P'\sigma = P'\sigma\sigma''$, then
by Lemma~\ref{lem:closing}\eqref{closing:equivP}
$P'\sigma\sigma'' \equivP Q'\sigma\sigma''$
and $\Sigma \vdash Q'\sigma = Q'\sigma\sigma''$.
So $\Sigma \vdash Q = Q'\sigma = (Q'\sigma'')\sigma$.
Since $P' \sigma' \equivP Q' \sigma'$, we have
$P'\sigma' = P'\sigma'\sigma'' \equivP Q'\sigma'\sigma'' = (Q'\sigma'')\sigma'$.
Since $\fn(Q') \cap \{ \vect n, \vect n'\} = \emptyset$, we have 
$\fn(Q'\sigma'') \cap \{ \vect n, \vect n'\} = \emptyset$.
We also have $\fv(Q'\sigma'') \subseteq \dom(\sigma) = \dom(\sigma')$,
so we get the desired result by using $Q'\sigma''$ instead of $Q'$.
\qedhere
\end{proof}

\begin{lemma}\label{lem:red-enveq}
Suppose that
  $\Res{\vect n}\sigma \enveq \Res{\vect n'}\sigma'$, 
  $\fn(P') \cap \{ \vect n, \vect n'\} = \emptyset$, and
  $\Sigma \vdash P' \sigma = P$.
If  $P \redP Q$, then $P' \sigma' \redP Q' \sigma'$ for some $Q'$
  such that $\fn(Q') \cap \{ \vect n, \vect n'\} = \emptyset$; 
  $\Sigma \vdash Q = Q' \sigma$; and, 
(*) if $\sigma$, $\sigma'$, and $P'\sigma$ are closed, then $Q'\sigma$ is closed.
\end{lemma}

\begin{proof} 
  We first prove the lemma without property (*),
by induction on the derivation of $P \redP Q$. 
\begin{itemize}
\item Case $\brn{Comm}'$. We have $P = \Snd{N}{M}.P_0 \parop \Rcv{N}{x}.Q_0 \redP P_0 \parop Q_0\{\subst{M}{x}\} = Q$. We rename $x$ so that $x \notin \dom(\sigma')$. Therefore, $P' = \Snd{N'}{M'}.P'_0 \parop \Rcv{N''}{x}.Q'_0$ 
for some $N'$, $M'$, $P'_0$, $N''$, $Q'_0$ such that
$\Sigma \vdash N'\sigma = N''\sigma = N$, 
$\Sigma \vdash M'\sigma = M$, $\Sigma \vdash P_0' \sigma = P_0$, and
$\Sigma \vdash Q_0'\sigma = Q_0$.
Let $Q' = P'_0 \parop Q'_0\{\subst{M'}{x}\}$.
We have 
\begin{align*}
P' \sigma' &= \Snd{N' \sigma'}{M'\sigma'}.P'_0\sigma' \parop \Rcv{N''\sigma'}{x}.Q'_0\sigma' \\
&\equiv \Snd{N' \sigma'}{M'\sigma'}.P'_0\sigma' \parop \Rcv{N'\sigma'}{x}.Q'_0\sigma' \\
&\redP P'_0\sigma' \parop Q'_0\sigma'\{\subst{M'\sigma'}{x}\} = Q'\sigma'
\end{align*}
since $\Sigma \vdash N' \sigma' = N'' \sigma'$ because $\Sigma \vdash N'\sigma = N''\sigma$, $(\fn(N') \cup \fn(N'')) \cap \{\vect n, \vect n'\} = \emptyset$, and $\Res{\vect n}\sigma \enveq \Res{\vect n'}\sigma'$.
Moreover, $\fn(Q') \cap \{ \vect n, \vect n'\} = \emptyset$, and $\Sigma \vdash Q = P_0 \parop Q_0\{\subst{M}{x}\} = P_0'\sigma \parop Q_0'\sigma\{\subst{M'\sigma}{x}\} = Q' \sigma$.

\item The cases $\brn{Then}'$ and $\brn{Else}'$ are similar: the equalities that trigger reductions happen both in $P$ and in $P' \sigma$.
\item The case in which we apply $\equivP$ holds by Lemma~\ref{lem:equiv-enveq} and induction hypothesis.
\item Case in which we apply a context. The reduction $P = \CTX[P_0]\redP Q = \CTX[Q_0]$
is derived from $P_0 \redP Q_0$. 
If $E$ contains a restriction $\Res{n}$ above the hole, we rename $n$ so that
$n \notin \{\vect n, \vect n'\}$.
Hence $P' = \CTX'[P'_0]$ with $\Sigma \vdash \CTX' \sigma = \CTX$,
$\Sigma \vdash P'_0 \sigma = P_0$, and
$\fn(P'_0) \cap \{\vect n,\vect n'\} = \emptyset$. By induction hypothesis, $P'_0 \sigma' \redP Q'_0 \sigma'$, $\fn(Q'_0) \cap \{ \vect n, \vect n'\} = \emptyset$, and $\Sigma \vdash Q_0 = Q'_0 \sigma$ for some $Q'_0$. Let $Q' = \CTX'[Q'_0]$.
Then $P' \sigma' = \CTX'\sigma'[P'_0\sigma']\redP \CTX'\sigma'[Q'_0\sigma'] = Q' \sigma'$, $\fn(Q') \cap \{ \vect n, \vect n'\} = \emptyset$, and $\Sigma \vdash Q = \CTX[Q_0] = \CTX'\sigma[Q'_0\sigma] = Q' \sigma$.
\end{itemize}
We now prove the lemma with property~(*) by applying 
Lemma~\ref{lem:closing}\eqref{closing:redP} to the reduction $P'\sigma \redP Q'\sigma$ for the process $Q'$ obtained above.
Let $Y = \fv(P'\sigma) \cup \fv(Q'\sigma) = \fv(Q'\sigma) = \fv(Q')\setminus \dom(\sigma)$ and let $\sigma''$ map $Y$ to pairwise distinct fresh names.
We have $P'\sigma\sigma'' = P'\sigma$ so a fortiori
$\Sigma \vdash P'\sigma = P'\sigma\sigma''$, then
by Lemma~\ref{lem:closing}\eqref{closing:redP},
$P'\sigma\sigma'' \redP Q'\sigma\sigma''$
and $\Sigma \vdash Q'\sigma = Q'\sigma\sigma''$.
So $\Sigma \vdash Q = Q'\sigma = (Q'\sigma'')\sigma$.
Since $P' \sigma' \redP Q' \sigma'$, we have
$P'\sigma' = P'\sigma'\sigma'' \redP Q'\sigma'\sigma'' = (Q'\sigma'')\sigma'$.
Since $\fn(Q') \cap \{ \vect n, \vect n'\} = \emptyset$, we have 
$\fn(Q'\sigma'') \cap \{ \vect n, \vect n'\} = \emptyset$.
We also have $\fv(Q'\sigma'') \subseteq \dom(\sigma) = \dom(\sigma')$,
so we get the desired result by using $Q'\sigma''$ instead of $Q'$.
\qedhere
\end{proof}

Lemma~\ref{lem:ltr-enveq} gives two variants of the same result:
if $\Res{\vect n}\sigma \enveq \Res{\vect n'}\sigma'$
and $P$ such that $\Sigma \vdash P' \sigma = P$ has a labelled transition,
then
$P'\sigma'$ has a corresponding labelled transition.
The two variants differ by the closure assumptions and conclusions.

\begin{lemma}\label{lem:ltr-enveq}
Suppose that  $\Res{\vect n}\sigma \enveq \Res{\vect n'}\sigma'$, 
  $\fn(P') \cap \{ \vect n, \vect n'\} = \emptyset$, 
  $\Sigma \vdash P' \sigma = P$, 
$P \ltrP{\alpha} A$, and
$\sigma$, $\sigma'$, and $P'\sigma$ are closed.
\begin{enumerate}
\item\label{prop:ltr-enveq-open} If $\alpha$ is an output or  $\alpha = \Rcv{N}{M}$ 
with some $M'$ such that 
$\Sigma \vdash M' \sigma = M$, 
$M'\sigma$ is closed, and 
$\fn(M') \cap \{ \vect n, \vect n'\} = \emptyset$;
then 
$P' \sigma' \ltrP{\alpha'\sigma'} A' \sigma'$, 
$\fn(A') \cap \{ \vect n, \vect n'\} = \emptyset$, 
$A \equiv A' \sigma$, $\fn(\alpha') \cap \{ \vect n, \vect n'\} = \emptyset$, 
$\Sigma \vdash \alpha = \alpha' \sigma$, and 
$A'\sigma$ is closed for some $A'$, $\alpha'$.
\item\label{prop:ltr-enveq-closed} If $\alpha = \Rcv{N}{x}$ and $x \notin \dom(\sigma)$,
then 
$P' \sigma' \ltrP{\Rcv{N'\sigma'}{x}} A' \sigma'$, 
$\fn(A') \cap \{ \vect n, \vect n'\} = \emptyset$, 
$A \equiv A' \sigma$, 
$\fn(N') \cap \{ \vect n, \vect n'\} = \emptyset$, 
$\Sigma \vdash N = N' \sigma$, 
$\fv(A') \subseteq \dom(\sigma) \cup \{x\}$, and 
$\fv(N') \subseteq \dom(\sigma) 
$ for some $A'$, $N'$. 
\end{enumerate}
\end{lemma}
\begin{proof}
  \emph{Property~1:}
By induction on the derivation of $P \ltrP{\alpha} A$.
\begin{itemize}
\item Case $\brn{In}'$. We have $P = \Rcv{N}{x}.P_0 \ltrP{\Rcv{N}{M}} P_0\{\subst{M}{x}\} = A$ and there exists $M'$ such that $\Sigma \vdash M' \sigma = M$ and $\fn(M') \cap \{ \vect n, \vect n'\} = \emptyset$. 
We rename $x$ so that $x \notin \dom(\sigma)$. So $P' = \Rcv{N'}{x}.P'_0$ with $\Sigma \vdash N = N'\sigma$ and $\Sigma \vdash P_0 = P_0'\sigma$.
Let $A' = P'_0\{\subst{M'}{x}\}$ and $\alpha' = \Rcv{N'}{M'}$.
Then we have $P' \sigma' = \Rcv{N'\sigma'}{x}.P'_0\sigma'
\ltrP{\Rcv{N'\sigma'}{M'\sigma'}} 
P'_0\sigma'\{\subst{M'\sigma'}{x}\} = A' \sigma'$,
$\fn(A') \cap \{ \vect n, \vect n'\} = \emptyset$, 
$\Sigma \vdash A = P'_0\sigma\{\subst{M'\sigma}{x}\} = A' \sigma$, 
$\fn(\alpha') \cap \{ \vect n, \vect n'\} = \emptyset$, and
$\Sigma \vdash \alpha = \alpha' \sigma$.
Since $P'\sigma$ is closed, $\fv(P'_0) \subseteq \dom(\sigma) \cup \{x\}$;
moreover $M'\sigma$ is closed, so $A'\sigma$ is closed.

\item Case $\brn{Out-Var}'$. We have $P = \Snd{N}{M}.P_0 \ltrP{\Res{x}\Snd{N}{x}}
P_0 \parop \{\subst{M}{x}\} = A$ with $x\notin \fv(\Snd{N}{M}.P_0)$.
So $P' = \Snd{N'}{M'}.P'_0$ with 
$\Sigma \vdash N = N'\sigma$, $\Sigma \vdash M = M'\sigma$, and $\Sigma \vdash P_0 = P_0'\sigma$.
Let $A' = P'_0 \parop \{\subst{M'}{x}\}$ and $\alpha' = \Res{x}\Snd{N'}{x}$.
We have $P' \sigma' = \Snd{N'\sigma'}{M'\sigma'}.P'_0\sigma'           
\ltrP{\Res{x}\Snd{N'\sigma'}{x}}
P'_0\sigma' \parop \{\subst{M'\sigma'}{x}\} = A' \sigma'$,
$\fn(A') \cap \{ \vect n, \vect n'\} = \emptyset$,
$\Sigma \vdash A = P'_0\sigma \parop \{\subst{M'\sigma}{x}\} = A' \sigma$,
$\fn(\alpha') \cap \{ \vect n, \vect n'\} = \emptyset$, and
$\Sigma \vdash \alpha = \alpha' \sigma$.
Since $P'\sigma$ is closed, $P'_0\sigma$ is closed;
moreover $M'\sigma$ is closed, so $A'\sigma$ is closed.

\item Case $\brn{Scope}'$. The transition $P = \Res{n} P_0 \ltrP{\alpha} \Res{n}A_0 = A$ is derived from $P_0 \ltrP{\alpha} A_0$, where $n$ does not occur in $\alpha$. We rename $n$ so that $n \notin \{\vect n, \vect n'\}$ and $n \notin\fn(\sigma) \cup \fn(\sigma')$. We have $P' = \Res{n} P'_0$ for some $P'_0$, 
so $\Sigma \vdash P'_0\sigma = P_0$ and $P'_0 \sigma$ is closed. By induction hypothesis, 
$P'_0 \sigma' \ltrP{\alpha'\sigma'} A'_0 \sigma'$, $\fn(A'_0) \cap \{ \vect n, \vect n'\} = \emptyset$, $\Sigma \vdash A_0 = A'_0 \sigma$, $\fn(\alpha') \cap \{ \vect n, \vect n'\} = \emptyset$, $\Sigma \vdash \alpha = \alpha' \sigma$, and $A'_0\sigma$ is closed for some $A'_0$, $\alpha'$.
Let $A' = \Res{n} A'_0$. 
Then $P' \sigma' \ltrP{\alpha'\sigma'} A' \sigma'$ by $\brn{Scope}'$, 
so we have the desired result.

\item Case $\brn{Par}'$. The transition $P = P_0 \parop Q_0 \ltrP{\alpha} A_0 \parop Q_0 = A$ is derived from $P_0 \ltrP{\alpha} A_0$, where $\bv(\alpha) \cap\fv(Q_0) = \emptyset$. We have $P' = P'_0 \parop Q'_0$ for some $P'_0$, $Q'_0$, so
$\Sigma \vdash P'_0 \sigma = P_0$, $\Sigma \vdash Q'_0 \sigma = Q_0$, and
$P'_0 \sigma$ and $Q'_0 \sigma$ are closed. 
By induction hypothesis, 
$P'_0 \sigma' \ltrP{\alpha'\sigma'} A'_0 \sigma'$, $\fn(A'_0) \cap \{ \vect n,
\vect n'\} = \emptyset$, $\Sigma \vdash A_0 = A'_0 \sigma$, $\fn(\alpha') \cap
\{ \vect n, \vect n'\} = \emptyset$, $\Sigma \vdash \alpha = \alpha' \sigma
$, and $A'_0\sigma$ is closed for some $A'_0$, $\alpha'$.
Let $A' = A'_0 \parop Q'_0$.
Then $P' \sigma' \ltrP{\alpha'\sigma'} A' \sigma'$ by $\brn{Par}'$,
since $\fv(Q'_0\sigma') = \emptyset$. 
Since $P'\sigma$ is closed, $Q'_0\sigma$ is closed, so
$A'\sigma$ is closed. Therefore, we have the desired result.

\item Case $\brn{Struct}'$ follows by Lemma~\ref{lem:equiv-enveq} and induction hypothesis.
\end{itemize}
\noindent \emph{Property~2:}
By Lemma~\ref{lem:caract-ltrP}, $P \equivP \Res{\vect n''}(\Rcv{N}{y}.P_1 \parop P_2)$,
$A \equiv \Res{\vect n''}(P_1\{\subst{x}{y}\} \parop P_2)$,
and $\{\vect n''\} \cap \fn(N) = \emptyset$,
for some $\vect n''$, $P_1$, $P_2$, $N$, $y$.
We rename $\vect n''$ so that $\{ \vect n''\} \cap (\fn(\sigma) \cup \fn(\sigma')) = \emptyset$, and we rename $y$ so that $y \notin \dom(\sigma)$.
By Lemma~\ref{lem:equiv-enveq}, $P'\sigma' \equivP Q'\sigma'$,
$\fn(Q') \cap \{ \vect n, \vect n'\} = \emptyset$, and $\Sigma \vdash \Res{\vect n''}(\Rcv{N}{y}.P_1 \parop P_2) = Q' \sigma$ for some $Q'$ such that $Q'\sigma$ is closed, so $\fv(Q') \subseteq \dom(\sigma) = \dom(\sigma')$.
Hence, $Q'$ is of the form $Q' = \Res{\vect n''}(\Rcv{N'}{y}.P'_1 \parop P_2')$ with $\Sigma \vdash N = N'\sigma$, $\Sigma \vdash P_1 = P'_1\sigma$, and $\Sigma \vdash P_2 = P_2'\sigma$.
Hence, by Lemma~\ref{lem:caract-ltrP}, $P' \sigma' \equivP Q' \sigma' =\Res{\vect n''}(\Rcv{N'\sigma'}{y}.P'_1 \sigma' \parop P'_2 \sigma') \ltrP{\Rcv{N'\sigma'}{x}} \Res{\vect n''}(P'_1\sigma'\{\subst{x}{y}\} \parop P'_2\sigma')$.
Let $A' = \Res{\vect n''}(P'_1\{\subst{x}{y}\} \parop P'_2)$.
Then $P' \sigma' \ltrP{\Rcv{N'\sigma'}{x}} A' \sigma'$, $\fn(A') \cap \{ \vect n, \vect n'\} = \emptyset$ and $\fn(N') \cap \{ \vect n, \vect n'\} = \emptyset$ because $\fn(Q') \cap \{ \vect n, \vect n'\} = \emptyset$, $A \equiv \Res{\vect n''}(P_1\{\subst{x}{y}\} \parop P_2) \equiv \Res{\vect n''}(P'_1\sigma\{\subst{x}{y}\} \parop P'_2\sigma) \equiv A' \sigma$, $\Sigma \vdash N = N' \sigma$, and $\fv(A') \subseteq \dom(\sigma) \cup \{x\}$ and $\fv(N') \subseteq \dom(\sigma) = \dom(\sigma')$ because $\fv(Q') \subseteq \dom(\sigma) = \dom(\sigma')$.
\end{proof}

\subsection{Labelled Bisimilarity Implies Observational Equivalence}\label{app:onedirection}

The goal of this section is to establish the lemmas needed in the
outline of the argument that labelled bisimilarity implies
observational equivalence in Section~\ref{sec:bigpf}.

\begin{lemma}\label{lem:rename-eqstr-red}
Let $A$ and $B$ be two extended processes.
Let $\sigma$ be a bijective renaming (a substitution that is a bijection from names to names).
We have:
\begin{itemize}
\item $A \equiv B$ if and only if $A \sigma \equiv B \sigma$,
\item $A \rightarrow B$ if and only if $A \sigma \rightarrow B \sigma$,
\item $A \ltr{\alpha} B$ if and only if $A \sigma \ltr{\alpha\sigma} B\sigma$.
\end{itemize}
Let $A'$, $B'$, and $\alpha'$ be obtained from $A$, $B$, and $\alpha$, respectively, by replacing all variables (including their occurrences in domains of active substitutions) with distinct variables. We have:
\begin{itemize}
\item $A \equiv B$ if and only if $A' \equiv B'$,
\item $A \rightarrow B$ if and only if $A' \rightarrow B'$,
\item $A \ltr{\alpha} B$ if and only if $A' \ltr{\alpha'} B'$.
\end{itemize}
\end{lemma}
\begin{proof}
The implications from left to right are proved by induction on the derivations. We use that the equational theory is closed under renaming of names and variables. 
The same argument also proves the converse implications, via the inverse renaming.
\end{proof}

\begin{lemma}\label{lem:ren-wkbisim}
Let $A$ and $B$ be two closed extended processes.
\begin{itemize}
\item Let $\sigma$ be a bijective renaming.
We have $A \wkbisim B$ if and only if $A \sigma \wkbisim B \sigma$.
\item Let $A'$ and $B'$ be obtained from $A$ and $B$, respectively, by replacing all variables (including their occurrences in domains of active substitutions) with distinct variables. We have
$A \wkbisim B$ if and only if $A' \wkbisim B'$.
\end{itemize}
\end{lemma}
\begin{proof}
To prove the first point, we define a relation $\rel$ by
$A' \rel B'$ if and only if $A' = A\sigma$, $B' = B\sigma$, and $A \wkbisim B$ for some $A$ and $B$. We show that $\rel$ satisfies the three properties of Definition~\ref{def:wkbisim}. Then ${\rel} \subseteq {\wkbisim}$, so if $A \wkbisim B$, then $A' = A\sigma \wkbisim B' = B\sigma$.
\begin{enumerate}
\item Property~\ref{ppone} comes from Lemma~\ref{lem:rename-static-eq}.
\item If $A' \rel B'$, $A' \rightarrow A'_1$, and $A'_1$ is closed, then by Lemma~\ref{lem:rename-eqstr-red},
$A = A'\sigma^{-1} \rightarrow A'_1 \sigma^{-1}$. We let $A'' = A'_1 \sigma^{-1}$, which is also closed.
So by definition of $\wkbisim$, $B \rightarrow^* B''$ and $A'' \wkbisim B''$ for some $B''$. By Lemma~\ref{lem:rename-eqstr-red}, $B' = B\sigma \rightarrow^* B''\sigma$. We let $B'_1 = B''\sigma$. We have $A'_1 \rel B'_1$ and $B' \rightarrow^* B'_1$. So Property~\ref{pptwo} holds.

\item The proof of Property~\ref{ppthree} is similar to the proof of Property~\ref{pptwo}.
\end{enumerate}
The same argument also proves the converse, via the inverse renaming.

The proof of the second point is similar.
\end{proof}

\begin{restate}{Lemma}{\ref{lem:bisim-context-closed}}
$\wkbisim$ is closed by application of closing evaluation contexts.
\end{restate}%
\begin{proof}
Let $A$ and $B$ be two closed extended processes such that $A \wkbisim B$,
and $\CTX$ be an evaluation context closing for $A$ and $B$.
Our goal is to show that $\CTX[A] \wkbisim \CTX[B]$.
We first rename the free names and variables of $\CTX$ by Lemma~\ref{lem:ren-wkbisim}, so that the obtained context is simple. Then by Lemma~\ref{lem:simplecontexts}, we construct a context $\CTX'$ of the form $\Res{\vect u}(\hole \parop C'')$ such that $\CTX \equiv \CTX'$. Since $\wkbisim$ is invariant by structural equivalence, it is sufficient to show that $\CTX'[A] \wkbisim \CTX'[B]$. 
Hence, it is sufficient to consider evaluation contexts of the form
  $\nuc{\hole}$, such that $\nuc{A}$ and $\nuc{B}$ are closed.

  To every relation $\rel$ on closed extended processes, we
  associate the relation ${\rel'} = \{ (\nuc A, \nuc B) \mid A \rel B, \nuc{\hole} \text{ closing for $A$ and $B$} \}$.
  We prove that, if $\rel$ is a labelled bisimulation,
  then $\rel'$ is a labelled bisimulation up to $\equiv$, hence ${\rel}
  \subseteq {\equiv\rel'\equiv} \subseteq {\wkbisim}$.
  For ${\rel} = {\wkbisim}$, this establishes that $\wkbisim$ is
  closed by application of evaluation contexts $\nuc{\hole}$.

  Assume $S \rel' T$, with $S = \nuc A$, $T = \nuc B$, and $A \rel B$.
Let $\pnf(A) = \Res{\vect n}(\sigma \parop P)$, $\pnf(B) = \Res{\vect n'}(\sigma' \parop P')$, $\pnf(C) = \Res{\vect n''}(\sigma'' \parop P'')$, $\vect u$ consist of names $\vect n'''$ and variables $\vect x$.
We suppose that $A$ or $C$ is not a plain process. (The case in which $A$ and $C$ are plain processes is simpler.) Since $A \rel B$, we have $\dom(A) = \dom(B)$, so we also have that $B$ or $C$ is not a plain process.
We rename $\vect n$, $\vect n'$, and $\vect n''$ so that they are disjoint,
the names of $\vect n$ and of $\vect n'$ are not free in $\sigma'' \parop P''$, and the names of
$\vect n''$ are not free in $\sigma \parop P$ nor in $\sigma' \parop P'$.
Since $A$ is closed, by Lemma~\ref{lem:pnf-closed}, $\pnf(A)$ is closed, so 
$P$ and the image of $\sigma$ have no free variables,
so they are not modified by $\sigma''$.
Similarly, $P'$ and the image of $\sigma'$ have no free variables,
so they are not modified by $\sigma''$.
Hence $\pnf(S) = \Res{\vect n''',\vect n,\vect n''}({(\sigma \parop \sigma''\sigma)}_{|\dom(\sigma \parop \sigma''\sigma) \setminus \{ \vect x \}} \parop P \parop P'' \sigma)$
and $\pnf(T) = \Res{\vect n''',\vect n',\vect n''}({(\sigma' \parop \sigma''\sigma')}_{|\dom(\sigma' \parop \sigma''\sigma') \setminus \{ \vect x \}} \parop P' \parop P'' \sigma')$.

We argue that $\rel'$ satisfies the three properties of a labelled bisimulation up to $\equiv$ (Definition~\ref{def:weakbisimstruct}). 
The proof of the first property is trivial; those of 
the last two properties (given in more detail below) go as follows. From a (labelled or internal) reduction of $S$, we infer a reduction of $\pnf(S)$, hence a reduction of $P \parop P''\sigma$ by a decomposition lemma (Lemma~\ref{lem:decomp-ltrpnf} or~\ref{lem:decomp-redpnf}), hence reductions of $P$ and/or $P''\sigma$ by another decomposition lemma (Lemma~\ref{lem:decomp-ltrP} or~\ref{lem:decomp-redP-closed}).
From a reduction of $P$, we infer a reduction of $A$, hence a reduction of $B$ since $\rel$ is labelled bisimulation, so a reduction of $P'$ by a decomposition lemma.
From a reduction of $P''\sigma$, we infer a reduction of $P''\sigma'$ using the static equivalence $A \enveq B$, which means that $\Res{n}\sigma \enveq \Res{n'}\sigma'$.
Therefore, in all cases, we obtain a reduction of $P' \parop P''\sigma'$, hence a reduction of $\pnf(T)$, so a reduction of $T$.
In more detail, the proof proceeds as follows.
\begin{enumerate}
\item  $S \enveq T$ immediately follows from $A \enveq B$ by Lemma~\ref{LEM:INVARIANT-STATIC-EQ}. 

  \item For every $S \rightarrow S'$ with $S'$ closed, we prove that $T \rightarrow^* T'$
  and $S' \equiv \rel' \equiv T'$ for some $T'$. 
By Lemma~\ref{lem:red-std-to-pnf}, we have $\pnf(S) \redpnf \pnf(S')$.
By Lemma~\ref{lem:decomp-redpnf}, we have $P \parop P''\sigma \redP Q$ and 
$\pnf(S') \equiv  \Res{\vect n''',\vect n,\vect n''}({(\sigma \parop \sigma''\sigma)}_{|\dom(\sigma \parop \sigma''\sigma) \setminus \{ \vect x \}} \parop Q)$ for some $Q$.
By Lemma~\ref{lem:decomp-redP-closed}, we have four cases:
\begin{enumerate}

\item $P \redP Q'$ and $Q \equiv Q'\parop P''\sigma$ for some closed process $Q'$.
By Lemmas~\ref{lem:equivpnf} and~\ref{lem:red-pnf-to-std}, $A \equiv \Res{\vect n}(\sigma \parop P) \rightarrow A'$ where $A' = \Res{\vect n}(\sigma \parop Q')$.
Since $A \rel B$ and $A'$ is closed, we have $B \rightarrow^* B'$ and $A' \rel B'$ for some $B'$.
By Lemma~\ref{lem:red-std-to-pnf}, $\pnf(B) \redpnf^* \pnf(B')$, 
so by Lemma~\ref{lem:decomp-redpnf-closed}, $P' \redP^* Q''$ and $\pnf(B') \equiv  \Res{\vect n'}(\sigma' \parop Q'')$ for some closed process $Q''$.
We rename $\vect n''$ so that $\{\vect n''\} \cap \fn(Q') = \emptyset$
and $\{\vect n''\} \cap \fn(Q'') = \emptyset$.
Hence, by Lemmas~\ref{lem:equivpnf} and~\ref{lem:red-pnf-to-std}, 
\begin{align*}
T &\equiv \Res{\vect n''',\vect n',\vect n''}({(\sigma' \parop \sigma''\sigma')}_{|\dom(\sigma' \parop \sigma''\sigma') \setminus \{ \vect x \}} \parop P' \parop P'' \sigma')\\
&\rightarrow^* \Res{\vect n''',\vect n',\vect n''}({(\sigma' \parop \sigma''\sigma')}_{|\dom(\sigma' \parop \sigma''\sigma') \setminus \{ \vect x \}} \parop Q'' \parop P'' \sigma') \\
&\quad \equiv \Res{\vect u}(\Res{\vect n'}(\sigma' \parop Q'') \parop \Res{\vect n''}(\sigma'' \parop P'')) \equiv \Res{\vect u}(B' \parop C)
\end{align*}
If there is at least one reduction step in this trace, we let $T' = \Res{\vect u}(B' \parop C)$; otherwise, we let $T' = T$.
In all cases, $T \rightarrow^* T'$ and $T' \equiv \Res{\vect u}(B' \parop C)$.
Since $A' \rel B'$,
\begin{align*}
S' &\equiv 
\Res{\vect n''',\vect n,\vect n''}({(\sigma \parop \sigma''\sigma)}_{|\dom(\sigma \parop \sigma''\sigma) \setminus \{ \vect x \}} \parop Q' \parop P''\sigma) \\
&\equiv 
\Res{\vect u}(\Res{\vect n}(\sigma \parop Q') \parop \Res{\vect n''}(\sigma'' \parop P''))\\
& \equiv \Res{\vect u}(A' \parop C)\,,
\end{align*}
and $\nuc{\hole}$ is closing for $A'$ and $B'$, we have $S' \equiv \rel' \equiv T'$.

\item $P''\sigma \redP Q'$ and $Q \equiv P \parop Q'$ for some closed process $Q'$.
Since $A \rel B$, we have $A \enveq B$, that is, $\Res{\vect n}\sigma \enveq \Res{\vect n'}\sigma'$.
By Lemma~\ref{lem:red-enveq}, $P'' \sigma' \redP Q'' \sigma'$, $\fn(Q'') \cap \{ \vect n, \vect n'\} = \emptyset$, and $\Sigma \vdash Q' = Q'' \sigma$ for some $Q''$ such that $Q''$ is closed, so $\fv(Q'') \subseteq \dom(\sigma) = \dom(\sigma')$.
So by Lemmas~\ref{lem:equivpnf} and~\ref{lem:red-pnf-to-std}, 
\begin{align*}
T &\equiv \Res{\vect n''',\vect n',\vect n''}({(\sigma' \parop \sigma''\sigma')}_{|\dom(\sigma' \parop \sigma''\sigma') \setminus \{ \vect x \}} \parop P' \parop P'' \sigma')\\
&\rightarrow \Res{\vect n''',\vect n',\vect n''}({(\sigma' \parop \sigma''\sigma')}_{|\dom(\sigma' \parop \sigma''\sigma') \setminus \{ \vect x \}} \parop P' \parop Q'' \sigma') \\
&\quad \equiv \Res{\vect u}(\Res{\vect n'}(\sigma' \parop P') \parop \Res{\vect n''}(\sigma'' \parop Q'')) \equiv \Res{\vect u}(B \parop C')
\end{align*}
where $C' = \Res{\vect n''}(\sigma'' \parop Q'')$.
Moreover, 
\begin{align*}
S' &\equiv 
\Res{\vect n''',\vect n,\vect n''}({(\sigma \parop \sigma''\sigma)}_{|\dom(\sigma \parop \sigma''\sigma) \setminus \{ \vect x \}} \parop P \parop Q''\sigma) \\
&\equiv 
\Res{\vect u}(\Res{\vect n}(\sigma \parop P) \parop \Res{\vect n''}(\sigma'' \parop Q'')) \\
&\equiv \Res{\vect u}(A \parop C')
\end{align*}
We let $T' = \Res{\vect u}(B \parop C')$. We have $T \rightarrow T'$. Since $\fv(Q'') \subseteq \dom(\sigma) = \dom(\sigma')$, $\Res{\vect u}(\hole \parop C')$ is closing for $A$ and $B$, and since $A \rel B$, we have $S' \equiv \rel' T'$.

\item $P \ltrP{\Rcv{N}{x}} A_1$, $P''\sigma \ltrP{\Res{x}\Snd{N}{x}} C_1$, and $Q \equiv \Res{x}(A_1 \parop C_1)$ for some $A_1$, $C_1$, $x$, and ground term $N$.
We rename $x$ so that $x \notin \dom(\sigma) = \dom(\sigma')$.
By Lemma~\ref{lem:caract-ltrP} applied twice,
$P \equivP \Res{\vect n_1}(\Rcv{N}{y}.P_1 \parop P_2)$,
$A_1 \equiv \Res{\vect n_1}(P_1\{\subst{x}{y}\} \parop P_2)$,
$\{\vect n_1\} \cap \fn(N) = \emptyset$ and
$P''\sigma \equivP \Res{\vect n_2}(\Snd{N}{M}.P_3 \parop P_4)$,
$C_1 \equiv \Res{\vect n_2}(P_3 \parop \{\subst{M}{x}\} \parop P_4)$,
$\{\vect n_2\} \cap \fn(N) = \emptyset$,
$x \notin \fv(\Snd{N}{M}.P_3 \parop P_4)$.
By Lemma~\ref{lem:closing}\eqref{closing:equivP}, we transform
$\Res{\vect n_1}(\Rcv{N}{y}.P_1 \parop P_2)$ and 
$\Res{\vect n_2}(\Snd{N}{M}.P_3 \parop P_4)$ into closed processes that
satisfy the same properties.
Since $A \rel B$, we have $A \enveq B$, that is, $\Res{\vect n}\sigma \enveq \Res{\vect n'}\sigma'$.
By Lemma~\ref{lem:equiv-enveq}, $P''\sigma' \equivP Q' \sigma'$,
$\fn(Q') \cap \{\vect n, \vect n'\} = \emptyset$,
and $\Sigma \vdash Q' \sigma = \Res{\vect n_2}(\Snd{N}{M}.P_3 \parop P_4)$
for some $Q'$ such that $Q'\sigma$ is closed, so $\fv(Q') \subseteq \dom(\sigma)$.
Then $Q'$ is of the form $Q' = \Res{\vect n_2}(\Snd{N'}{M'}.P'_3 \parop P'_4)$,
with $\Sigma \vdash N'\sigma = N$, $\Sigma \vdash M'\sigma = M$,
$\Sigma \vdash P'_3 \sigma = P_3$, $\Sigma \vdash P'_4 \sigma = P_4$,
and $\fv(\Snd{N'}{M'}) \subseteq \dom(\sigma) = \dom(A)$.
We rename $\vect n_1$ and $\vect n_2$ so that $\{\vect n_1\} \cap \fn(M) = \emptyset$, $\{\vect n_1\} \cap \{\vect n_2\} = \emptyset$,
$\{\vect n_1\} \cap (\fn(P'_3\sigma) \cup \fn(P'_4\sigma) \cup \fn(\sigma \parop \sigma''\sigma)) = \emptyset$, 
$\{\vect n_2\} \cap (\fn(P_1) \cup \fn(P_2) \cup \fn(\sigma \parop \sigma''\sigma)) = \emptyset$.
By Lemma~\ref{lem:caract-ltrP},
$P \equivP \Res{\vect n_1}(\Rcv{N}{y}.P_1 \parop P_2)
\ltrP{\Rcv{N}{M}}
\Res{\vect n_1}(P_1\{\subst{M}{y}\} \parop P_2)$.
By definition of $\ltrpnf{\Rcv{N'}{M'}}$, $A \equiv \Res{\vect n}(\sigma \parop P) 
\ltrpnf{\Rcv{N'}{M'}} A'$ where $A' = \Res{\vect n}(\sigma \parop \Res{\vect n_1}(P_1\{\subst{M}{y}\} \parop P_2))$.
(The elements of $\vect n$ do not occur in $\Rcv{N'}{M'}$ since $\fn(Q') \cap \{\vect n, \vect n'\} = \emptyset$.)
So by Lemma~\ref{lem:redalpha-pnf-to-std}, $A \ltr{\Rcv{N'}{M'}} A'$.
Since $A'$ is closed and $A \rel B$, we have $B \rightarrow^* \ltr{\Rcv{N'}{M'}} B'' \rightarrow^* B'$
and $A' \rel B'$ for some $B'$, $B''$.
By Lemmas~\ref{lem:red-std-to-pnf} and~\ref{lem:redalpha-std-to-pnf},
$\pnf(B) = \Res{\vect n'}(\sigma' \parop P') \redpnf^* \ltrpnf{\Rcv{N'}{M'}} B''$.
By Lemmas~\ref{lem:decomp-redpnf-closed} and~\ref{lem:decomp-ltrpnf},
$P' \redP^* \ltrP{\Rcv{N'\sigma'}{M'\sigma'}} B'''$ and $B'' \equiv \Res{\vect n'}(\sigma' \parop B''')$ for some $B'''$.
By Lemma~\ref{lem:caract-ltrP}, $P' \redP^* \equivP \Res{\vect n_3}(\Rcv{N'\sigma'}{z}.P_5 \parop P_6)$,
$B''' \equiv \Res{\vect n_3}(P_5\{\subst{M'\sigma'}{z}\} \parop P_6)$,
and $\{\vect n_3\} \cap \fn(\Rcv{N'\sigma'}{M'\sigma'}) = \emptyset$ for some $\vect n_3$, $P_5$, and $P_6$.
We rename $\vect n''$ so that $\{ \vect n'' \} \cap \fn(\Res{\vect n_1}(P_1\{\subst{M}{y}\} \parop P_2)) = \emptyset$
and $\{ \vect n'' \} \cap \fn(B''') = \emptyset$. 
Then we have
\begin{align*}
S' &\equiv 
\Res{\vect n''',\vect n,\vect n''}({(\sigma \parop \sigma''\sigma)}_{|\dom(\sigma \parop \sigma''\sigma) \setminus \{ \vect x \}} \parop \Res{x}(A_1 \parop C_1))\\
& \equiv 
\Res{\vect u}\Res{\vect n}\Res{\vect n''}(\sigma \parop \sigma''\sigma \parop \Res{x}(\Res{\vect n_1}(P_1\{\subst{x}{y}\} \parop P_2) \parop \Res{\vect n_2}(P_3 \parop \{\subst{M}{x}\} \parop P_4)) \\
&\equiv 
\Res{\vect u}\Res{\vect n}\Res{\vect n''}(\sigma \parop \sigma''\sigma \parop \Res{x}(\Res{\vect n_1}(P_1\{\subst{x}{y}\} \parop P_2) \parop \Res{\vect n_2}(P'_3\sigma \parop \{\subst{M}{x}\} \parop P'_4\sigma)) \\
&\equiv 
\Res{\vect u}\Res{\vect n}\Res{\vect n''}\Res{\vect n_1}\Res{\vect n_2}(\sigma \parop \sigma''\sigma \parop P_1\{\subst{M}{y}\} \parop P_2 \parop P'_3\sigma \parop P'_4\sigma) \\
&\equiv 
\Res{\vect u}\Res{\vect n}\Res{\vect n''}\Res{\vect n_1}\Res{\vect n_2}(\sigma \parop \sigma'' \parop P_1\{\subst{M}{y}\} \parop P_2 \parop P'_3 \parop P'_4) \\
&\equiv 
\Res{\vect u,\vect n_2}\Res{\vect n}\Res{\vect n''}(\sigma \parop \Res{\vect n_1}(P_1\{\subst{M}{y}\} \parop P_2) \parop \sigma'' \parop (P'_3 \parop P'_4)) \\
&\equiv \Res{\vect u,\vect n_2}(A' \parop C')
\end{align*}
where $C' = \Res{\vect n''}(\sigma'' \parop (P'_3 \parop P'_4))$.
We have 
\begin{align*}
T &\equiv \Res{\vect n''',\vect n',\vect n''}({(\sigma' \parop \sigma''\sigma')}_{|\dom(\sigma' \parop \sigma''\sigma') \setminus \{ \vect x \}} \parop P' \parop P'' \sigma')\\
&\rightarrow^*\equiv \Res{\vect u}\Res{\vect n'}\Res{\vect n''}(\sigma' \parop \sigma''\sigma' \parop \Res{\vect n_3}(\Rcv{N'\sigma'}{z}.P_5 \parop P_6) \\
&\phantom{\rightarrow^*\equiv \Res{\vect u}\Res{\vect n'}\Res{\vect n''}}\parop \Res{\vect n_2}(\Snd{N'\sigma'}{M'\sigma'}.P'_3\sigma' \parop P'_4\sigma'))\\
&\rightarrow \Res{\vect u,\vect n_2}\Res{\vect n'}\Res{\vect n''}(\sigma' \parop \sigma''\sigma' \parop \Res{\vect n_3}(P_5\{\subst{M'\sigma'}{z}\} \parop P_6) \parop (P'_3 \sigma' \parop P'_4 \sigma'))\\
&\equiv \Res{\vect u,\vect n_2}\Res{\vect n'}\Res{\vect n''}(\sigma' \parop \Res{\vect n_3}(P_5\{\subst{M'\sigma'}{z}\} \parop P_6) \parop \sigma'' \parop (P'_3 \parop P'_4))\\
&\equiv \Res{\vect u,\vect n_2}(\Res{\vect n'}(\sigma' \parop B''') \parop C')\\
&\equiv \Res{\vect u,\vect n_2}(B'' \parop C') \rightarrow^* \Res{\vect u}(B' \parop C')
\end{align*}
We let $T' = \Res{\vect u,\vect n_2}(B' \parop C')$.
Hence, $T \rightarrow^* T'$. Since $\fv(P'_3 \parop P'_4) \subseteq \fv(Q') \subseteq \dom(\sigma)$, $\Res{\vect u,\vect n_2}(\hole \parop C')$ is closing for $A'$ and $B'$, and moreover $A' \rel B'$, so $S' \equiv \rel' T'$.

\item $P \ltrP{\Res{x}\Snd{N}{x}} A_1$, $P''\sigma \ltrP{\Rcv{N}{x}} C_1$, and $Q \equiv \Res{x}(A_1 \parop C_1)$ for some $A_1$, $C_1$, $x$, and ground term $N$.
We rename $x$ so that $x \notin \dom(\sigma) = \dom(\sigma')$.
By Lemma~\ref{lem:closing}\eqref{closing:ltrP}, we transform $A_1$ into a closed extended process that satisfies the same properties.
Since $A \rel B$, we have $A \enveq B$, that is, $\Res{\vect n}\sigma \enveq \Res{\vect n'}\sigma'$.
By Lemma~\ref{lem:ltr-enveq}\eqref{prop:ltr-enveq-closed}, $P'' \sigma' \ltrP{\Rcv{N'\sigma'}{x}} C_2 \sigma'$, $\fn(C_2) \cap \{ \vect n, \vect n'\} = \emptyset$, $C_1 \equiv C_2 \sigma$, $\fn(N') \cap \{ \vect n, \vect n'\} = \emptyset$, $\Sigma \vdash N = N' \sigma$, $\fv(C_2) \subseteq \dom(\sigma) \cup \{x\}$, and $\fv(N') \subseteq \dom(\sigma) = \dom(A)$ for some $C_2$ and $N'$.
By definition of $\ltrpnf{\Res{x}\Snd{N'}{x}}$, $A \equiv \Res{\vect n}(\sigma \parop P) 
\ltrpnf{\Res{x}\Snd{N'}{x}} A'$ where $A' = \Res{\vect n}(\sigma \parop A_1)$,
so by Lemma~\ref{lem:redalpha-pnf-to-std}, $A \ltr{\Res{x}\Snd{N'}{x}} A'$.
Since $A'$ is closed and $A \rel B$, we have $B \rightarrow^* \ltr{\Res{x}\Snd{N'}{x}} B'' \rightarrow^* B'$
and $A' \rel B'$ for some $B'$, $B''$.
By Lemmas~\ref{lem:red-std-to-pnf} and~\ref{lem:redalpha-std-to-pnf},
$\pnf(B) = \Res{\vect n'}(\sigma' \parop P') \redpnf^* \ltrpnf{\Res{x}\Snd{N'}{x}} B''$.
By Lemmas~\ref{lem:decomp-redpnf-closed} and~\ref{lem:decomp-ltrpnf},
$P' \redP^* \ltrP{\Res{x}\Snd{N'\sigma'}{x}} B'''$ and $B'' \equiv \Res{\vect n'}(\sigma' \parop B''')$ for some $B'''$.
By Lemma~\ref{lem:comp-ltr}, $P' \parop P'' \sigma' \redP^* \redP \Res{x} (B''' \parop C_2 \sigma')$.
We rename $\vect n''$ so that $\{ \vect n'' \} \cap \fn(A_1) = \emptyset$
and $\{ \vect n'' \} \cap \fn(B''') = \emptyset$. 
Moreover, 
\begin{align*}
S' &\equiv 
\Res{\vect n''',\vect n,\vect n''}({(\sigma \parop \sigma''\sigma)}_{|\dom(\sigma \parop \sigma''\sigma) \setminus \{ \vect x \}} \parop \Res{x}(A_1 \parop C_1))\\
& \equiv 
\Res{x,\vect u}\Res{\vect n}\Res{\vect n''}(\sigma \parop \sigma''\sigma \parop A_1 \parop C_2 \sigma) \\
&\equiv 
\Res{x,\vect u}(\Res{\vect n}(\sigma \parop A_1) \parop \Res{\vect n''}(\sigma'' \parop C_2)) \\
&\equiv \Res{x,\vect u}(A' \parop C')
\end{align*}
where $C' =  \Res{\vect n''}(\sigma'' \parop C_2)$.
We have 
\begin{align*}
T &\equiv \Res{\vect n''',\vect n',\vect n''}({(\sigma' \parop \sigma''\sigma')}_{|\dom(\sigma' \parop \sigma''\sigma') \setminus \{ \vect x \}} \parop P' \parop P'' \sigma')\\
&\rightarrow^* \rightarrow \Res{\vect n''',\vect n',\vect n''}({(\sigma' \parop \sigma''\sigma')}_{|\dom(\sigma' \parop \sigma''\sigma') \setminus \{ \vect x \}} \parop \Res{x}(B'''\parop C_2 \sigma')) \\
&\quad \equiv
\Res{x,\vect u}(\Res{n'}(\sigma' \parop B''') \parop \Res{n''}(\sigma'' \parop C_2)) \equiv \Res{x,\vect u}(B'' \parop C')\\
&\quad \rightarrow^* \Res{x,\vect u}(B' \parop C')
\end{align*}
We let $T' = \Res{x,\vect u}(B' \parop C')$.
We have $T \rightarrow^* T'$ and, since $A' \rel B'$ and $\Res{x,\vect u}(\hole \parop C')$ is closing for $A'$ and $B'$, 
$S' \equiv \rel' T'$.

\end{enumerate}

\item For every $S \ltr{\alpha} S'$ with $S'$ closed and $\fv(\alpha) \subseteq \dom(S)$, we prove that $T \rightarrow^* \ltr{\alpha} \rightarrow^* T'$ and $S' \equiv \rel' \equiv T'$ for some $T'$.
We rename $\vect n''', \vect n, \vect n', \vect n''$ so that these names do 
not occur in $\alpha$.
By Lemma~\ref{lem:redalpha-std-to-pnf}, we have $\pnf(S) \ltrpnf{\alpha} S'$.
By Lemma~\ref{lem:decomp-ltrpnf}, we have $P \parop P'' \sigma \ltrP{\alpha'} A_0$,
$S' \equiv  \Res{\vect n''',\vect n,\vect n''}({(\sigma \parop \sigma''\sigma)}_{|\dom(\sigma \parop \sigma''\sigma) \setminus \{ \vect x \}} \parop A_0)$, and 
$\bv(\alpha) \cap \dom(\sigma\parop \sigma''\sigma) \setminus \{ \vect x \} = \emptyset$  for $\alpha' = \alpha (\sigma \parop \sigma''\sigma)_{|\dom(\sigma \parop \sigma''\sigma) \setminus \{ \vect x \}} = \alpha \sigma''\sigma$ and some
$A_0$. 
We rename $\vect x$ so that $\bv(\alpha) \cap \{ \vect x \} = \emptyset$.
By Lemma~\ref{lem:decomp-ltrP},
we have two cases:
\begin{enumerate}
\item $P \ltrP{\alpha'} A_1$ and $A_0 \equiv A_1 \parop P''\sigma$ for some $A_1$.
By Lemma~\ref{lem:closing}\eqref{closing:ltrP}, we transform $A_1$ into
a closed extended process that satisfies the same properties.
We have $\alpha' = (\alpha \sigma'')\sigma$,
the elements of $\vect n$ do not occur in $\alpha \sigma''$,
because they do not occur in $\alpha$ nor in $\sigma''$.
We also have $\bv(\alpha') \cap \dom(\sigma) = \emptyset$.
By definition of $\ltrpnf{\alpha\sigma''}$, we have $A \equiv \Res{\vect n}(\sigma \parop P) \ltrpnf{\alpha\sigma''} A'$ where $A' = \Res{\vect n}(\sigma \parop A_1)$, so by Lemma~\ref{lem:redalpha-pnf-to-std}, $A \ltr{\alpha\sigma''} A'$.
Since $\fv(\alpha) \subseteq \dom(S) \subseteq \dom(\sigma) \cup \dom(\sigma'')$, we have $\fv(\alpha\sigma'') \subseteq \dom(\sigma) = \dom(A)$.
Since $A'$ is closed and $A \rel B$, we have $B \rightarrow^* \ltr{\alpha\sigma''} B'' \rightarrow^* B'$
and $A' \rel B'$ for some $B'$ and $B''$.
By Lemmas~\ref{lem:red-std-to-pnf} and~\ref{lem:redalpha-std-to-pnf},
$\pnf(B) = \Res{\vect n'}(\sigma' \parop P') \redpnf^* \ltrpnf{\alpha\sigma''} B''$.
By Lemmas~\ref{lem:decomp-redpnf-closed} and~\ref{lem:decomp-ltrpnf},
$P' \redP^* \ltrP{\alpha\sigma''\sigma'} B'''$, $B'' \equiv \Res{\vect n'}(\sigma' \parop B''')$, and $\bv(\alpha) \cap \dom(\sigma') = \emptyset$ for some $B'''$.
Hence by $\brn{Par}'$, $P' \parop P''\sigma' \redP^* \ltrP{\alpha\sigma''\sigma'} B''' \parop P''\sigma'$. By definition of $\ltrpnf{\alpha}$,
\[\begin{split} 
\pnf(T) &= \Res{\vect n''',\vect n',\vect n''}({(\sigma' \parop \sigma''\sigma')}_{|\dom(\sigma' \parop \sigma''\sigma') \setminus \{ \vect x \}} \parop P' \parop P'' \sigma')\\
& \redpnf^* \ltrpnf{\alpha} \Res{\vect n''',\vect n',\vect n''}({(\sigma' \parop \sigma''\sigma')}_{|\dom(\sigma' \parop \sigma''\sigma') \setminus \{ \vect x \}} \parop B''' \parop P'' \sigma')\,.
\end{split}\]
(We have $\bv(\alpha) \cap \fv({(\sigma' \parop \sigma''\sigma')}_{|\dom(\sigma' \parop \sigma''\sigma') \setminus \{ \vect x \}}) = \emptyset$ because we have 
$\bv(\alpha) = \bv(\alpha')$,
$\fv({(\sigma' \parop \sigma''\sigma')}_{|\dom(\sigma' \parop \sigma''\sigma') \setminus \{ \vect x \}}) = \dom(\sigma') \cup \dom(\sigma'') \setminus \{\vect x\} = \dom(\sigma) \cup \dom(\sigma'') \setminus \{ \vect x\}$,  and 
$\bv(\alpha') \cap (\dom(\sigma) \cup \dom(\sigma'') \setminus \{ \vect x \}) = \emptyset$.)
We rename $\vect n''$ so that $\{\vect n''\} \cap \fn(A_1) = \emptyset$ and $\{\vect n''\} \cap \fn(B''') = \emptyset$. 
By Lemmas~\ref{lem:equivpnf}, \ref{lem:red-pnf-to-std}, and~\ref{lem:redalpha-pnf-to-std}, we have
\begin{align*}
T &\equiv \pnf(T) \\
&\rightarrow^* \ltr{\alpha} \Res{\vect n''',\vect n',\vect n''}({(\sigma' \parop \sigma''\sigma')}_{|\dom(\sigma' \parop \sigma''\sigma') \setminus \{ \vect x \}} \parop B''' \parop P'' \sigma')\\
&\equiv \Res{\vect u}(\Res{\vect n'}(\sigma' \parop B''') \parop \Res{\vect n''}(\sigma'' \parop P''))\\
&\equiv \Res{\vect u}(B'' \parop C) \\
&\rightarrow^* \Res{\vect u}(B' \parop C)
\end{align*}
Moreover,
\begin{align*}
S' &\equiv  \Res{\vect n''',\vect n,\vect n''}({(\sigma \parop \sigma''\sigma)}_{|\dom(\sigma \parop \sigma''\sigma) \setminus \{ \vect x \}} \parop A_0)\\
&\equiv \Res{\vect n''',\vect n,\vect n''}({(\sigma \parop \sigma''\sigma)}_{|\dom(\sigma \parop \sigma''\sigma) \setminus \{ \vect x \}} \parop A_1 \parop P''\sigma)\\
&\equiv \Res{\vect u}(\Res{\vect n}(\sigma \parop A_1) \parop \Res{\vect n''}(\sigma'' \parop P''))\\
&\equiv \Res{\vect u}(A' \parop C)
\end{align*}
We let $T' = \Res{\vect u}(B' \parop C)$.
Then we have $T \rightarrow^* \ltr{\alpha} \rightarrow^* T'$
and $\nuc{\hole}$ is closing for $A'$ and $B'$ so $S' \equiv \rel' T'$.

\item $P''\sigma \ltrP{\alpha'} A_1$ and $A_0 \equiv P \parop A_1$ for some $A_1$.
Since $A \rel B$, we have $A \enveq B$, that is, $\Res{\vect n}\sigma \enveq \Res{\vect n'}\sigma'$.
We have $\Sigma \vdash (\alpha\sigma'')\sigma = \alpha'$, so if $\alpha' = \Rcv{N}{M}$, then we have $\alpha\sigma'' = \Rcv{N'}{M'}$, $\Sigma \vdash M'\sigma = M$, and $\fn(M') \cap \{\vect n, \vect n'\} = \emptyset$ for some $N', M'$.
By Lemma~\ref{lem:ltr-enveq}\eqref{prop:ltr-enveq-open}, $P'' \sigma' \ltrP{\alpha''\sigma'} A''\sigma'$, $\fn(A'') \cap \{\vect n, \vect n'\} = \emptyset$, $A_1 \equiv A''\sigma$, $\fn(\alpha'') \cap \{\vect n, \vect n'\} = \emptyset$, $\Sigma \vdash \alpha' = \alpha'' \sigma$, and $A''\sigma$ is closed for some $A'', \alpha''$.
By \brn{Par'}, $P' \parop P''\sigma' \ltrP{\alpha''\sigma'} P' \parop A'' \sigma'$.
We have $\Sigma \vdash \alpha\sigma''\sigma = \alpha' = \alpha''\sigma$,
$\fn(\alpha\sigma'') \cap \{\vect n, \vect n'\} = \emptyset$,
$\fn(\alpha'') \cap \{\vect n, \vect n'\} = \emptyset$,
and $\Res{\vect n}\sigma \enveq \Res{\vect n'}\sigma'$, so
$\Sigma \vdash \alpha\sigma''\sigma' = \alpha''\sigma'$ by definition of 
static equivalence.
By definition of $\ltrpnf{\alpha}$,
\[\begin{split}
\pnf(T) &= \Res{\vect n''',\vect n',\vect n''}({(\sigma' \parop \sigma''\sigma')}_{|\dom(\sigma' \parop \sigma''\sigma') \setminus \{ \vect x \}} \parop P' \parop P'' \sigma')\\
&\quad\ltrpnf{\alpha} \Res{\vect n''',\vect n',\vect n''}({(\sigma' \parop \sigma''\sigma')}_{|\dom(\sigma' \parop \sigma''\sigma') \setminus \{ \vect x \}} \parop P' \parop A'' \sigma')\,.
\end{split}\]
(We have $\bv(\alpha''\sigma') \cap \fv({(\sigma' \parop \sigma''\sigma')}_{|\dom(\sigma' \parop \sigma''\sigma') \setminus \{ \vect x \}})  = \emptyset$ because
$\bv(\alpha''\sigma') = \bv(\alpha'') = \bv(\alpha')$, 
$\fv({(\sigma' \parop \sigma''\sigma')}_{|\dom(\sigma' \parop \sigma''\sigma') \setminus \{ \vect x \}}) = \dom(\sigma') \cup \dom(\sigma'') \setminus \{\vect x\} = \dom(\sigma) \cup \dom(\sigma'') \setminus \{ \vect x\}$, and $\bv(\alpha') \cap (\dom(\sigma) \cup \dom(\sigma'') \setminus \{ \vect x \}) = \emptyset$.)
Hence, by Lemmas~\ref{lem:equivpnf}, \ref{lem:red-pnf-to-std}, and~\ref{lem:redalpha-pnf-to-std}, we have
\begin{align*}
T &\equiv \pnf(T) \\
&\ltr{\alpha} \Res{\vect n''',\vect n',\vect n''}({(\sigma' \parop \sigma''\sigma')}_{|\dom(\sigma' \parop \sigma''\sigma') \setminus \{ \vect x \}} \parop P' \parop A'' \sigma')\\
&\equiv \Res{\vect u}(\Res{\vect n'}(\sigma' \parop P') \parop \Res{\vect n''}(\sigma'' \parop A''))\\
&\equiv \Res{\vect u}(B \parop C')
\end{align*}
where $C' = \Res{\vect n''}(\sigma'' \parop A'')$. Moreover,
\begin{align*}
S' &\equiv  \Res{\vect n''',\vect n,\vect n''}({(\sigma \parop \sigma''\sigma)}_{|\dom(\sigma \parop \sigma''\sigma) \setminus \{ \vect x \}} \parop A_0)\\
&\equiv  \Res{\vect n''',\vect n,\vect n''}({(\sigma \parop \sigma''\sigma)}_{|\dom(\sigma \parop \sigma''\sigma) \setminus \{ \vect x \}} \parop P \parop A_1)\\
&\equiv  \Res{\vect n''',\vect n,\vect n''}({(\sigma \parop \sigma''\sigma)}_{|\dom(\sigma \parop \sigma''\sigma) \setminus \{ \vect x \}} \parop P \parop A''\sigma)\\
&\equiv \Res{\vect u}(\Res{\vect n}(\sigma \parop P) \parop \Res{\vect n''}(\sigma'' \parop A''))\\
&\equiv \Res{\vect u}(A \parop C')
\end{align*}
We let $T' = \Res{\vect u}(B \parop C')$. We have $T \ltr{\alpha} T'$ and since $A \rel B$ and $\Res{\vect u}(\hole \parop C')$ is closing for $A$ and $B$, we have $S' \equiv \rel' T'$.
\qedhere

\end{enumerate}

\end{enumerate}
\end{proof}

\begin{restate}{Lemma}{\ref{lem:caract-barb}}
Let $A$ be a closed extended process.
We have $A \barb{a}$ if and only if $A \rightarrow^* \ltr{\nu x.\Snd{a}{x}} A'$ for some fresh variable $x$ and some $A'$.
\end{restate}%
\begin{proof}
In order to establish this claim, we argue that
$A \equiv \CTX[\Snd{a}{M}.P]$ for some evaluation context $\CTX[\hole]$ that does not bind $a$ if and only if $A \ltr{\nu x.\Snd{a}{x}} A'$ for some fresh variable $x$ and some $A'$. 

For the implication from left to right, let $x$ be a fresh variable. We derive
\begin{align*}
&\Snd{a}{M}.P \ltr{\nu x.\Snd{a}{x}} P \parop \{\subst{M}{x}\} \tag*{by \brn{Out-Var}}\\
&\CTX[\Snd{a}{M}.P] \ltr{\nu x.\Snd{a}{x}} \CTX[P \parop \{\subst{M}{x}\}]\tag*{by \brn{Par} and \brn{Scope}}\\
&A \ltr{\nu x.\Snd{a}{x}} \CTX[P \parop \{\subst{M}{x}\}]\tag*{by \brn{Struct}, since $A \equiv \CTX[\Snd{a}{M}.P]$}
\end{align*}

Conversely, if $A \ltr{\nu x.\Snd{a}{x}} A'$ for some fresh variable $x$ and some $A'$, then we show by induction on the derivation that $A \equiv \CTX[\Snd{a}{M}.P]$ for some evaluation context $\CTX[\hole]$ that does not bind $a$. In case \brn{Out-Var}, the context $\CTX$ is empty. In case \brn{Scope}, a restriction that does not bind $a$ is added to $\CTX$. In case \brn{Par}, a parallel composition is added to $\CTX$. In case \brn{Struct}, the context $\CTX$ is unchanged.
\end{proof}

\subsection{Observational Equivalence Implies Labelled Bisimilarity}\label{app:twodirection}

Finally, the goal of this section is to establish the lemmas needed in
the outline of the argument that observational equivalence implies
labelled bisimilarity in Section~\ref{sec:bigpf}. The section also
contains a corollary, namely that observational equivalence and static
equivalence coincide on frames.

\begin{lemma}\label{lem:p-notin-fnA-preserved}
Let $P$ be a plain process. The existence of $P'$ such that $\Sigma \vdash P = P'$ and $p \notin \fn(P')$ is preserved by structural equivalence ($\equivP$) and reduction ($\redP$) of $P$.

Let $A$ be a normal process.
The existence of $A'$ such that $\Sigma \vdash A = A'$ and $p \notin \fn(A')$ is preserved by structural equivalence ($\equivpnf$) and reduction ($\redpnf$) of $A$.
\end{lemma}
\begin{proof}
\emph{Property~1:} Suppose that $P \equivP Q$, $\Sigma \vdash P = P'$, and $p \notin \fn(P')$.
We show that there exists $Q'$ such that $\Sigma \vdash Q = Q'$ and $p \notin \fn(Q')$, 
by induction on the derivation of $P \equivP Q$. We consider as base cases
the application of each rule under an evaluation context, in the two directions,
and use induction only for transitivity.
\begin{itemize}
\item Case $\brn{Rewrite}'$, under an evaluation context $\CTX$. We have $\CTX[P_1\{\subst{M}{x}\}] \equivP \CTX[P_1\{\subst{N}{x}\}]$, $\Sigma \vdash M = N$, $\Sigma \vdash \CTX[P_1\{\subst{M}{x}\}] = P'$, and $p \notin \fn(P')$. Since $\Sigma \vdash \CTX[P_1\{\subst{N}{x}\}] = \CTX[P_1\{\subst{M}{x}\}] = P'$, we have the result with $Q' = Q$.
\item Case $\brn{Par-C}'$, under an evaluation context $\CTX$. We have $\CTX[P_1 \parop Q_1] \equivP \CTX[Q_1 \parop P_1]$, $\Sigma \vdash \CTX[P_1 \parop Q_1] = P'$, and $p \notin \fn(P')$. Since $\Sigma \vdash \CTX[P_1 \parop Q_1] = P'$, we have $P' = \CTX'[P'_1 \parop Q'_1]$ with $\Sigma \vdash \CTX = \CTX'$, $\Sigma \vdash P_1 = P'_1$, and $\Sigma \vdash Q_1 = Q'_1$. Let $Q' = \CTX'[Q'_1 \parop P'_1]$. We have $\Sigma \vdash \CTX[Q_1 \parop P_1] = Q'$ and $p \notin \fn(Q') = \fn(P')$.
\item All other base cases are handled similarly to case $\brn{Par-C}'$.
\item The case of transitivity follows by applying the induction hypothesis twice.
\end{itemize}
\emph{Property~2:} Suppose that $P \redP Q$, $\Sigma \vdash P = P'$, and $p \notin \fn(P')$.
We show that there exists $Q'$ such that $\Sigma \vdash Q = Q'$ and $p \notin \fn(Q')$,
by induction on the derivation of $P \redP Q$. Again, we consider as base cases
the application of each rule under an evaluation context,
and use induction only for the application of $\equivP$.
\begin{itemize}
\item Case $\brn{Comm}'$, under an evaluation context $\CTX$. We have $\CTX[\Snd{N}{M}.P_1 \parop \Rcv{N}{x}.Q_1] \redP \CTX[P_1 \parop Q_1\{\subst{M}{x}\}]$, $\Sigma \vdash \CTX[\Snd{N}{M}.P_1 \parop \Rcv{N}{x}.Q_1] = P'$, and $p \notin \fn(P')$. Since $\Sigma \vdash \CTX[\Snd{N}{M}.P_1 \parop \Rcv{N}{x}.Q_1] = P'$, we have $P' = \CTX'[\Snd{N'}{M'}.P'_1 \parop \Rcv{N''}{x}.Q'_1]$ with $\Sigma \vdash \CTX = \CTX'$, $\Sigma \vdash M = M'$, $\Sigma \vdash P_1 = P'_1$, and $\Sigma \vdash Q_1 = Q'_1$. Let $Q' = \CTX'[P'_1 \parop Q'_1\{\subst{M'}{x}\}]$.
We have $\Sigma \vdash \CTX[P_1 \parop Q_1\{\subst{M}{x}\}] = Q'$ and $p \notin \fn(Q')$ since $\fn(Q') \subseteq \fn(P')$.

\item Case $\brn{Then}'$, under an evaluation context $\CTX$. We have $\CTX[\IfThenElse{M}{M}{P_1}{Q_1}] \redP \CTX[P_1]$,  $\Sigma \vdash \CTX[\IfThenElse{M}{M}{P_1}{Q_1}] = P'$, and $p \notin \fn(P')$. Since $\Sigma \vdash \CTX[\IfThenElse{M}{M}{P_1}{Q_1}] = P'$, we have $P' = \CTX'[\IfThenElse{M'}{M''}{P'_1}{Q'_1}]$ with $\Sigma \vdash \CTX = \CTX'$ and $\Sigma \vdash P_1 = P'_1$. Let $Q' = \CTX'[P'_1]$.
We have $\Sigma \vdash \CTX[P_1] = Q'$ and $p \notin \fn(Q')$ since $\fn(Q') \subseteq \fn(P')$.

\item Case $\brn{Else}'$ is handled similarly to case $\brn{Then}'$.
\item In case we additionally apply $\equivP$, we conclude using Property~1 and the induction hypothesis.
\end{itemize}
\emph{Property~3:} Suppose that $A \equivpnf B$, $\Sigma \vdash A = A'$, and $p \notin \fn(A')$.
We show that there exists $B'$ such that $\Sigma \vdash B = B'$ and $p \notin \fn(B')$, 
by induction on the derivation of $A \equivpnf B$.
\begin{itemize}
\item Case $\brn{Plain}''$ follows from Property~1.
\item Case $\brn{New-Par}''$. We have $\Res{\vect n}(\sigma \parop \Res{n'}P) \equivpnf \Res{\vect n, n'}(\sigma \parop P)$
with $n' \notin \fn(\sigma)$,  $\Sigma \vdash \Res{\vect n}(\sigma \parop \Res{n'}P) = A'$, and $p \notin \fn(A')$.
If $p \in \{ \vect n, n'\}$, then we have the result with $B' = \Res{\vect n, n'}(\sigma \parop P)$.
Otherwise, since $\Sigma \vdash \Res{\vect n}(\sigma \parop \Res{n'}P) = A'$, we have $A' = \Res{\vect n}(\sigma' \parop \Res{n'}P')$
with $\Sigma \vdash \sigma = \sigma'$ and $\Sigma \vdash P = P'$.
Let $B' = \Res{\vect n, n'}(\sigma' \parop P')$. We have $\Sigma \vdash \Res{\vect n, n'}(\sigma \parop P) = B'$ and $p \notin \fn(B')$ since $\fn(B') \subseteq \fn(A')$.
\item Case $\brn{New-Par}''$ reversed. We have $\Res{\vect n, n'}(\sigma \parop P) \equivpnf \Res{\vect n}(\sigma \parop \Res{n'}P)$
with $n' \notin \fn(\sigma)$,  $\Sigma \vdash \Res{\vect n,n'}(\sigma \parop P) = A'$, and $p \notin \fn(A')$.
If $p \in \{ \vect n\}$, then we have the result with $B' = \Res{\vect n}(\sigma \parop \Res{n'}P)$.
If $p = n'$, then we also have the result with $B' = \Res{\vect n}(\sigma \parop \Res{n'}P)$ because $n' \notin \fn(\sigma)$.
Otherwise, since $\Sigma \vdash \Res{\vect n,n'}(\sigma \parop P) = A'$, we have $A' = \Res{\vect n,n'}(\sigma' \parop P')$
with $\Sigma \vdash \sigma = \sigma'$ and $\Sigma \vdash P = P'$.
Let $B' = \Res{\vect n}(\sigma' \parop \Res{n'}P')$. We have $\Sigma \vdash \Res{\vect n}(\sigma \parop \Res{n'}P) = B'$ and $p \notin \fn(B')$.
\item Case $\brn{New-C}''$ is handled similarly to case $\brn{New-Par}''$.
\item Case $\brn{Rewrite}''$. We have $\Res{\vect n}(\sigma \parop P) \equivpnf \Res{\vect n}(\sigma' \parop P)$ with $\Sigma \vdash \sigma = \sigma'$, $\Sigma \vdash \Res{\vect n}(\sigma \parop P) = A'$, and $p \notin \fn(A')$. Since
$\Sigma \vdash \Res{\vect n}(\sigma' \parop P) = \Res{\vect n}(\sigma \parop P) = A'$, we have the result with $B' = A'$.
\item The case of transitivity follows by applying the induction hypothesis twice.
\end{itemize}
\emph{Property~4:} Suppose that $A \redpnf B$, $\Sigma \vdash A = A'$, and $p \notin \fn(A')$.
We show that there exists $B'$ such that $\Sigma \vdash B = B'$ and $p \notin \fn(B')$.
Suppose $\Res{\vect n}(\sigma \parop P) \redpnf \Res{\vect n}(\sigma \parop Q)$ with $P \redP Q$, 
$\Sigma \vdash \Res{\vect n}(\sigma \parop P) = A'$, and $p \notin \fn(A')$. 
If $p \in \{\vect n\}$, then we have the result with $B' = \Res{\vect n}(\sigma \parop Q)$.
Otherwise, since
$\Sigma \vdash \Res{\vect n}(\sigma \parop P) = A'$, we have $A' = \Res{\vect n}(\sigma' \parop P')$ with 
$\Sigma \vdash \sigma = \sigma'$ and $\Sigma \vdash P = P'$. Since $p \notin \fn(A')$ and $p \notin \{\vect n\}$,
$p \notin \fn(\sigma') \cup \fn(P')$. By Property~2, there exists $Q'$ such that $\Sigma \vdash Q = Q'$ and $p \notin \fn(Q')$. Let $B' = \Res{\vect n}(\sigma' \parop Q')$. We have $\Sigma \vdash \Res{\vect n}(\sigma \parop Q) = B'$ and $p \notin \fn(B')$.
In case we additionally apply $\equivpnf$, we conclude using Property~3 and the induction hypothesis.
\end{proof}

\begin{lemma}\label{lem:barb-obvious}
If $p \notin \fn(A)$, then $A \not\barb{p}$.
\end{lemma}
\begin{proof}
In order to obtain a contradiction, suppose that $A \barb{p}$, that is, that
$A \rightarrow^* \equiv \CTX[\Snd{p}{M}.P]$ for some $M$, $P$, and evaluation context $\CTX[\hole]$ that does not bind $p$.
Hence, $\pnf(A) \redpnf^* \equivpnf \CTX[\Snd{p}{M}.P]$ for some $M$, $P$, and evaluation context $\CTX[\hole]$ that does not bind $p$.
Let $A_1 = \pnf(A)$. We have $p \notin \fn(A_1)$.
By Lemma~\ref{lem:p-notin-fnA-preserved}, the existence of $A'$ such that $\Sigma \vdash A_1 = A'$ and $p \notin \fn(A')$ is preserved by structural equivalence and reduction of $A_1$, so there exists $A'$ such that $\Sigma \vdash \CTX[\Snd{p}{M}.P] = A'$ and $p \notin \fn(A')$. Hence, there exists $N$ such that $\Sigma \vdash p = N$ and $p \notin \fn(N)$. 
Since the equational theory is preserved by substitution of terms for names, for all $N'$,  $\Sigma \vdash p\{\subst{N'}{p}\} = N\{\subst{N'}{p}\}$, that is $\Sigma \vdash N' = N$, which contradicts the assumption that the equational theory is non-trivial.
\end{proof}

\begin{lemma}\label{lem:red-no-fn}
If $p \notin \fn(P)$ and $P \redP^* \ltrP{\Res{x}\Snd{N}{x}} A$ or 
$P \redP^* \ltrP{\Rcv{N}{M}} A$, then $\Sigma \vdash N \neq p$.
\end{lemma}
\begin{proof}
The proof uses ideas similar to the proof of Lemma~\ref{lem:barb-obvious}.
By Lemma~\ref{lem:caract-ltrP}, 
$P \redP^* \equivP \Res{\vect n}(\Snd{N}{M}.P_1 \parop P_2)$
for some $\vect n$, $M$, $P_1$, $P_2$ with 
$\{ \vect n \} \cap \fn(N) = \emptyset$, or 
$P \redP^* \equivP \Res{\vect n}(\Rcv{N}{x}.P_1 \parop P_2)$ for some
for some $\vect n$, $x$, $P_1$, $P_2$ with 
$\{ \vect n \} \cap \fn(N) = \emptyset$.
By Lemma~\ref{lem:p-notin-fnA-preserved}, the existence of $P'$ such that $\Sigma \vdash P = P'$
and $p \notin \fn(P')$ is preserved by structural equivalence and reduction of $P$, so there exists
$N'$ such that $\Sigma \vdash N = N'$ and $p \notin \fn(N')$.
If we had $\Sigma \vdash N = p$, then we would have $\Sigma \vdash p = N'$ and $p \notin \fn(N')$, which yields a contradiction as in the proof of Lemma~\ref{lem:barb-obvious}. So $\Sigma \vdash N \neq p$.
\end{proof}

\begin{lemma}\label{lem:enveq-gen}
${\bicong} \subseteq {\enveq}$.
\end{lemma}
\begin{proof}
If $A$ and $B$ are observationally equivalent, then 
$A \parop C$ and $B \parop C$ have the same barbs for every $C$ 
with $\fv(C)\subseteq\dom(A)$.
In particular, 
$A \parop C$ and $B \parop C$ have the same barb $\barb{a}$ for every $C$
of the special form $\IfThen{M}{N}{\Snd{a}{s}}$, where $a$ does not occur 
in $A$ or $B$ and $\fv(C)\subseteq\dom(A)$, that is, 
$\fv(M) \cup \fv(N) \subseteq \dom(A)$. 
We obtain that $A$ and $B$ are statically equivalent, using the following
property:
assuming that $A$ is closed, $\fv(M) \cup \fv(N) \subseteq \dom(A)$, and $a$ does not occur in $A$, we have $(M=N)\frameof{A}$ if and only if
$A \parop \IfThen{M}{N}{\Snd{a}{s}} \barb{a}$.
We show this property below.

Let $\pnf(A) = \Res{\vect n}(\sigma \parop P)$.
We rename $\vect n$ so that $\{ \vect n \} \cap (\fn(M) \cup \fn(N) \cup \{a\}) = \emptyset$. 
If $(M=N)\frameof{A}$, then $M\sigma = N\sigma$, so
$A \parop \IfThen{M}{N}{\Snd{a}{s}} \equiv \Res{\vect n}(\sigma \parop P \parop 
\IfThen{M\sigma}{N\sigma}{\Snd{a}{s}}) \rightarrow \Res{\vect n}(\sigma \parop P \parop \Snd{a}{s})$, so 
we conclude that $A \parop \IfThen{M}{N}{\Snd{a}{s}} \barb{a}$.
Conversely, in order to obtain a contradiction, suppose that $(M \neq N)\frameof{A}$ and $A \parop \IfThen{M}{N}{\Snd{a}{s}} \barb{a}$.
Lemma~\ref{lem:caract-barb} implies that
$A \parop \IfThen{M}{N}{\Snd{a}{s}} \rightarrow^* \ltr{\Res{x}\Snd{a}{x}} A'$ for some fresh variable $x$ and some $A'$.
So $\pnf(A \parop \IfThen{M}{N}{\Snd{a}{s}}) = \Res{\vect n}(\sigma \parop P \parop 
\IfThen{M\sigma}{N\sigma}{\Snd{a}{s}}) \redpnf^* \ltrpnf{\Res{x}\Snd{a}{x}} A'$
by Lemmas~\ref{lem:red-std-to-pnf} and~\ref{lem:redalpha-std-to-pnf}.
Then $P \parop 
\IfThen{M\sigma}{N\sigma}{\Snd{a}{s}} \redP^* \ltrP{\Res{x}\Snd{a}{x}} A''$,
$A' \equiv \Res{\vect n}(\sigma \parop A'')$, 
and $x \notin \dom(\sigma)$
for some $A''$ by Lemmas~\ref{lem:decomp-redpnf-closed}
and~\ref{lem:decomp-ltrpnf}.
We have $a \notin \fn(\pnf(A))$, so $a \notin \fn(P)$.
We show by induction on the length of the trace, that it is impossible to have
$P \parop \IfThen{M\sigma}{N\sigma}{\Snd{a}{s}} \redP^* \ltrP{\Res{x}\Snd{a}{x}} A''$.
\begin{itemize}

\item If this trace contains a single step, then $P \parop \IfThen{M\sigma}{N\sigma}{\Snd{a}{s}} \ltrP{\Res{x}\Snd{a}{x}} A''$, so by Lemma~\ref{lem:decomp-ltrP}, 
$P \ltrP{\Res{x}\Snd{a}{x}}$, which yields a contradiction by Lemma~\ref{lem:red-no-fn}.

\item If this trace contains several steps, the first step is an internal reduction, so by Lemma~\ref{lem:decomp-redP-closed}, either $P$ reduces, and we conclude by induction hypothesis, or $\IfThen{M\sigma}{N\sigma}{\Snd{a}{s}}$ reduces to $\nil$ and $P \parop \nil \redP^* \ltrP{\Res{x}\Snd{a}{x}} A''$, which yields a contradiction by Lemma~\ref{lem:red-no-fn}.
\qedhere
\end{itemize}
\end{proof}

\begin{lemma}\label{lem:rename-free-names}
Let $\vect n$ be pairwise distinct names.
Let $\vect n'$ be pairwise distinct names that do not occur in $P$ nor in $P'$.

If $P\equivP P'$ and $\Sigma \vdash P = P\{\subst{\vect n'}{\vect n}\}$,
then $P\{\subst{\vect n'}{\vect n}\} \equivP P'\{\subst{\vect n'}{\vect n}\}$
and $\Sigma \vdash P' = P'\{\subst{\vect n'}{\vect n}\}$.

If $P \redP P'$ and $\Sigma \vdash P = P\{\subst{\vect n'}{\vect n}\}$,
then $P\{\subst{\vect n'}{\vect n}\} \redP P'\{\subst{\vect n'}{\vect n}\}$
and $\Sigma \vdash P' = P'\{\subst{\vect n'}{\vect n}\}$.
\end{lemma}
\begin{proof}
By induction on the derivations of $P\equivP P'$ and $P \redP P'$, respectively.
\end{proof}

\begin{restate}{Lemma}{\ref{lem:caract-RcvNM}}
Let $A$ be a closed extended process.
Let $N$ and $M$ be terms such that $\fv(\Snd{N}{M}) \subseteq \dom(A)$. Let $p$ be a name that does not occur in $A$, $M$, and $N$.
\begin{enumerate}
\item If $A \ltr{\Rcv{N}{M}} A'$ and $p$ does not occur in $A'$, 
then $A \parop T^p_{\Rcv N M} \rightarrow \rightarrow A'$ and $A' \not\barb{p}$.
\item If $A \parop T^p_{\Rcv N M} \rightarrow^* A'$ and
$A' \not\barb{p}$, then
$A \rightarrow^* \ltr{\Rcv{N}{M}} \rightarrow^* A'$.
\end{enumerate}
\end{restate}%
\begin{proof}
  \emph{Property~1:}
Let $\pnf(A) = \Res{\vect n}(\sigma \parop P)$. We rename $\vect n$ so that these names do not occur in $N$, $M$, $p$.
By Lemma~\ref{lem:redalpha-std-to-pnf}, $\pnf(A) \ltrpnf{\Rcv{N}{M}} A'$.
By Lemma~\ref{lem:decomp-ltrpnf},  $P \ltrP{\Rcv{N\sigma}{M\sigma}} A''$ and $A' \equiv \Res{\vect n}(\sigma \parop A'')$ for some $A'$.
By Lemma~\ref{lem:caract-ltrP}, $P \equivP \Res{\vect n'}(\Rcv{N\sigma}{x'}.P_1 \parop P_2)$, $A'' \equiv \Res{\vect n'}(P_1\{\subst{M\sigma}{x'}\} \parop P_2)$,
$\{\vect n'\} \cap \fn(\Rcv{N\sigma}{M\sigma}) = \emptyset$,
for some $\vect n'$, $P_1$, $P_2$, $x'$.
We rename $\vect n'$ so that $p \notin \{\vect n'\}$.
Hence, by Lemmas~\ref{lem:equivpnf} and~\ref{lem:struct-pnf-to-std}, 
{\allowdisplaybreaks\begin{align*}
A \parop \Snd{p}{p} \parop \Snd{N}{M}.\Rcv{p}{x}
&\equiv \pnf(A) \parop \Snd{p}{p} \parop \Snd{N}{M}.\Rcv{p}{x}\\
&\equiv \Res{\vect n}(\sigma \parop \Res{\vect n'}(\Rcv{N\sigma}{x'}.P_1 \parop P_2)) \parop \Snd{p}{p}\parop  \Snd{N}{M}.\Rcv{p}{x}\\
&\equiv \Res{\vect n}(\sigma \parop \Res{\vect n'}(\Rcv{N\sigma}{x'}.P_1 \parop P_2 \parop \Snd{p}{p}\parop  \Snd{N\sigma}{M\sigma}.\Rcv{p}{x}))\\
&\rightarrow \Res{\vect n}(\sigma \parop \Res{\vect n'}(P_1\{\subst{M\sigma}{x'}\} \parop P_2  \parop \Snd{p}{p}\parop  \Rcv{p}{x}))\\
&\rightarrow \Res{\vect n}(\sigma \parop \Res{\vect n'}(P_1\{\subst{M\sigma}{x'}\} \parop P_2))\\
&\equiv \Res{\vect n}(\sigma \parop A'')\\
&\equiv A'
\end{align*}}%
Since $p \notin \fn(A')$, we have $A' \not\barb{p}$ by Lemma~\ref{lem:barb-obvious}.

\medskip
\noindent \emph{Property~2:} Let $\pnf(A) = \Res{\vect n}(\sigma \parop P)$. 
By Lemma~\ref{lem:pnf-closed}, $\pnf(A)$ is closed.
We rename $\vect n$ so that these names do not occur in $N$, $M$, $p$.
Then $\pnf(A \parop \Snd{p}{p} \parop \Snd{N}{M}.\Rcv{p}{x}) = \Res{\vect n}(\sigma \parop P \parop \Snd{p}{p} \parop \Snd{N\sigma}{M\sigma}.\Rcv{p}{x})$.
By Lemma~\ref{lem:red-std-to-pnf}, $\pnf(A \parop \Snd{p}{p} \parop \Snd{N}{M}.\Rcv{p}{x}) \redpnf^* \pnf(A')$.
By Lemma~\ref{lem:decomp-redpnf-closed} applied several times,
$P \parop \Snd{p}{p} \parop \Snd{N\sigma}{M\sigma}.\Rcv{p}{x} \redP^* P'$ and $\pnf(A') \equiv \Res{\vect n}(\sigma \parop P')$ for some closed process $P'$.
Since $A' \not\barb{p}$, we have $P' \not\barb{p}$. (If we had $P' \barb{p}$, we would immediately obtain $A' \barb{p}$ by definition of $\barb{p}$.)

We prove that, if $P$ is a closed process, $P \parop \Snd{p}{p} \parop \Rcv{p}{x} \redP^* P'$,
$P' \not\barb{p}$, and $p \notin \fn(P)$, then $P \equivP\redP^* P'$,
by induction on the length of the trace.
Since $P \parop \Snd{p}{p} \parop \Rcv{p}{x} \barb{p}$,
the trace $P \parop \Snd{p}{p} \parop \Rcv{p}{x} \redP^* P'$ has at 
least one step: $P \parop \Snd{p}{p} \parop \Rcv{p}{x} \redP P_1 \redP^* P'$ .
By Lemmas~\ref{lem:decomp-redP-closed}, \ref{lem:decomp-ltrP}, and~\ref{lem:red-no-fn}, the only cases that can happen in the first step are:
\begin{itemize}
\item $P \redP P''$ and $P'' \parop \Snd{p}{p} \parop \Rcv{p}{x} \equiv P_1 \redP^* P'$ for some closed process $P''$. As above this trace has at least one step, so $P'' \parop \Snd{p}{p} \parop \Rcv{p}{x} \redP^* P'$. By Lemma~\ref{lem:rename-free-names}, we rename $p$ inside $P''$ so that $p \notin \fn(P'')$, and we obtain the desired result by induction hypothesis.
\item $\Snd{p}{p} \ltrP{\Res{y}\Snd{N}{y}} A_1 \equiv \{\subst{p}{y}\}$,
$\Rcv{p}{x} \ltrP{\Rcv{N}{y}} A_2 \equiv \nil$, 
$\Sigma \vdash N = p$,
$P_1 \equiv P \parop \Res{y}(A_1 \parop A_2) \equiv P \parop \Res{y}(\{\subst{p}{y}\} \parop \nil) \equiv P$
so $P \parop \Snd{p}{p} \parop \Rcv{p}{x} \redP P \equivP P_1 \redP^* P'$,
so we obtain $P \equivP\redP^* P'$ as desired.
\end{itemize}

Next, we prove that, if $P \parop \Snd{p}{p} \parop \Snd{N\sigma}{M\sigma}.\Rcv{p}{x}$ is a closed process, $P \parop \Snd{p}{p} \parop \Snd{N\sigma}{M\sigma}.\Rcv{p}{x} \redP^* P'$, $P' \not\barb{p}$, and $p \notin \fn(P) \cup \fn(N\sigma) \cup \fn(M\sigma)$, then $P \redP^* \ltrP{\Rcv{N\sigma}{M\sigma}} \redP^* P'$,
by induction on the length of the trace.
Since $P \parop \Snd{p}{p} \parop \Snd{N\sigma}{M\sigma}.\Rcv{p}{x} \barb{p}$,
the trace $P \parop \Snd{p}{p} \parop \Snd{N\sigma}{M\sigma}.\Rcv{p}{x} \redP^* P'$ has at 
least one step: $P \parop \Snd{p}{p} \parop \Snd{N\sigma}{M\sigma}.\Rcv{p}{x} \redP P_1 \redP^* P'$.
By Lemmas~\ref{lem:decomp-redP-closed}, \ref{lem:decomp-ltrP}, and~\ref{lem:red-no-fn}, the only cases that can happen in the first step are:
\begin{itemize}
\item $P \redP P''$ and $P'' \parop \Snd{p}{p} \parop \Snd{N\sigma}{M\sigma}.\Rcv{p}{x} \equiv P_1 \redP^* P'$ for some closed process $P''$. As above this trace has at least one step, so $P'' \parop \Snd{p}{p} \parop \Snd{N\sigma}{M\sigma}.\Rcv{p}{x} \redP^* P'$. By Lemma~\ref{lem:rename-free-names}, we rename $p$ inside $P''$ so that $p \notin \fn(P'')$, and we obtain the desired result by induction hypothesis.
\item $P \ltrP{\Rcv{N'}{y}} B$, $\Snd{N\sigma}{M\sigma}.\Rcv{p}{x} \ltrP{\Res{y}\Snd{N'}{y}} B'$,
and $P_1 \equiv \Res{y}(B \parop \Snd{p}{p} \parop B')$.
By Lemma~\ref{lem:caract-ltrP},
$P \equivP \Res{\vect n'}(\Rcv{N'}{z}.P_2 \parop P_3)$,
$B \equiv \Res{\vect n'}(P_2\{\subst{y}{z}\} \parop P_3)$, and
$\{\vect n'\} \cap \fn(N') = \emptyset$ for some $\vect n'$, $z$, $P_2$, and $P_3$.
By Lemma~\ref{lem:decomp-ltrP},
$\Sigma \vdash N\sigma = N'$, $y \notin \fv(\Snd{N\sigma}{M\sigma}.\Rcv{p}{x})$, and $B' \equiv \Rcv{p}{x} \parop \{\subst{M\sigma}{y}\}$.
We rename $\vect n'$ so that these names do not appear in $M\sigma$
and are distinct from $p$.
By Lemma~\ref{lem:rename-free-names}, we rename $p$ inside 
$\Res{\vect n'}(\Rcv{N'}{z}.P_2 \parop P_3)$ so that $p \notin \fn(\Res{\vect n'}(\Rcv{N'}{z}.P_2 \parop P_3))$, so $p \notin \fn(P_2) \cup \fn(P_3)$.
Hence 
$P_1 \equiv \Res{y}(\Res{\vect n'}(P_2\{\subst{y}{z}\} \parop P_3) \parop \Snd{p}{p} \parop \{\subst{M\sigma}{y}\} \parop \Rcv{p}{x}) 
\equiv \Res{\vect n'}(P_2\{\subst{M\sigma}{z}\} \parop P_3) \parop \Snd{p}{p} \parop \Rcv{p}{x}$.
We have
$P \equivP \Res{\vect n'}(\Rcv{N'}{z}.P_2 \parop P_3)
\ltrP{\Rcv{N\sigma}{M\sigma}} \Res{\vect n'}(P_2\{\subst{M\sigma}{z}\} \parop P_3)
$. Let $P_4 = \Res{\vect n'}(P_2\{\subst{M\sigma}{z}\} \parop P_3)$. 
We have then $P \ltrP{\Rcv{N\sigma}{M\sigma}} P_4$ and $P_4 \parop \Snd{p}{p} \parop \Rcv{p}{x} \equiv P_1 \redP^* P'$. 
By Lemma~\ref{lem:closing}\eqref{closing:ltrP}, we transform $P_4$ into a closed process that satisfies the same properties. Since $P' \not\barb{p}$, this trace has at least one step, so $P_4 \parop \Snd{p}{p} \parop \Rcv{p}{x} \redP^* P'$. Since $p \notin \fn(P_4)$,
by the property shown above, $P_4 \equivP\redP^* P'$, so
$P \ltrP{\Rcv{N\sigma}{M\sigma}} \redP^* P'$.

\end{itemize}

To sum up, we have $A \equiv \pnf(A) = \Res{\vect n}(\sigma \parop P)$, 
$P \redP^* P_5 \ltrP{\Rcv{N\sigma}{M\sigma}} P_6 \redP^* P'$, and $A'
\equiv \pnf(A') \equiv \Res{\vect n}(\sigma \parop P')$.
So $\Res{\vect n}(\sigma \parop P) \redpnf^* \Res{\vect n}(\sigma \parop P_5) \ltrpnf{\Rcv{N}{M}} \Res{\vect n}(\sigma \parop P_6) \redpnf^* \Res{\vect n}(\sigma \parop P')$. Hence by Lemmas~\ref{lem:red-pnf-to-std} 
and~\ref{lem:redalpha-pnf-to-std},
$A \rightarrow^* \ltr{\Rcv{N}{M}} \rightarrow^* A'$.
\end{proof}

\begin{restate}{Lemma}{\ref{lem:caract-SndN}}
Let $A$ be a closed extended process.
Let $N$ be a term such that $\fv(N) \subseteq \dom(A)$. Let $p$ and $q$ be names that do not occur in $A$ and $N$.
\begin{enumerate}
\item If $A \ltr{\Res{x}\Snd{N}{x}} A'$ and $p$ and $q$ do not occur in $A'$, 
then $A \parop T^{p,q}_{\Res{x}\Snd N x}  \rightarrow\rightarrow \Res{x}(A' \parop \Snd{q}{x})$,
$\Res{x}(A' \parop \Snd{q}{x}) \not\barb{p}$, and $x \notin \dom(A)$.
\item Let $x$ be a variable such that $x \notin \dom(A)$.
If $A \parop T^{p,q}_{\Res{x}\Snd N x} \rightarrow^* A''$ and $A'' \not\barb{p}$, then
$A \rightarrow^*\allowbreak \ltr{\Res{x}\Snd{N}{x}}\allowbreak \rightarrow^* A'$
and $A'' \equiv \Res{x}(A' \parop \Snd{q}{x})$ for some $A'$.
\end{enumerate}
\end{restate}%
\begin{proof}
\emph{Property~1:}
Let $\pnf(A) = \Res{\vect n}(\sigma \parop P)$. 
By Lemma~\ref{lem:pnf-closed}, $\pnf(A)$ is closed.
We rename $\vect n$ so that these names do not occur in $N$, $p$, and $q$.
By Lemma~\ref{lem:redalpha-std-to-pnf}, $\pnf(A) \ltrpnf{\Res{x}\Snd{N}{x}} A'$.
By Lemma~\ref{lem:decomp-ltrpnf},  $P \ltrP{\Res{x}\Snd{N\sigma}{x}} A''$, $A' \equiv \Res{\vect n}(\sigma \parop A'')$, and $x \notin \dom(\sigma)$ for some $A''$.
By Lemma~\ref{lem:caract-ltrP}, $P \equivP \Res{\vect n'}(\Snd{N\sigma}{M}.P_1 \parop P_2)$, $A'' \equiv \Res{\vect n'}(P_1 \parop \{\subst{M}{x}\} \parop P_2)$,
$\{\vect n'\} \cap \fn(N\sigma) = \emptyset$,
and $x \notin \fv(\Snd{N\sigma}{M}.P_1 \parop P_2))$
for some $\vect n'$, $P_1$, $P_2$, $M$.
We rename $\vect n'$ so that $p,q \notin \{\vect n'\}$ and $y$ so that $y \notin \fv(M)$.
Hence, by Lemmas~\ref{lem:equivpnf} and~\ref{lem:struct-pnf-to-std}, 
{\allowdisplaybreaks\begin{align*}
A \parop \Snd{p}{p} \parop \Rcv{N}{x}.\Rcv{p}{y}.\Snd{q}{x}
&\equiv \pnf(A) \parop \Snd{p}{p} \parop \Rcv{N}{x}.\Rcv{p}{y}.\Snd{q}{x}\\
&\equiv \Res{\vect n}(\sigma \parop \Res{\vect n'}(\Snd{N\sigma}{M}.P_1 \parop P_2)) \parop \Snd{p}{p}\parop  \Rcv{N}{x}.\Rcv{p}{y}.\Snd{q}{x}\\
&\equiv \Res{\vect n}(\sigma \parop \Res{\vect n'}(\Snd{N\sigma}{M}.P_1 \parop P_2 \parop \Snd{p}{p}\parop  \Rcv{N}{x}.\Rcv{p}{y}.\Snd{q}{x}))\\
&\rightarrow \Res{\vect n}(\sigma \parop \Res{\vect n'}(P_1 \parop P_2  \parop \Snd{p}{p}\parop  \Rcv{p}{y}.\Snd{q}{M}))\\
&\rightarrow \Res{\vect n}(\sigma \parop \Res{\vect n'}(P_1 \parop P_2 \parop \Snd{q}{M}))\\
&\equiv\Res{x}\Res{\vect n}(\sigma \parop \Res{\vect n'}(P_1 \parop \{\subst{M}{x}\} \parop P_2 \parop \Snd{q}{x}))\\
&\equiv \Res{x}(\Res{\vect n}(\sigma \parop A'') \parop \Snd{q}{x})\\
&\equiv \Res{x}(A' \parop \Snd{q}{x})
\end{align*}}%
Since $p \notin \fn(\Res{x}(A' \parop \Snd{q}{x}))$, we have $\Res{x}(A' \parop \Snd{q}{x}) \not\barb{p}$ by Lemma~\ref{lem:barb-obvious}.

\medskip
\noindent \emph{Property~2:}
Let $\pnf(A) = \Res{\vect n}(\sigma \parop P)$. By Lemma~\ref{lem:pnf-closed}, $\pnf(A)$ is closed. We rename $\vect n$ so that these names do not occur in $N$, $p$, $q$.
Then $\pnf(A \parop \Snd{p}{p} \parop \Rcv{N}{x}.\Rcv{p}{y}.\Snd{q}{x}) = \Res{\vect n}(\sigma \parop P \parop \Snd{p}{p} \parop \Rcv{N\sigma}{x}.\Rcv{p}{y}.\Snd{q}{x})$.
By Lemma~\ref{lem:red-std-to-pnf}, $\pnf(A \parop \Snd{p}{p} \parop \Rcv{N}{x}.\Rcv{p}{y}.\Snd{q}{x}) \redpnf^* \pnf(A'')$.
By Lemma~\ref{lem:decomp-redpnf-closed} applied several times,
$P \parop \Snd{p}{p} \parop \Rcv{N\sigma}{x}.\Rcv{p}{y}.\Snd{q}{x} \redP^* P''$ and $\pnf(A'') \equiv \Res{\vect n}(\sigma \parop P'')$ for some closed process $P''$. Since $A'' \not\barb{p}$, we have $P'' \not\barb{p}$.

We prove that, if $P_2$ and $M'$ are closed, $P_2 \parop \Snd{q}{M'} \equivP \redP^* P''$, and $q \notin \fn(P_2)$, then $P'' \equiv P_3 \parop \Snd{q}{M'}$ and $P_2 \redP^* P_3$ for some closed process $P_3$, by induction on the length of the trace $P_2 \parop \Snd{q}{M'} \equivP \redP^* P''$. If this trace has zero reduction steps, then the result holds obviously with $P_3 = P_2$. If this trace has at least one reduction step, then $P_2 \parop \Snd{q}{M'} \redP P_4 \redP^* P''$, so by Lemmas~\ref{lem:decomp-redP-closed} and~\ref{lem:red-no-fn}, the only case that can happen is that $P_2 \redP P_2'$ and $P_2' \parop \Snd{q}{M'} \equiv P_4 \redP^* P''$ for some closed process $P'_2$. By Lemma~\ref{lem:rename-free-names}, we rename $q$ inside $P_2'$ so that $q \notin \fn(P'_2)$, and we obtain the desired result by induction hypothesis.

Next, we prove that, if $P_1$ and $M'$ are closed, $P_1 \parop \Snd{p}{p} \parop \Rcv{p}{y}.\Snd{q}{M'} \redP^* P''$, $P'' \not\barb{p}$, and $p, q\notin\fn(P_1)$, then $P'' \equiv P_3 \parop \Snd{q}{M'}$ and $P_1 \redP^* P_3$ for some closed process $P_3$, by induction on the length of the trace.
Since $P_1 \parop \Snd{p}{p} \parop \Rcv{p}{y}.\Snd{q}{M'} \barb{p}$,
the trace $P_1 \parop \Snd{p}{p} \parop \Rcv{p}{y}.\Snd{q}{M'} \redP^* P''$ has 
at least one step: $P_1 \parop \Snd{p}{p} \parop \Rcv{p}{y}.\Snd{q}{M'} \redP P_1' \redP^* P''$.
By Lemmas~\ref{lem:decomp-redP-closed}, \ref{lem:decomp-ltrP}, and~\ref{lem:red-no-fn}, the only cases that can happen in the first step are:
\begin{itemize}

\item $P_1 \redP P_1''$ and $P_1''  \parop \Snd{p}{p} \parop \Rcv{p}{y}.\Snd{q}{M'} \equiv P_1' \redP^* P''$ for some closed process $P_1''$. As above this trace has at least one step, so $P_1''  \parop \Snd{p}{p} \parop \Rcv{p}{y}.\Snd{q}{M'} \redP^* P''$. By Lemma~\ref{lem:rename-free-names}, we rename $p$ and $q$ inside $P_1''$ so that $p, q \notin \fn(P_1'')$, and we obtain the desired result by induction hypothesis.

\item $\Snd{p}{p} \ltrP{\Res{z}\Snd{N}{z}} A_1 \equiv \{\subst{p}{z}\}$,
$\Rcv{p}{y}.\Snd{q}{M'} \ltrP{\Rcv{N}{z}} A_2 \equiv \Snd{q}{M'}$, 
$P_1' \equiv P_1 \parop \Res{z}(A_1 \parop A_2) 
\equiv P_1 \parop \Res{z}(\{\subst{p}{z}\} \parop \Snd{q}{M'}) \equiv P_1 \parop \Snd{q}{M'}$
so $P_1 \parop \Snd{p}{p} \parop \Rcv{p}{y}.\Snd{q}{M'} \redP P_1 \parop \Snd{q}{M'} \equivP P_1' \redP^* P''$ and $q \notin \fn(P_1)$,
so by the property shown above, $P'' \equiv P_3 \parop \Snd{q}{M'}$ and
$P_1 \redP^* P_3$ for some closed process $P_3$, as desired.

\end{itemize}

Finally, we prove that, if $P$ and $N\sigma$ are closed, $P \parop \Snd{p}{p} \parop \Rcv{N\sigma}{x}.\Rcv{p}{y}.\Snd{q}{x} \redP^* P''$, $P'' \not\barb{p}$, and $p,q \notin \fn(P) \cup \fn(N\sigma)$, then $P \redP^* \ltrP{\Res{x}\Snd{N\sigma}{x}} \rightarrow^* B$ and $P'' \equiv \Res{x}(B \parop \Snd{q}{x})$ for some $B$,
by induction on the length of the trace.
Since $P \parop \Snd{p}{p} \parop \Rcv{N\sigma}{x}.\Rcv{p}{y}.\Snd{q}{x} \barb{p}$,
the trace $P \parop \Snd{p}{p} \parop \Rcv{N\sigma}{x}.\Rcv{p}{y}.\Snd{q}{x} \redP^* P''$ has at 
least one step: $P \parop \Snd{p}{p} \parop \Rcv{N\sigma}{x}.\Rcv{p}{y}.\Snd{q}{x} \redP P_1 \redP^* P''$.
By Lemmas~\ref{lem:decomp-redP-closed}, \ref{lem:decomp-ltrP}, and~\ref{lem:red-no-fn}, the only cases that can happen in the first step are:
\begin{itemize}

\item $P \redP P'$ and $P' \parop \Snd{p}{p} \parop \Rcv{N\sigma}{x}.\Rcv{p}{y}.\Snd{q}{x} \equiv P_1 \redP^* P''$ for some closed process $P'$. As above this trace has at least one step, so $P' \parop \Snd{p}{p} \parop \Rcv{N\sigma}{x}.\Rcv{p}{y}.\Snd{q}{x} \redP^* P''$. By Lemma~\ref{lem:rename-free-names}, we rename $p$ and $q$ inside $P'$ so that $p, q \notin \fn(P')$, and we obtain the desired result by induction hypothesis.

\item $P \ltrP{\Res{z}\Snd{N'}{z}} B'$, 
$\Rcv{N\sigma}{x}.\Rcv{p}{y}.\Snd{q}{x} \ltrP{\Rcv{N'}{z}} B''$,
and $P_1 \equiv \Res{z}(B' \parop \Snd{p}{p} \parop B'')$.
By Lemma~\ref{lem:caract-ltrP},
$P \equivP \Res{\vect n'}(\Snd{N'}{M'}.P_2 \parop P_3)$,
$B' \equiv \Res{\vect n'}(P_2 \parop \{\subst{M'}{z}\} \parop P_3)$,
$\{\vect n'\} \cap \fn(N') = \emptyset$, and
$z \notin \fv(\Snd{N'}{M'}.P_2 \parop P_3)$.
By Lemma~\ref{lem:decomp-ltrP},
$\Sigma \vdash N\sigma = N'$ and $B'' \equiv \Rcv{p}{y}.\Snd{q}{z}$.
Using Lemma~\ref{lem:closing}\eqref{closing:equivP}, 
we can guarantee that $N'$, $M'$, $P_2$, $P_3$ are closed.
We rename $\vect n'$ so that these names are distinct from $p$ and $q$.
By Lemma~\ref{lem:rename-free-names}, we rename $p$ and $q$ inside
$\Res{\vect n'}(\Snd{N'}{M'}.P_2 \parop P_3)$ so that 
$p, q \notin \fn(\Res{\vect n'}(\Snd{N'}{M'}.P_2 \parop P_3))$.
So $P_1 \equiv \Res{z}(B' \parop \Snd{p}{p} \parop B'')
\equiv \Res{z}(\Res{\vect n'}(P_2 \parop \{\subst{M'}{z}\} \parop P_3)
\parop \Snd{p}{p} \parop \Rcv{p}{y}.\Snd{q}{z})
\equiv \Res{\vect n'}(P_2 \parop P_3 \parop \Snd{p}{p} \parop \Rcv{p}{y}.\Snd{q}{M'})$.
Since $P_1 \redP^* P''$ and this trace has at least one step because $P_1 \barb{p}$ and $P'' \not\barb{p}$, we have $\Res{\vect n'}(P_2 \parop P_3 \parop \Snd{p}{p} \parop \Rcv{p}{y}.\Snd{q}{M'}) \redP^* P''$, so by Lemma~\ref{lem:decomp-redP},
$P_2 \parop P_3 \parop \Snd{p}{p} \parop \Rcv{p}{y}.\Snd{q}{M'} \redP^* P_4$
and $P'' \equiv \Res{\vect n'}P_4$ for some $P_4$.
Since $p,q \notin \fn(P_2 \parop P_3)$, by the previous result, $P_2 \parop P_3 \redP^* P_5$
and $P_4 \equiv P_5 \parop \Snd{q}{M'}$ for some closed process $P_5$.
Therefore, we have $P'' \equiv \Res{\vect n'}(P_5 \parop \Snd{q}{M'})
\equiv \Res{x}(\Res{\vect n'}(P_5 \parop \{\subst{M'}{x}\}) \parop \Snd{q}{x})$
and $P \equivP \Res{\vect n'}(\Snd{N'}{M'}.P_2 \parop P_3)
\ltrP{\Res{x}\Snd{N\sigma}{x}} \Res{\vect n'}(P_2 \parop \{\subst{M'}{x}\} \parop P_3)\rightarrow^* \Res{\vect n'}(P_5 \parop \{\subst{M'}{x}\})$.
Let $B \eqdef \Res{\vect n'}(P_5 \parop \{\subst{M'}{x}\})$.
Then we have $P \ltrP{\Res{x}\Snd{N\sigma}{x}} \rightarrow^* B$ and
$P'' \equiv \Res{x}(B \parop \Snd{q}{x})$.
\end{itemize}

To sum up, we have $A \equiv \pnf(A) = \Res{\vect n}(\sigma \parop P)$,
$P \redP^* \ltrP{\Res{x}\Snd{N\sigma}{x}} \rightarrow^* B$, and $P'' \equiv \Res{x}(B \parop \Snd{q}{x})$, so $A'' \equiv \pnf(A'') \equiv \Res{\vect n}(\sigma \parop P'') \equiv \Res{\vect n}(\sigma \parop \Res{x}(B \parop \Snd{q}{x}))
\equiv \Res{x}(\Res{\vect n}(\sigma \parop B) \parop \Snd{q}{x})$
since $x\notin \fv(\sigma)$.
Let $A' \eqdef \Res{\vect n}(\sigma \parop B)$.
So $\pnf(A) \redpnf^* \ltrpnf{\Res{x}\Snd{N}{x}} \rightarrow^* A'$. Hence by Lemmas~\ref{lem:red-pnf-to-std} and~\ref{lem:redalpha-pnf-to-std}, $A \rightarrow^* \ltr{\Res{x}\Snd{N}{x}} \rightarrow^* A'$
and $A'' \equiv \Res{x}(A' \parop \Snd{q}{x})$.
\end{proof}

\begin{lemma}\label{lem:ren-bicong}
Let $A$ and $B$ be two closed extended processes.
\begin{itemize}
\item Let $\sigma$ be a bijective renaming.
We have $A \bicong B$ if and only if $A \sigma \bicong B\sigma$.
\item Let $A'$ and $B'$ be obtained from $A$ and $B$, respectively, by replacing all variables (including their occurrences in domains of active substitutions) with distinct variables. We have
$A \bicong B$ if and only if $A' \bicong B'$.
\end{itemize}
\end{lemma}
\begin{proof}
To prove the first point, we define a relation $\rel$ by
$A' \rel B'$ if and only if $A' = A\sigma$, $B' = B\sigma$, and $A \bicong B$ for some $A$ and $B$. We show that $\rel$ satisfies the three properties of Definition~\ref{def:bicong}. Then ${\rel} \subseteq {\bicong}$, so if $A \bicong B$, then $A' = A \sigma \bicong B' = B\sigma$.
\begin{enumerate}
\item If $A' \rel B'$ and $A' \barb{a}$, then 
$A' \rightarrow^* \equiv \CTX[\Snd{a}{M}.P]$ for some evaluation context
$\CTX$ that does not bind $a$. Then, by Lemma~\ref{lem:rename-eqstr-red},
$A = A'\sigma^{-1} \rightarrow^* \equiv C\sigma^{-1}[\Snd{a\sigma^{-1}}{M\sigma^{-1}}.P\sigma^{-1}]$, so $A \barb{a\sigma^{-1}}$.
By definition of $\bicong$, $B \barb{a\sigma^{-1}}$, so $B' \barb{a}$ as above.

\item If $A' \rel B'$, $A' \rightarrow A'_1$, and $A'_1$ is closed, then by Lemma~\ref{lem:rename-eqstr-red},
$A = A'\sigma^{-1} \rightarrow A'_1 \sigma^{-1}$. We let $A'' = A'_1 \sigma^{-1}$, which is also closed.
So by definition of $\bicong$, $B \rightarrow^* B''$ and $A'' \bicong B''$ for some $B''$. By Lemma~\ref{lem:rename-eqstr-red}, $B' = B\sigma \rightarrow^* B''\sigma$. We let $B'_1 = B''\sigma$. We have $A'_1 \rel B'_1$ and $B' \rightarrow^* B'_1$. So Property~\ref{bctwo} holds.

\item If $A' \rel B'$, then $A = A'\sigma^{-1} \bicong B'\sigma^{-1} = B$,
so $\CTX[A']\sigma^{-1} = \CTX\sigma^{-1}[A] \bicong \CTX\sigma^{-1}[B] = \CTX[B']\sigma^{-1}$,
hence $\CTX[A'] \rel \CTX[B']$.

\end{enumerate}
The same argument also proves the converse, via the inverse renaming.

The proof of the second point is similar.
\end{proof}

\begin{lemma}\label{lem:bicong-comm}
If $M$ is ground, $\fv(P) \subseteq \{x\}$, and $a \notin \fn(P) \cup \fn(M)$, then $\Res{a}(\Snd{a}{M} \parop \Rcv{a}{x}.P) \bicong P\{\subst{M}{x}\}$.
\end{lemma}
\begin{proof}
By Lemma~\ref{lem:th1-first-dir}, it is enough to prove that
$\Res{a}(\Snd{a}{M} \parop \Rcv{a}{x}.P) \wkbisim P\{\subst{M}{x}\}$.
Let $A_1 = \Res{a}(\Snd{a}{M} \parop \Rcv{a}{x}.P)$ and $B_1 = P\{\subst{M}{x}\}$.
Let ${\rel} = \{ (A, B) \mid A$ and $B$ are closed extended processes, $A \equiv A_1$ and $B \equiv B_1$, or $A \equiv B_1$ and $B \equiv A_1 \} \cup \{ (A, B) \mid A$ and $B$ are closed extended processes and $A \equiv B\}$.
We show that $\rel$ is a labelled bisimulation: $\rel$ is symmetric and
\begin{enumerate}
\item We have $A_1 \enveq B_1$ since $\frameof{A_1} = \nil = \frameof{B_1}$. 
Hence, if $A \rel B$, then $A \enveq B$.

\item If $A_1 \rightarrow A'$ and $A'$ is closed, then 
$A' \equiv B_1$. (This point can be proved in detail by using 
partial normal forms.)

Hence, if $A \rel B$, $A \rightarrow A'$, and $A'$ is closed, then 
\begin{itemize}
\item either $A \equiv A_1$ and $B \equiv B_1$, so $A' \equiv B_1 \equiv B$,
hence with $B' \eqdef B$, $B \rightarrow^* B'$ and $A' \rel B'$.

\item or $A \equiv B_1$ and $B \equiv A_1$, so
$B \equiv A_1 \rightarrow B_1 \equiv A \rightarrow A'$,
hence with $B' \eqdef A'$, $B \rightarrow^* B'$ and $A' \rel B'$.

\item or $A \equiv B$, so with $B' \eqdef A'$, $B \equiv A \rightarrow A' = B'$,
and $A' \rel B'$.
\end{itemize}

\item $A_1$ does not reduce by $\ltr{\alpha}$, for any $\alpha$.
(This point can be proved in detail by using partial normal forms.)
Hence, if $A \rel B$, $A \ltr{\alpha} A'$, and $A'$ is closed, then 
\begin{itemize}
\item either $A \equiv A_1$ and $B \equiv B_1$, so $A_1 \ltr{\alpha} A'$.
This case is impossible. 

\item or $A \equiv B_1$ and $B \equiv A_1$, so
$B \equiv A_1 \rightarrow B_1 \equiv A \ltr{\alpha} A'$,
hence with $B' \eqdef A'$, $B \rightarrow \ltr{\alpha} B'$ and $A' \rel B'$.

\item or $A \equiv B$, so with $B' \eqdef A'$, $B \equiv A \ltr{\alpha} A' = B'$,
and $A' \rel B'$.
\end{itemize}
\end{enumerate}
Therefore, ${\rel} \subseteq {\wkbisim}$, so $A_1 \wkbisim B_1$.
\end{proof}

\begin{corollary}\label{cor:bicong-comm}
If $A$ is a closed extended process, $x \in \dom(A)$, $\fv(P) \subseteq \dom(A)$, and $a \notin \fn(P)$, then $A \parop \Res{a}(\Snd{a}{x} \parop \Rcv{a}{x}.P) \bicong A \parop P$.
\end{corollary}
\begin{proof}
Let $\pnf(A) = \Res{\vect n}(\sigma \parop P')$.
We rename $\vect n$ so that $\{\vect n\} \cap \fn(P) = \emptyset$.
Let $\sigma' = \sigma_{|\dom(\sigma)\setminus\{x\}}$. Let $a' \notin \fn(P) \cup \fn(\sigma)$. We have
\begin{align*}
A \parop \Res{a}(\Snd{a}{x} \parop \Rcv{a}{x}.P) 
&\equiv \Res{\vect n}(\sigma \parop P' \parop \Res{a'}(\Snd{a'}{x\sigma} \parop \Rcv{a'}{x}.P\sigma')\\
&\bicong \Res{\vect n}(\sigma \parop P' \parop P\sigma'\{\subst{x\sigma}{x}\})
\tag*{by Lemma~\ref{lem:bicong-comm}}\\
&= \Res{\vect n}(\sigma \parop P' \parop P\sigma)\\
&\equiv A \parop P
\end{align*} 
\end{proof}

\begin{restate}{Lemma}{\ref{lem:extrusion}}
Let $A$ and $B$ be two closed extended processes with a same domain that contains $\vect{x}$.
Let $\CTX_{\vect x}[\hole] \eqdef \nu \vect{x}.( \prod_{x \in \vect{x}}\Snd {n_x} x \parop \hole\,)$
using names $n_x$ that do not occur in $A$ or $B$.
If $\CTX_{\vect x}[A] \bicong \CTX_{\vect x}[B]$, then $A \bicong B$.
\end{restate}%
\begin{proof}
  We rely on the following property: if $A$ is a closed extended process with $\{\vect x\} \subseteq \dom(A)$ and $\CTX_{\vect x}[A] \rightarrow
  C'$, then $A \rightarrow A'$ and $C' \equiv
  \CTX_{\vect x}[A']$ for some closed extended process $A'$, proved as follows. Let $\pnf(A) = \Res{\vect n}(\sigma \parop P)$. We rename $\vect n$ so that $\{ \vect n \} \cap \{ \vect n_x \} = \emptyset$. Then $\pnf(\CTX_{\vect x}[A]) = \Res{\vect n}(\sigma_{|\dom(\sigma)\setminus\{\vect x\}} \parop \prod_{x \in \vect{x}}\Snd {n_x}{x\sigma} \parop P)$. By Lemma~\ref{lem:red-std-to-pnf}, $\pnf(\CTX_{\vect x}[A]) \redpnf \pnf(C')$. By Lemma~\ref{lem:decomp-redpnf}, $\prod_{x \in \vect{x}}\Snd {n_x}{x\sigma} \parop P \redP P'$ and 
$\pnf(C') \equiv \Res{\vect n}(\sigma_{|\dom(\sigma)\setminus\{\vect x\}} \parop P')$
for some $P'$.
By Lemmas~\ref{lem:decomp-redP-closed} and~\ref{lem:red-no-fn}, since $\{ \vect n_x \} \cap \fn(P) = \emptyset$, the only case that can happen is $P \redP P''$ and $P' \equiv \prod_{x \in \vect{x}}\Snd {n_x}{x\sigma} \parop P''$ for some closed process $P''$.
Let $A' = \Res{\vect n}(\sigma \parop P'')$.
Then $A \equiv \pnf(A) = \Res{\vect n}(\sigma \parop P) \rightarrow \Res{\vect n}(\sigma \parop P'') = A'$ and
$C' \equiv \pnf(C') 
\equiv \Res{\vect n}(\sigma_{|\dom(\sigma)\setminus\{\vect x\}} \parop P') 
\equiv \Res{\vect n}(\sigma_{|\dom(\sigma)\setminus\{\vect x\}} \parop \prod_{x \in \vect{x}}\Snd {n_x}{x\sigma} \parop P'')
\equiv \Res{\vect x}(\prod_{x \in \vect{x}}\Snd {n_x}{x} \parop \Res{\vect n}(\sigma \parop P''))
\equiv \CTX_{\vect x}[A']$.
 
Let $\rel$ be the relation that
collects all closed extended processes $A$ and~$B$ with a same domain that contains $\vect{x}$, such that $\CTX_{\vect x}[A] \bicong \CTX_{\vect x}[B]$, for some $\vect x$ and some names $\vect n_x$ that do not occur in $A$ or $B$.
We show that $\rel$ is an observational bisimulation.

  Assume $A \rel B$. 
\begin{itemize}
\item If $A \rightarrow A'$ and $A'$ is closed, then $\CTX_{\vect x}[A] \rightarrow \CTX_{\vect x}[A']$.
By bisimulation hypothesis, $\CTX_{\vect x}[B] \rightarrow^* C' \bicong \CTX_{\vect x}[A']$.
By induction on the number of reductions and using partial normal forms,
we build $B \rightarrow^* B'$ such that $C' \equiv \CTX_{\vect x}[B']$ for some closed extended process $B'$
and conclude using $A' \rel B'$.

\item We have $\CTX_{\vect x}[A] \barb{n}$ if and only if $n = n_x$ for some $x \in \vect{x}$ 
or $A \barb{n}$, and similarly for $B$. Hence,
if $A \barb{n}$, then $\CTX_{\vect x}[A] \barb{n}$, so $\CTX_{\vect x}[B] \barb{n}$.
By Lemma~\ref{lem:barb-obvious},
since $A \barb{n}$, we have $n \neq n_x$ for all $x \in \vect{x}$, 
so $B \barb{n}$.

\item 
For the congruence property, we suppose that $A \rel B$, and we want to show
that $\CTX[A] \rel \CTX[B]$ for all closing evaluation contexts $\CTX$.
Using Lemma~\ref{lem:ren-bicong}, we show that $\rel$ is invariant
by renaming of free names and variables, so we can rename the free names
and variables of $\CTX$, so that the obtained context is simple.
Then by Lemma~\ref{lem:simplecontexts}, we construct a context $\CTX'$ of 
the form $\Res{\vect u}(\hole \parop C'')$ such that $\CTX \equiv \CTX'$.
Hence, it is sufficient to show that $\CTX'[A] \rel \CTX'[B]$.

Let $\pnf(C'') = \Res{\vect n}(\sigma \parop P)$. Let $\vect u = \vect m\vect z$. We rename $\vect n$ so that $\{\vect n\} \cap (\fn(A) \cup \fn(B)) = \emptyset$.
Since $A \rel B$, $\CTX_{\vect x}[A] \bicong \CTX_{\vect x}[B]$ for some $\vect x$.
Using Lemma~\ref{lem:ren-bicong}, we rename $\vect n_x$ so that
$\{ \vect n_x \} \cap (\{ \vect n, \vect m\} \cup \fn(P) \cup \fn(\sigma)) = \emptyset$.
Let $n'_x$ for $x \in (\vect x \cup \dom(\sigma))\setminus \vect z$ be fresh names.
Let $\CTX_1[\hole] = \Res{\vect m,\vect n, \vect n_x,\vect z\setminus(\vect x \cup \dom(\sigma))}(\hole \parop \widetilde{\Rcv{n_x}{x}}.(P \parop \prod_{x \in \vect x \setminus \vect z} \Snd{n'_x}{x} \parop \prod_{y \in \dom(\sigma) \setminus \vect z} \Snd{n'_y}{y\sigma}))$, where $\widetilde{\Rcv{n_x}{x}}$ stands for $\Rcv{n_{x_1}}{x_1}\ldots \Rcv{n_{x_k}}{x_k}$ when $\vect x = x_1, \ldots, x_k$.
\begin{align*}
&\CTX_{(\vect x \cup \dom(\sigma)) \setminus \vect z }[\CTX'[A]] \\
&\quad\equiv \Res{(\vect x \cup \dom(\sigma)) \setminus \vect z}(\Res{\vect m}\Res{\vect z}(A \parop \Res{\vect n}(\sigma \parop P)) \parop \prod_{x \in \vect x \setminus \vect z} \Snd{n'_x}{x} \parop \prod_{y \in \dom(\sigma) \setminus \vect z} \Snd{n'_y}{y})\\
&\quad\equiv \Res{\vect m,\vect n,\vect z\cup\vect x\cup\dom(\sigma)}
(A \parop P \parop \sigma \parop \prod_{x \in \vect x \setminus \vect z} \Snd{n'_x}{x} \parop \prod_{y \in \dom(\sigma) \setminus \vect z} \Snd{n'_y}{y\sigma})\\
&\quad\equiv \Res{\vect m,\vect n, \vect z\setminus(\vect x \cup \dom(\sigma))}\Res{\vect x}(A \parop P \parop \prod_{x \in \vect x \setminus \vect z} \Snd{n'_x}{x} \parop \prod_{y \in \dom(\sigma) \setminus \vect z} \Snd{n'_y}{y\sigma})\\
&\quad\bicong \Res{\vect m,\vect n, \vect z\setminus(\vect x \cup \dom(\sigma))}\Res{\vect x}\\
&\phantom{\quad{}\equiv {}\quad}(A \parop \Res{\vect n_x}(\prod_{x \in \vect x} \Snd{n_x}{x} \parop \widetilde{\Rcv{n_x}{x}}.(P \parop \prod_{x \in \vect x \setminus \vect z} \Snd{n'_x}{x} \parop \prod_{y \in \dom(\sigma) \setminus \vect z} \Snd{n'_y}{y\sigma}))
\end{align*}
by Corollary~\ref{cor:bicong-comm} applied several times, so
\begin{align*}
&\CTX_{(\vect x \cup \dom(\sigma)) \setminus \vect z }[\CTX'[A]] \\
&\quad\equiv \Res{\vect m,\vect n, \vect n_x,\vect z\setminus(\vect x \cup \dom(\sigma))}\\
&\phantom{\quad{}\equiv {}\quad}(\Res{\vect x}(A \parop \prod_{x \in \vect x} \Snd{n_x}{x})\parop \widetilde{\Rcv{n_x}{x}}.(P \parop \prod_{x \in \vect x \setminus \vect z} \Snd{n'_x}{x} \parop \prod_{y \in \dom(\sigma) \setminus \vect z} \Snd{n'_y}{y\sigma})) \\
&\quad\equiv \CTX_1[\CTX_{\vect x}[A]]
\end{align*}
By the same argument, $\CTX_{(\vect x \cup \dom(\sigma)) \setminus \vect z }[\CTX'[B]] \bicong \CTX_1[\CTX_{\vect x}[B]]$.
Since $\CTX_{\vect x}[A] \bicong \CTX_{\vect x}[B]$, we also have $\CTX_1[\CTX_{\vect x}[A]] \bicong \CTX_1[\CTX_{\vect x}[B]]$, so by transitivity of $\bicong$, $\CTX_{(\vect x \cup \dom(\sigma)) \setminus \vect z }[\CTX'[A]] \bicong \CTX_{(\vect x \cup \dom(\sigma)) \setminus \vect z }[\CTX'[B]]$. Hence,
$\CTX'[A] \rel \CTX'[B]$.
\end{itemize}
Since $\rel$ is an observational bisimulation, ${\rel} \subseteq {\bicong}$,
so $\CTX_{\vect x}[A] \bicong \CTX_{\vect x}[B]$ implies $A \bicong B$.
\end{proof}

\begin{corollary}\label{cor:enveq-frame}
Observational equivalence and static equivalence coincide on frames.
\end{corollary}
\begin{proof}
Since frames do not reduce, static equivalence and labelled bisimilarity
coincide on frames. By Theorem~\ref{THM:OBSERVATIONAL-LABELED}, we can 
then conclude.
\end{proof}

\section{Proof of Lemma~\ref{LEM:DISCLOSURE}}\label{app:disclosure}

The \emph{image} of a substitution $\sigma = \{ \subst{M_1}{x_1}, \dots, \subst{M_n}{x_n} \}$
is the set of terms $\{M_1, \allowbreak \dots, \allowbreak M_n\}$.
We denote by $\rho$ a bijective renaming.
We denote by $\sigma\rho$ the substitution obtained by applying the renaming $\rho$
to the terms in the image of $\sigma$, that is,
when $\sigma = \{ \subst{M_1}{x_1}, \dots, \subst{M_n}{x_n} \}$,
$\sigma\rho = \{\subst{M_1\rho}{x_1}, \dots, \subst{M_n\rho}{x_n} \}$.

\begin{lemma}\label{lem:equivframes}
Let $\Res{\vect n}\sigma$ and $\Res{\vect n'}\sigma'$ be two frames such that $\Res{\vect n}\sigma \equiv \Res{\vect n'}\sigma'$, and $M$ and $N$ be two terms such that $\fv(M) \cup \fv(N) \subseteq \dom(\sigma) = \dom(\sigma')$ and $\{\vect n, \vect n'\} \cup (\fn(M) \cup \fn(N)) = \emptyset$. If $\Sigma \vdash M\sigma = N\sigma$, then $\Sigma \vdash M\sigma' = N\sigma'$.
\end{lemma}
\begin{proof}
Let us prove the following result:
\begin{quote}
Suppose $\Res{\vect n}(\sigma \parop P) \equivpnf \Res{\vect n'}(\sigma' \parop P')$ and $\fv(M) \cup \fv(N) \subseteq \dom(\sigma) = \dom(\sigma')$.
Let $\rho$ be a bijective renaming that maps names in $\vect n$ to names not in $\fn(\Res{\vect n}\sigma) \cup \fn(M) \cup \fn(N)$ and leaves names in $(\fn(\Res{\vect n}\sigma) \cup \fn(M) \cup \fn(N))\setminus \{\vect n\}$ unchanged, and
$\rho'$ be a bijective renaming that maps names in $\vect n'$ to names not in $\fn(\Res{\vect n'}\sigma') \cup \fn(M) \cup \fn(N)$ and leaves names in $(\fn(\Res{\vect n'}\sigma') \cup \fn(M) \cup \fn(N))\setminus \{\vect n'\}$ unchanged.

We have $\Sigma \vdash M(\sigma\rho) = N(\sigma\rho)$ if and only if
$\Sigma \vdash M(\sigma'\rho') = N(\sigma'\rho')$.
\end{quote}
This result is proved by induction on the derivation of $\Res{\vect n}(\sigma \parop P) \equivpnf \Res{\vect n'}(\sigma' \parop P')$.
\begin{itemize}

\item Transitivity and symmetry: obvious.

\item Reflexivity: The renamings $\rho$ and $\rho'$ map names in $\vect n$ to names not in $\fn(\Res{\vect n}\sigma) \cup \fn(M) \cup \fn(N)$ and leave names in $(\fn(\Res{\vect n}\sigma) \cup \fn(M) \cup \fn(N))\setminus \{\vect n\}$ unchanged. Let $\rho''$ be a bijective renaming that maps $\vect n\rho$ to $\vect n\rho'$ and leaves names in $\fn(\Res{\vect n}\sigma) \cup \fn(M) \cup \fn(N)$ unchanged.
If $\Sigma \vdash M(\sigma\rho) = N(\sigma\rho)$, then
$\Sigma \vdash M(\sigma\rho)\rho'' = N(\sigma\rho)\rho''$,
so $\Sigma \vdash M(\sigma\rho') = N(\sigma\rho')$.
The converse is proved is the same way, using $\rho''^{-1}$ 
instead of $\rho''$.

\item Cases $\brn{Plain}''$ and $\brn{New-C}''$: These cases are proved by
  the same proof as for reflexivity, since the desired property does
  not depend on the process $P$ nor on the order of $\vect n$.

\item Case $\brn{New-Par}''$: $\Res{\vect n}(\sigma \parop \Res{n'}P) \equivpnf \Res{\vect n, n'}(\sigma \parop P)$ where $n' \notin \fn(\sigma)$.
Let $\rho$ be a bijective renaming that maps names in $\vect n$ to names not in $\fn(\Res{\vect n}\sigma) \cup \fn(M) \cup \fn(N)$ and leaves names in $(\fn(\Res{\vect n}\sigma) \cup \fn(M) \cup \fn(N))\setminus \{\vect n\}$ unchanged, and
$\rho'$ be a bijective renaming that maps names in $\vect n, n'$ to names not in $\fn(\Res{\vect n, n'}\sigma) \cup \fn(M) \cup \fn(N)$ and leaves names in $(\fn(\Res{\vect n,n'}\sigma) \cup \fn(M) \cup \fn(N))\setminus \{\vect n, n'\}$ unchanged.
Let $\rho''$ be a bijective renaming that maps $\vect n \rho$ to $\vect n\rho'$ and that leaves names in  $\fn(\Res{\vect n}\sigma) \cup \fn(M) \cup \fn(N)$ unchanged. (Since $n' \notin \fn(\sigma)$, $\fn(\Res{\vect n}\sigma) = \fn(\Res{\vect n, n'}\sigma)$, so the names $\vect n\rho'$ do not collide with $\fn(\Res{\vect n}\sigma) \cup \fn(M) \cup \fn(N)$, hence $\rho''$ exists.)

If $\Sigma \vdash M (\sigma\rho) = N (\sigma\rho)$, then
$\Sigma \vdash M (\sigma\rho)\rho'' = N (\sigma\rho)\rho''$,
so $\Sigma \vdash M (\sigma\rho') = N (\sigma\rho')$.
(We have $\sigma \rho \rho'' = \sigma\rho'$ because $n' \notin \fn(\sigma)$.)

The converse is proved in the same way, using $\rho''^{-1}$ 
instead of $\rho''$.

\item Case $\brn{Rewrite}''$: $\Res{\vect n}(\sigma \parop P) \equivpnf \Res{\vect n}(\sigma' \parop P)$ where $\dom(\sigma) = \dom(\sigma')$, 
$\Sigma \vdash x\sigma = x\sigma'$ for all $x \in \dom(\sigma)$,
and $(\fv(x\sigma) \cup \fv(x\sigma')) \cap \dom(\sigma)  = \emptyset$
for all $x \in \dom(\sigma)$.

Let $\rho$ be a bijective renaming that maps names in $\vect n$ to names not in $\fn(\Res{\vect n}\sigma) \cup \fn(M) \cup \fn(N)$ and leaves names in $(\fn(\Res{\vect n}\sigma) \cup \fn(M) \cup \fn(N))\setminus \{\vect n\}$ unchanged, and
$\rho'$ be a bijective renaming that maps names in $\vect n$ to names not in $\fn(\Res{\vect n}\sigma') \cup \fn(M) \cup \fn(N)$ and leaves names in $(\fn(\Res{\vect n}\sigma') \cup \fn(M) \cup \fn(N))\setminus \{\vect n\}$ unchanged.

Let $\rho''$ be a bijective renaming that maps names in $\vect n$ to names not in $\fn(\Res{\vect n}\sigma) \cup \fn(\Res{\vect n}\sigma') \cup \fn(M) \cup \fn(N)$ and leaves names in $(\fn(\Res{\vect n}\sigma) \cup \fn(\Res{\vect n}\sigma') \cup \fn(M) \cup \fn(N))\setminus \{\vect n\}$ unchanged.

The renaming $\rho''$ a fortiori maps names in $\vect n$ to names not in $\fn(\Res{\vect n}\sigma) \cup \fn(M) \cup \fn(N)$ and leaves names in $(\fn(\Res{\vect n}\sigma) \cup \fn(M) \cup \fn(N))\setminus \{\vect n\}$ unchanged, so by the case of reflexivity 
$\Res{\vect n}(\sigma \parop P) \equivpnf \Res{\vect n}(\sigma \parop P)$, 
we have
$\Sigma \vdash M(\sigma\rho) = N(\sigma\rho)$ if and only if
$\Sigma \vdash M(\sigma\rho'') = N(\sigma\rho'')$.

Similarly, $\Sigma \vdash M(\sigma'\rho') = N(\sigma'\rho')$ if and only if
$\Sigma \vdash M(\sigma'\rho'') = N(\sigma'\rho'')$.

Moreover, for all $x \in \dom(\sigma)$, $\Sigma \vdash x\sigma = x\sigma'$,
so $\Sigma \vdash x\sigma\rho'' = x\sigma'\rho''$,
hence $\Sigma \vdash M(\sigma\rho') = M(\sigma'\rho'')$
and $\Sigma \vdash N(\sigma\rho') = N(\sigma'\rho'')$,
therefore $\Sigma \vdash M(\sigma\rho'') = N(\sigma\rho'')$
if and only if $\Sigma \vdash M(\sigma'\rho'') = N(\sigma'\rho'')$.

We can then conclude that $\Sigma \vdash M(\sigma\rho) = N(\sigma\rho)$
if and only if $\Sigma \vdash M(\sigma'\rho') = N(\sigma'\rho')$.

\end{itemize}

The lemma is an easy consequence of this result: since $\Res{\vect n}\sigma \equiv \Res{\vect n'}\sigma'$, we have $\Res{\vect n}(\sigma \parop \nil) \equivpnf \Res{\vect n'}(\sigma' \parop \nil)$ by Lemma~\ref{lem:struct-std-to-pnf}. We conclude by applying the previous result taking $P = P' = \nil$ and the identity for $\rho$ and $\rho'$.
\end{proof}

\begin{lemma}\label{lem:red-inside-Res.s.s/x}
Let $A$ be a closed extended process.
If $\Res{s}(\{\subst{s}{x}\} \parop A ) \rightarrow B'$,
then there exists a closed extended process $A'$ such that
$A \rightarrow A'$ and $B' \equiv \Res{s}(\{\subst{s}{x}\} \parop A')$.
\end{lemma}
\begin{proof}
Let $\pnf(A) = \Res{\vect n}(\sigma \parop P)$. We rename $\vect n$ so that $s \notin \{ \vect n\}$.
By Lemma~\ref{lem:red-std-to-pnf}, 
$\pnf(\Res{s}(\{\subst{s}{x}\} \parop A )) = \Res{s,\vect n}(\{\subst{s}{x}\} \parop \sigma \parop P)
\redpnf \pnf(B')$.
By Lemma~\ref{lem:decomp-redpnf-closed},
$P \redP P'$ and $\pnf(B') \equiv \Res{s,\vect n}(\{\subst{s}{x}\} \parop \sigma \parop P')$ for some
closed process $P'$.
Let $A' = \Res{\vect n}(\sigma \parop P')$.
Hence, $A \equiv \Res{\vect n}(\sigma \parop P) \rightarrow \Res{\vect n}(\sigma \parop P') = A'$ 
and $B' \equiv \pnf(B') \equiv \Res{s,\vect n}(\{\subst{s}{x}\} \parop \sigma \parop P') \equiv \Res{s}(\{\subst{s}{x}\} \parop A')$.
\end{proof}

\begin{lemma}\label{lem:ltr-inside-Res.s.s/x}
Let $A$ be a closed extended process and $\alpha$ be such that $\fv(\alpha) \subseteq \dom(A) \cup\{x\}$ and $s \notin \fn(\alpha)$.
If $\Res{s}(\{\subst{s}{x}\} \parop A) \ltr{\alpha} B'$,
then there exists a closed extended process $A'$ such that
$A \ltr{\alpha\{\subst{s}{x}\}} A'$ and $B' \equiv \Res{s}(\{\subst{s}{x}\} \parop A')$.
\end{lemma}
\begin{proof}
Let $\pnf(A) = \Res{\vect n}(\sigma \parop P)$. We rename $\vect n$ so that $s \notin \{ \vect n\}$ and the elements of $\vect n$ do not occur in $\alpha$.
By Lemma~\ref{lem:redalpha-std-to-pnf}, 
$\pnf(\Res{s}(\{\subst{s}{x}\} \parop A )) = \Res{s,\vect n}(\{\subst{s}{x}\} \parop \sigma \parop P)
\ltrpnf{\alpha} B'$.
By Lemma~\ref{lem:decomp-ltrpnf}, 
$P \ltrP{\alpha(\{\subst{s}{x}\} \parop \sigma)} B''$ and $B' \equiv \Res{s,\vect n}(\{\subst{s}{x}\} \parop \sigma \parop B'')$ for some $B''$.

Next, we show that we can choose $B''$ so that it is closed.
Let $\alpha' = \alpha(\{\subst{s}{x}\} \parop \sigma)$.
By Lemma~\ref{lem:caract-ltrP}, for some $\vect n'$, $P_1$, $P_2$, $B_1$, $N$, $M$, $P'$, $y$, we have
$P \equivP \Res{\vect n'}(P_1 \parop P_2)$,
$B'' \equiv \Res{\vect n'}(B_1 \parop P_2)$,
$\{\vect n'\} \cap \fn(\alpha') = \emptyset$,
$\bv(\alpha) \cap \fv(P_1 \parop P_2) = \emptyset$,
and
one of the following two cases holds:
\begin{enumerate}
\item $\alpha' = \Rcv{N}{M}$, $P_1 = \Rcv{N}{y}.P'$, and $B_1 = P'\{\subst{M}{y}\}$; or
\item $\alpha' = \Res{y}\Snd{N}{y}$, $P_1 = \Snd{N}{M}.P'$, and $B_1 = P' \parop \{\subst{M}{y}\}$.
\end{enumerate}
Let $\sigma'$ be a substitution that maps variables of $\fv(B'')\setminus\dom(B'')$ to distinct fresh names. We rename $y$ so that $y \notin \fv(B'')\setminus\dom(B'')$.
Since $P$ and $\alpha'$ are closed, $P = P\sigma'$ and $\alpha' = \alpha'\sigma'$.
By Lemma~\ref{lem:closing}\eqref{closing:redP}, $P \equivP \Res{\vect n'}(P_1\sigma' \parop P_2\sigma')$ and $\Sigma \vdash \Res{\vect n'}(P_1 \parop P_2) = \Res{\vect n'}(P_1\sigma' \parop P_2\sigma')$, so $\Sigma \vdash P_1 = P_1\sigma'$ and $\Sigma \vdash P_2 = P_2 \sigma'$.
By Lemma~\ref{lem:instance-equiv}, $B''\sigma' \equiv \Res{\vect n'}(B_1\sigma' \parop P_2\sigma')$.
Finally, one of the following two cases holds:
\begin{enumerate}
\item $\alpha' = \Rcv{N}{M}$, $P_1\sigma' = \Rcv{N}{y}.P'\sigma'$, and $B_1\sigma' = P'\sigma'\{\subst{M}{y}\}$; or
\item $\alpha' = \Res{y}\Snd{N}{y}$, $P_1\sigma' = \Snd{N}{M\sigma'}.P'\sigma'$, and $B_1\sigma' = P'\sigma' \parop \{\subst{M\sigma'}{y}\}$.
\end{enumerate}
Hence, by Lemma~\ref{lem:caract-ltrP}, 
$P \ltrP{\alpha'} B''\sigma'$.
Moreover, $\Sigma \vdash B_ 1= B_1\sigma'$ because $\Sigma \vdash P_1 = P_1\sigma'$, so $\Sigma \vdash \Res{\vect n'}(B_1 \parop P_2) = \Res{\vect n'}(B_1\sigma' \parop P_2\sigma')$, hence $B''\sigma' \equiv \Res{\vect n'}(B_1\sigma' \parop P_2\sigma') \equiv \Res{\vect n'}(B_1 \parop P_2) \equiv B''$,
so $B' \equiv \Res{s,\vect n}(\{\subst{s}{x}\} \parop \sigma \parop B''\sigma')$. Hence, by replacing $B''$ with $B''\sigma'$, we obtain the same properties as above, and additionally $B''\sigma'$ is closed.

Let $A' = \Res{\vect n}(\sigma \parop B''\sigma')$.
Hence, $A \equiv \Res{\vect n}(\sigma \parop P) \ltrpnf{\alpha\{\subst{s}{x}\}} A'$ by definition of $\ltrpnf{\alpha\{\subst{s}{x}\}}$, so $A \ltr{\alpha\{\subst{s}{x}\}} A'$ by Lemma~\ref{lem:redalpha-pnf-to-std},
and $B' \equiv \Res{s,\vect n}(\{\subst{s}{x}\} \parop \sigma \parop B''\sigma')
\equiv \Res{s}(\{\subst{s}{x}\} \parop A')$.
\end{proof}

\begin{restatenamed}{Lemma}{\ref{LEM:DISCLOSURE}}{Name disclosure}
  Let $A$ and $B$ be closed extended processes and $x$ be a variable such that
  $x \notin \dom(A)$.
We have $A \wkbisim B$ if and only if
\eqns{
  \nu n.(\{\subst{n}{x}\}
  \parop A )
  & \wkbisim & 
  \nu n.(\{\subst{n}{x}\}
  \parop B )}%
\end{restatenamed}%
\begin{proof}
The direct implication follows from context closure of
$\wkbisim$. Conversely, we show that the relation $\rel$ defined by
  $A \rel B$ if and only if $A$ and $B$ are closed extended processes and
  $\Res{s}(\{\subst{s}{x}\} \parop A )
  \wkbisim  
  \Res{s}(\{\subst{s}{x}\} \parop B )$ for some $x \notin \dom(A)$ 
  is a labelled bisimulation.
\begin{enumerate}
\item The relation $\rel$ is symmetric, because $\wkbisim$ is.
\item We suppose that $A \rel B$ and show that $A \enveq B$.
Since $A \rel B$, we have $\Res{s}(\{\subst{s}{x}\} \parop A )
  \wkbisim  
  \Res{s}(\{\subst{s}{x}\} \parop B )$ for some $x \notin \dom(A)$, so
$\Res{s}(\{\subst{s}{x}\} \parop A )
  \enveq
  \Res{s}(\{\subst{s}{x}\} \parop B )$.
We have $\dom(A) = \dom(\Res{s}(\{\subst{s}{x}\} \parop A )) \setminus \{ x\} =
\dom(\Res{s}(\{\subst{s}{x}\} \parop B )) \setminus \{ x\} = \dom(B)$.
Let $M$, $N$ be two terms such that $\fv(M) \cup \fv(N) \subseteq \dom(A)$.
Let $M' = M\{\subst{x}{s}\}$ and $N' = N\{\subst{x}{s}\}$.
We show that $(M=N)\frameof{A}$ if and only if $(M'=N')\frameof{\Res{s}(\{\subst{s}{x}\} \parop A)}$.

If $(M=N)\frameof{A}$, then $\frameof{A} \equiv \Res{\vect n}\sigma$, $\Sigma \vdash M \sigma = N \sigma$, and $\{\vect n\} \cap (\fn(M) \cup \fn(N)) = \emptyset$ for some $\vect n$ and $\sigma$. If $s \in \fn(M) \cup \fn(N)$, we know that $s \notin \{ \vect n\}$. Otherwise, we rename $\vect n$ so that $s \notin \{ \vect n\}$, while preserving the previous properties. So  $\frameof{\Res{s}(\{\subst{s}{x}\} \parop A)}
\equiv \Res{s,\vect n}(\sigma \parop \{\subst{s}{x}\})$, $\Sigma \vdash M'(\sigma \parop \{\subst{s}{x}\}) = N'(\sigma \parop \{\subst{s}{x}\})$, and $\{s,\vect n\} \cap (\fn(M') \cup \fn(N')) = \emptyset$, so $(M'=N')\frameof{\Res{s}(\{\subst{s}{x}\} \parop A)}$.

Conversely, if $(M'=N')\frameof{\Res{s}(\{\subst{s}{x}\} \parop A)}$, then
$\frameof{\Res{s}(\{\subst{s}{x}\} \parop A)} \equiv \Res{\vect n'}\sigma'$, $\Sigma \vdash M' \sigma' = N' \sigma'$, and $\{\vect n'\} \cap (\fn(M') \cup \fn(N')) = \emptyset$ for some $\vect n'$ and $\sigma'$.
We have $\frameof{A} \equiv \Res{\vect n}\sigma$ for some $\vect n$ and $\sigma$. We rename $\vect n$ so that $(\{s\} \cup \fn(N) \cup \fn(M)) \cap  \{ \vect n\} = \emptyset$, so $(\fn(N') \cup \fn(M')) \cap  \{ \vect n\} = \emptyset$. Then $\frameof{\Res{s}(\{\subst{s}{x}\} \parop A)} \equiv \Res{s,\vect n}(\sigma \parop \{\subst{s}{x}\})$, so
$\Res{\vect n'}\sigma' \equiv \Res{s,\vect n}(\sigma \parop \{\subst{s}{x}\})$.
By Lemma~\ref{lem:equivframes}, $\Sigma \vdash M'(\sigma \parop \{\subst{s}{x}\}) = N'(\sigma \parop \{\subst{s}{x}\})$,
so $\Sigma \vdash M\sigma = N\sigma$, hence $(M=N)\frameof{A}$.

Symmetrically, $(M=N)\frameof{B}$ if and only if $(M'=N')\frameof{\Res{s}(\{\subst{s}{x}\} \parop B)}$.
Moreover, $(M'=N')\frameof{\Res{s}(\{\subst{s}{x}\} \parop A)}$ if and only if $(M'=N')\frameof{\Res{s}(\{\subst{s}{x}\} \parop B)}$, because $\Res{s}(\{\subst{s}{x}\} \parop A) \enveq \Res{s}(\{\subst{s}{x}\} \parop B)$.
Therefore, $(M=N)\frameof{A}$ if and only if $(M=N)\frameof{B}$, so $A \enveq B$.

\item We suppose that $A \rel B$, $A \rightarrow A'$, and $A'$ is closed, and we show that
$B \rightarrow^* B'$ and $A' \rel B'$ for some $B'$. 
For some $x \notin \dom(A)$, we have $\Res{s}(\{\subst{s}{x}\} \parop A )
  \wkbisim  
  \Res{s}(\{\subst{s}{x}\} \parop B )$, $\Res{s}(\{\subst{s}{x}\} \parop A ) \rightarrow \Res{s}(\{\subst{s}{x}\} \parop A' )$,
and $\Res{s}(\{\subst{s}{x}\} \parop A' )$ is closed, so 
$\Res{s}(\{\subst{s}{x}\} \parop B ) \rightarrow^* B''$ and $\Res{s}(\{\subst{s}{x}\} \parop A' ) \wkbisim B''$ for some $B''$.

By Lemma~\ref{lem:red-inside-Res.s.s/x} applied several times, $B \rightarrow^* B'$ and 
$B'' \equiv \Res{s}(\{\subst{s}{x}\} \parop B' )$ for some closed extended process $B'$, so
 $\Res{s}(\{\subst{s}{x}\} \parop A' ) \wkbisim \Res{s}(\{\subst{s}{x}\} \parop B' )$,
which shows that $A' \rel B'$.

\item We suppose that $A \rel B$, $A \ltr{\alpha} A'$, $A'$ is closed, and $\fv(\alpha) \subseteq \dom(A)$, and we show that
$B \rightarrow^* \ltr{\alpha}\rightarrow^* B'$ and $A' \rel B'$ for some $B'$. 

For some $x \notin \dom(A)$, we have $\Res{s}(\{\subst{s}{x}\} \parop A )
  \wkbisim  
  \Res{s}(\{\subst{s}{x}\} \parop B )$.
First, we rename $x$ in this equivalence so that $x \notin \bv(\alpha)$, 
by Lemma~\ref{lem:ren-wkbisim}.

Let $\alpha' = \alpha\{\subst{x}{s}\}$.
By Lemma~\ref{lem:redalpha-std-to-pnf}, we have $\pnf(A) \ltrpnf{\alpha} A'$, so 
there exist $\vect n$, $\sigma$, $P$, $\alpha''$, and $A''$
such that
$\pnf(A) \equivpnf \Res{\vect n}(\sigma \parop P)$, 
$P \ltrP{\alpha''} A''$, $A' \equiv \Res{\vect n}(\sigma \parop A'')$, 
$\fv(\sigma) \cap \bv(\alpha'') = \emptyset$, 
$\Sigma \vdash \alpha \sigma = \alpha''$, and the elements of $\vect n$ 
do not occur in $\alpha$.
We rename $\vect n$ so that $s \notin \{\vect n\}$.
Since $A$ is closed, $\pnf(A)$ is closed, so by Lemma~\ref{lem:closing}\eqref{closing:equivP},
we can arrange that $\Res{\vect n}(\sigma \parop P)$ is also closed,
by substituting fresh names for its free variables.

We have $\Res{s}(\{\subst{s}{x}\} \parop A) \equiv \Res{s,\vect n}(\{\subst{s}{x}\} \parop \sigma \parop P)$,
so by Lemma~\ref{lem:struct-std-to-pnf}, $\pnf(\Res{s}(\{\subst{s}{x}\} \parop A)) \equivpnf \Res{s,\vect n}(\{\subst{s}{x}\} \parop \sigma \parop P)$ since $\Res{s,\vect n}(\{\subst{s}{x}\} \parop \sigma \parop P)$
is in partial normal form, because $x \notin \fv(P)$ and $x \notin \fv(\sigma)$, since $\Res{\vect n}(\sigma \parop P)$ is closed and $x \notin \dom(A) = \dom(\sigma)$.
Moreover, $P \ltrP{\alpha''} A''$, $\Res{s}(\{\subst{s}{x}\} \parop A') \equiv \Res{s,\vect n}(\{\subst{s}{x}\} \parop \sigma \parop A'')$, 
$\fv(\{\subst{s}{x}\} \parop \sigma) \cap \bv(\alpha'') = \emptyset$ because $x \notin \bv(\alpha'') = \bv(\alpha)$, 
$\Sigma \vdash \alpha' (\{\subst{s}{x}\} \parop \sigma) = \alpha \sigma = \alpha''$, and the elements of $s,\vect n$ 
do not occur in $\alpha'$.
Therefore, $\pnf(\Res{s}(\{\subst{s}{x}\} \parop A)) \ltrpnf{\alpha'} \Res{s}(\{\subst{s}{x}\} \parop A')$.

So $\Res{s}(\{\subst{s}{x}\} \parop A) \equiv \pnf(\Res{s}(\{\subst{s}{x}\} \parop A)) \ltr{\alpha'}\Res{s}(\{\subst{s}{x}\} \parop A')$ using Lemmas~\ref{lem:equivpnf} and~\ref{lem:redalpha-pnf-to-std}, so $\Res{s}(\{\subst{s}{x}\} \parop A) \ltr{\alpha'}\Res{s}(\{\subst{s}{x}\} \parop A')$ by \brn{Struct}.

Since $\Res{s}(\{\subst{s}{x}\} \parop A)\wkbisim \Res{s}(\{\subst{s}{x}\} \parop B)$, we have $\Res{s}(\{\subst{s}{x}\} \parop B) \rightarrow^* \ltr{\alpha'} \rightarrow^* B''$ and $\Res{s}(\{\subst{s}{x}\} \parop A') \wkbisim B''$ for some $B''$.
By Lemma~\ref{lem:red-inside-Res.s.s/x} applied several times and
Lemma~\ref{lem:ltr-inside-Res.s.s/x}, $B \rightarrow^* \ltr{\alpha}\rightarrow^* B'$ and $B'' \equiv \Res{s}(\{\subst{s}{x}\} \parop B' )$ for some closed extended process $B'$, so
 $\Res{s}(\{\subst{s}{x}\} \parop A' ) \wkbisim \Res{s}(\{\subst{s}{x}\} \parop B' )$,
which shows that $A' \rel B'$.

\end{enumerate}
Since $\rel$ is a labelled bisimulation and $\wkbisim$ is the largest labelled bisimulation, we have ${\rel} \subseteq {\wkbisim}$.
If $\Res{s}(\{\subst{s}{x}\} \parop A ) \wkbisim  
  \Res{s}(\{\subst{s}{x}\} \parop B )$, then $A \rel B$, so $A \wkbisim B$.
\end{proof}

\section{Proofs for Section~\ref{SUBSEC:REFINING}}\label{app:refining}

\begin{restate}{Lemma}{\ref{lem:not-so-new}}
The following three properties are equivalent:
\begin{enumerate}
\item the variables $\vect x$ resolve to $\vect M$ in $A$;
\item there exists $A'$ such that $A \equiv \smxvect \parop A'$;
\item $(\vect x = \vect M)\frameof{A}$
and the substitution $\smxvect$ is cycle-free.
\end{enumerate}
\end{restate}%
\begin{proof}
The implication from~\ref{equiv:resolve} to~\ref{equiv:caract-equiv} is immediate, with $A' = \nu \vect x.A$. 
The implication from~\ref{equiv:caract-equiv} to~\ref{equiv:caract-phi} is also obvious. Let us prove the implication from~\ref{equiv:caract-phi} to~\ref{equiv:resolve}. Since $(\vect x = \vect M)\frameof{A}$, we have $\{\vect x\} \subseteq \dom(\frameof{A}) = \dom(A)$, so $A \equiv \Res{\vect n}(\{\subst{\vect M'}{\vect x}\} \parop \sigma \parop P)$ for some $\vect n$, $\vect M'$, $\sigma$, and $P$ such that the variables of $\dom(A)$ do not occur in $\vect M'$, the image of $\sigma$, nor $P$. We rename $\vect n$ so that these names do not occur in $\vect M$.
Since $(\vect x = \vect M)\frameof{A}$, we have $\vect M' = \vect M\{\subst{\vect M'}{\vect x}\} \sigma = \vect M \smxvect \sigma$ using that $\smxvect$ is
cycle-free, 
so
$A \equiv \Res{\vect n}(\smxvect \parop \sigma \parop P)$. 
Since the names $\vect n$ do not occur in $\vect M$,
$A \equiv \smxvect \parop \Res{\vect n}(\sigma \parop P) \equiv \smxvect \parop \Res{\vect x}A$, which proves~\ref{equiv:resolve}.
\end{proof}

\begin{restate}{Lemma}{\ref{lem:output-correspondence}}
  $A \ltr{ \nu\vect{x}.\Snd{N}{M}} A'$ if and only if, for some $z$ that does
  not occur in any of $A$, $A'$, $\vect{x}$, $N$, and $M$,
  $A \ltr{\nu z.\Snd{N}{z} } \nu\vect{x}. (\{\subst{M}{z}\} \parop A')$,
  $\{\vect x\} \subseteq \fv(M) \setminus \fv(N)$,
  and the variables $\vect{x}$ are solvable in $\{\subst{M}{z}\} \parop A'$. 
\end{restate}%
\begin{proof}
We prove the implication from left to right by induction
on the derivation of $A \ltr{ \nu\vect{x}.\Snd{N}{M}} A'$.
Precisely, we prove the result for all $z$ that do not occur in the
derivation of $A \ltr{ \nu\vect{x}.\Snd{N}{M}} A'$.
\begin{itemize}
\item Case \brn{Out-Term}. We have $A = \Snd{N}{M}.P \ltr{\Snd{N}{M}} P = A'$ and
$\vect x$ is empty. Let $z \notin \fv(\Snd{N}{M}.P)$.
By \brn{Out-Var}, $A = \Snd{N}{M}.P \ltr{\Res{z}\Snd{N}{z}} P \parop \{\subst{M}{z}\} \equiv \{\subst{M}{z}\} \parop A'$, so by \brn{Struct},
$A \ltr{\Res{z}\Snd{N}{z}} \{\subst{M}{z}\} \parop A'$.

\item Case \brn{Open-Var}. The transition $A = \Res{\vect x}B \ltr{\Res{\vect x}\Snd{N}{M}} A'$ is derived from $B \ltr{\Snd{N}{M}} A'$ with $\{ \vect x\} \subseteq \fv(M) \setminus \fv(N)$ and $\vect x$ solvable in $\{\subst{M}{z'}\} \parop A'$ for some $z' \notin fv(A') \cup \{ \vect x \}$.
By induction hypothesis, $B \ltr{\Res{z}\Snd{N}{z}} \{\subst{M}{z}\} \parop A'$ for all $z$ that do not occur in the derivation of $B \ltr{\Snd{N}{M}} A'$, so $z$ does not occur in $A = \Res{\vect x}B \ltr{\Res{\vect x}\Snd{N}{M}} A'$ since $\{ \vect x\} \subseteq \fv(M)$.
By \brn{Scope}, $A = \Res{\vect x}B \ltr{\Res{z}\Snd{N}{z}} \Res{\vect x}(\{\subst{M}{z}\} \parop A')$, since $\{ \vect x\}\cap \fv(N) = \emptyset$.

\item Case \brn{Scope}. The transition $A = \Res{u}B \ltr{\Snd{N}{M}} \Res{u}B' = A'$ is derived from $B \ltr{\Snd{N}{M}} B'$, where $u$ does not occur in $\Snd{N}{M}$. (The restriction of the rule \brn{Scope} guarantees that $\vect x$ is empty.) By induction hypothesis, $B \ltr{\Res{z}\Snd{N}{z}} \{\subst{M}{z}\} \parop B'$ for all $z$ that do not occur in the derivation of $B \ltr{\Snd{N}{M}} B'$. Let $z$ be a variable that does not occur in the derivation of $A = \Res{u}B \ltr{\Snd{N}{M}} \Res{u}B' = A'$. Since the derivation of $A = \Res{u}B \ltr{\Snd{N}{M}} \Res{u}B' = A'$ includes the derivation of $B \ltr{\Snd{N}{M}} B'$, $z$ does not occur in the derivation of $B \ltr{\Snd{N}{M}} B'$. Hence, we have $B \ltr{\Res{z}\Snd{N}{z}} \{\subst{M}{z}\} \parop B'$, so by \brn{Scope}, $A = \Res{u}B \ltr{\Res{z}\Snd{N}{z}} \Res{u}(\{\subst{M}{z}\} \parop B')$, since $u$ does not occur in $\Res{z}\Snd{N}{z}$. Moreover, $\Res{u}(\{\subst{M}{z}\} \parop B') \equiv \{\subst{M}{z}\} \parop \Res{u}B' = \{\subst{M}{z}\} \parop A'$ since $u$ does not occur in $\{\subst{M}{z}\}$. So by \brn{Struct}, $A \ltr{\Res{z}\Snd{N}{z}} \{\subst{M}{z}\} \parop A'$.

\item Case \brn{Par}. The transition $A = B \parop C \ltr{\Res{\vect x}\Snd{N}{M}} B' \parop C = A'$ is derived from $B \ltr{\Res{\vect x}\Snd{N}{M}} B'$, with $\{ \vect x \} \cap \fv(C) = \emptyset$. By induction hypothesis, $B \ltr{\Res{z}\Snd{N}{z}} \Res{\vect x}(\{\subst{M}{z}\} \parop B')$, $\{\vect x\} \subseteq \fv(M) \setminus \fv(N)$, and the variables $\vect x$ are solvable in $\{\subst{M}{z}\} \parop B'$, for all $z$ that do not occur in the derivation of $B \ltr{\Snd{N}{M}} B'$. Let $z$ be a variable that does not occur in the derivation of $A = B \parop C \ltr{\Res{\vect x}\Snd{N}{M}} B' \parop C = A'$. Since the derivation of $A = B \parop C \ltr{\Res{\vect x}\Snd{N}{M}} B' \parop C = A'$ includes the derivation of $B \ltr{\Snd{N}{M}} B'$, $z$ does not occur in the derivation of $B \ltr{\Snd{N}{M}} B'$. Hence, we have $B \ltr{\Res{z}\Snd{N}{z}} \Res{\vect x}(\{\subst{M}{z}\} \parop B')$, so by \brn{Par}, $B \parop C \ltr{\Res{z}\Snd{N}{z}} \Res{\vect x}(\{\subst{M}{z}\} \parop B') \parop C$, since $z \notin \fv(C)$. Moreover, $\Res{\vect x}(\{\subst{M}{z}\} \parop B') \parop C \equiv \Res{\vect x}(\{\subst{M}{z}\} \parop (B' \parop C)) = \Res{\vect x}(\{\subst{M}{z}\} \parop A')$ since $\{ \vect x \} \cap \fv(C) = \emptyset$, so by \brn{Struct}, $A \ltr{\Res{z}\Snd{N}{z}} \Res{\vect x}(\{\subst{M}{z}\} \parop A')$. Moreover, the variables $\vect x$ are solvable in $\{\subst{M}{z}\} \parop A'$: assuming that the variables $\vect x$ resolve to $\vect M$ in $\{\subst{M}{z}\} \parop B'$, we have
\begin{align*}
\{\subst{\vect M}{\vect x}\} \parop \Res{\vect x}(\{\subst{M}{z}\} \parop A')
&\equiv \{\subst{\vect M}{\vect x}\} \parop \Res{\vect x}(\{\subst{M}{z}\} \parop (B' \parop C))\\
&\equiv \{\subst{\vect M}{\vect x}\} \parop \Res{\vect x}(\{\subst{M}{z}\} \parop B') \parop C\tag*{since $\{ \vect x \} \cap \fv(C) = \emptyset$}\\
&\equiv \{\subst{M}{z}\} \parop B' \parop C
\tag*{since $\vect x$ resolve to $\vect M$ in $\{\subst{M}{z}\} \parop B'$}\\
&\equiv \{\subst{M}{z}\} \parop A'
\end{align*}

\item Case \brn{Struct}. The transition $A \ltr{\Res{\vect x}\Snd{N}{M}} A'$ is derived from $B \ltr{\Res{\vect x}\Snd{N}{M}} B'$, $A \equiv B$ and $A' \equiv B'$. By induction hypothesis, $B \ltr{\Res{z}\Snd{N}{z}} \Res{\vect x}(\{\subst{M}{z}\} \parop B')$, $\{\vect x\} \subseteq \fv(M) \setminus \fv(N)$, and the variables $\vect x$ are solvable in $\{\subst{M}{z}\} \parop B'$, for all $z$ that do not occur in the derivation of $B \ltr{\Snd{N}{M}} B'$.
  Let $z$ be a variable that does not occur in the derivation of $A \ltr{\Res{\vect x}\Snd{N}{M}} A'$. Since the derivation of $A \ltr{\Res{\vect x}\Snd{N}{M}} A'$ includes the derivation of $B \ltr{\Snd{N}{M}} B'$, $z$ does not occur in the derivation of $B \ltr{\Snd{N}{M}} B'$. Hence, we have $B \ltr{\Res{z}\Snd{N}{z}} \Res{\vect x}(\{\subst{M}{z}\} \parop B')$ and $\Res{\vect x}(\{\subst{M}{z}\} \parop B')\equiv \Res{\vect x}(\{\subst{M}{z}\} \parop A')$, so by \brn{Struct}, 
$A \ltr{\Res{z}\Snd{N}{z}} \Res{\vect x}(\{\subst{M}{z}\} \parop A')$.
Moreover, the variables $\vect x$ are solvable in $\{\subst{M}{z}\} \parop B'$ and $\{\subst{M}{z}\} \parop B' \equiv \{\subst{M}{z}\} \parop A'$, 
so by Definition~\ref{defn:derived}, the variables $\vect x$ are solvable in $\{\subst{M}{z}\} \parop A'$.

\end{itemize}

Let us now prove the implication from right to left. For this proof, we use the notion of partial normal form introduced in Appendix~\ref{app:bigpfpnf}.
We have $A \ltr{ \nu z.\Snd{N}{z} } \nu\vect{x}. (\{\subst{M}{z}\} \parop A')$ where the
  variables $\vect{x}$ are solvable in $\{\subst{M}{z}\} \parop  A'$,
$\{\vect x\} \subseteq \fv(M) \setminus \fv(N)$, and $z$ does not occur in
$A$, $A'$, $\vect x$, $N$, $M$.
By Lemma~\ref{lem:redalpha-std-to-pnf}, we have $\pnf(A) \ltrpnf{ \nu z.\Snd{N}{z} } \nu\vect{x}. (\{\subst{M}{z}\} \parop A')$.
By definition of $\ltrpnf{\nu z.\Snd{N}{z}}$, we have 
$\pnf(A) \equiv \Res{\vect n}(\sigma \parop P)$, 
$P \ltrP{\Res{z}\Snd{N'}{z}} B'$,
$\Res{\vect x}(\{\subst{M}{z}\} \parop A') \equiv \Res{\vect n}(\sigma \parop B')$,
$z \notin \fv(\sigma)$, $\Sigma \vdash N\sigma = N'$,
and the elements of $\vect n$ do not occur in $N$,
for some $\vect n$, $\sigma$, $P$, $N'$, $B'$.
By Lemma~\ref{lem:caract-ltrP}, we have 
$P \equivP \Res{\vect n'}(\Snd{N'}{M'}.P_1 \parop P_2)$,
$B' \equiv \Res{\vect n'}(P_1 \parop \{\subst{M'}{z}\} \parop P_2)$,
$\{\vect n'\} \cap \fn(N') = \emptyset$,
$z \notin \fv(P_1 \parop P_2)$
for some $\vect n'$, $P_1$, $P_2$, $N'$, $M'$.
Hence, we have
\begin{align*}
&A \equiv \Res{\vect n}(\sigma \parop \Res{\vect n'}(\Snd{N'}{M'}.P_1 \parop P_2))\\
&\Res{\vect x}(\{\subst{M}{z}\} \parop A') \equiv \Res{\vect n}(\sigma \parop \Res{\vect n'}(P_1 \parop \{\subst{M'}{z}\} \parop P_2))
\end{align*}
We rename the names in $\vect n'$ so that they do not occur in $\sigma$ nor in $N$. Then
\begin{align*}
&A \equiv \Res{\vect n, \vect n'}(\sigma \parop \Snd{N'}{M'}.P_1 \parop P_2)\\
&\Res{\vect x}(\{\subst{M}{z}\} \parop A') \equiv \Res{\vect n,\vect n'}(\sigma \parop P_1 \parop \{\subst{M'}{z}\} \parop P_2)
\end{align*}
We instantiate the variables using $\sigma$, so that the variables of $\dom(\sigma)$ do not occur in the image of $\sigma$ nor in $N', M', P_1, P_2$.
Furthermore, let $\sigma'$ be a substitution that maps $\vect x$ to distinct fresh names. By Lemma~\ref{lem:struct-std-to-pnf}, 
\[\pnf(\Res{\vect x}(\{\subst{M}{z}\} \parop A')) \equivpnf \Res{\vect n,\vect n'}((\sigma \parop \{\subst{M'}{z}\}) \parop (P_1  \parop P_2))\]
Moreover, $\Sigma \vdash \pnf(\Res{\vect x}(\{\subst{M}{z}\} \parop A')) = \pnf(\Res{\vect x}(\{\subst{M}{z}\} \parop A'))\sigma'$
because $\{\vect x\} \cap \fv(\pnf(\Res{\vect x}(\{\subst{M}{z}\} \parop A'))) = \emptyset$.
Therefore, by Lemma~\ref{lem:inv-subst},
$\Sigma \vdash \Res{\vect n,\vect n'}((\sigma \parop \{\subst{M'}{z}\}) \parop (P_1  \parop P_2)) = \Res{\vect n,\vect n'}((\sigma \parop \{\subst{M'}{z}\}) \parop (P_1  \parop P_2))\sigma'$, so
$\Sigma \vdash M' = M'\sigma'$, $\Sigma \vdash \sigma = \sigma\sigma'$, $\Sigma \vdash P_1 = P_1 \sigma'$, and $\Sigma \vdash P_2 = P_2 \sigma'$, so by replacing $\sigma$ with $\sigma\sigma'$, $M'$ with $M'\sigma'$, $P_1$ with $P_1 \sigma'$, and $P_2$ with $P_2\sigma'$, we obtain 
\begin{align*}
&A \equiv \Res{\vect n, \vect n'}(\sigma \parop \Snd{N}{M'}.P_1 \parop P_2)\\
&\Res{\vect x}(\{\subst{M}{z}\} \parop A') \equiv \Res{\vect n,\vect n'}(\sigma \parop P_1 \parop \{\subst{M'}{z}\} \parop P_2)
\end{align*}
and the variables $\vect x$ are not free in the right-hand sides of these equivalences.

The variables $\vect x$ resolve to some $\vect M$ in $\{\subst{M}{z}\} \parop A'$, so
\[
\{\subst{M}{z}\} \parop A' \equiv \{\subst{\vect M}{\vect x}\} \parop \Res{\vect x}(\{\subst{M}{z}\} \parop A') \equiv \{\subst{\vect M}{\vect x}\} \parop \Res{\vect n,\vect n'}(\sigma \parop P_1 \parop \{\subst{M'}{z}\} \parop P_2)
\]
We rename the names $\vect n, \vect n'$ so that they do not occur in $\vect M$.
Hence
\begin{align*}
&\{\subst{M}{z}\} \parop A' \equiv \Res{\vect n,\vect n'}(\sigma  \parop \{\subst{M'}{z}\} \parop \{\subst{\vect M}{\vect x}\} \parop P_1 \parop P_2)\\
&A' \equiv \Res{z}(\{\subst{M}{z}\} \parop A') \equiv 
\Res{z,\vect n,\vect n'}(\sigma \parop \{\subst{M'}{z}\} \parop \{\subst{\vect M}{\vect x}\} \parop P_1 \parop P_2)
\end{align*}
By Lemma~\ref{lem:not-so-new}, $(z = M)\frameof{\{\subst{M}{z}\} \parop A'}$,
so $(z = M) \Res{\vect n, \vect n'}(\sigma \parop \{\subst{M'}{z}\} \parop \{\subst{\vect M}{\vect x}\})$.
We rename the names $\vect n, \vect n'$ so that they do not occur in $M$.
Therefore, 
\begin{align*}
A &\equiv \Res{\vect x, z, \vect n, \vect n'}(\sigma \parop  \{\subst{M'}{z}\} \parop \{\subst{\vect M}{\vect x}\} \parop \Snd{N}{z}.P_1 \parop P_2)\\
&\equiv \Res{\vect x, z, \vect n, \vect n'}(\sigma \parop  \{\subst{M'}{z}\} \parop \{\subst{\vect M}{\vect x}\} \parop \Snd{N}{M}.P_1 \parop P_2)
\end{align*}
So we derive
{\allowdisplaybreaks\begin{align*}
& \Snd{N}{M}.P_1 \ltr{\Snd{N}{M}} P_1 \tag*{by \brn{Out-Term}}\\
& \Snd{N}{M}.P_1 \parop \sigma \parop  \{\subst{M'}{z}\} \parop \{\subst{\vect M}{\vect x}\} \parop P_2 \ltr{\Snd{N}{M}} 
P_1 \parop \sigma \parop  \{\subst{M'}{z}\} \parop \{\subst{\vect M}{\vect x}\} \parop P_2\tag*{by \brn{Par}}\\
&\Res{z,\vect n,\vect n'}(\Snd{N}{M}.P_1 \parop \sigma \parop  \{\subst{M'}{z}\} \parop \{\subst{\vect M}{\vect x}\} \parop P_2 ) \ltr{\Snd{N}{M}} 
\Res{z,\vect n,\vect n'}(P_1 \parop \sigma \parop  \{\subst{M'}{z}\} \parop \{\subst{\vect M}{\vect x}\} \parop P_2)
\tag*{by \brn{Scope}, since $z,\vect n,\vect n'$ do not occur in $\Snd{N}{M}$}\\
&\Res{z,\vect n,\vect n'}(\sigma \parop  \{\subst{M'}{z}\} \parop \{\subst{\vect M}{\vect x}\} \parop \Snd{N}{M}.P_1 \parop P_2)
\ltr{\Snd{N}{M}} A'
\tag*{by \brn{Struct}}\\
&\Res{\vect x, z, \vect n, \vect n'}(\sigma \parop  \{\subst{M'}{z}\} \parop \{\subst{\vect M}{\vect x}\} \parop \Snd{N}{M}.P_1 \parop P_2) \ltr{\Res{\vect x}\Snd{N}{M}} A'
\tag*{by \brn{Open-Var}, since $\{\vect x\} \subseteq \fv(M) \setminus \fv(N)$}\\*[-0.5mm]
&\tag*{and the variables $\vect x$ are solvable in $\{\subst{M}{z}\}\parop A'$}\\
&A \ltr{\Res{\vect x}\Snd{N}{M}} A'\tag*{by \brn{Struct}}
\end{align*}}%
\end{proof}

\begin{restate}{Lemma}{\ref{lem:single-var-output-correspondence}}
$A \ltr{ \nu x.\Snd{N}{x}} A'$ in the refined semantics 
if and only if 
$A \ltr{ \nu x.\Snd{N}{x}} A'$ in the simple semantics.
\end{restate}%
\begin{proof}
Suppose that $A \ltr{ \nu x.\Snd{N}{x}} A'$
in the refined semantics.
By Lemma~\ref{lem:output-correspondence}, for some variable $z$ that 
does not occur in this transition, we have $A \ltr{ \nu
    z.\Snd{N}{z} } \nu x. (\{\subst{x}{z}\} \parop A')$
in the simple semantics.
Since $x \in \dom(A')$, $A' \equiv \Res{\vect n}(\{\subst{M}{x}\} \parop A'')$
for some $\vect n$ and some $M$ and $A''$ that do not contain $x$ nor $z$, so
\[\nu x. (\{\subst{x}{z}\} \parop A') \equiv \nu x. (\{\subst{x}{z}\} \parop \Res{\vect n}(\{\subst{M}{x}\} \parop A'')) \equiv \Res{\vect n}(\{\subst{M}{z}\} \parop A'')\]
Hence by \brn{Struct}, $A \ltr{ \nu z.\Snd{N}{z} } \Res{\vect n}(\{\subst{M}{z}\} \parop A'')$.
By renaming $z$ into $x$ and $x$ into $z$ everywhere in the derivation of this transition, we obtain
$A \ltr{ \nu x.\Snd{N}{x} } \Res{\vect n}(\{\subst{M}{x}\} \parop A'')$,
since $z$ and $x$ are not free in $A$, $N$, $A''$, $M$.
Since we have $\Res{\vect n}(\{\subst{M}{x}\} \parop A'') \equiv A'$, 
we obtain $A \ltr{ \nu x.\Snd{N}{x} } A'$ by \brn{Struct} 
in the simple semantics.

Conversely, suppose that $A \ltr{ \nu x.\Snd{N}{x} } A'$
in the simple semantics.
Since $x \in \dom(A')$, $A' \equiv \Res{\vect n}(\{\subst{M}{x}\} \parop A'')$
for some $\vect n$ and some $M$ and $A''$ that do not contain $x$, so
by \brn{Struct}, $A \ltr{ \nu x.\Snd{N}{x} } \Res{\vect n}(\{\subst{M}{x}\} \parop A'')$.
By renaming $x$ into a fresh variable $z$ everywhere in the derivation of this transition,
$A \ltr{ \nu z.\Snd{N}{z} } \Res{\vect n}(\{\subst{M}{z}\} \parop A'')$,
since $x$ is not free in $A$, $M$, $A''$.
Moreover, $\Res{\vect n}(\{\subst{M}{z}\} \parop A'') \equiv \Res{x}(\{\subst{x}{z}\} \parop \Res{\vect n}(\{\subst{M}{x}\} \parop A'')) \equiv \Res{x}(\{\subst{x}{z}\} \parop A')$, so by \brn{Struct}, we obtain
$A \ltr{ \nu z.\Snd{N}{z} } \Res{x}(\{\subst{x}{z}\} \parop A')$.

The variable $x$ resolves to $z$ in $\{\subst{x}{z}\} \parop A'$, because
\begin{align*}
\{\subst{z}{x}\} \parop \Res{x}(\{\subst{x}{z}\} \parop A')
&\equiv \{\subst{z}{x}\} \parop \Res{x}(\{\subst{x}{z}\} \parop \Res{\vect n}(\{\subst{M}{x}\} \parop A''))\\
&\equiv \{\subst{z}{x}\} \parop \Res{\vect n}(\{\subst{M}{z}\} \parop A'')\\
&\equiv \Res{\vect n}(\{\subst{z}{x}\} \parop \{\subst{M}{z}\} \parop A'')\\
&\equiv \Res{\vect n}(\{\subst{M}{x}\} \parop \{\subst{x}{z}\} \parop A'')\\
&\equiv \{\subst{x}{z}\} \parop \Res{\vect n}(\{\subst{M}{x}\} \parop A'')\\
&\equiv \{\subst{x}{z}\} \parop A'
\end{align*}
Therefore, by Lemma~\ref{lem:output-correspondence},
$A \ltr{ \nu x.\Snd{N}{x} } A'$
in the refined semantics.
\end{proof}

\begin{restate}{Theorem}{\ref{thm:relate-lts}}
  Let $\altbisim$ be the relation of labelled
  bisimilarity obtained by applying Definition~\ref{def:wkbisim} to
  the refined semantics. We have ${\wkbisim} = {\altbisim}$.
\end{restate}%
\begin{proof}
By Lemma~\ref{lem:single-var-output-correspondence}, 
$\altbisim$ is a simple-labelled bisimulation, and thus ${\altbisim} \subseteq {\wkbisim}$. 
Conversely, to show that $\wkbisim$ is a refined-labelled bisimulation,
it suffices to prove its bisimulation property for any refined output label.

Assume $A \wkbisim B$, $A \ltr{ \nu\vect{x}.\Snd{N}{M}} A'$, $A'$ is closed,
and $\fv(\nu\vect{x}.\Snd{N}{M}) \subseteq \dom(A)$.
By Lemma~\ref{lem:output-correspondence}, 
we have 
\[A \ltr{ \nu z.\Snd{N}{z} } A^\circ = \nu\vect{x}.(\{\subst{M}{z}\} \parop A')\]
for some fresh variable~$z$, 
where $\{\vect x\} \subseteq \fv(M) \setminus \fv(N)$ and $\vect{x}$ resolves to $\vect M$
in $\{\subst{M}{z}\} \parop A'$:
\begin{equation}
\{\subst{M}{z}\} \parop A' 
\equiv 
\{\subst{\vect M}{\vect x}\} \parop \nu{\vect{x}}.(\{\subst{M}{z}\} \parop A') 
\equiv 
\{\subst{\vect M}{\vect x}\} \parop A^\circ
\label{eq:struct-equiv1}\end{equation}
Let $\CTX[\_] = \nu z.( \{\subst{\vect M}{\vect x}\} \parop \_ )$.  Using
the structural equivalence above and structural rearrangements, we
obtain $\CTX[A^\circ] \equiv \nu z.( \{\subst{M}{z}\} \parop A') \equiv A'$.
By labelled bisimulation hypothesis on the simple output transition above,
we have $B \rightarrow^* B_1 \ltr{\nu z.\Snd{N}{z}} B_2 \rightarrow^* B^\circ$ with $A^\circ \wkbisim B^\circ$ for some $B_1, B_2, B^\circ$.
By instantiating all variables in $\fv(B_2) \setminus \dom(B_2)$ with 
fresh names in the derivation of this reduction, we obtain the same property
and additionally $B_2$ is closed.
By Theorem~\ref{THM:OBSERVATIONAL-LABELED}, labelled bisimilarity is closed 
by application of closing contexts. Using $\CTX[\_]$, we obtain $A' \wkbisim \CTX[B^\circ]$. Let $B' = \CTX[B^\circ]$.

Let us first show that $B_2 \equiv \Res{\vect x}(\{\subst{M}{z}\} \parop \CTX[B_2])$.
By Lemma~\ref{lem:not-so-new}, we have $(z = M)\frameof{\{\subst{M}{z}\} \parop A'}$
and by the structural equivalence~\eqref{eq:struct-equiv1},
$(\vect x = \vect M)\frameof{\{\subst{M}{z}\} \parop A'}$,
so $(z = M\smxvect)\frameof{\{\subst{M}{z}\} \parop A'}$,
so $(z = M\smxvect)\frameof{\Res{\vect x}(\{\subst{M}{z}\} \parop A')}$
since the variables $\vect x$ do not occur in $M\smxvect$.
Hence $(z = M\smxvect)\frameof{A^\circ}$.
Since $A^\circ \wkbisim B^\circ \enveq B_2$, we have $A^\circ \enveq B_2$,
so $(z = M\smxvect)\frameof{B_2}$.
Since $z \in \dom(B_2)$, 
we have $B_2 \equiv \Res{\vect n}(\{\subst{N}{z}\} \parop B_3)$
for some $\vect n$, $N$ and $B_3$ such that $z$ is not free in $B_3$.
We rename $\vect n$ so that these names do not occur in $\vect M$ nor in $M$.
Then $\CTX[B_2] \equiv \Res{\vect n}(\{\subst{\vect M\{\subst{N}{z}\}}{\vect x}\} \parop B_3)$, so
\[\Res{\vect x}(\{\subst{M}{z}\} \parop \CTX[B_2]) \equiv
\Res{\vect n}(\{\subst{M\{\subst{\vect M\{\subst{N}{z}\}}{\vect x}\}}{z}\} \parop B_3) \equiv
\Res{\vect n}(\{\subst{N}{z}\} \parop B_3) \equiv B_2\]
because $(N = M\{\subst{\vect M\{\subst{N}{z}\}}{\vect x}\}) \frameof{B_3}$ since $(z = M\smxvect)\frameof{B_2}$.
So we have the desired structural equivalence $B_2 \equiv \Res{\vect x}(\{\subst{M}{z}\} \parop \CTX[B_2])$.

Then 
\[B_1 \ltr{\nu z.\Snd{N}{z}} \Res{\vect x}(\{\subst{M}{z}\} \parop \CTX[B_2])\]
Moreover, $\vect{x}$ resolves to $\vect M$
in $\{\subst{M}{z}\} \parop A'$ and 
\[\{\subst{M}{z}\} \parop A' \equiv \{\subst{M}{z}\} \parop \CTX[A^\circ]
\wkbisim \{\subst{M}{z}\} \parop \CTX[B^\circ] \enveq \{\subst{M}{z}\} \parop \CTX[B_2]\]
so  $\{\subst{M}{z}\} \parop A' \enveq \{\subst{M}{z}\} \parop \CTX[B_2]$,
so by Lemma~\ref{lem:enveq-preserves-resolve},
$\vect{x}$ resolves to $\vect M$ in $\{\subst{M}{z}\} \parop \CTX[B_2]$.
Hence, by Lemma~\ref{lem:output-correspondence}, $B_1 \ltr{\Res{\vect x}\Snd{N}{M}} \CTX[B_2]$.
Hence 
\[B \rightarrow^* B_1 \ltr{\Res{\vect x}\Snd{N}{M}} \CTX[B_2] \rightarrow^* \CTX[B^{\circ}] = B'\] 
so we have $A' \wkbisim B'$ and $B  \rightarrow^* \ltr{\Res{\vect x}\Snd{N}{M}} \rightarrow^* B'$, which concludes the proof.
\end{proof}

\section{Proofs for Section~\ref{SUBSEC:CONSTRUCTING}}\label{app:constructing}

In this appendix, we suppose that the signature $\Sigma$ satisfies the 
assumptions of Theorem~\ref{th:mac} and write $\rew$ for its convergent 
rewrite system.
In particular, since $\rew$ terminates, the left-hand sides of its rewrite rules 
cannot be variables. 

\newcommand{\onlykey}{$k$ occurs only as $\Const{mac}(k,\cdot)$}

In preparation for the proof of Theorem~\ref{th:mac}, we study the effect of the translation $\tr{\cdot}$ 
on the semantics of terms and processes where {\onlykey},
relying on the partial normal forms defined in Appendix~\ref{app:bigpfpnf}.

\begin{lemma}\label{lem:h-inj}
$\Sigma \vdash M_1 = M_2$ if and only if 
$\Sigma \vdash \hbin{k}{M_1} = \hbin{k}{M_2}$.
\end{lemma}
\begin{proof}
The implication from left to right is obvious. Conversely, suppose that $\Sigma \vdash \hbin{k}{M_1} = \hbin{k}{M_2}$. Let $M'_1$ and $M'_2$ be the normal forms under $\rew$ of $M_1$ and $M_2$ respectively. Hence $\Sigma \vdash \hbin{k}{M'_1} = \hbin{k}{M'_2}$. 
For some $n, n' \geq 1$, 
we have $M'_1 = \cons{N_1}{\cons{\dots}{N_n}}$
and $M'_2 = \cons{N'_1}{\cons{\dots}{N'_{n'}}}$ where the root symbols of 
$N_n$ and $N'_{n'}$ are not $::$. 
We compute the normal form of $\hbin{k}{M'_1}$:
\begin{itemize}
\item If $n = 1$, then $\hbin{k}{M'_1} = \hbin{k}{N_1}$ is irreducible since 
$N_1$ is irreducible and the rewrite rules with $\Const{h}$ at the root 
of the left-hand side apply only to terms with $::$ at the root.

\item If $n > 1$ and $N_n = \Const{nil}$, then
$\hbin{k}{M'_1}$ reduces to $\fbin{\dots(\fbin{k}{N_1}, \dots)}{N_{n-1}}$,
and this term is irreducible since $N_1, \ldots, N_{n-1}$ are irreducible
as subterms of an irreducible term,
and no rewrite rule contains $\Const{f}$ in its left-hand side.

\item If $n > 1$ and $N_n \neq \Const{nil}$, then 
$\hbin{k}{M'_1}$ reduces to 
$\Const{h}(\Const{f}(\dots(\fbin{k}{N_1}, \dots), \allowbreak N_{n-2}), \allowbreak \cons{N_{n-1}}{N_n})$
and this term is irreducible since $N_1, \ldots, N_n$ are irreducible,
no rewrite rule contains $\Const{f}$ in its left-hand side,
and no rewrite rule with $\Const{h}$ at the root applies since
$N_n$ is not $\Const{nil}$ and does not contain $::$ at the root.
\end{itemize}
and similarly compute a normal form of $\hbin{k}{M'_2}$.
Their equality implies 
$n = n'$ and $N_i = N'_i$ for all $i \leq n$,
hence $M'_1 = M'_2$ and $\Sigma \vdash M_1 = M_2$.
\end{proof}

\begin{lemma}\label{lem:tr-red}
If $M_1 \rightarrow_{\rew} M_2$ and {\onlykey} in $M_1$,
then $\tr{M_1} \rightarrow_{\rew} \tr{M_2}$ and {\onlykey} in $M_2$.
\end{lemma}
\begin{proof}
We have $M_1 = C[M_3\sigma]$ and $M_2 = C[M_4\sigma]$
for some rewrite rule $M_3 \rightarrow M_4$ of $\rew$, term context $C$, 
and substitution $\sigma$.
Hence $\tr{M_1} = \tr{C[M_3\sigma]}$.
Furthermore, $\Const{mac}$ and $k$ do not occur in $M_3$ and $M_3$ 
is not a variable, so $\tr{C[M_3\sigma]} = \tr{C}[M_3\tr{\sigma}]$.
Since {\onlykey} in $M_1$ and $\Const{mac}$ does 
not occur in $M_3$, 
{\onlykey} in $C$ and in the image of $\sigma$.
Furthermore, $k$ does not occur in $M_4$.
Therefore, {\onlykey} in $M_2 = C[M_4\sigma]$
and $\tr{M_2} = \tr{C[M_4\sigma]} = \tr{C}[M_4\tr{\sigma}]$.
We can then conclude that $\tr{M_1} \rightarrow_{\rew} \tr{M_2}$.
\end{proof}

\begin{lemma}\label{lem:tr-inj}
Suppose that {\onlykey} in $M_1$ and $M_2$.
We have $\Sigma \vdash M_1 = M_2$ if and only if 
$\Sigma \vdash \tr{M_1} = \tr{M_2}$.
\end{lemma}
\begin{proof}
Let us first prove the implication from left to right.
If $\Sigma \vdash M_1 = M_2$, then $M_1 \rightarrow_{\rew}^* M'$ and
$M_2 \rightarrow_{\rew}^* M'$ for some $M'$.
By Lemma~\ref{lem:tr-red}, $\tr{M_1} \rightarrow_{\rew}^* \tr{M'}$ and
$\tr{M_2} \rightarrow_{\rew}^* \tr{M'}$,
so $\Sigma \vdash \tr{M_1} = \tr{M_2}$.

Conversely, suppose that $\Sigma \vdash \tr{M_1} = \tr{M_2}$.
Let $M'_1$ and $M'_2$ be the normal forms under $\rew$ of $M_1$ and $M_2$, respectively. By Lemma~\ref{lem:tr-red}, {\onlykey} in $M'_1$ and $M'_2$, $\tr{M_1} \rightarrow_{\rew}^* \tr{M'_1}$, and $\tr{M_2} \rightarrow_{\rew}^* \tr{M'_2}$, so $\Sigma \vdash \tr{M'_1} = \tr{M'_2}$. 
We show by induction on the total size of the terms $M'_1$ and $M'_2$ that, 
if {\onlykey} in $M'_1$ and $M'_2$, $M'_1$ and $M'_2$ are irreducible under $\rew$, and $\Sigma \vdash \tr{M'_1} = \tr{M'_2}$, then $M'_1 = M'_2$:
\begin{itemize}

\item First suppose that $M'_1$ and $M'_2$ are not of the form $\Const{mac}(k, \cdot)$. 

Since $\Const{f}$ does not occur on the left-hand sides of rewrite rules of $\rew$, 
if a rewrite rule of $\rew$ could be applied at the root of $\tr{M'_1}$ or $\tr{M'_2}$, then it would match only symbols above occurrences of $\Const{f}(\dots)$ in  $\tr{M'_1}$ or $\tr{M'_2}$, hence, only symbols above occurrences of $\Const{mac}(k,\cdot)$ in $M'_1$ or $M'_2$. Moreover, by induction hypothesis, if subterms of $\tr{M'_1}$ or $\tr{M'_2}$ are equal, the corresponding subterms of $M'_1$ or $M'_2$ are also equal. Hence, the same rewrite rule would also apply at the root of $M'_1$ or $M'_2$, which is impossible since $M'_1$ and $M'_2$ are irreducible.

Hence, the equality $\Sigma \vdash \tr{M'_1} = \tr{M'_2}$ is equivalent to the
equality between the immediate subterms of $\tr{M'_1}$ and $\tr{M'_2}$,
and we conclude by induction.

\item Now suppose that $M'_1 = \Const{mac}(k, M''_1)$ and $M'_2 = \Const{mac}(k, M''_2)$.
Then $\tr{M'_1} = \fbin{k}{\hbin{k}{\tr{M''_1}}}$
and $\tr{M'_2} = \fbin{k}{\hbin{k}{\tr{M''_2}}}$.
Since $\Sigma \vdash \tr{M'_1} = \tr{M'_2}$, we have $\Sigma \vdash \hbin{k}{\tr{M''_1}} = \hbin{k}{\tr{M''_2}}$ since no rewrite rule applies to $\Const{f}(\cdot,\cdot)$, so by Lemma~\ref{lem:h-inj}, $\Sigma \vdash \tr{M''_1} = \tr{M''_2}$. By induction hypothesis, $M''_1 = M''_2$, so $M'_1 = M'_2$.

\item Finally, if $M'_1 = \Const{mac}(k, M''_1)$ and $M'_2$ is not of the form $\Const{mac}(k, \cdot)$, then $\tr{M'_1} =  \Const{f}(k, \allowbreak \Const{h}(k, \allowbreak \tr{M''_1}))$ and $\tr{M'_2}$ is not of the form $\fbin{k}{\cdot}$ because {\onlykey} in $M'_2$, so $\Sigma \vdash \tr{M'_1} \neq \tr{M'_2}$: this case cannot happen.
Symmetrically, the case $M'_2 = \Const{mac}(k, M''_2)$ and $M'_1$ is not of the form $\Const{mac}(k, \cdot)$ cannot happen.

\end{itemize}
From this result, we easily conclude that $\Sigma \vdash M_1 = M_2$.
\end{proof}

\begin{lemma}\label{lem:ltrP-tr1}
Suppose that $P_0$ is closed, $\alpha = \Res{x}\Snd{N'}{x}$ or $\alpha = \Rcv{N'}{M'}$ for some ground term $N'$, and {\onlykey} in $P_0$ and $\alpha$.

If $P_0 \ltrP{\alpha} A$, then $\tr{P_0} \ltrP{\tr{\alpha}} \tr{A'}$ and
$A \equiv A'$ for some $A'$ where {\onlykey} and, moreover,
\begin{itemize}
\item 
when $\alpha = \Res{x}\Snd{N'}{x}$, $A' = \CTX[\{\subst{M}{x}\}]$ 
where $\CTX$ is a closed plain evaluation context (with no active
substitutions and no variable restrictions) and $M$ is a ground term;

\item when $\alpha = \Rcv{N'}{M'}$, $A'$ is a plain process with
$\fv(A') \subseteq \fv(M')$.
\end{itemize}
\end{lemma}
\begin{proof}
We proceed by induction on the syntax of $P_0$ and apply Lemma~\ref{lem:decomp-ltrP}
to decompose $P_0 \ltrP{\alpha} A$, with the following cases:
\begin{enumerate}
\item $P_0 = P \parop Q$ and either $P \ltrP{\alpha} A'$ and $A \equiv A' \parop Q$, or $Q \ltrP{\alpha} A'$ and $A \equiv P \parop A'$, for some $P$, $Q$, and $A'$. 
In the first case, by induction hypothesis, 
$\tr{P} \ltrP{\tr{\alpha}} \tr{A''}$ and  $A' \equiv A''$ for some $A''$ where \onlykey.
By $\brn{Par}'$, since $\tr{Q}$ is closed, 
$\tr{P_0} = \tr{P} \parop \tr{Q} \ltrP{\tr{\alpha}} \tr{A''}\parop \tr{Q} = \tr{A'' \parop Q}$ and $A'' \parop Q \equiv A' \parop Q \equiv A$. 
Furthermore, {\onlykey} in $A'' \parop Q$. The second case is symmetric.

\item $P_0 = \Res{n} P$, $P \ltrP{\alpha} A'$, and $A \equiv \Res{n} A'$ for some $P$, $A'$, and $n$ that does not occur in $\alpha$. We rename $n$ so that $n \neq k$. 
By induction hypothesis, 
$\tr{P} \ltrP{\tr{\alpha}} \tr{A''}$ and $A' \equiv A''$ 
for some $A''$ where \onlykey.
By $\brn{Scope}'$, $\tr{P_0} = \Res{n} \tr{P} \ltrP{\alpha} \Res{n} \tr{A''} = \tr{\Res{n}A''}$ and $\Res{n}A'' \equiv \Res{n}A' \equiv A$. 
Furthermore, {\onlykey} in $\Res{n}A''$.

\item $P_0 = \Repl{P}$, $P \ltrP{\alpha} A'$, and $A \equiv A' \parop \Repl{P}$ for some $P$ and $A'$. 
By induction hypothesis, 
$\tr{P} \ltrP{\tr{\alpha}} \tr{A''}$ and $A' \equiv A''$ 
for some $A''$  where {\onlykey}. We have $\tr{P_0} = \Repl{\tr{P}} \equivP \tr{P} \parop \Repl{\tr{P}} \ltrP{\tr{\alpha}} \tr{A''} \parop \Repl{\tr{P}}$ by $\brn{Par}'$, since $\Repl{\tr{P}}$ is closed. Hence by $\brn{Struct}'$, $\tr{P_0} \ltrP{\tr{\alpha}} \tr{A'' \parop \Repl{P}}$ and $A'' \parop \Repl{P} \equiv A' \parop \Repl{P} \equiv A$. Furthermore, {\onlykey} in $A'' \parop \Repl{P}$.

\item $P_0 = \Rcv{N}{x}.P$, $\alpha = \Rcv{N'}{M'}$, $\Sigma \vdash N = N'$, and $A \equiv P\{\subst{M'}{x}\}$ for some $N$, $x$, $P$, $N'$, and $M'$.
By Lemma~\ref{lem:tr-inj}, $\Sigma \vdash \tr{N} = \tr{N'}$, so we have $\tr{P_0} = \Rcv{\tr{N}}{x}.\tr{P} \equivP \Rcv{\tr{N'}}{x}.\tr{P} \ltrP{\tr{\alpha}} \tr{P}\{\subst{\tr{M'}}{x}\}$ by $\brn{In}'$. Since {\onlykey} in $M'$, the substitution $\tr{P}\{\subst{\tr{M'}}{x}\}$ does not create new occurrences of $\Const{mac}$ with key $k$, so $\tr{P}\{\subst{\tr{M'}}{x}\} = \tr{P\{\subst{M'}{x}\}}$. By $\brn{Struct}'$, we obtain $\tr{P_0} \ltrP{\tr{\alpha}} \tr{P\{\subst{M'}{x}\}}$, and we have $A \equiv P\{\subst{M'}{x}\}$. Furthermore, {\onlykey} in $P\{\subst{M'}{x}\}$.

\item $P_0 = \Snd{N}{M}.P$, $\alpha = \Res{x}\Snd{N'}{x}$, $\Sigma \vdash N = N'$, $x \notin \fv(P_0)$, and $A \equiv P \parop \{\subst{M}{x}\}$ for some $N$, $M$, $P$, $x$, and $N'$. By Lemma~\ref{lem:tr-inj}, $\Sigma \vdash \tr{N} = \tr{N'}$, so we have $\tr{P_0} = \Snd{\tr{N}}{\tr{M}}.\tr{P} \equivP \Snd{\tr{N'}}{\tr{M}}.\tr{P} \ltrP{\tr{\alpha}} \tr{P} \parop \{\subst{\tr{M}}{x}\} = \tr{P \parop \{\subst{M}{x}\}}$ by $\brn{Out-Var}'$. By $\brn{Struct}'$, we obtain $\tr{P_0} \ltrP{\tr{\alpha}} \tr{P \parop \{\subst{M}{x}\}}$, and we have $A \equiv P \parop \{\subst{M}{x}\}$. Furthermore, {\onlykey} in $P \parop \{\subst{M}{x}\}$.
\qedhere
\end{enumerate}
\end{proof}

\begin{lemma}\label{lem:redP-tr1}
If $P_0 \redP R$ for some closed process $P_0$ where \onlykey,
then $\tr{P_0} \redP \tr{R'}$ 
and $R' \equiv R$ for some closed process $R'$ where {\onlykey}.
\end{lemma}
\begin{proof}
We define the size of processes by induction on the syntax, such that $\size(\Repl{P}) = 1 + 2 \times \size(P)$ and, when $P$ is not a replication, $\size(P)$ is one plus the size of the immediate subprocesses of $P$.
We proceed by induction on the size of $P_0$. 
By Lemma~\ref{lem:decomp-redP}, we decompose $P_0 \redP R$, with the following cases:
\begin{enumerate}
\item $P_0 = P \parop Q$ for some $P$ and $Q$, and one of the following cases holds:
\begin{enumerate}
\item $P \redP P'$ and $R \equiv P' \parop Q$ for some $P'$,
\item $P \ltrP{\Rcv{N}{x}} A$, $Q \ltrP{\Res{x}\Snd{N}{x}} B$, and $R \equiv \Res{x}(A \parop B)$ for some $A$, $B$, $x$, and ground term $N$, 
\end{enumerate}
and two symmetric cases obtained by swapping $P$ and $Q$.

In case (a), by induction hypothesis, $\tr{P} \redP \tr{P''}$ and $P'' \equiv P'$ for some closed process $P''$ where \onlykey.
Hence $\tr{P_0} = \tr{P} \parop \tr{Q} \redP \tr{P''} \parop \tr{Q} = \tr{P'' \parop Q}$ and $P'' \parop Q \equiv P' \parop Q \equiv R$. Furthermore, {\onlykey} in $P'' \parop Q$.

In case (b), by Lemma~\ref{lem:ltrP-tr1}, 
$\tr{P} \ltrP{\Rcv{\tr{N}}{x}} \tr{P_1}$ and $A \equiv P_1$ for some $P_1$ where {\onlykey} and $\fv(P_1) \subseteq \{ x\}$;
and $\tr{Q} \ltrP{\Res{x}\Snd{\tr{N}}{x}} \tr{B'}$ for some $B' = \CTX_2[\{\subst{M_2}{x}\}]$ such that $B \equiv B'$, {\onlykey} in $B'$, $\CTX_2$ is a closed plain evaluation context and $M_2$ is a ground term. 
By Lemma~\ref{lem:comp-ltr}, 
$\tr{P_0} = \tr{P} \parop \tr{Q} \redP R'$ and $R' \equiv \Res{x}(\tr{P_1} \parop \tr{B'}) = \Res{x}(\tr{P_1} \parop \tr{\CTX_2}[\{\subst{\tr{M_2}}{x}\}])$ for some $R'$. 
We rename the bound names of $\CTX_2$ so that they do not occur in $P_1$.
Let $R'' = \CTX_2[P_1\{\subst{M_2}{x}\}]$.
The process $R''$ is closed and such that {\onlykey}.
We have $R' \equiv \tr{R''}$, so $\tr{P_0}\redP \tr{R''}$ and
$R'' \equiv \Res{x}(P_1 \parop B') \equiv \Res{x}(A \parop B) \equiv R$. 
The last two cases are symmetric.

\item $P_0 = \Res{n}P$, $P \redP Q'$, and $R \equiv \Res{n} Q'$ for some $n$, $P$, and $Q'$. We rename $n$ so that $n \neq k$. 
By induction hypothesis, 
$\tr{P} \redP \tr{Q''}$ and $Q' \equiv Q''$ for some closed process $Q''$ where \onlykey.
Hence $\tr{P_0} = \Res{n}\tr{P} \redP \Res{n}\tr{Q''} = \tr{\Res{n}Q''}$ and $\Res{n}Q'' \equiv \Res{n}Q' \equiv R$. Furthermore, {\onlykey} in $\Res{n}Q''$.

\item $P_0 = \Repl{P}$, $P \parop P \redP Q'$, and $R \equiv Q' \parop \Repl{P}$ for some $P$ and $Q'$. 
By induction hypothesis, 
$\tr{P \parop P} \redP \tr{Q''}$ and $Q' \equiv Q''$ for some closed process $Q''$ where \onlykey. 
Hence $\tr{P_0} = \Repl{\tr{P}} \equivP \tr{P} \parop \tr{P} \parop \Repl{\tr{P}} \redP \tr{Q''} \parop \Repl{\tr{P}} = \tr{Q'' \parop \Repl{P}}$ and $Q'' \parop \Repl{P} \equiv Q' \parop \Repl{P} \equiv R$. Furthermore, {\onlykey} in $Q'' \parop \Repl{P}$.

\item $P_0 = \IfThenElse{M}{N}{P}{Q}$ and either $\Sigma \vdash M = N$ and $R \equiv P$, or $\Sigma \vdash M \neq N$ and $R \equiv Q$, for some $M$, $N$, $P$, and $Q$. 

In the first case, by Lemma~\ref{lem:tr-inj}, $\Sigma \vdash \tr{M} = \tr{N}$, so $\tr{P_0} = \kw{if}$ $\tr{M} = \tr{N}$ $\kw{then}$ $\tr{P}$ $\kw{else}$ $\tr{Q} \redP \tr{P}$ and we know that $P \equiv R$ and {\onlykey} in $P$. 
In the second case, by Lemma~\ref{lem:tr-inj}, $\Sigma \vdash \tr{M} \neq \tr{N}$, so $\tr{P_0} = \IfThenElse{\tr{M}}{\tr{N}}{\tr{P}}{\tr{Q}} \redP \tr{Q}$ and we know that $Q \equiv R$ and {\onlykey} in $Q$.
\qedhere
\end{enumerate}
\end{proof}

\begin{lemma}\label{lem:ltrP-tr2}
Suppose that $P_0$ is closed, $\alpha = \Res{x}\Snd{N'}{x}$ or $\alpha = \Rcv{N'}{M'}$ for some ground term $N'$, and {\onlykey} in $P_0$ and $\alpha$.

If $\tr{P_0} \ltrP{\tr{\alpha}} A$, then $P_0 \ltrP{\alpha} A'$ and $A \equiv \tr{A'}$
for some $A'$ where \onlykey. 
Furthermore, when $\alpha = \Res{x}\Snd{N'}{x}$, $A' = \CTX[\{\subst{M}{x}\}]$ where $\CTX$ is a closed plain evaluation context and $M$ is a ground term, and when $\alpha = \Rcv{N'}{M'}$, $A'$ is a plain process with $\fv(A') \subseteq \fv(M')$.
\end{lemma}
\begin{proof}
We proceed by structural induction on $P_0$, with the following cases:
\begin{itemize}
\item $P_0 = P \parop Q$. Then $\tr{P_0} = \tr{P} \parop \tr{Q}$, so by Lemma~\ref{lem:decomp-ltrP}, either $\tr{P} \ltrP{\tr{\alpha}} A'$ and $A \equiv A' \parop \tr{Q}$, or $\tr{Q} \ltrP{\tr{\alpha}} A'$ and $A \equiv \tr{P} \parop A'$, for some $A'$. In the first case, by induction hypothesis, 
$P \ltrP{\alpha} A''$ and $A' \equiv \tr{A''}$ for some $A''$ where \onlykey. 
By $\brn{Par}'$, since $Q$ is closed, we have $P_0 = P \parop Q \ltrP{\alpha} A''\parop Q$ and $\tr{A'' \parop Q} \equiv A' \parop \tr{Q} \equiv A$. 
Furthermore, {\onlykey} in $A'' \parop Q$. The second case is symmetric.

\item $P_0 = \Res{n} P$. We rename $n$ so that $n \neq k$ and $n$ does not occur in $\alpha$. Then $\tr{P_0} = \Res{n}\tr{P}$, so by Lemma~\ref{lem:decomp-ltrP}, 
we have $\tr{P} \ltrP{\tr{\alpha}} A'$ and $A \equiv \Res{n} A'$ for some $A'$. 
By induction hypothesis, $P \ltrP{\alpha} A''$ and $A' \equiv \tr{A''}$
for some $A''$ where \onlykey.
By $\brn{Scope}'$, $P_0 = \Res{n} P \ltrP{\alpha} \Res{n} A''$ and $\tr{\Res{n}A''} = \Res{n}\tr{A''} \equiv \Res{n}A' \equiv A$. Furthermore, {\onlykey} in $\Res{n}A''$.

\item $P_0 = \Repl{P}$. Then $\tr{P_0} = \Repl{\tr{P}}$, so by Lemma~\ref{lem:decomp-ltrP}, we have $\tr{P} \ltrP{\tr{\alpha}} A'$, and $A \equiv A' \parop \Repl{\tr{P}}$ for some $A'$. 
By induction hypothesis, $P \ltrP{\alpha} A''$ and $A' \equiv \tr{A''}$ 
for some $A''$ where \onlykey. 
We have $P_0 = \Repl{P} \equivP P \parop \Repl{P} \ltrP{\alpha} A'' \parop \Repl{P}$ by $\brn{Par}'$, since $\Repl{P}$ is closed. Hence by $\brn{Struct}'$, $P_0 \ltrP{\alpha} A'' \parop \Repl{P}$ and $\tr{A'' \parop \Repl{P}} \equiv A' \parop \Repl{\tr{P}} \equiv A$. Furthermore, {\onlykey} in $A'' \parop \Repl{P}$.

\item $P_0 = \Rcv{N}{x}.P$. Then $\tr{P_0} = \Rcv{\tr{N}}{x}.\tr{P}$, so by Lemma~\ref{lem:decomp-ltrP}, we have $\tr{\alpha} = \Rcv{N'}{M'}$, $\Sigma \vdash \tr{N} = N'$, and $A \equiv \tr{P}\{\subst{M'}{x}\}$ for some $N'$ and $M'$.
Hence $\alpha = \Rcv{N''}{M''}$, $N' = \tr{N''}$, and $M' = \tr{M''}$ for some $N''$ and $M''$. We have $\Sigma \vdash \tr{N} = \tr{N''}$, so
by Lemma~\ref{lem:tr-inj}, $\Sigma \vdash N = N''$, so we have
$P_0 = \Rcv{N}{x}.P \equivP \Rcv{N''}{x}.P \ltrP{\alpha} P\{\subst{M''}{x}\}$ by $\brn{In}'$. 
By $\brn{Struct}'$, we obtain $P_0 \ltrP{\alpha} P\{\subst{M''}{x}\}$.
The name {\onlykey} in $M''$, so the substitution $\tr{P}\{\subst{\tr{M''}}{x}\}$ does not create new occurrences of $\Const{mac}(k,\cdot)$, so $\tr{P\{\subst{M''}{x}\}} = \tr{P}\{\subst{\tr{M''}}{x}\} = \tr{P}\{\subst{M'}{x}\} \equiv A$. Furthermore, {\onlykey} in $P\{\subst{M''}{x}\}$.

\item $P_0 = \Snd{N}{M}.P$. Then $\tr{P_0} = \Snd{\tr{N}}{\tr{M}}.\tr{P}$, so by Lemma~\ref{lem:decomp-ltrP}, $\tr{\alpha} = \Res{x}\Snd{N'}{x}$, $\Sigma \vdash \tr{N} = N'$, $x \notin \fv(\tr{P_0}) = \fv(P_0)$, and $A \equiv \tr{P} \parop \{\subst{\tr{M}}{x}\}$ for some $x$ and $N'$. Hence $\alpha = \Res{x}\Snd{N''}{x}$ and $N' = \tr{N''}$ for some $N''$. We have $\Sigma \vdash \tr{N} = \tr{N''}$, so by Lemma~\ref{lem:tr-inj}, $\Sigma \vdash N = N''$, so we have $P_0 = \Snd{N}{M}.P \equivP \Snd{N''}{M}.P \ltrP{\alpha} P \parop \{\subst{M}{x}\}$ by $\brn{Out-Var}'$. Hence by $\brn{Struct}'$, we obtain $P_0 \ltrP{\alpha} P \parop \{\subst{M}{x}\}$, and we have $\tr{P \parop \{\subst{M}{x}\}} = \tr{P} \parop \{\subst{\tr{M}}{x}\} \equiv A$. Furthermore, {\onlykey} in $P \parop \{\subst{M}{x}\}$.

\item $P_0$ is neither $\nil$ nor a conditional, because by Lemma~\ref{lem:decomp-ltrP}, $\tr{P_0}$ would not have a labelled transition.
\qedhere
\end{itemize}
\end{proof}

\begin{lemma}\label{lem:ltrP-alpha-is-tr}
Suppose that $P_0$ is a closed process where \onlykey.
If $\tr{P_0} \ltrP{\alpha} A$
with $\alpha = \Res{x}\Snd{N'}{x}$ or $\alpha = \Rcv{N'}{M'}$, 
then $\Sigma \vdash N' = \tr{N}$
for some ground term $N$ where {\onlykey}.
\end{lemma}
\begin{proof}
We proceed by structural induction on $P_0$, with the following cases:
\begin{itemize}
\item $P_0 = P \parop Q$. Then $\tr{P_0} = \tr{P} \parop \tr{Q}$, so by Lemma~\ref{lem:decomp-ltrP}, either $\tr{P} \ltrP{\alpha} A'$ and $A \equiv A' \parop \tr{Q}$, or $\tr{Q} \ltrP{\alpha} A'$ and $A \equiv \tr{P} \parop A'$, for some $A'$. In both cases, the result follows immediately from the induction hypothesis.

\item $P_0 = \Res{n} P$. We rename $n$ so that $n \neq k$ and $n$ does not occur in $\alpha$. Then $\tr{P_0} = \Res{n}\tr{P}$, so by Lemma~\ref{lem:decomp-ltrP}, $\tr{P} \ltrP{\alpha} A'$, and $A \equiv \Res{n} A'$ for some $A'$.
The result follows immediately from the induction hypothesis.

\item $P_0 = \Repl{P}$. Then $\tr{P_0} = \Repl{\tr{P}}$, so by Lemma~\ref{lem:decomp-ltrP}, $\tr{P} \ltrP{\alpha} A'$, and $A \equiv A' \parop \Repl{\tr{P}}$ for some $A'$. The result follows immediately from the induction hypothesis.

\item $P_0 = \Rcv{N}{x}.P$. Then $\tr{P_0} = \Rcv{\tr{N}}{x}.\tr{P}$, so by Lemma~\ref{lem:decomp-ltrP}, $\alpha = \Rcv{N'}{M'}$, $\Sigma \vdash \tr{N} = N'$, and $A \equiv \tr{P}\{\subst{M'}{x}\}$ for some $N'$ and $M'$.
Moreover, since $N$ occurs in $P_0$, $N$ is ground and {\onlykey} 
in $N$, so the result holds.

\item $P_0 = \Snd{N}{M}.P$. Then $\tr{P_0} = \Snd{\tr{N}}{\tr{M}}.\tr{P}$, so by Lemma~\ref{lem:decomp-ltrP}, $\alpha = \Res{x}\Snd{N'}{x}$, $\Sigma \vdash \tr{N} = N'$, $x \notin \fv(\tr{P_0}) = \fv(P_0)$, and $A \equiv \tr{P} \parop \{\subst{\tr{M}}{x}\}$ for some $x$ and $N'$. 
Moreover, since $N$ occurs in $P_0$, $N$ is ground and {\onlykey} in $N$, so
the result holds.

\item $P_0$ is neither $\nil$ nor a conditional, because by Lemma~\ref{lem:decomp-ltrP}, $\tr{P_0}$ would not have a labelled transition.
\qedhere
\end{itemize}
\end{proof}

\begin{lemma}\label{lem:ltrP-eq-alpha}
If $P \ltrP{\alpha} A$ and $\Sigma \vdash \alpha = \alpha'$, then
$P \ltrP{\alpha'} A$.
\end{lemma}
\begin{proof}
We proceed by induction on the derivation of $P \ltrP{\alpha} A$.
\begin{itemize}
\item Case $\brn{In}'$. We have $P = \Rcv{N}{x}.P'$,
$\alpha = \Rcv{N}{M}$, and $A = P'\{\subst{M}{x}\}$ for some $N$, $M$, 
$x$, and $P'$.
Since $\Sigma \vdash \alpha = \alpha'$, we have $\alpha' = \Rcv{N'}{M'}$,
$\Sigma \vdash N' = N$, and $\Sigma \vdash M' = M$ for some $N'$ and $M'$.
Hence $P \equivP \Rcv{N'}{x}.P' \ltrP{\alpha'} P'\{\subst{M'}{x}\} \equiv A$
by $\brn{In}'$, so $P \ltrP{\alpha'} A$ by $\brn{Struct}'$.

\item Case $\brn{Out-Var}'$. We have $P = \Snd{N}{M}.P'$,
$\alpha = \Res{x}\Snd{N}{x}$, and $A = P\parop \{\subst{M}{x}\}$
for some $N$, $M$, $x$, and $P'$.
Since $\Sigma \vdash \alpha = \alpha'$, we have $\alpha' = \Res{x}\Snd{N'}{x}$
and $\Sigma \vdash N' = N$ for some $N'$.
Hence $P \equivP \Snd{N'}{M}.P' \ltrP{\alpha'} P\parop \{\subst{M}{x}\}$
by $\brn{Out-Var}'$, so $P \ltrP{\alpha'} A$ by $\brn{Struct}'$.

\item The other cases follow easily from the induction hypothesis.
In the case $\brn{Scope}'$, we rename the bound name $n$ so that it
does not occur in $\alpha$.
In the case $\brn{Par}'$, we use that $\bv(\alpha') = \bv(\alpha)$.
\qedhere
\end{itemize}
\end{proof}

\begin{lemma}\label{lem:redP-tr2}
Suppose that $P_0$ is a closed process where {\onlykey}.
If $\tr{P_0} \redP R$, then $P_0 \redP R'$ and $R \equiv \tr{R'} $
for some closed process $R'$ where {\onlykey}.
\end{lemma}
\begin{proof}
We proceed by induction on the size of $P_0$, with the same definition of size
as in the proof of Lemma~\ref{lem:redP-tr1}. The following cases may occur:
\begin{enumerate}
\item $P_0 = P \parop Q$. Then $\tr{P_0} = \tr{P} \parop \tr{Q}$, so by Lemma~\ref{lem:decomp-redP}, one of the following cases holds:
\begin{enumerate}
\item $\tr{P} \redP P'$ and $R \equiv P' \parop \tr{Q}$ for some $P'$,
\item $\tr{P} \ltrP{\Rcv{N}{x}} A$, $\tr{Q} \ltrP{\Res{x}\Snd{N}{x}} B$, and $R \equiv \Res{x}(A \parop B)$ for some $A$, $B$, $x$, and ground term $N$, 
\end{enumerate}
and two symmetric cases obtained by swapping $P$ and $Q$.
In the first case, by induction hypothesis, $P \redP P''$ and $\tr{P''} \equiv P'$
for some closed process $P''$ where \onlykey. 
Hence $P_0 = P \parop Q \redP P'' \parop Q$ and $\tr{P'' \parop Q} = \tr{P''} \parop \tr{Q} \equiv P' \parop \tr{Q} \equiv R$. Furthermore, {\onlykey} in $P'' \parop Q$.
In the second case, by Lemma~\ref{lem:ltrP-alpha-is-tr}, $\Sigma \vdash N = \tr{N'}$ for some ground term $N'$ where {\onlykey}.
By Lemma~\ref{lem:ltrP-eq-alpha}, $\tr{P} \ltrP{\Rcv{\tr{N'}}{x}} A$ and $\tr{Q} \ltrP{\Res{x}\Snd{\tr{N'}}{x}} B$.
By Lemma~\ref{lem:ltrP-tr2}, $P \ltrP{\Rcv{N}{x}} P_1$ and $A \equiv \tr{P_1}$
for some $P_1$ where {\onlykey} and $\fv(P_1) \subseteq \{ x\}$; and $Q \ltrP{\Res{x}\Snd{N}{x}} B'$ and $B \equiv \tr{B'}$ for some $B' = \CTX_2[\{\subst{M_2}{x}\}]$ 
where {\onlykey} in $B'$, $\CTX_2$ is a closed plain evaluation context, and $M_2$ is a ground term.
By Lemma~\ref{lem:comp-ltr}, $P_0 = P \parop Q \redP R'$ and $R' \equiv \Res{x}(P_1 \parop B') = \Res{x}(P_1 \parop \CTX_2[\{\subst{M_2}{x}\}])$ for some $R'$. 
We rename the bound names of $\CTX_2$ so that they do not occur in $P_1$.
Let $R'' = \CTX_2[P_1\{\subst{M_2}{x}\}]$.
The process $R''$ is closed and {\onlykey} in $R''$.
We have $R' \equiv R''$, so $P_0\redP R''$ and
$\tr{R''} \equiv \tr{\Res{x}(P_1 \parop B')} \equiv \Res{x}(A \parop B) \equiv R$. 
The last two cases are symmetric.

\item $P_0 = \Res{n}P$. We rename $n$ so that $n \neq k$. Then $\tr{P_0} = \Res{n}\tr{P}$, so by Lemma~\ref{lem:decomp-redP}, $\tr{P} \redP Q'$, and $R \equiv \Res{n} Q'$ for some $Q'$. By induction hypothesis, 
$P \redP Q''$ and $Q' \equiv \tr{Q''}$ for some closed process $Q''$ where {\onlykey}. 
Hence $P_0 = \Res{n}P \redP \Res{n}Q''$ and $\tr{\Res{n}Q''} = \Res{n}\tr{Q''} \equiv \Res{n}Q' \equiv R$. Furthermore, {\onlykey} in $\Res{n}Q''$.

\item $P_0 = \Repl{P}$. Then $\tr{P_0} = \Repl{\tr{P}}$, so by Lemma~\ref{lem:decomp-redP}, $\tr{P \parop P} = \tr{P} \parop \tr{P} \redP Q'$, and $R \equiv Q' \parop \Repl{\tr{P}}$ for some $Q'$. By induction hypothesis, 
$P \parop P \redP Q''$ and $Q' \equiv \tr{Q''}$ for some closed process $Q''$ where {\onlykey}. 
Hence $P_0 = \Repl{P} \equivP P \parop P \parop \Repl{P} \redP Q'' \parop \Repl{P}$ and $\tr{Q'' \parop \Repl{P}} = \tr{Q''} \parop \Repl{\tr{P}} \equiv Q' \parop \Repl{\tr{P}} \equiv R$. Furthermore, {\onlykey} in $Q'' \parop \Repl{P}$.

\item $P_0 = \IfThenElse{M}{N}{P}{Q}$. Then $\tr{P_0} = 
\IfThen{\tr{M}}{\tr{N}}{\tr{P}}\ \kw{else}$ $\tr{Q}
$, so by Lemma~\ref{lem:decomp-redP}, either $\Sigma \vdash \tr{M} = \tr{N}$ and $R \equiv \tr{P}$, or $\Sigma \vdash \tr{M} \neq \tr{N}$ and $R \equiv \tr{Q}$. 
In the first case, by Lemma~\ref{lem:tr-inj}, $\Sigma \vdash M = N$, so $P_0 = \IfThenElse{M}{N}{P}{Q} \redP P$ and we know that $\tr{P} \equiv R$ and {\onlykey} in $P$. 
In the second case, by Lemma~\ref{lem:tr-inj}, $\Sigma \vdash M \neq N$, so $P_0 = \IfThenElse{M}{N}{P}{Q} \redP Q$ and we know that $\tr{Q} \equiv R$ and {\onlykey} in $Q$.

\item $P_0$ is not $\nil$, an input, or an output, because by Lemma~\ref{lem:decomp-redP}, $\tr{P_0}$ would not reduce.
\qedhere
\end{enumerate}
\end{proof}

\begin{restate}{Theorem}{\ref{th:mac}}
Suppose that the signature $\Sigma$ is equipped with 
an equational theory generated by a convergent rewrite system such that
$\Const{mac}$ and $\Const{f}$ do not occur in the left-hand sides of rewrite rules;
the only rewrite rules with $\Const{h}$ at the root of the left-hand side are
those of~\eqref{eq:h-iter} and~\eqref{eq:h-end} oriented from left to right;
there are no rewrite rules with $::$ nor $\Const{nil}$ at the root of the left-hand side;
and names do not occur in rewrite rules.
  Suppose that $C$ is closed and the name $k$ appears only as first argument of $\Const{mac}$ in $C$. 
Then $\nu k. C \bicong \nu k.\tr{C}$.
\end{restate}%
\begin{proof}
Let $\rel$ relate all closed extended processes $A$ and $B$ 
such that
 $A \equiv \Res{k}C$ and $B \equiv \Res{k}\tr{C}$
for some closed normal process $C$ 
where {\onlykey}.

We show that $\rel \cup \rel^{-1}$ is a labelled bisimulation. It is symmetric by construction. Assume that $A \rel B$ for some 
$C = \Res{\vect n}(\sigma \parop P)$ where {\onlykey}. 
In particular, $k$ does not occur in $\vect n$,
$A$ and $B$ are closed, 
$A \equiv \Res{k}C$, and 
$B \equiv \Res{k}\tr{C}$.
\begin{enumerate}
\item We show that $A \enveq B$.

Let $M$ and $N$ be two terms such that $\fv(M) \cup \fv(N) \subseteq \dom(A) = \dom(B) = \dom(\sigma)$.
We have $\frameof{A} \equiv \Res{k,\vect n}\sigma$ and
$\frameof{B} \equiv \Res{k,\vect n}\tr{\sigma}$.
We rename $k$ and $\vect n$ so that these names do not occur in $M$ and $N$.
Then 
\begin{align*}
(M=N)\frameof{A}
\Leftrightarrow {}&\Sigma \vdash M\sigma = N\sigma\\
\Leftrightarrow {}&\Sigma \vdash \tr{M\sigma} = \tr{N\sigma} \tag*{by Lemma~\ref{lem:tr-inj}}\\
\Leftrightarrow {}&\Sigma \vdash M\tr{\sigma} = N\tr{\sigma}
\end{align*}
since $k$ does not occur in $M$ and $N$ and {\onlykey} in $\sigma$, so 
\begin{align*}
(M=N)\frameof{A} \Leftrightarrow {}&(M=N)\frameof{B}
\end{align*}
Therefore, $A \enveq B$.

\item We first show that, if $A \ltr{\alpha} A'$, $A'$ is closed, and 
    $\fv(\alpha) \subseteq \dom(A)$,
    then $B \rightarrow^*
    \ltr{\alpha} \rightarrow^* B'$ and $A' \rel B'$ for some $B'$.

We have $\Res{k,\vect n}(\sigma \parop P) \ltr{\alpha} A'$, 
so by Lemma~\ref{lem:redalpha-std-to-pnf},
$\Res{k,\vect n}(\sigma \parop P) \ltrpnf{\alpha} A'$.
We rename $k$ and $\vect n$ so that these names do not occur in $\alpha$.
By Lemma~\ref{lem:decomp-ltrpnf}, $P \ltrP{\alpha\sigma} A'_1$, $A' \equiv \Res{k,\vect n}(\sigma \parop A'_1)$, and $\bv(\alpha) \cap \dom(\sigma) = \emptyset$ for some $A'_1$.
By Lemma~\ref{lem:ltrP-tr1}, $\tr{P} \ltrP{\tr{\alpha\sigma}} \tr{A''_1}$ 
and $A'_1 \equiv A''_1$ for some $A''_1$ where {\onlykey}.
Furthermore, when $\alpha\sigma = \Res{x}\Snd{N'}{x}$, $A''_1 = \CTX[\{\subst{M}{x}\}]$ where $\CTX$ is a closed plain evaluation context and $M$ is a ground term, and when $\alpha\sigma = \Rcv{N'}{M'}$, $A''_1$ is a plain process with $\fv(A''_1) \subseteq \fv(M') = \emptyset$, so $A''_1$ is a closed plain process.
Since $k$ does not occur in $\alpha$ and {\onlykey} in $\sigma$, 
we have $\tr{\alpha\sigma} = \alpha \tr{\sigma}$.
Therefore, $B \equiv \Res{k}\tr{C} \equiv \Res{k,\vect n}(\tr{\sigma} \parop \tr{P}) \ltrpnf{\alpha} \Res{k,\vect n}(\tr{\sigma} \parop \tr{A''_1})$.
Let $C' = \pnf(\Res{\vect n}(\sigma \parop A''_1))$.
The process $C'$ is a closed normal process
where {\onlykey}
We have $\Res{k}C' \equiv \Res{k,\vect n}(\sigma \parop A''_1) \equiv \Res{k,\vect n}(\sigma \parop A'_1) \equiv A'$.
Let $B' = \Res{k}\tr{C'}$.
We have $A' \rel B'$.
Given the form of $A''_1$, we can show that $\tr{C'} \equiv \Res{\vect n}(\tr{\sigma} \parop \tr{A''_1})$, so $B \ltr{\alpha} B'$.

Next, we show that, if $B \ltr{\alpha} B'$, $B'$ is closed, and 
    $\fv(\alpha) \subseteq \dom(B)$,
    then $A \rightarrow^*
    \ltr{\alpha} \rightarrow^* A'$ and $A' \rel B'$ for some $A'$.

We have $\Res{k,\vect n}(\tr{\sigma} \parop \tr{P}) \ltr{\alpha} B'$, 
so by Lemma~\ref{lem:redalpha-std-to-pnf},
$\Res{k,\vect n}(\tr{\sigma} \parop \tr{P}) \ltrpnf{\alpha} B'$.
We rename $k$ and $\vect n$ so that these names do not occur in $\alpha$.
By Lemma~\ref{lem:decomp-ltrpnf}, $\tr{P} \ltrP{\alpha\tr{\sigma}} B'_1$, $B' \equiv \Res{k,\vect n}(\tr{\sigma} \parop B'_1)$, and $\bv(\alpha) \cap \dom(\tr{\sigma}) = \bv(\alpha) \cap \dom(\sigma) = \emptyset$ for some $B'_1$.
Since $k$ does not occur in $\alpha$ and {\onlykey} in $\sigma$, 
we have $\tr{\alpha\sigma} = \alpha \tr{\sigma}$.
By Lemma~\ref{lem:ltrP-tr2}, 
$P \ltrP{\alpha\sigma} A'_1$ and $B'_1 \equiv \tr{A'_1}$ 
for some $A'_1$ where {\onlykey}.
Furthermore, when $\alpha\sigma = \Res{x}\Snd{N'}{x}$, $A'_1 = \CTX[\{\subst{M}{x}\}]$ where $\CTX$ is a closed plain evaluation context and $M$ is a ground term, and when $\alpha\sigma = \Rcv{N'}{M'}$, $A'_1$ is a plain process with $\fv(A'_1) \subseteq \fv(M') = \emptyset$, so $A'_1$ is a closed plain process.
Therefore, $A \equiv \Res{k}C \equiv \Res{k,\vect n}(\sigma \parop P) \ltrpnf{\alpha} \Res{k,\vect n}(\sigma \parop A'_1)$.
Let $C' = \pnf(\Res{\vect n}(\sigma \parop A'_1))$.
The process $C'$ is a closed normal process
where {\onlykey}.
Given the form of $A'_1$, we can show that $\Res{k}\tr{C'} \equiv \Res{k,\vect n}(\tr{\sigma} \parop \tr{A'_1}) \equiv \Res{k,\vect n}(\tr{\sigma} \parop B'_1) \equiv B'$.
Let $A' = \Res{k}C'$.
We have $A' \rel B'$ and $\Res{k,\vect n}(\sigma \parop A'_1) \equiv \Res{k}C' = A'$ so $A \ltr{\alpha} A'$.

\item We first show that, if $A \rightarrow A'$ for some closed $A'$,
  then $B \rightarrow^* B'$ and $A' \rel B'$ for some $B'$.  

We have $\Res{k,\vect n}(\sigma \parop P) \rightarrow A'$, 
so by Lemma~\ref{lem:red-std-to-pnf}, 
$\Res{k,\vect n}(\sigma \parop P) \redpnf \pnf(A')$.
By Lemma~\ref{lem:decomp-redpnf},
$P \redP P'$ and $\pnf(A') \equiv \Res{k,\vect n}(\sigma \parop P')$ for some $P'$.
By Lemma~\ref{lem:redP-tr1}, 
$\tr{P} \redP \tr{P''}$ and $P' \equiv P''$ 
for some closed process $P''$ where {\onlykey}.
Let $C' = \Res{\vect n}(\sigma \parop P'')$.
The process $C'$ is a closed normal process
where {\onlykey}.
We have $\Res{k}C' \equiv \Res{k,\vect n}(\sigma \parop P'') \equiv \pnf(A') \equiv A'$.
Let $B' = \Res{k}\tr{C'}$.
We have $A' \rel B'$ and
$B \equiv \Res{k}\tr{C} \equiv \Res{k,\vect n}(\tr{\sigma}\parop\tr{P})
\redpnf \Res{k,\vect n}(\tr{\sigma}\parop\tr{P''}) = \Res{k}\tr{C'} = B'$,
so $B \rightarrow B'$. 

Next, we show that, if $B \rightarrow B'$ for some closed $B'$,
then $A \rightarrow^* A'$ and $A' \rel B'$ for some $A'$.

We have $\Res{k,\vect n}(\tr{\sigma} \parop \tr{P}) \rightarrow B'$, 
so by Lemma~\ref{lem:red-std-to-pnf}, 
$\Res{k,\vect n}(\tr{\sigma} \parop \tr{P}) \redpnf \pnf(B')$.
By Lemma~\ref{lem:decomp-redpnf},
$\tr{P} \redP P'$ and $\pnf(B') \equiv \Res{k,\vect n}(\tr{\sigma} \parop P')$ for some $P'$.
By Lemma~\ref{lem:redP-tr2}, $P \redP P''$ and $P' \equiv \tr{P''}$ for some closed process $P''$ where {\onlykey}.
Let $C' = \Res{\vect n}(\sigma \parop P'')$.
The process $C'$ is a closed normal process where {\onlykey}. 
We have $\Res{k}\tr{C'} \equiv \Res{k,\vect n}(\tr{\sigma} \parop \tr{P''}) \equiv \Res{k,\vect n}(\tr{\sigma} \parop P') \equiv \pnf(B') \equiv B'$.
Let $A' = \Res{k}C'$.
We have $A' \rel B'$ and
$A \equiv \Res{k}C \equiv \Res{k,\vect n}(\sigma\parop P)
\redpnf \Res{k,\vect n}(\sigma\parop P'') = \Res{k}C' = A'$,
so $A \rightarrow A'$. 

\end{enumerate}
Therefore, ${\rel} \subseteq {\wkbisim}$ and, by Theorem~\ref{THM:OBSERVATIONAL-LABELED}, ${\rel} \subseteq {\bicong}$.

Finally, when $C$ is a closed extended process where {\onlykey},
we have $\Res{k}C \rel \Res{k}\tr{C}$ because $\pnf(C)$ is a closed normal process where {\onlykey} such that $\pnf(C) \equiv C$.
We thus obtain $\Res{k}C \bicong \Res{k}\tr{C}$.
\end{proof}

\begin{restate}{Corollary}{\ref{cor:mac}}
Suppose that the signature $\Sigma$ is equipped with the equational theory
defined by the equations~\eqref{eq:fst}, \eqref{eq:snd}, \eqref{eq:seq}, \eqref{eq:++}, \eqref{eq:h-iter}, and
\eqref{eq:h-end}.
Suppose that $C$ is closed and the name $k$ appears only as first argument of $\Const{mac}$ in $C$. 
Then $\nu k. C \bicong \nu k.\tr{C}$.
\end{restate}%
\begin{proof}
We define the rewrite system $\rew$ by orienting the equations~\eqref{eq:fst}, \eqref{eq:snd}, \eqref{eq:seq}, \eqref{eq:++}, \eqref{eq:h-iter}, and~\eqref{eq:h-end} from left to right:
{\allowdisplaybreaks\begin{eqnarray}
  \Const{fst}( (x,y) ) & \rightarrow  & x \label{red:fst}\\
  \Const{snd}( (x,y) ) & \rightarrow  & y \label{red:snd}\\
\Const{hd}(\cons{x}{y}) &\rightarrow& x\label{red:hd}\\
\Const{tl}(\cons{x}{y}) &\rightarrow& y\label{red:tl}\\
\Const{nil} \bappend x &\rightarrow& \cons{x}{\Const{nil}}\label{red:nil++}\\
(\cons{x}{y}) \bappend z &\rightarrow& \cons{x}{(y\bappend z)}\label{red:cons++}\\
\hbin{x}{\cons{y_0}{\cons{y_1}{z}}} &\rightarrow& \hbin{\fbin{x}{y_0}}{\cons{y_1}{z}}\label{red:h-iter}\\
\hbin{x}{\cons{y}{\Const{nil}}} &\rightarrow& \fbin{x}{y}\label{red:h-end}
\end{eqnarray}}%
In order to prove that $\rew$ terminates, we order terms $M$
lexicographically, using:
\begin{enumerate}
\item the size of $M$; then
\item the number of occurrences of the $\bappend$ symbol in $M$; then
\item the number of occurrences of the $::$ symbol in $M$; then 
\item the sum, over all occurrences of $\bappend$ in $M$, of
the lengths of the first arguments of $\bappend$, 
computed as follows: $\len(\cons{N_1}{N_2}) = 1 + \len(N_2)$,
$\len(N_1 \bappend N_2) = 1 + \len(N_1)$,
and $\len(N) = 0$ for all other terms.
\end{enumerate}
This ordering is well-founded.
Rules~\eqref{red:fst}, \eqref{red:snd}, \eqref{red:hd}, \eqref{red:tl}, and~\eqref{red:h-end} decrease the size.
Rule~\eqref{red:nil++} preserves the size and decreases the number of occurrences of $\bappend$.
Rule~\eqref{red:cons++} preserves the size and the numbers of occurrences of $\bappend$ and $::$ but it decreases the sum above, 
because the length of the first argument decreases for the occurrence of $\bappend$ modified by rule~\eqref{red:cons++} ($\len(N) < \len(\cons{M}{N})$) and is unchanged for all other occurrences of $\bappend$ in the term.
Rule~\eqref{red:h-iter} preserves the size and the number of occurrences of $\bappend$; it decreases the number of occurrences of $::$.
Therefore, if $M$ reduces to $M'$ by any of these rules, 
we have $M' < M$. 
This property shows that $\rew$ terminates.
(The termination of $\rew$ can also be proved using well-known techniques.
For instance, it can be proved using a lexicographic path ordering,
provided the second argument of $\Const{h}$, which decreases by~\eqref{red:h-iter}, is considered
before the first one, which increases, in the lexicographic ordering.)

The rewrite system $\rew$ is confluent because there are no critical pairs between the rules. Hence $\rew$ is convergent. Since $\rew$ generates
the equational theory under consideration, we conclude by Theorem~\ref{th:mac}.
\end{proof}

\section{Proofs for Section~\ref{SUBSEC:INDIF}}\label{app:indif}

In Lemma~\ref{lem:key-stat-eq} and Corollary~\ref{cor:key-stat-eq},
we suppose that the signature $\Sigma$ is equipped with
an equational theory generated by a convergent rewrite system $\rew$.
Since $\rew$ terminates, the left-hand side of its rewrite rules cannot be 
variables. 
We suppose that the rewrite rules of $\rew$ do not contain names.
We denote by $\theta$ a substitution and by $\rho$ a variable renaming.
We first study active substitutions from variables to hash computations,
that is, terms whose root symbols range over functions that do not occur on the left-hand side of $\rew$.

\begin{lemma}\label{lem:key-stat-eq}
  Suppose $\Sigma$ is equipped with an equational theory generated by a convergent rewrite system $\rew$.
  Let $\theta$ be a closed substitution that ranges over pairwise distinct terms modulo $\Sigma$, each
  of the form $f(k, M)$ where $f$ does not occur on the left-hand side
  of the rules of $\rew$.
  Let $\sigma$ map the same variables to
  pairwise distinct names $\vect{a}$.  We have
  $\Res{k}\theta \enveq \Res{\vect{a}}\sigma$.
\end{lemma}
\begin{proof}
More explicitly, let $\theta = \{(\subst{M_i}{x_i})_{i = 1..n}\}$,
$\sigma = \{(\subst{a_i}{x_i})_{i = 1..n}\}$, and $\vect a = a_1, \dots, a_n$.

We first prove the property
\begin{quote}
\substinj:
if, moreover, $\theta$ ranges over syntactically pairwise distinct terms, then
$N_1\theta  = N_2\theta$ and $k \notin \fn(N_1) \cup \fn(N_2)$ implies $N_1 = N_2$.
\end{quote}
Let $N_1'$ be obtained from $N_1$ by replacing the occurrences of $x_1, \ldots, x_n$
with pairwise distinct variables $y_1, \ldots, y_{n'}$, and let $(i_j)_{j = 1.. n'}$ and $\rho = \{(\subst{x_{i_j}}{y_j})_{j = 1..n'}\}$ be such that $N_1 = N'_1 \rho$.
We have $N'_1\rho\theta = N_2\theta $.
Since 
$k$ does not occur in $N_2$ or $N'_1$, and 
$k$ occurs as first argument of the root function symbol of $M_i$ for $i = 1..n$ 
and $M_{i_j}$ for $j = 1..n'$, 
the terms $N_2$ and $N'_1$ are equal up to some variable renaming.
Since each variable $y_j$ occurs once in $N'_1$, 
we have $N_2 = N'_1\rho'$ for some $(i'_j)_{j = 1.. n'}$ 
and $\rho' = \{(\subst{x_{i'_j}}{y_j})_{j = 1.. n'}\}$.
We have $N'_1\rho'\theta = N'_1\rho\theta$, so
for all $j = 1..n'$ we have $y_j\rho'\theta = y_j\rho\theta$,
so $M_{i'_j} = M_{i_j}$.
Since $M_1, \ldots, M_n$ are pairwise distinct, we have $i'_j = i_j$,
so $\rho' = \rho$.
Hence $N_1 = N'_1 \rho  = N'_1 \rho' = N_2$.

Let us now prove the lemma itself.
We first reduce $M_1, \dots, M_n$ into irreducible form under $\rew$.
By Lemma~\ref{LEM:INVARIANT-STATIC-EQ}, it is enough to prove static equivalence on these reduced terms.
Moreover, they are still of the form $f(k,M)$ with the same condition on $f$.
(Indeed, the left-hand sides of rewrite rules do not contain $f$, so
the rewrite rules apply to strict subterms of $f(k,M)$; and $k$ is
irreducible, so the rewrite rules apply only to the terms $M$ within
$f(k,M)$.)

Let $N_1$, $N_2$ be two terms with $\fv(N_1) \cup \fv(N_2) \subseteq \{x_1, \ldots, x_n\}$.
We need to show that $(N_1 = N_2) \Res{k}\theta$
if and only if $(N_1 =  N_2) \Res{\vect{a}}\sigma$.
We rename $k, \vect a$ so that $(\fn(N_1) \cup \fn(N_2)) \cap \{k, \vect a\} = \emptyset$.
We have $(N_1 = N_2) \Res{k}\theta  $
if and only if $\Sigma \vdash N_1\theta  = N_2\theta $,
and $(N_1 =  N_2) \sigma$
if and only if $\Sigma \vdash N_1 \sigma = N_2 \sigma$.
We show that $\Sigma \vdash N_1\theta  = N_2\theta  $ if and only if 
$\Sigma \vdash N_1 = N_2$ if and only if 
$\Sigma \vdash N_1 \sigma = N_2 \sigma$.

Since the equational theory is closed under substitution of terms for variables and names, we have that $\Sigma \vdash N_1 = N_2$ implies $\Sigma \vdash N_1\theta  = N_2\theta  $, 
$\Sigma \vdash N_1 = N_2$ implies $\Sigma \vdash N_1\sigma = N_2 \sigma$, and
$\Sigma \vdash N_1 \sigma = N_2 \sigma$ implies $\Sigma \vdash N_1 = N_2$ (by substituting $x_i$ for $a_i$ for $i = 1..n$). Hence, we just have to show that $\Sigma \vdash N_1\theta  = N_2\theta  $ implies
$\Sigma \vdash N_1 = N_2$. We can restrict our attention to the case in which $N_1$ and $N_2$ are irreducible under $\rew$, since the equality of the initial terms is equivalent to the equality of their reduced forms. 

Suppose that $\Sigma \vdash N_1\theta  = N_2\theta  $, with $N_1, \allowbreak N_2, \allowbreak M_1, \allowbreak \ldots, \allowbreak M_n$ irreducible under $\rew$. We first show that $N_1\theta $ is irreducible under $\rew$. In order to derive a contradiction, suppose that $N_1\theta $ is reducible by a rewrite rule $N_3 \rightarrow N_4$ of $\rew$. Then there exists a term context $C$ and a substitution $\sigma$ such that $C[N_3\sigma] = N_1\theta $.
Let $N'_1$ be obtained from $N_1$ by renaming the occurrences of $x_1, \ldots, x_n$ into pairwise distinct variables $y_1, \ldots, y_{n'}$, and let $(i_j)_{j = 1..n'}$ and 
$\rho = \{ (\subst{x_{i_j}}{y_j})_{j = 1.. n'} \}$ such that $N_1 = N'_1\rho $.
We have $C[N_3\sigma] = N_1'\rho\theta$.
The position of the hole of $C$ cannot be inside $M_{i_j}$, since
otherwise $M_{i_j}$ would be reducible by $N_3 \rightarrow N_4$.
Hence, the position of the hole of $C$ is inside $N_1'$, so
$N_3\sigma = N_1''\rho\theta$,
$C = C'\rho\theta$,
and $N_1' = C'[N_1'']$
for some subterm $N_1''$ of $N_1'$ and term context $C'$.

Let $\rho'$ be a variable renaming such that $N_3'\rho' = N_3$ and all variable
occurrences in $N_3'$ are fresh and pairwise distinct.
We have $N_3' \rho' \sigma = N_1'' \rho \theta$.
Since the function symbols $f$ at the root of $M_{i_j}$ do not occur in $N_3$, all occurrences of $M_{i_j}$ are in $z\rho'\sigma$ for some $z \in \fv(N_3')$. 
Hence, for all $z \in \fv(N_3')$, 
there exists a subterm $N_z$ of $N_1''$ such that 
$z\rho'\sigma = N_z\rho\theta$ and 
$N_1'' = N_3'\{(\subst{N_z}{z})_{z \in \fv(N_3')} \}$. 
Furthermore, when $z$ and $z'$ are distinct variables of $N_3'$ such that $z\rho' = z'\rho'$, we have
$z\rho'\sigma =z'\rho'\sigma$, 
so 
$
N_z\rho\theta = N_{z'}\rho\theta$ and, by \substinj{}, $N_z\rho = N_{z'}\rho$.

For each variable $y$ of $N_3$, let us choose one variable $z_y$ of $N_3'$ such that $z_y\rho' = y$.
Let us define $\sigma'$ by $y \sigma' = N_{z_y} \rho$.
Since for all $z, z' \in \fv(N_3')$, we have $z\rho' = z'\rho'$ implies $N_z\rho = N_{z'}\rho$, 
we have for all $z \in \fv(N_3')$, $z\rho'\sigma'=N_z\rho$.
Let $C''= C'\rho$.
We have
\[
C''[N_3\sigma'] =C''[N_3'\rho'\sigma'] 
= C''[N_3'\{(\subst{N_z}{z})_{z \in \fv(N_3')}\}\rho]
= C''[N_1''\rho]
= C'[N_1'']\rho
= N_1'\rho 
= N_1
\]
Hence $N_1$ would be reducible by $N_3 \rightarrow N_4$, which is a contradiction. Therefore, $N_1\theta $ is irreducible. Similarly, $N_2\theta $ is irreducible. Hence $\Sigma \vdash N_1\theta  = N_2\theta  $ implies $N_1\theta  = N_2\theta $. By \substinj{}, $N_1 = N_2$, so a fortiori $\Sigma\vdash N_1 = N_2$.
\end{proof}

\begin{corollary}\label{cor:key-stat-eq}
Suppose $\Sigma$ is equipped with an equational theory generated by  
a convergent rewrite system $\rew$.
Let $\theta$ be a closed substitution that ranges over terms of the form $f(k, M)$ where
each $f$ does not occur on the left-hand side of
the rules of $\rew$.
Let $\sigma$ map the same variables to names $\vect a$ 
such that, 
for all $x, y \in \dom(\theta)$, we have $x\sigma =  y\sigma$ if and only if $\Sigma \vdash x\theta  = y\theta$.
We have $\Res{k}\theta  \enveq \Res{\vect{a}}\sigma$.

\end{corollary}
\begin{proof} We factor $\theta$ and $\sigma$ into 
  $\rho\theta'$ and $\rho\sigma'$ where $\theta'$ and $\sigma'$
  range over pairwise distinct terms modulo $\Sigma$ and $\rho$ is a variable renaming. 
We apply Lemma~\ref{lem:key-stat-eq} and conclude by Lemma~\ref{LEM:INVARIANT-STATIC-EQ}.
\end{proof}

Our next lemma confirms that, with the equations \eqref{eq:pairs-all}, all terms are pairs.

\begin{lemma}\label{lem:pairs-stat-eq}
Suppose $\Sigma$ is equipped with an equational theory that contains the equations~\eqref{eq:pairs-all}. 
We have $\Res{a_1,a_2}\{\subst{(a_1,a_2)}{x}\} \enveq \Res{a}\{\subst{a}{x}\}$.
\end{lemma}
\begin{proof}
Let $M$ and $N$ be two terms such that $\fv(M) \cup \fv(N) \subseteq \{x\}$.
We rename $a, a_1, a_2$ so that $\{a, a_1, a_2 \} \cap (\fn(M) \cup \fn(N)) = \emptyset$.

If $(M = N) \Res{a_1,a_2}\{\subst{(a_1,a_2)}{x}\}$, then
$\Sigma \vdash M\{\subst{(a_1,a_2)}{x}\} = N\{\subst{(a_1,a_2)}{x}\}$.
Since the equational theory is closed under substitution of any term for names,
we have $\Sigma \vdash M\{\subst{(a_1,a_2)}{x}\}\{\subst{\Const{fst}(a)}{a_1}, \subst{\Const{snd}(a)}{a_2}\} = N\{\subst{(a_1,a_2)}{x}\}\{\subst{\Const{fst}(a)}{a_1}, \subst{\Const{snd}(a)}{a_2}\}$, that is,
$\Sigma \vdash M\{\subst{(\Const{fst}(a),\Const{snd}(a))}{x}\} = N\{\subst{(\Const{fst}(a),\Const{snd}(a))}{x}\}$,
so $\Sigma \vdash M\{\subst{a}{x}\} = N\{\subst{a}{x}\}$ by the equation 
$(\Const{fst}(x),\Const{snd}(x)) = x$.
Hence $(M = N) \Res{a}\{\subst{a}{x}\}$.

Conversely, suppose that $(M = N) \Res{a}\{\subst{a}{x}\}$. Hence
$\Sigma \vdash M\{\subst{a}{x}\} = N\{\subst{a}{x}\}$.
Since the equational theory is closed under substitution of any term for names,
we have $\Sigma \vdash M\{\subst{a}{x}\}\{\subst{(a_1,a_2)}{a}\} = N\{\subst{a}{x}\}\{\subst{(a_1,a_2)}{a}\}$, that is,
$\Sigma \vdash M\{\subst{(a_1,a_2)}{x}\} = N\{\subst{(a_1,a_2)}{x}\}$,
so $(M = N) \Res{a_1,a_2}\{\subst{(a_1,a_2)}{x}\}$.

Therefore, $(M = N) \Res{a_1,a_2}\{\subst{(a_1,a_2)}{x}\}$ if and only if 
$(M = N) \Res{a}\{\subst{a}{x}\}$, so $\Res{a_1,a_2}\{\subst{(a_1,a_2)}{x}\} \enveq \Res{a}\{\subst{a}{x}\}$.
\end{proof}

\begin{lemma}\label{lem:ACconv}
The equational theory defined by equations~\eqref{eq:seq}, \eqref{eq:++}, \eqref{eq:pairs-all}, \eqref{eq:h}, \eqref{eq:h2nil}, \eqref{eq:h2cons},
and~\eqref{eq:neseq}
 is generated by a convergent rewrite system $\rew$.
\end{lemma}
\begin{proof}
We define $\rew$ by orienting all equations from left to right, as follows:
{\allowdisplaybreaks\begin{eqnarray}
\Const{hd}(\cons{x}{y}) &\rightarrow& x\label{red:hd2}\\
\Const{tl}(\cons{x}{y}) &\rightarrow& y\label{red:tl2}\\
\Const{nil} \bappend x &\rightarrow& \cons{x}{\Const{nil}}\label{red:nil++2}\\
(\cons{x}{y}) \bappend z &\rightarrow& \cons{x}{(y\bappend z)}\label{red:cons++2}\\
\neseq(\cons{x}{\cons{y}{z}}) &\rightarrow& \neseq(\cons{y}{z})\label{red:neseq-nil}\\
\neseq(\cons{x}{\Const{nil}}) &\rightarrow& \Const{true}\label{red:neseq-cons}\\
  \Const{fst}( (x,y) ) & \rightarrow  & x \label{red:fst2}\\
  \Const{snd}( (x,y) ) & \rightarrow  & y \label{red:snd2}\\
(\Const{fst}(x),\Const{snd}(x)) & \rightarrow  & x \label{red:pair}\\
\Const{h}(k,z) & \rightarrow  &  \Const{h}_2(k,(0,0),z)\label{red:h}\\
\Const{h}_2(k,x,\Const{nil}) & \rightarrow  &  \Const{fst}(x) \label{red:h2nil}\\
\Const{h}_2(k,x,\cons{y}{z}) & \rightarrow  &  \Const{h}_2(k,\Const{f}(k,(x,y)),z) \label{red:h2cons}
\end{eqnarray}}%
To prove that $\rew$ terminates, we order terms $M$ lexicographically, as follows:
\begin{enumerate}
\item by $\hval(M)$, where $\hval$ is defined by
\begin{align*}
\hval(\Const{h}_2(M_1, M_2, M_3)) &= \hval(M_2) + \hval(M_3) + \len(M_3)\\
\hval(\Const{h}(M_1, M_2)) &= \hval(M_2) + \len(M_2) + 1\\
\hval(f(M_1, \ldots, M_n)) &= \hval(M_1) + \dots + \hval(M_n)\\
&\qquad \text{ for all other functions}\\
\hval(M) &= 0\text{ when $M$ is a variable or a name}
\end{align*}
and the $\len$ of a term is defined by
{\allowdisplaybreaks\begin{align*}
&\len(\cons{M}{N}) = 1 + \len(N)\\*
&\quad\text{when the symbol $::$ has sort $\Block \times \BlockList \rightarrow \BlockList$}\\
&\len(M \bappend N) = 1 + \len(M)\\
&\len(f(M_1, \ldots, M_n)) = \max(\len(M_1), \ldots, \len(M_n))\\*
&\quad\text{where $f$ is a function symbol other than}\\*
&\qquad ::\hspace*{3mm}:\;\Block \times \BlockList \rightarrow \BlockList\\*
&\qquad \bappend\;:\;\BlockList \times \Block \rightarrow \BlockList\\*
&\quad\begin{minipage}{10.5cm}
such that the sort of the result of $f$ may contain $\BlockList$, that is,
this sort is $\BlockList$, $\BlockBlockList$, $\BlockBlockListList$,
or one of the sorts of pairs used in the syntactic sugar for
$\Rcv{\ell}{x,t,s}$ and $\Snd{\ell}{x,t,s}$.
\end{minipage}\\
&\len(M) = 0\text{ for all other terms $M$}; 
\end{align*}}%
\item then by the size of $M$; 
\item then by the number of occurrences of the $\bappend$ symbol in $M$; 
\item then by the sum of the lengths of the first arguments of $\bappend$ in $M$.

\end{enumerate}
This ordering is well-founded. By induction on $C$, we show that, for all term contexts $C$, 
\begin{itemize}
\item if $\len(M') \leq \len(M)$, then $\len(C[M']) \leq \len(C[M])$;
\item if $\hval(M') \leq \hval(M)$ and $\len(M') \leq \len(M)$, then $\hval(C[M']) \leq \hval(C[M])$;
\item if $\hval(M') < \hval(M)$ and $\len(M') \leq \len(M)$, then $\hval(C[M']) < \hval(C[M])$.
\end{itemize}
We notice that terms of sorts $\Bool$, $\Block$, $\BlockPair$, and $\BlockT$ have length 0.
For all rewrite rules $M \rightarrow M'$ above and all substitutions $\sigma$, 
we show that $\len(M'\sigma) \leq \len(M\sigma)$ by inspecting each rule.
For all rules except~\eqref{red:h}, \eqref{red:h2nil}, and~\eqref{red:h2cons} and all substitutions $\sigma$, we have $\hval(M'\sigma) \leq \hval(M\sigma)$ because 
\[\hval(M\sigma) = \sum_{x \in \fv(M)} \hval(x\sigma) \times (\text{number of occurrences of $x$ in $M$})\] 
and similarly for $M'$, and all variables $x$ occur at least as many times in $M$ as in $M'$.
We have 
$\hval(\Const{h}(M_1,M_2)) =  \hval(M_2) + \len(M_2) + 1$
and
$\hval(\Const{h}_2(M_1, (0,0), M_2)) = \hval(M_2) + \len(M_2)$,
so rule~\eqref{red:h} decreases $\hval$.
We have
$\hval(\Const{h}_2(M_1, M_2, \Const{nil})) = \hval(M_2)$,
so rule~\eqref{red:h2nil} preserves $\hval$.
We have 
$\hval(\Const{h}_2(M_1, M_2, \cons{M_3}{M_4})) = \hval(M_2) + \hval(M_3) + \hval(M_4) + \len(M_4) +1$
and 
\begin{align*}
&\hval(\Const{h}_2(M_1,\Const{f}(M_1,(M_2,M_3)),M_4)) \\
&\quad = \hval(\Const{f}(M_1,(M_2,M_3))) + \hval(M_4) + \len(M_4) \\
&\quad = \hval(M_1) + \hval(M_2) + \hval(M_3) + \hval(M_4) + \len(M_4) \\
&\quad = \hval(M_2) + \hval(M_3) + \hval(M_4) + \len(M_4)
\end{align*}
since $\hval(M_1) = 0$ because $M_1$ is a variable or a name since no function returns sort $\Key$.
Hence rule~\eqref{red:h2cons} decreases $\hval$.
Therefore, we have:
\begin{itemize}
\item Rules~\eqref{red:hd2}, \eqref{red:tl2}, \eqref{red:neseq-nil}, \eqref{red:neseq-cons}, \eqref{red:fst2}, \eqref{red:snd2}, \eqref{red:pair}, \eqref{red:h2nil} do not increase $\hval$ and decrease the size.
\item Rule~\eqref{red:nil++2} does not increase $\hval$, preserves the size and decreases the number of occurrences of $\bappend$.
\item Rule~\eqref{red:cons++2} does not increase $\hval$, preserves the size and the number of occurrences of $\bappend$, and decreases the sum because the length of the first argument decreases for the occurrence of $\bappend$ modified by rule~\eqref{red:cons++2} ($\len(N) < \len(\cons{M}{N})$) and is unchanged for all other occurrences of $\bappend$ in the term.
\item Rules~\eqref{red:h} and~\eqref{red:h2cons} decrease $\hval$.
\end{itemize}
Therefore, if $M$ reduces to $M'$ by any of these rules, then $M'$ is smaller than $M$ 
in a well-founded lexicographic ordering, and thus $\rew$ terminates.
(The termination of $\rew$ can also be proved using well-known techniques.
For instance, it can be proved using a lexicographic path ordering,
provided the third argument of $\Const{h}_2$, which decreases by~\eqref{red:h2cons}, is considered
before the second one, which increases, in the lexicographic ordering.)

The only critical pairs between these rules are:
\begin{itemize}
\item between rules~\eqref{red:fst2} and~\eqref{red:pair}:
$\Const{fst}((\Const{fst}(x), \Const{snd}(x)))$ reduces to $\Const{fst}(x)$ by both rules, and
$(\Const{fst}((x,\allowbreak y)), \allowbreak \Const{snd}((x,\allowbreak y)))$ reduces to $(x,y)$ by~\eqref{red:pair} or by~\eqref{red:fst2} and~\eqref{red:snd2}, 
so these two critical pairs are joinable.
\item between rules~\eqref{red:snd2} and~\eqref{red:pair}, symmetrically.
\end{itemize} 
Since all critical pairs are joinable, $\rew$ is confluent,
so it is convergent.
\end{proof}

\begin{lemma}\label{lem:eq-g}
Suppose that $\Sigma$ is equipped with the equational theory of Lemma~\ref{lem:ACconv}.
If $\Sigma \vdash \Const f(k, (\dots \Const f(k,\allowbreak ((0,0),M_1)) \ldots, M_n)) =
\Const f(k, (\dots \Const f(k,((0,0),M'_1)) \ldots, M'_{n'}))$, then
$n = n'$ and $\Sigma \vdash M_i = M'_i$ for all $i = 1..n$.
\end{lemma}
\begin{proof}
We proceed by induction on $n$.
\begin{itemize}
\item If $n = n' = 0$, the result holds trivially.
\item If $n = 0$ and $n' > 0$, then 
$\Sigma \vdash (0,0) = \Const f(k, M)$ for some term $M$ and, after reducing under 
$\rew$ of Lemma~\ref{lem:ACconv}, 
$(0,0) = \Const f(k, M')$ for some term $M'$. 
This equality does not hold, so this case is excluded.
By symmetry, the case $n>0$ and $n' = 0$ is also excluded.

\item If $n>0$ and $n'>0$, then 
$\Sigma \vdash \Const f(k,\allowbreak  (\dots \Const f(k,\allowbreak ((0,0),M_1)) \dots, \allowbreak M_n)) = \Const f(k,\allowbreak (\dots \Const f(k,\allowbreak ((0,0),\allowbreak M'_1)) \ldots, \allowbreak M'_{n'}))$
implies 
$\Sigma \vdash \Const f(k, \allowbreak (\dots \Const f(k,\allowbreak ((0,0),\allowbreak M_1)) \dots, \allowbreak M_{n-1})) = \Const f(k, \allowbreak (\dots \Const f(k,\allowbreak ((0,0),\allowbreak M'_1)) \ldots, \allowbreak M'_{n'-1}))$ 
and $\Sigma \vdash M_n=M'_{n'}$. 
By induction hypothesis, $n = n'$ and $\Sigma \vdash M_i = M'_i$
for all $i \leq n-1$.\qedhere
\end{itemize}
\end{proof}

\begin{restate}{Theorem}{\ref{th:indif}}
$\Res{k}{(A^0_h \parop A^0_f)} 
\bicong 
\Res{k}{( A^1_h \parop A^1_f)} $.
\end{restate}%
\begin{proof}
In this proof, we use uppercase letters $X, Y, Z, S, \ldots$ for terms substituted for variables named with the corresponding lowercase letters $x, y, z, s, \ldots$ during execution.
We first extend the notations of Section~\ref{SUBSEC:INDIF}
with intermediate processes parametrized by terms, which we will use
to define our candidate bisimulation.
{\allowdisplaybreaks\begin{align*}
A^0_{\mathit{hi}}(Y) & {} = \IfThen{\neseq(Y)}{\Const{true}}{\Snd{c'_h}{ \Const{h}(k,Y) }} \\
A^1_{\mathit{hi}}(Y) & {} = \IfThen{\neseq(Y)}{\Const{true}}{\Snd{c'_h}{ \Const{h}'(k,Y) }} \\[1ex]
A^1_f(S) & {} = \Res{\ell,c_s}({ \Repl{\Rcv{c_s}{s}. \Rcv{c_f}{x}. \Snd{\ell}{x,s,s}}  \parop \Repl{Q} \parop \Snd{c_s}{S}})
\\
A^1_{\mathit{fi}}(X, S) & {} = \Res{\ell,c_s}({ \Repl{\Rcv{c_s}{s}. \Rcv{c_f}{x}. \Snd{\ell}{x,s,s}}  \parop \Repl{Q} \parop \Snd{\ell}{X,S,S}})
\end{align*}}%
Hence, for hash requests
we have 
$A^0_h \ltr{\Rcv{c_h}{Y}} A^0_h \parop A^0_{\mathit{hi}}$ and
$A^1_h \ltr{\Rcv{c_h}{Y}} A^1_h \parop A^1_{\mathit{hi}}$; 
and for compression requests we have 
$A^1_f(S) \rightarrow \ltr{\Rcv{c_f}{X}} A^1_{\mathit{fi}}(X,S) \rightarrow^* A^1_f(S') \parop 
\Snd{c'_f}{X'}$ for some $S'$ and $X'$ with, initially, $A^1_f = A^1_f(\cons{((0,0),\Const{nil})}{\Const{nil}})$.

Consider traces that interleave
inputs $\Rcv{c_f}{X_i}$ for $i \in \fset$, 
outputs $\Res{x_i}{\Snd{c'_f}{x_i}}$ for $i \in \fsetOut \subseteq \fset$, 
inputs $\Rcv{c_h}{Y_i}$ for $i \in \hset$, and 
outputs $\Res{h_i}{\Snd{c'_h}{h_i}}$ for $i \in \hsetOut \subseteq \hset$,
for some disjoint index sets $\fset$ and $\hset$, such that 
variables $x_j$ or $h_j$ may occur in $X_i$ or $Y_i$ only when $j < i$.
We let $\rel$ be the smallest relation closed by reductions within $A^1_f(S)$ or $A^1_{\mathit{fi}}(X_{i_0}, S)$,
but not in $A^0_{\mathit{his}}$ or $A^1_{\mathit{his}}$, such that 
\begin{align*}
&\Res{k, \vect{x}}{(A^0_h \parop A^0_{\mathit{his}} \parop A^0_f     \parop \sigma^0 \parop O)} 
\;\rel\;
\Res{k, \vect{x}}{( A^1_h \parop A^1_{\mathit{his}} \parop A^1_f(S) \parop \sigma^1 \parop O)} \\
\text{and }&\Res{k, \vect{x}}{(A^0_h \parop A^0_{\mathit{his}} \parop A^0_f  \parop \sigma^0 \parop O)} 
\;\rel\;
\Res{k, \vect{x}}{( A^1_h \parop A^1_{\mathit{his}} \parop A^1_{\mathit{fi}}(X_{i_0}, S) \parop \sigma^1 \parop O')} 
\end{align*}
where the following conditions hold:
\begin{itemize} 
\item 
$J = \hsetOut \uplus \hsetBOut \uplus \hsetDrop \uplus \hsetInt$,
$I = \fsetOut \uplus \fsetBOut = \fsetRel \uplus \fsetNRel$ and, in
the second case of the definition of $\rel$, $i_0 \in \fsetBOut$
is the greatest index in $\hset$ and $\fset$.

Intuitively, 
$\hset$ collects the indices of all hash requests processed so far, partitioned into
$\hsetInt$, for requests before the test $\neseq(Y_i) = \Const{true}$;
$\hsetDrop$, for requests after failing the test;
$\hsetBOut$, for requests after passing the test but before the output; and
$\hsetOut$, for requests after passing the test and performing the output.
And $\fset$ collects the indices of all compression requests processed so far,
partitioned into
$\fsetBOut$, for requests before the output and
$\fsetOut$, for requests after the output;
and also into 
$\fsetRel$, for requests that must be made consistent with the hash function, and
$\fsetNRel$, for unrelated requests;
$i_0$ is the index of the current compression request.

\item $A^0_{\mathit{his}} = \prod_{i \in \hsetInt} A^0_{\mathit{hi}}(Y_i)$ 
and $A^1_{\mathit{his}} = \prod_{i \in \hsetInt} A^1_{\mathit{hi}}(Y_i)$.

These processes represent requests before the test $\neseq(Y_i) = \Const{true}$.

\item 
$\vect{x} = \{ x_i \mid i \in \fsetBOut \} \cup \{ h_i \mid i \in \hsetBOut \}$
are pairwise distinct variables, and the name $k$ and the variables $\vect x$ do not occur in any $(X_i)_{i \in \fset}$ 
or $(Y_i)_{i \in \hset}$.

\item $\neseq(Y_i) = \Const{true}$ for $i \in \hsetOut \cup \hsetBOut$, and $\neseq(Y_i) \neq \Const{true}$ for $i \in \hsetDrop$.

\item 
$O = \prod_{i \in \fsetBOut} \Snd{c'_f}{x_i} \parop \prod_{i \in \hsetBOut} \Snd{c'_h}{h_i}$ and

$O' = \prod_{i \in \fsetBOut \setminus \{ i_0\}} \Snd{c'_f}{x_i} \parop \prod_{i \in \hsetBOut} \Snd{c'_h}{h_i}$.

These parallel compositions represent pending request
outputs, and each output transition consists of removing one message
from $O$ and one restriction on the corresponding variable in
$\vect{x}$.

\item $S$ is (any list representation of) a finite map from pairs of blocks to lists of blocks
that maps $(0,0)$ to $\Const{nil}$ 
and $(\Const{h}'(k, M), \Const{f}_c(k, M))$
to $M$ for some lists $M = \cons{M_1}{\cons{\dots}{\cons{M_n}{\Const{nil}}}}$
with $n > 0$.
The range of $S$ is prefix-closed, that is, if $S$ maps a pair to $M\bappend M'$, 
then it also maps a pair to~$M$. 

The variables $x_i$ and $h_i$ do not occur in $S$.

For every $i \in \fset$, $S$ maps $\Const{fst}(X_i\sigma^1)$ to some
list $M$ if and only if $i \in \fsetRel$; then $S$ also maps
$x_i\sigma^1$ to $M \bappend \Const{snd}(X_i\sigma^1)$, except when
$i = i_0$ in the second case in the definition of $\rel$.

\item $\sigma^0 = \sigma_h^0 \parop \sigma_f^0 \parop \sigma_{\mathit{fo}}^0 $
and $\sigma^1 = \sigma^1_h \parop \sigma^1_f \parop \sigma^1_{\mathit{fo}}$, where

$\sigma^0_h = \{(\subst{\Const h(k,Y_i)}{h_i})_{i \in \hsetOut \cup \hsetBOut} \}$ and 
$\sigma^1_h = \{(\subst{\Const h'(k,Y_i)}{h_i})_{i \in \hsetOut \cup \hsetBOut} \}$,

$\sigma^0_f = \{(\subst{\Const f(k,X_i)}{x_i})_{i \in \fsetRel} \}$ and 
$\sigma^1_f = \{(\subst{(\Const h'(k, Z'_i), \Const f_c(k, Z'_i))}{x_i})_{i \in \fsetRel} \}$ 
where $S$ maps $\Const{fst}(X_i\sigma^1)$ to $Z_i$
and $Z'_i$ is $Z_i \bappend \Const{snd}(X_i)$,

$\sigma^0_{\mathit{fo}} = \{(\subst{\Const f(k,X_i)}{x_i})_{i \in \fsetNRel}\}$ and 
$\sigma^1_{\mathit{fo}} = \{(\subst{\Const f'(k,X_i)}{x_i})_{i \in \fsetNRel}\}$.

\end{itemize}
With $\sigma^1$ defined in the second case of $\rel$, for instance, we have
\[A^1_{\mathit{fi}}(X_{i_0}\sigma^1, S) \rightarrow^* A^1_f(\cons{
(x_{i_0}\sigma^1, M \bappend \Const{snd}(X_{i_0}\sigma^1) 
)}{S}) \parop \Snd{c_f'}{x_{i_0} \sigma^1}\]
when $S$ maps $\Const{fst}(X_{i_0}\sigma^1)$ to $M$,
and $A^1_{\mathit{fi}}(X_{i_0}\sigma^1, S)  \rightarrow^* A^1_f(S) \parop \Snd{c_f'}{x_{i_0} \sigma^1}$ otherwise.
Hence, the second case of the definition of $\rel$ reduces to the first one.
However, the first case is useful for the initial case, and 
the second case is useful after inputs $\Rcv{c_f}{X_i}$.
Taking $S = \cons{((0,0),\Const{nil})}{\Const{nil}}$ and $\fset = \hset = \emptyset$,
the first case yields $\Res{k}(A^0_h \parop A^0_f) \rel \Res{k}(A^1_h \parop A^1_f)$,
so $\rel$ includes our target observational equivalence.
We show that $\rel \cup \rel^{-1}$ is a labelled bisimulation.
\begin{enumerate}
\item We show that, if $A \rel B$, then $A \enveq B$.
To this end, we prove the two properties below by induction on the number of variables in the domain of
$\sigma^0$ and $\sigma^1$.
\begin{enumerate}
\item[P1.] $\Res{k} \sigma^0 \enveq \Res{k} \sigma^1$ and
\item[P2.] if $M = \cons{M_1}{\cons{\dots}{\cons{M_n}{\Const{nil}}}}$
for some $n \geq 1$
contains neither $k$ nor the variable with greatest index in $\dom(\sigma^0) = \dom(\sigma^1)$, 
then for all $i \in \fset$ we have 
$\Sigma \vdash \Const{fst}(x_i \sigma^0) = \Const{h}(k,M\sigma^0)
\Longleftrightarrow
\Sigma \vdash \Const{fst}(x_i \sigma^1) = \Const{h}'(k,M\sigma^1)$.
\end{enumerate}

For all $i \in \fset$, $x_i \sigma^0  = \Const{f}(k, X_i\sigma^0)$ 
and for all $i \in \hsetOut \cup \hsetBOut$, $h_i \sigma^0 = \Const{fst}(\Const{f}(k, M))$ for some term $M$.
Hence, by Corollary~\ref{cor:key-stat-eq}, 
\[\Res{k} \sigma^0 \enveq \Res{\vect a}\{(\subst{x_i \sigma_0}{x_i})_{i \in \fset}, (\subst{\Const{fst}(h_i \sigma_0)}{h_i})_{i \in \hsetOut \cup \hsetBOut}\}\]
where the following conditions hold: 
\begin{itemize}
\item $x_i \sigma_0$ for $i \in \fset$ and $h_i \sigma_0$ for $i \in \hsetOut \cup \hsetBOut$ are names in $\vect a$.
\item For all $i, j \in \fset$, $x_i \sigma_0 = x_j \sigma_0$ if and only if $\Sigma \vdash \Const{f}(k, X_i\sigma^0 ) = \Const{f}(k, X_j\sigma^0)$, that is,
$\Sigma \vdash X_i\sigma^0 = X_j\sigma^0$.
\item For all $i, j \in \hsetOut \cup \hsetBOut$, $h_i \sigma_0 = h_j \sigma_0$ if and only if $\Sigma \vdash \Const f(k, \allowbreak (\dots \Const f(k,\allowbreak ((0,0),\allowbreak M_1)) \ldots, \allowbreak M_n)) =
\Const f(k, (\dots \Const f(k,((0,0),M'_1)) \ldots, M'_{n'}))$ where $Y_i\sigma^0 = \cons{M_1}{\cons{\dots}{ \cons{M_n}{\Const{nil}}}}$
and $Y_j\sigma^0 = \cons{M'_1}{\cons{\dots}{ \cons{M'_n}{\Const{nil}}}}$, that is,
$\Sigma \vdash Y_i\sigma^0 = Y_j\sigma^0$, by Lemma~\ref{lem:eq-g}.
\item For all $i\in \fset$ and $j \in \hsetOut \cup \hsetBOut$, $x_i \sigma_0 = h_j \sigma_0$ if and only if $\Sigma \vdash \Const{f}(k, X_i\sigma^0) = \Const f(k, \allowbreak (\dots \Const f(k,\allowbreak ((0,0),\allowbreak M_1)) \ldots, \allowbreak M_n))$
where $Y_j\sigma^0 = \cons{M_1}{\cons{\dots}{ \cons{M_n}{\Const{nil}}}}$. In this case, we have
$\Sigma \vdash \Const{fst}(x_i \sigma^0) = \Const{fst}(\Const{f}(k, X_i\sigma^0)) = \Const{fst}(\Const f(k, \allowbreak (\dots \Const f(k,\allowbreak ((0,0),\allowbreak M_1)) \ldots, \allowbreak M_n)))
= \Const{h}(k,Y_j\sigma^0)$. Conversely, if $\Sigma \vdash \Const{fst}(x_i \sigma^0) = \Const{h}(k,Y_j\sigma^0)$, then
$\Sigma \vdash \Const{fst}(\Const{f}(k, X_i\sigma^0)) = \Const{fst}(\Const f(k, \allowbreak (\dots \Const f(k,\allowbreak ((0,0),\allowbreak M_1)) \ldots, \allowbreak M_n)))$.
Since these terms do not reduce at the root under the rewrite system $\rew$ of Lemma~\ref{lem:ACconv},
we have $\Sigma \vdash \Const{f}(k, X_i\sigma^0) = \Const f(k, \allowbreak (\dots \Const f(k,\allowbreak ((0,0),\allowbreak M_1)) \ldots, \allowbreak M_n))$.
Therefore, $x_i \sigma_0 = h_j \sigma_0$ if and only if $\Sigma \vdash \Const{fst}(x_i \sigma^0) = \Const{h}(k,Y_j\sigma^0)$.
\end{itemize}
By Lemma~\ref{lem:pairs-stat-eq}, we can replace the names $x_i \sigma_0$ and $h_i \sigma_0$ with pairs
$(x_i \sigma_1, \allowbreak x_i \sigma_2)$ and $(h_i \sigma_1, h_i \sigma_2)$ respectively.
Thus \[\Res{k} \sigma^0 \enveq \Res{\vect a'}\{(\subst{(x_i \sigma_1, x_i \sigma_2)}{x_i})_{i \in \fset}, \allowbreak (\subst{h_i \sigma_1}{h_i})_{i \in \hsetOut \cup \hsetBOut}\}\]
where the following conditions hold:
\begin{itemize}
\item $x_i \sigma_1, x_i \sigma_2$ for $i \in \fset$ and $h_i \sigma_1$ for $i \in \hsetOut \cup \hsetBOut$ are names in $\vect a'$.
\item For all $i, j \in \fset$, $x_i \sigma_1 \neq  x_j \sigma_2$.
\item For all $i\in \fset$ and $j \in \hsetOut \cup \hsetBOut$, $x_i \sigma_2 \neq  h_j \sigma_1$.
\item For all $i, j \in \fset$, $x_i \sigma_1 = x_j \sigma_1 \Longleftrightarrow x_i \sigma_2 = x_j \sigma_2 \Longleftrightarrow \Sigma \vdash X_i\sigma^0 = X_j\sigma^0$.
\item For all $i, j \in \hsetOut \cup \hsetBOut$, $h_i \sigma_1 = h_j \sigma_1 \Longleftrightarrow \Sigma \vdash Y_i\sigma^0 = Y_j\sigma^0$.
\item For all $i\in \fset$ and $j \in \hsetOut \cup \hsetBOut$, $x_i \sigma_1 = h_j \sigma_1
\Longleftrightarrow \Sigma \vdash \Const{fst}(x_i \sigma^0) = \Const{h}(k,Y_j\sigma^0)$.
\end{itemize}

For all $i \in \fsetNRel$, $x_i \sigma^1  = \Const{f}'(k, X_i\sigma^1)$,
for all $i \in \fsetRel$, $x_i \sigma^1  = (\Const{h}'(k, Z'_i\sigma^1), \allowbreak \Const{f}_c(k, Z'_i\sigma^1))$,
and for all $i \in \hsetOut \cup \hsetBOut$, $h_i \sigma^1 = \Const{h}'(k, Y_i\sigma^1)$.
Hence, by Corollary~\ref{cor:key-stat-eq}, 
\[\Res{k} \sigma^1 \enveq \Res{\vect a}\{
(\subst{x_i \sigma_3}{x_i})_{i \in \fsetNRel}, 
(\subst{(x_i \sigma_4, x_i \sigma_5)}{x_i})_{i \in \fsetRel}, 
(\subst{h_i \sigma_4}{h_i})_{i \in \hsetOut \cup \hsetBOut}\}\]
where the following conditions hold:
\begin{itemize}
\item $x_i \sigma_3$ for $i \in \fsetNRel$, $x_i \sigma_4$ and $x_i \sigma_5$ for $i \in \fsetRel$, and $h_i \sigma_4$ for $i \in \hsetOut \cup \hsetBOut$ are names in $\vect a$.
\item For all $i, j \in \fsetRel$, $x_i \sigma_4 \neq x_j \sigma_5$.
\item For all $i \in \fsetNRel$ and $j \in \fsetRel$, $x_i \sigma_3 \neq x_j \sigma_4$ and $x_i \sigma_3 \neq x_j \sigma_5$.
\item For all $i \in \fsetNRel$ and $j \in \hsetOut \cup \hsetBOut$, $x_i \sigma_3 \neq h_j \sigma_4$.
\item For all $i \in \fsetRel$ and $j \in \hsetOut \cup \hsetBOut$, $x_i \sigma_5 \neq h_j \sigma_4$.
\item For all $i, j \in \fsetNRel$, $x_i \sigma_3 = x_j \sigma_3\Longleftrightarrow \Sigma \vdash \Const{f}'(k, X_i\sigma^1) = \Const{f}'(k, X_j\sigma^1) \Longleftrightarrow 
\Sigma \vdash X_i\sigma^1 = X_j\sigma^1$.
\item For all $i, j \in \hsetOut \cup \hsetBOut$, $h_i \sigma_4 = h_j \sigma_4\Longleftrightarrow \Sigma \vdash Y_i\sigma^1 = Y_j\sigma^1$.
\item For all $i \in \fsetRel$ and $j \in \hsetOut \cup \hsetBOut$, $x_i \sigma_4 = h_j \sigma_4\Longleftrightarrow \Sigma \vdash \Const{fst}(x_i\sigma^1) = \Const{h}'(k, Y_j\sigma^1)$.
\item For all $i, j \in \fsetRel$, $x_i \sigma_4 = x_j \sigma_4\Longleftrightarrow x_i \sigma_5 = x_j \sigma_5\Longleftrightarrow \Sigma \vdash Z'_i\sigma^1 = Z'_j\sigma^1$.
In this case, by definition of $Z'_i$, $\Sigma \vdash Z_i\sigma^1 = Z_j\sigma^1$ and $\Sigma \vdash \Const{snd}(X_i \sigma^1) = \Const{snd}(X_j \sigma^1)$.
Since $S$ maps $\Const{fst}(X_i \sigma^1)$ to $Z_i$
and $\Const{fst}(X_j \sigma^1)$ to $Z_j$,
we have  $Z_i = Z_i\sigma^1$ and $Z_j = Z_j\sigma^1$, so
either $\Sigma \vdash Z_i = Z_j = \Const{nil}$ and $\Sigma \vdash \Const{fst}(X_i \sigma^1) = (0,0) = \Const{fst}(X_j \sigma^1)$
or $\Sigma \vdash Z_i = Z_j \neq \Const{nil}$ and $\Sigma \vdash \Const{fst}(X_i \sigma^1) = (\Const{h'}(k, Z_i), \Const{f}_c(k, Z_i)) = (\Const{h'}(k, Z_j), \Const{f}_c(k, Z_j)) = \Const{fst}(X_j \sigma^1)$.
So in both cases, $\Sigma \vdash X_i \sigma^1 = X_j \sigma^1$. Conversely, if $\Sigma \vdash X_i \sigma^1 = X_j \sigma^1$, then $\Sigma \vdash Z'_i\sigma^1 = Z'_j\sigma^1$,
since $Z'_i$ is computed from $X_i$.
Therefore, for all $i, j \in \fsetRel$, $x_i \sigma_4 = x_j \sigma_4\Longleftrightarrow x_i \sigma_5 = x_j \sigma_5\Longleftrightarrow \Sigma \vdash X_i \sigma^1 = X_j \sigma^1$.
\end{itemize}
By Lemma~\ref{lem:pairs-stat-eq}, we can replace the names $x_i \sigma_3$ for $i \in \fsetNRel$ with pairs
$(x_i \sigma_4, \allowbreak x_i \sigma_5)$.
Thus $\Res{k} \sigma^1 \enveq \Res{\vect a'}\{
(\subst{(x_i \sigma_4, x_i \sigma_5)}{x_i})_{i \in \fset}, 
(\subst{h_i \sigma_4}{h_i})_{i \in \hsetOut \cup \hsetBOut}\}$
where the following conditions hold:
\begin{itemize}
\item $x_i \sigma_4$ and $x_i \sigma_5$ for $i \in \fset$, and $h_i \sigma_4$ for $i \in \hsetOut \cup \hsetBOut$ are names in $\vect a'$.
\item For all $i, j \in \fset$, $x_i \sigma_4 \neq x_j \sigma_5$.
\item For all $i \in \fset$ and $j \in \hsetOut \cup \hsetBOut$, $x_i \sigma_5 \neq h_j \sigma_4$.
\item For all $i, j \in \fset$, $x_i \sigma_4 = x_j \sigma_4\Longleftrightarrow x_i \sigma_5 = x_j \sigma_5\Longleftrightarrow \Sigma \vdash X_i \sigma^1 = X_j \sigma^1$.
Indeed, when $i\in\fsetRel$ and $j\in\fsetNRel$, we have $x_i \sigma_4 \neq x_j \sigma_4$, $x_i \sigma_5 \neq x_j \sigma_5$,
and $\Sigma \vdash X_i \sigma^1 \neq X_j \sigma^1$ since $\Const{fst}(X_i\sigma^1)$ and $\Const{fst}(X_j\sigma^1)$ are not mapped to the same value by $S$.
When $i$ and $j$ are both in $\fsetRel$ or both in $\fsetNRel$, the result comes from the equivalences shown above.
\item For all $i, j \in \hsetOut \cup \hsetBOut$, $h_i \sigma_4 = h_j \sigma_4\Longleftrightarrow \Sigma \vdash Y_i\sigma^1 = Y_j\sigma^1$.
\item For all $i \in \fset$ and $j \in \hsetOut \cup \hsetBOut$, $x_i \sigma_4 = h_j \sigma_4\Longleftrightarrow \Sigma \vdash \Const{fst}(x_i\sigma^1) = \Const{h}'(k, Y_j\sigma^1)$.
Indeed, if $i \in \fsetNRel$, we have $x_i \sigma_4 \neq h_j \sigma_4$ and $\Sigma \vdash \Const{fst}(x_i\sigma^1) \neq \Const{h}'(k, Y_j\sigma^1)$.
When $i \in \fsetRel$, the result comes from an equivalence shown above.
\end{itemize}
Since $X_i$ and $Y_i$ contain variables $x_j$ and $h_j$ only with $j < i$, 
the variable of $\dom(\sigma^0)$ with the greatest
index does not occur in $X_i$ and $Y_i$, so by induction hypothesis, we have 
$\Sigma \vdash X_i \sigma^0 = X_j \sigma^0\Longleftrightarrow \Sigma \vdash X_i \sigma^1 = X_j \sigma^1$
and 
$\Sigma \vdash Y_i \sigma^0 = Y_j \sigma^0\Longleftrightarrow \Sigma \vdash Y_i \sigma^1 = Y_j \sigma^1$,
so it suffices to show property P2 to obtain $\Res{k} \sigma^0 \enveq \Res{k} \sigma^1$.

Let us now show property P2, by induction on $n$. 
Let $M = \cons{M_1}{\cons{\dots}{\cons{M_n}{\Const{nil}}}}$
be a term that does not contain $k$ nor the variable with greatest index in $\dom(\sigma^0) = \dom(\sigma^1)$, $n \geq 1$, and $i \in \fset$.

Suppose that $\Sigma \vdash \Const{fst}(x_i \sigma^0) = \Const{h}(k,M\sigma^0)$.
We have 
\begin{align*}
\Const{fst}(x_i \sigma^0) &= \Const{fst}(\Const{f}(k, X_i\sigma^0))\\
\Const{h}(k,M\sigma^0) &= \Const{fst}( \Const f(k,(\dots \Const f(k,((0,0),M_1))\dots,M_n)))\sigma^0
\end{align*}
so
\[\Sigma \vdash X_i\sigma^0 = (\dots \Const f(k,((0,0),M_1))\dots,M_n)\sigma^0\]
If $n = 1$, we have $\Sigma \vdash X_i\sigma^0 = ((0,0), M_n\sigma^0)$.
Since $X_i$ and $M_n$ do not contain the variable of $\dom(\sigma^0)$ with the greatest
index, we have $\Sigma \vdash X_i\sigma^1 = ((0,0), M_n\sigma^1)$ by induction hypothesis,
so $S$ maps $\Const{fst}(X_i\sigma^1) = (0,0)$ to $Z_i = \Const{nil}$,
hence $Z'_i = Z_i \bappend \Const{snd}(X_i)$, so
$\Sigma \vdash Z'_i \sigma^1 = \Const{nil} \bappend M_n\sigma^1 = \cons{M_n\sigma^1}{\Const{nil}} = M\sigma^1$
and $\Sigma \vdash \Const{fst}(x_i\sigma^1) = \Const{h}'(k, Z'_i\sigma^1) = \Const{h}'(k, M\sigma^1)$.

If $n> 1$, let $M' = \cons{M_1}{\cons{\dots}{\cons{M_{n-1}}{\Const{nil}}}}$ and $H = \Const f(k,(\dots \Const f(k,((0,0),\allowbreak M_1))\allowbreak \dots,M_{n-1}))$.
We have $\Sigma \vdash X_i\sigma^0 = (H, M_n)\sigma^0$.
Since $k$ does not occur in $X_i$ and $X_i \sigma^0$ is of the form $(H\sigma^0,M_n\sigma^0) = (\Const{f}(k, \allowbreak \cdot), \cdot)$, 
there exists $j_0\in \fset$ such that $\Sigma \vdash H\sigma^0 = x_{j_0}\sigma^0$ and $x_{j_0}$ occurs in $X_i$, so
$j_0 < i$.
Thus $\Sigma \vdash\Const{fst}(x_{j_0}\sigma^0) = \Const{fst}(H\sigma^0) = \Const{h}(k, M' \sigma^0)$,
so by induction hypothesis, 
\[\Sigma \vdash\Const{fst}(x_{j_0}\sigma^1) = \Const{h}'(k, M' \sigma^1)\]
By construction of $\sigma^1$, we have 
\[\Sigma \vdash\Const{snd}(x_{j_0}\sigma^1) = \Const{f}_c(k, M' \sigma^1)\]
Moreover, we have $\Sigma \vdash X_i\sigma^0 = (x_{j_0}, M_n) \sigma^0$
and $X_i$, $x_{j_0}$, and $M_n$ do not contain the variable of $\dom(\sigma^0)$ with
the greatest index, so we have $\Sigma \vdash X_i\sigma^1 = (x_{j_0}, M_n) \sigma^1$ by induction hypothesis.
Hence $\Sigma \vdash \Const{fst}(X_i \sigma^1) = x_{j_0}\sigma^1 = (\Const{h}'(k, M' \sigma^1), \Const{f}_c(k, M' \sigma^1))$.
Hence $S$ maps $\Const{fst}(X_i \sigma^1)$ to $Z_i$ such that $\Sigma \vdash Z_i = M' \sigma^1$,
so $\Sigma \vdash Z'_i \sigma^1 = Z_i \bappend \Const{snd}(X_i)\sigma^1 = 
M\sigma^1$,
so $\Sigma \vdash \Const{fst}(x_i\sigma^1) = \Const{h}'(k, Z'_i\sigma^1) = \Const{h}'(k, M\sigma^1)$.

Conversely, suppose that $\Sigma \vdash \Const{fst}(x_i \sigma^1) = \Const{h}'(k,M\sigma^1)$. Thus $i \in \fsetRel$ and $\Sigma \vdash Z'_i\sigma^1 = M\sigma^1$.
So $S$ maps $\Const{fst}(X_i\sigma^1)$ to some $Z_i$.

If $Z_i = \Const{nil}$, then $\Sigma \vdash \Const{fst}(X_i\sigma^1) = (0,0)$.
Since $X_i$ does not contain the variable of $\dom(\sigma^0)$ with the greatest
index, we have $\Sigma \vdash \Const{fst}(X_i\sigma^0) = (0,0)$ by induction hypothesis.
Moreover $Z'_i = Z_i \bappend \Const{snd}(X_i)$, so
$\Sigma \vdash M\sigma^1 = Z'_i\sigma^1 = \cons{\Const{snd}(X_i\sigma^1)}{\Const{nil}}$.
Since $M$ and $X_i$ do not contain the variable of $\dom(\sigma^0)$ with the greatest
index, we have $\Sigma \vdash M\sigma^0 = \cons{\Const{snd}(X_i\sigma^0)}{\Const{nil}}$ by induction hypothesis.
We obtain
$\Sigma \vdash \Const{h}(k, M\sigma^0) = \Const{fst}(\Const{f}(k, ((0,0), \allowbreak \Const{snd}(X_i\sigma^0))) = \Const{fst}(\Const{f}(k, X_i\sigma^0)) = \Const{fst}(x_i\sigma^0)$.

If $Z_i \neq \Const{nil}$, then  $\Sigma \vdash \Const{fst}(X_i\sigma^1) = (\Const{h}'(k, Z_i), \Const{f}_c(k, Z_i))$ and $Z'_i = Z_i \bappend \Const{snd}(X_i)$.
Since  $\Sigma \vdash Z'_i\sigma^1 = M\sigma^1$, we have
\begin{align*}
&\Sigma \vdash Z_i = (\cons{M_1}{\cons{\dots}{\cons{M_{n-1}}{\Const{nil}}}})\sigma^1\\
&\Sigma \vdash \Const{snd}(X_i\sigma^1) = M_n\sigma^1
\end{align*}
Since $k$ does not occur in $X_i$, there exists $j_0 \in \fset$ such that
$\Sigma \vdash x_{j_0}\sigma^1 = (\Const{h}'(k, Z_i), \Const{f}_c(k, Z_i)) = \Const{fst}(X_i\sigma^1)$.
Hence
\[\Sigma \vdash \Const{fst}(x_{j_0}\sigma^1) = \Const{h}'(k, Z_i) = \Const{h}'(k, (\cons{M_1}{\cons{\dots}{\cons{M_{n-1}}{\Const{nil}}}})\sigma^1)\]
By induction hypothesis, 
\begin{align*}
\Sigma \vdash \Const{fst}(x_{j_0}\sigma^0) &= \Const{h}(k, (\cons{M_1}{\cons{\dots}{\cons{M_{n-1}}{\Const{nil}}}})\sigma^0)\\
& = \Const{fst}(\Const f(k, (\dots \Const f(k,((0,0),M_1)) \ldots, M_{n-1}))\sigma^0)
\end{align*}
Since $x_{j_0}\sigma^0$ is of the form $\Const{f}(k, \cdot)$, we obtain 
\[\Sigma \vdash x_{j_0}\sigma^0 = \Const f(k, (\dots \Const f(k,((0,0),M_1)) \ldots, M_{n-1}))\sigma^0\]
We have $\Sigma \vdash \Const{fst}(X_i\sigma^1) = x_{j_0}\sigma^1$
and $\Sigma \vdash \Const{snd}(X_i\sigma^1) = M_n\sigma^1$,
so $\Sigma \vdash X_i\sigma^1 = (x_{j_0}\sigma^1, M_n\sigma^1)$
Since $X_i$, $x_{j_0}$, and $M_n$ do not contain the variable of $\dom(\sigma^0)$ with the greatest
index, we have $\Sigma \vdash X_i\sigma^0 = (x_{j_0}\sigma^0, M_n\sigma^0)$ by induction hypothesis.
So $\Sigma \vdash \Const{fst}(x_i\sigma^0) = \Const{fst}(\Const{f}(k, X_i\sigma^0)) = \Const{fst}(\Const{f}(k, \allowbreak (x_{j_0}\sigma^0, \allowbreak M_n\sigma^0))) = \Const{h}(k, M\sigma^0)$.

\item We first show that, if $A \rel B$, $A \ltr{\alpha} A'$, $A'$ is closed, and 
    $\fv(\alpha) \subseteq \dom(A)$,
    then $B \rightarrow^*
    \ltr{\alpha} \rightarrow^* B'$ and $A' \rel B'$ for some $B'$.
The only possible labelled transitions in $A$ are as follows:
\begin{itemize}
\item $A^0_h$ performs an input with label $\alpha = \Rcv{c_h}{Y_i}$,
  with $\fv(Y_i) \subseteq \dom(\sigma^0) \setminus \{\vect x\}$,
creating a process $A^0_{\mathit{hi}}(Y_i)$. 
The process $A^1_h$ can perform
the same input, creating a process
$A^1_{\mathit{hi}}(Y_i)$.
The resulting extended processes are still in $\rel$, by adding to
$\hsetInt$ an index $i$ greater than those already in $\fset$ and $\hset$.

\item $A^0_f$ performs an input with label $\alpha = \Rcv{c_f}{X_i}$,
  with $\fv(X_i) \subseteq \dom(\sigma^0) \setminus \{\vect x\}$,
creating a process $\Snd{c'_f}{ \Const{f}(k,X_i) } \linebreak[2]\equiv \Res{x_i}(\Snd{c'_f}{x_i} \parop \{\subst{\Const{f}(k,X_i)}{x_i}\})$.
A reduced form of $A^1_f(S)$ or $A^1_{\mathit{fi}}(X_{i_0}, S)$ can perform the
same input (possibly after internal reductions). We keep performing
internal reductions after the input, until the output on $c'_f$ is enabled.
Hence a new process 
\[\Snd{c'_f}{ \Const{f}'(k,X_i) }  \equiv \Res{x_i}(\Snd{c'_f}{x_i} \parop \{\subst{\Const{f}'(k,X_i)}{x_i}\})\]
or 
\[\Snd{c'_f}{ (\Const h'(k, Z'_i), \Const f_c(k, Z'_i)) } \equiv 
\Res{x_i}(\Snd{c'_f}{x_i} \parop \{ \subst{(\Const h'(k, Z'_i), \Const f_c(k, Z'_i))}{x_i}\})\]
appears, 
depending on whether $S$ maps $\Const{fst}(X_i)$ to some $Z_i$ or not, and for
$Z'_i$ given in the definition of $\rel$.
The resulting extended processes are still in $\rel$, 
by adding to $\fsetBOut$ an index greater than those already in $\fset$ and $\hset$.

\item $O$ performs an output with label $\alpha = \Res{h_i}{\Snd{c'_h}{h_i}}$.
(We arrange that the bound variable of $\alpha$ has the same name
as the variable used internally by the output that we perform.)
In this case, $h_i$ is removed from $\vect{x}$ and $\Snd{c'_h}{h_i}$ is removed from $O$.
The process $O$ can perform the same output on the right-hand side, hence
we remain in $\rel$ by moving the index $i$ from $\hsetBOut$ to $\hsetOut$.

\item $O$ performs an output with label $\alpha = \Res{x_i}{\Snd{c'_f}{x_i}}$.
In this case, $x_i$ is removed from $\vect{x}$ and $\Snd{c'_f}{x_i}$ is removed from $O$.
If we are in the second case of the definition of $\rel$ with $i = i_0$, 
we first reduce $A^1_{\mathit{fi}}(X_{i_0}, S)$ until we arrive at the first case of the
definition of $\rel$.
The process $O$ can then perform the same output on the right-hand side, hence
we remain in $\rel$ by moving the index $i$ from $\fsetBOut$ to $\fsetOut$.
\end{itemize}
A detailed proof that these are the only possible labelled transitions of~$A$
uses the partial normal forms introduced in Appendix~\ref{app:bigpfpnf}
and the decomposition lemmas proved in Appendix~\ref{app:comppnf}.
This comment also applies to other case distinctions below in the proof of Theorem~\ref{th:indif}.

Conversely, we show that, if $A \rel B$, $B \ltr{\alpha} B'$, $B'$ is closed, and 
    $\fv(\alpha) \subseteq \dom(B)$,
    then $A \rightarrow^*
    \ltr{\alpha} \rightarrow^* A'$ and $A' \rel B'$ for some $A'$.
The only possible labelled transitions in $B$ are as follows:
\begin{itemize}
\item $A^1_h$ performs an input with label $\alpha = \Rcv{c_h}{Y_i}$, with $\fv(Y_i) \subseteq \dom(\sigma^1) \setminus \{\vect x\}$.
The process $A^1_h$ can perform the same input, and we remain 
in $\rel$ by adding to $\hsetInt$ an index $i$ greater than those already in $\fset$ and~$\hset$.

\item (A reduced form of) $A^1_f(S)$ performs an input with label $\alpha = \Rcv{c_f}{X_i}$, with $\fv(Y_i) \subseteq \dom(\sigma^1) \setminus \{\vect x\}$.
After the input, $A^1_f(S)$ is transformed into the process $A^1_{\mathit{fi}}(X_i, S)$.
The process $A^0_f$ can perform the same input.
A new process $\Snd{c'_f}{ \Const{f}(k,X_i) } \equiv \Res{x_i}(\Snd{c'_f}{x_i} \parop \{\subst{\Const{f}(k,X_i)}{x_i}\})$
appears on left-hand side.
We remain in $\rel$ by adding to $\fsetBOut$
an index $i_0$ greater than those already in $\fset \cup \hset$.
(On the right-hand side, the variable $x_{i_0}$ is defined but not used.)

When a reduced form of $A^1_{\mathit{fi}}(X_{i_0}, S)$ performs an input with label $\alpha = \Rcv{c_f}{X_i}$,
$A^1_{\mathit{fi}}(X_{i_0}, S)$ has first been reduced so that the configuration is in
the first case of the definition of $\rel$, with the considered input in $A^1_f(S)$,
so this case is already treated above.

\item $O$ performs an output with label $\alpha = \Res{h_i}{\Snd{c'_h}{h_i}}$.
The process $O$ performs the same output on the left-hand side and 
we remain in $\rel$ by moving the index $i$ from $\hsetBOut$ to $\hsetOut$.

\item $O$ performs an output with label $\alpha = \Res{x_i}{\Snd{c'_f}{x_i}}$.
The process $O$ performs the same output on the left-hand side, and 
we remain in $\rel$, by moving the index $i$ from $\fsetBOut$ to $\fsetOut$.

When a reduced form of $A^1_{\mathit{fi}}(X_{i_0}, S)$ performs an output with label $\alpha = \Res{x_i}{\Snd{c'_f}{x_i}}$,
$A^1_{\mathit{fi}}(X_{i_0}, S)$ has first been reduced so that the configuration is in
the first case of the definition of $\rel$, with the considered output included in $O$,
so this case is already treated above.
\end{itemize}

\item We first show that, if $A \rel B$, $A \rightarrow A'$, and $A'$ is closed, then $B
    \rightarrow^* B'$ and $A' \rel B'$ for some~$B'$.

The only processes that can be reduced in $A$ are processes $A^0_{\mathit{hi}}(Y_i)$ inside $A^0_{\mathit{his}}$.
If $\neseq(Y_i) = \Const{true}$, then $A^0_{\mathit{hi}}(Y_i)$ reduces to
\[\Snd{c'_h}{ \Const{h}(k,Y_i) } \equiv \Res{h_i}(\Snd{c'_h}{ h_i } \parop \{ \subst{\Const{h}(k,Y_i)}{h_i} \})\]
and similarly $A^1_{\mathit{hi}}(Y_i)$ reduces to
\[\Snd{c'_h}{ \Const{h}'(k,Y_i) } \equiv \Res{h_i}(\Snd{c'_h}{ h_i } \parop \{ \subst{\Const{h}'(k,Y_i)}{h_i} \})\]
and we remain in $\rel$ by moving $i$ from $\hsetInt$ to $\hsetBOut$.
(The value of $\neseq(Y_i)$ remains unchanged when we instantiate $Y_i$ with
$\sigma^0$ or $\sigma^1$ because the image of these substitutions does not contain lists.)
If $\neseq(Y_i) \neq \Const{true}$, then $A^0_{\mathit{hi}}(Y_i)$ reduces to $\nil$
and similarly $A^1_{\mathit{hi}}(Y_i)$ reduces to $\nil$,
and we remain in $\rel$ by moving $i$ from $\hsetInt$ to $\hsetDrop$.

Conversely, we show that, if $A \rel B$, $B \rightarrow B'$, and $B'$ is closed, then $A
    \rightarrow^* A'$ and $A' \rel B'$ for some $A'$.

The only reductions in $B$ are due to processes $A^1_{\mathit{hi}}(Y_i)$ within $A^1_{\mathit{his}}$, 
and $A^1_f(S)$ or $A^1_{\mathit{fi}}(X_{i_0}, S)$. The first case can be handled similarly to the case
in which $A^0_{\mathit{hi}}(Y_i)$ reduces. In the second case, we remain in $\rel$ with $A' = A$.

\end{enumerate}
Therefore, ${\rel} \subseteq {\wkbisim}$. By Theorem~\ref{THM:OBSERVATIONAL-LABELED}, ${\rel} \subseteq {\bicong}$.
So $\Res{k}{(A^0_h \parop A^0_f)} 
\bicong 
\Res{k}{( A^1_h \parop A^1_f)} $.
\end{proof}

\section*{Acknowledgments}

We thank Rocco De Nicola, Andy Gordon, Tony Hoare, and Phil Rogaway
for discussions that contributed to this work. Georges Gon\-thier and
Jan J\"urjens suggested improvements to the presentation of the
conference version of this paper. 
Steve Kremer and Ben Smyth provided helpful comments on a draft of this paper.


\newcommand{\noopsort}[1]{}

\end{document}